\documentclass[reqno]{amsart}
\usepackage{amssymb}
\usepackage{amsfonts}

\usepackage{mathrsfs}
\usepackage{caption}
\usepackage{mathtools}
\usepackage{enumitem}
\usepackage{pifont}
\usepackage[numbers,sort&compress]{natbib}
\usepackage{multirow}
\makeatletter
\@addtoreset{equation}{section}
\makeatother
\renewcommand{\theequation}{\arabic{section}.\arabic{equation}}

\DeclareMathOperator{\re}{Re}

\DeclareMathOperator{\sgn}{sgn}
\DeclareMathOperator{\spa}{span}

\newtheorem{theorem}{Theorem}[section]

\newtheorem{corollary}[theorem]{Corollary}
\newtheorem{proposition}[theorem]{Proposition}
\newtheorem{lemma}[theorem]{Lemma}
\newtheorem{definition}[theorem]{Definition}
\newtheorem{remark}[theorem]{Remark}

\usepackage{hyperref}
\allowdisplaybreaks[4]

\usepackage{amsthm}

\setlist[enumerate]{leftmargin=*} 

\begin{document}

\title[3D Fermi gas in the Gross-Pitaevskii regime]{The second order Huang-Yang formula to the 3D Fermi gas: the Gross-Pitaevskii regime}
\date{\today}
\author[X. Chen]{Xuwen Chen}
\address{Department of Mathematics, University of Rochester, Rochester, NY 14627, USA}

\email{xuwenmath@gmail.com}

\author[J. Wu]{Jiahao Wu}
\address{School of Mathematical Sciences, Peking University, Beijing, 100871, China}
\email{wjh00043@gmail.com}

\author[Z. Zhang]{Zhifei Zhang}
\address{School of Mathematical Sciences, Peking University, Beijing, 100871, China}

\email{zfzhang@math.pku.edu.cn}

\begin{abstract}
   For a system of $N$ Fermions of spin $1/2$, with its interaction potential of scattering length $a$, the classical Huang-Yang formula \cite{huangyang} states that the energy density $e(\rho)$ is of the form
  \begin{equation*}
    e(\rho)=\frac{3}{5}(3\pi^2)^{\frac{2}{3}}\rho^{\frac{5}{3}}+2\pi a\rho^2
    +\frac{12}{35}(11-2\ln2)3^{\frac{1}{3}}\pi^{\frac{2}{3}}a^2\rho^{\frac{7}{3}}
    +o(\rho^{\frac{7}{3}}).
  \end{equation*}
  We consider a general system of $N$ Fermions of spin $1/\mathbf{q}$ with the scattering length of the interaction potential at the scale $a\propto N^{-1}$, that is, in the Gross-Pitaevskii regime. We prove the 2nd order ground state energy approximation corresponding to the Huang-Yang formula. The thermodynamic limit case shares a similar logic, and could be dealt with in a separate paper.
\end{abstract}

\maketitle
\tableofcontents

\section{Introduction}\label{intro}
\par During 1950s, the formulations of the quantum-mechanical many-body problem for Boson or Fermion systems, and their physical interpretations, were famous problems among physicists. K.Huang, T.D.Lee and C.N.Yang investigated this problem, and have reached two famous formulae. For a system of $N$ Fermions of spin $1/2$, with half particles spin up and the other half spin down, and a two-body interaction potential of scattering length $a_0$, inside a box with edge length $L$, the classical Huang-Yang formula \cite{huangyang} predicts up to the second order that, the energy density $e(\rho)$, which is defined by the ground state energy $E_N$ per volume $L^3$, in the thermodynamic limit (i.e. in the limit where firstly, $N\to\infty$ with $\rho=N/L^3$ fixed, and then $\rho\to0$), is of the form
  \begin{equation}\label{no1}
    e(\rho)=\frac{3}{5}(3\pi^2)^{\frac{2}{3}}\rho^{\frac{5}{3}}+2\pi a_0\rho^2
    +\frac{12}{35}(11-2\ln2)3^{\frac{1}{3}}\pi^{\frac{2}{3}}a_0^2\rho^{\frac{7}{3}}
    +o(\rho^{\frac{7}{3}}).
  \end{equation}
For Boson systems, with particle density $\rho$ and scattering length $a_0$, Lee, Huang and Yang wrote down the energy density in the thermodynamic limit, in \cite{huangyang,LHY106,TDLEECNYANG112},
 \begin{equation}\label{no2}
     e(\rho)
     =4\pi\rho^2 a_0\left(1+\frac{128}{15\sqrt{\pi}}(\rho a_0^3)^{1/2}+
     o((\rho a_0^3)^{1/2})\right).
\end{equation}

\par For decades, mathematicians long to verify these two result rigorously up to the mathematical standard. Many important and interesting rigorous studies of Bosons were collected and produced by Lieb, Seiringer, Solovej and Yngvason in \cite{lieb2005mathematics}. For Fermions, \cite{2005fermiphy} gave a rigorous proof to the first order of Huang-Yang formula. Moreover, one can see \cite{workshopfermi} for many significant and pioneering references.

\par For (\ref{no2}), the second order Lee-Huang-Yang formula for Boson systems in the Gross-Pitaevskii regime has been mathematically proved in the first work \cite{2018Bogoliubov} by C.Boccato, C.Brennecke, S.Cenatiempo, and B.Schlein, after many preparations like in \cite{boccatoBrenCena2020optimal,BEC2018,BEC2002}. Then the second order Lee-Huang-Yang formula in the thermodynamics limit had been proved in \cite{bose2ndthermo1,bose2ndthermo2} by S.Fournais and J.P.Solovej, with delicate space localization and treatment of the boundary condition. For the low temperature $T>0$ case, \cite{haberberger2023,haberberger2024} gave the Lee-Huang-Yang formula, with localization of the volume. At the moment, (\ref{no1}) for the Fermion case is still open. For energy approximation to Fermion systems in the thermodynamics limit, \cite{dilutefermiBog,fermiupper,giacomelli2024} gave the result up to the first order via Bogoliubov theory, then very recently, an upper bound towards the second order Huang-Yang formula has been obtained in \cite{2024huangyangformulalowdensityfermi}. Particularly, the ground state energy expansion for spin-polarized Fermi gas in the thermodynamic regime can be seen in \cite{polarizedupper,polarizedlower}. Another important approximation problem is the mean-field limit. For Boson systems, \cite{meanfield2011seiR} first gave a proof of Bogoliubov theory in the mean field regime. Later on, \cite{spectrum} predicts the lowest part of the spectrum of the many-body Hamiltonian of Bosons in the mean-ﬁeld regime, covering a large class of interacting Bose gases. For Fermion systems, \cite{meanfieldfermi2019,meanfieldfermi2021} extracted the correlation energy of weakly interacting Fermi gases in high density regime, via the method of Bosonization\footnote{We are not using Bosonization in this paper.} developed recently. More recent works include, for example \cite{Seiringerfermi,Phanfermimeanfield,christiansen2024correlationenergyelectrongas}.


\par In this paper, we prove the second order behaviour of Fermion systems (\ref{no1}) in the \textbf{Gross-Pitaevskii regime}, in other words, the scattering length of the interaction potential is at the scale $aL^{-1}\propto N^{-1}$. During our study, we normalize $L=1$. The Gross-Pitaevskii (G-P) limit is not only highly relevant\footnote{The dynamics case of the GP limit has also generated many nice work, see for example, \cite{dynamicChenHol,dy2006,dy2007,dy2009,dy2010}.}, and has served as the base case to the thermodynamic limit in terms of mathematics, as its proof provides important instinct for the thermodynamic limit, but also practical in the aspect of experimental and computational physics. In fact, it could be considered as a diagonal case of the thermodynamic limit in the sense that $\lim_{\rho=1/N^2\to0}$, comparing to the original definition of thermodynamic limit $\lim_{\rho\to0}\lim_{N\to\infty,\rho\,\text{fixed}}$. Moreover, different from the Boson case, for Fermion systems, in both of the G-P and thermodynamic regimes, we can deduce a-priori excitation estimates (see Lemma \ref{lemma prori BEC} and \cite{2005fermiphy,dilutefermiBog}), for a family of approximate ground state $\{\psi_j\}$, say the \textbf{Gross-Pitaevskii} regime:
\begin{equation*}
 \langle\mathcal{N}_{ex}\psi_j,\psi_j\rangle\lesssim N^{\frac{2}{3}},
\end{equation*}
and for the \textbf{Thermodynamic} regime:
\begin{equation*}
 \langle\mathcal{N}_{ex}\psi_j,\psi_j\rangle\lesssim L^3\rho^{\frac{7}{6}}.
\end{equation*}
 Here $\mathcal{N}_{ex}$ is the number of excitation operator, which represents the number of excited particles of the Fermion system. These estimates make it possible for us to use directly the Bogoliubov renormalization techniques to improve this excitation estimate, which is the key to the lower bound. Therefore, the thermodynamic limit case shares a similar logic, and could be dealt with in a separate paper.

\par We consider a system of $N$ Fermions of spin $1/\mathbf{q}$ with $\mathbf{q}\in\mathbb{N}$ the number of spin states, trapped in the 3D torus $\Lambda_{L}=[-\frac{L}{2},\frac{L}{2}]^3\subset\mathbb{R}^3$ with periodic boundary conditions. Throughout the discussion of this paper, we always normalize $L=1$, and we may denote $\Lambda=[-\frac{1}{2},\frac{1}{2}]^3$. Denote the set of $\mathbf{q}$ spin states by $\mathcal{S}_\mathbf{q}=\{\varsigma_i,i=1,\dots,\mathbf{q}\}$. In particular, we let $\mathcal{S}_2=\{\uparrow,\downarrow\}$. For $i=1,\dots,N$, let ${z}_i=(x_i,\sigma_i)\in\Lambda_L\times\mathcal{S}_\mathbf{q}$ describes the position and the spin state of the i-th particle. The wave function of the system should be in the Hilbert space $L_a^2\big((\Lambda_L\times\mathcal{S}_\mathbf{q})^N\big)\eqqcolon \mathcal{H}^{\wedge N}$ consisting of functions that are anti-symmetric with respect to permutations of the N particles, which is appropriate for describing a system of Fermions. Our N-body Hamilton operator is given by
\begin{align}\label{Hamiltonian}
 H_{N}=\sum_{j=1}^{N}-\Delta_{{x}_j}+\sum_{1\leq i<j\leq N}v_{a}({x}_i-{x}_j)
\end{align}
acting on the Hilbert space $\mathcal{H}^{\wedge N}$ with
\begin{equation}\label{interaction potential}
  v_a({x})=\frac{1}{a^2}v\left(\frac{{x}}{a}\right)
\end{equation}
 and $a=N^{-1}$, i.e. in the Gross-Pitaevskii regime. We require the interaction potential $v$ to be non-negative, radially-symmetric and compactly supported in a 3D ball. Moreover, we assume $v$ has scattering length $\mathfrak{a}_0\geq0$ (therefore the scattering length of $v_a$ is $a\mathfrak{a}_0$). From here on out, we will always assume a function to be with these three properties whenever it is referred to as the \textit{interaction potential}.

\par The main object of our study is the ground state energy of $H_N$, which we denote by
\begin{equation}\label{ground state energy}
  E_{N,\mathbf{q}}=\inf_{\substack{\psi\in \mathcal{H}^{\wedge N}\\
\Vert\psi\Vert_2=1}}\langle H_N\psi,\psi\rangle.
\end{equation}
In fact, we mainly consider the case that there are $N_{\varsigma_i}$ Fermions of spin $\varsigma_i$, and $N=\sum_{i=1}^\mathbf{q}N_{\varsigma_i}$. Without loss of generality, we demand for each spin $\varsigma_i$, $N_{\varsigma_i}$ is at the scale of $N$. The admissible wave functions describing such a system form a subspace of $\mathcal{H}^{\wedge N}$ given by
\begin{equation}\label{subspace H^wedge N}
\begin{aligned}
  \mathcal{H}^{\wedge N}(\{N_{\varsigma_i}\})=\big\{\psi\in\mathcal{H}^{\wedge N}\vert\, &\text{${z}_j=(x_j,\sigma_j)\in\Lambda_L\times\mathcal{S}_\mathbf{q}$,}\\
   &\text{if there exists some fixed $1\leq i\leq q$}\\
   &\text{such that $\#\{\sigma_j=\varsigma_i\}\neq N_{\varsigma_i}$}
   \Rightarrow\psi(z_1,\dots,z_N)=0\big\}.
\end{aligned}
\end{equation}
Similarly, we denote the ground state energy of $H_N$ on $\mathcal{H}^{\wedge N}(\{N_{\varsigma_i}\})$ by
\begin{equation}\label{ground state energy N_i}
   E_{\{N_{\varsigma_i}\},\mathbf{q}}=\inf_{\substack{\psi\in \mathcal{H}^{\wedge N}(\{N_{\varsigma_i}\})\\
\Vert\psi\Vert_2=1}}\langle H_N\psi,\psi\rangle.
\end{equation}

\par For $\sigma\in \mathcal{S}_\mathbf{q}$, we write the Fermi ball\footnote{The Fermi ball is also known as the Fermi sea in some literature.} $B_F^{\sigma}=\{k\in2\pi\mathbb{Z}^3,\,
\vert k\vert\leq k_F^{\sigma}\}$ such that its cardinality $\#B_F^{\sigma}=N_\sigma$. For simplicity, in this paper we will always consider suitable $N_{\sigma}$ that precisely fills the Fermi ball\footnote{If we consider arbitrary $N_\sigma$ that may not completely fills the Fermi ball, we will encounter a complex algebraic correction term involving the Dedekind $\zeta$-function which is not of our interests (See Remark \ref{remark4}).}. For any integer $R>1$, we denote the set of integer points inside a 3-dimensional ball of radius $R^{1/2}$ by
\begin{equation}\label{lattice in 3d ball}
  \mathcal{B}_R\coloneqq\{x\in\mathbb{Z}^3,\,\big\vert\, \vert x\vert^2\leq R\}.
\end{equation}
As one can see in \cite{HeathBrown+1999+883+892}, as $R$ grows larger, the number of lattice points in $\mathcal{B}_R$ can be approximated by
\begin{equation}\label{approximate lattice 3d ball}
  \#\mathcal{B}_R=\frac{4}{3}\pi R^{\frac{3}{2}}+O(R^{\frac{21}{32}+\varepsilon})
\end{equation}
for any $\varepsilon>0$. It follows directly that
\begin{equation}\label{N_sigma approximate}
  N_{\sigma}=\frac{4}{3}\pi\Big(\frac{k_F^{\sigma}}{2\pi}\Big)^3+
  O\Big(\Big(\frac{k_F^{\sigma}}{2\pi}\Big)^{\frac{21}{16}+\varepsilon}\Big)
\end{equation}
for any $\varepsilon>0$ when $N_{\sigma}\to\infty$. Therefore we obtain
\begin{equation}\label{k_F^sigma}
  k_F^{\sigma}=(6\pi^2)^{\frac{1}{3}}N_{\sigma}^{\frac{1}{3}}
  +O(N_{\sigma}^{-\frac{11}{48}+\varepsilon}).
\end{equation}
We can always choose $\vert k_F^{\sigma}\vert^2\in4\pi^2\mathbb{Z}$, so that for any $k\notin B^\sigma_F$, we have the lower bound
\begin{equation}\label{lower bound k_F}
  \vert k\vert^2-\vert k_F^{\sigma}\vert^2\geq 4\pi^2.
\end{equation}
On the other hand, the number of lattice points on the 2-dimensional sphere $\partial\mathcal{B}_R$ (if nonempty), as $R$ grows larger, can be approximated by
\begin{equation}\label{approximate lattice 2d sphere}
  \#\partial\mathcal{B}_R=R^{\frac{1}{2}+\varepsilon}
\end{equation}
for any $\varepsilon>0$. This result can be seen in, for example, \cite[(1.1)]{Linniksergodicmethod}. As a direct consequence, the number of lattice points on the Fermi surface $\partial B_{F}^\sigma$ is of order $N_\sigma^{1/3+\varepsilon}$. This fact is crucial in our discussions below.

\par With the aforementioned notations, our main result is the following 2nd order Huang-Yang formula on the ground state energy, which is (\ref{no1}) in the Gross-Pitaevskii regime.
\begin{theorem}\label{core}
Let $H_N$ be given in (\ref{Hamiltonian}), assume $a=N^{-1}$, and $N=\sum_{i=1}^\mathbf{q}N_{\varsigma_i}$ such that for each spin $\varsigma_i$, $N_{\varsigma_i}$ is at the scale of $N$ and precisely fills the Fermi ball $B_F^{\varsigma_i}$, the following formula holds when $N$ is large enough:
  \begin{equation}\label{core formula}
  \begin{aligned}
    E_{\{N_{\varsigma_i}\},\mathbf{q}}=\sum_{i=1}^{\mathbf{q}}\sum_{k\in B_F^{\varsigma_i}}\vert k\vert^2
    +4\pi a\mathfrak{a}_0\sum_{i\neq j}N_{\varsigma_i}N_{\varsigma_j}+E^{(2)}+O(N^{\frac{1}{3}-\frac{1}{360}}).
  \end{aligned}
  \end{equation}
  with $E^{(2)}$ given by (\ref{rewrt E^(2)}), and we write it formally as
  \begin{equation}\label{E^(2)}
    E^{(2)}=2(4\pi a\mathfrak{a}_0)^2
    \sum_{\substack{\sigma,\nu\in\mathcal{S}_\mathbf{q}\\ \sigma\neq\nu}}
    \sum_{\substack{p,q\\k\neq0}}\Big(\frac{\chi_{p\in B^{\sigma}_F}\chi_{q\in B^{\nu}_F}}{2\vert k\vert^2}-\frac{
    \chi_{p-k\notin B^{\sigma}_F}\chi_{q+k\notin B^{\nu}_F}
    \chi_{p\in B^{\sigma}_F}\chi_{q\in B^{\nu}_F}}{\big(\vert q+k\vert^2+\vert p-k\vert^2-\vert q\vert^2-\vert p\vert^2\big)}\Big).
  \end{equation}
  The order of $E^{(2)}$ is not greater than $N^{1/3+\varepsilon}$ for any fixed $\varepsilon>0$ small enough as long as $N$ is correspondingly large enough.
\end{theorem}
\begin{remark}\label{remark1}
 The series stated in (\ref{E^(2)}) converges in the sense that for some cut-off constant $c_N=4\sup_{\sigma}k_F^\sigma$ (depending on $N$),
      \begin{equation}\label{rewrt E^(2)}
        \begin{aligned}
         \frac{E^{(2)}}{2(4\pi a\mathfrak{a}_0)^2}&\coloneqq
    \sum_{\substack{\sigma,\nu\in\mathcal{S}_\mathbf{q}\\ \sigma\neq\nu}}
    \sum_{\substack{p,q,k\\0<\vert k\vert\leq c_N}}\Big(\frac{\chi_{p\in B^{\sigma}_F}\chi_{q\in B^{\nu}_F}}{2\vert k\vert^2}-\frac{
    \chi_{p-k\notin B^{\sigma}_F}\chi_{q+k\notin B^{\nu}_F}
    \chi_{p\in B^{\sigma}_F}\chi_{q\in B^{\nu}_F}}{\big(\vert q+k\vert^2+\vert p-k\vert^2-\vert q\vert^2-\vert p\vert^2\big)}\Big)\\
    &+\frac{1}{2}\sum_{\substack{\sigma,\nu\in\mathcal{S}_\mathbf{q}\\ \sigma\neq\nu}}
    \sum_{\substack{p,q,k,\\\vert k\vert> c_N}}
    \Big(\frac{\chi_{p\in B^{\sigma}_F}\chi_{q\in B^{\nu}_F}}{\vert k\vert^2}
    -\frac{
    \chi_{p\in B^{\sigma}_F}\chi_{q\in B^{\nu}_F}}{\big(\vert q+k\vert^2+\vert p-k\vert^2-\vert q\vert^2-\vert p\vert^2\big)}\\
    &\quad\quad\quad\quad\quad\quad\quad\quad
    -\frac{
    \chi_{p\in B^{\sigma}_F}\chi_{q\in B^{\nu}_F}}{\big(\vert q-k\vert^2+\vert p+k\vert^2-\vert q\vert^2-\vert p\vert^2\big)}\Big)
        \end{aligned}
      \end{equation}
      This rearrangement is legitimate since $\vert k\vert> c_N$ ensures $\vert p\pm k\vert>k_F^{\sigma}$ and $\vert q\pm k\vert>k_F^\nu$. The first term on the right hand side of (\ref{rewrt E^(2)}) is a finite sum (for fixed $N$), and the second term converges absolutely. It is easy to figure out that the second term on the right hand side of (\ref{rewrt E^(2)}) is of order $N^{7/3}$, while for the first term, it is more subtle since it involves the number of lattice point on the Fermi surface $\partial B_F^\sigma$ (see the discussion below (\ref{C_2+C_3 conclusion}), which points out that it is of order $N^{7/3+\varepsilon}$). A thorough analysis of $E^{(2)}$ is given in Lemma \ref{Ground State}.
\end{remark}

\begin{remark}\label{remark2}
  As usual and pointed out by, for example, \cite[(1.13)]{2018Bogoliubov} or \cite{beyondGP}, to recover the Huang-Yang formula in the thermodynamics limit, we can rescale (\ref{core formula}) by  formally dividing (\ref{core formula}) by $L^5$, substituting $a\mathfrak{a}_0$ by $a_0/L$ with $\rho=N/L^3$ fixed as $L\to\infty$. We also set $\mathbf{q}=2$, $N_{\uparrow}=N_{\downarrow}=N/2$. Then the first term of (\ref{core formula}) becomes
      \begin{equation}\label{0th}
        \frac{2}{(2\pi)^3}\int_{\vert x\vert<(3\pi^2\rho)^{\frac{1}{3}}}\vert x\vert^2dx
        =\frac{3}{5}(3\pi^2)^{\frac{2}{3}}\rho^{\frac{5}{3}}.
      \end{equation}
      the second term of (\ref{core formula}) turns to
      \begin{equation}\label{1st}
        2\pi a_0\rho^2
      \end{equation}
      while the remaining terms have the leading order (with $r_0=(3\pi^2\rho)^{\frac{1}{3}}$)
      \begin{equation}\label{2nd}
      \begin{aligned}
      &\frac{4(4\pi a_0)^2}{(2\pi)^9}
      \int_{\vert q\vert<r_0}\int_{\vert p\vert<r_0}\int_{k\in\mathbb{R}^3}
      \Big(\frac{1}{2\vert k\vert^2}-
      \frac{\chi_{p-k\notin B^{\sigma}_F}\chi_{q+k\notin B^{\nu}_F}}
    {\big(\vert q+k\vert^2+\vert p-k\vert^2-\vert q\vert^2-\vert p\vert^2\big)}\Big)\\
        &=-\frac{8(4\pi a_0)^2}{(2\pi)^9}
        \int_{\vert x_1\vert<r_0}
        \int_{\vert x_2\vert<r_0}
        \int_{\vert x_3\vert<r_0}
        \int_{x_4}
        \frac{\delta(x_1+x_2=x_3+x_4)}{\vert x_1\vert^2+\vert x_2\vert^2
        -\vert x_3\vert^2-\vert x_4\vert^2}\\
        &=\frac{12}{35}(11-2\ln2)3^{\frac{1}{3}}\pi^{\frac{2}{3}}a_0^2\rho^{\frac{7}{3}}
        \end{aligned}
      \end{equation}
      exactly matching the Huang-Yang formula. We will give a brief proof in Appendix \ref{analysis} showing the first equality in the middle of (\ref{2nd}) holds. For the evaluation of the integral on the second line, see, for example, \cite{soviet}\footnote{In \cite[Appendix]{soviet}, the authors gave a sketch of how to compute the integral in the middle of (\ref{2nd}). Though the final result of $E^{(2)}$ was correct, the integral shown in equation \cite[(9)]{soviet} was 2 times larger than the correct value. The reason maybe that the choice of $U$ in equation \cite[(2)]{soviet} should be $U=2\pi a\hbar^2/m$ and the choice of $Q_{ik}$ in equation \cite[(4)]{soviet} should be $Q_{ik}=\chi_{ik}$ rather than $\chi_{ik}/2$}. Similarly, if we set $\mathbf{q}=4$, we can recover \cite[(64)]{huangyang} in the original Huang-Yang's paper.
\end{remark}

\begin{remark}\label{remark3}
   For Boson systems in the Gross-Pitaevskii regime, there is a term which is considered as a finite-volume correction to the scattering length $\mathfrak{a}_0$ in the second order energy approximation of order $1$, see \cite[(1.11)]{2018Bogoliubov} and \cite[(1.18)]{me}. For Fermion systems, we also come up with the exact same correction term during our computation, see (\ref{define e_0}), but its order $1$ is obviously less than the second order term, which is no less than $N^{1/3}$, and therefore this correction is absorbed into the small error.
\end{remark}

\begin{remark}\label{remark4}
If we consider arbitrary $N_\sigma$ that may not perfectly fill the Fermi ball, we will encounter a complex algebraic correction term involves the Dedekind $\zeta$-function. More precisely, there is a correction term to (\ref{core formula}) given by
          \begin{equation}\label{correction lattice}
            4\pi a\mathfrak{a}_0\sum_{i\neq j}
            \Big(\#B_F^{\varsigma_i}\#B_F^{\varsigma_j}-N_{\varsigma_i}N_{\varsigma_j}\Big).
          \end{equation}
          The difference $(\#B_F^{\varsigma_i}-N_{\varsigma_i})$ counts the cardinality of a subset of the Fermi surface $\partial B_{F}^\sigma$. Due to the observation that the cardinality of Fermi surface $\partial B_{F}^\sigma$ is of order $N_\sigma^{1/3+\varepsilon}$, this correction term is of order $N^{1/3+\varepsilon}$. Under the scale argument presented in Remark \ref{remark2}, we know that this term vanishes in the corresponding thermodynamics limit, and is hence not the intrinsic quality of Fermion systems, but a discrete surface correction to the Gross-Pitaevskii limit.
\end{remark}

\subsection{Organization of the Paper}\
\par We here sketch the grand scheme of the proof of Theorem \ref{core}. To obtain the second order formula (\ref{core formula}), we would first need to prove an \textquotedblleft optimal BEC\textquotedblright\, type excitation result, like the ones in \cite{spectrum,bogsimplified}, for our Fermion
systems. We are going to conjugate $H_N$ with several suitable
unitary operators, which are more
well-known as renormalizations. After three steps of renormalizations, where we collect the results in Propositions \ref{p1}-\ref{p3}, we analyse the order and compute the explicit form of the second order energy in Lemma \ref{Ground State}. With Proposition \ref{p3} and Lemma \ref{Ground State}, the excitation Hamiltonian has the form
\begin{equation*}
\mathcal{U}^*H_N\mathcal{U}
=E_0+\mathcal{K}_s+e^{-\tilde{B}}\mathcal{V}_4e^{\tilde{B}}+
\mathcal{E},
\end{equation*}
where $E_0$ is the expected energy expansion up to the second order. $\mathcal{K}_s$ and $\mathcal{V}_4$ are excited kinetic and potential energy operators respectively, and both of them are non-negative operators. $\mathcal{E}$ is a small error that can be absorbed by $\mathcal{K}_s$ and $e^{-\tilde{B}}\mathcal{V}_4e^{\tilde{B}}$ for approximate ground states. These preparations allow us to derive Proposition \ref{Optimal BEC} via the method of localization estimates. Proposition \ref{Optimal BEC}, which is considered as an excitation estimate:
\begin{equation*}
    H_N-E_0\gtrsim \mathcal{N}_{ex}-N^{\frac{1}{3}-\frac{1}{360}},
  \end{equation*}
 is the key to the lower bound in Theorem \ref{core}, since the number of excitation operator $\mathcal{N}_{ex}$ is non-negative. The upper bound (\ref{uppper}) also follows by Proposition \ref{p3} and Lemma \ref{Ground State} by choosing a suitable test function. Here, the concluding argument in reaching the lower bound of Theorem \ref{core} has more than one version (see for example \cite{meanfieldfermi2021}), though they are sort of equivalent, once Lemma \ref{Ground State} and Proposition \ref{Optimal BEC} are proved.
\par We start by rewriting $H_N$ using creation and annihilation operators
\begin{equation*}
  H_N=\mathcal{K}+\mathcal{V}=\sum_{k,\sigma}\vert k\vert^2a_{k,\sigma}^*a_{k,\sigma}
  +\frac{1}{2}\sum_{k,p,q,\sigma,\nu}\hat{v}_ka^*_{p-k,\sigma}a^*_{q+k,\nu}a_{q,\nu}a_{p,\sigma},
\end{equation*}
where $\hat{v}_k$ is the Fourier coefficient of $v_a$ on $\Lambda$, see (\ref{define v_k}). Inspired by \cite{2018Bogoliubov,me,bogsimplified}, we conjugate the Hamiltonian $H_N$ with a unitary operator, the quadratic renormalization $e^{B}$ (we leave the precise definitions of $a$, $a^*$ to Section \ref{fock space formalism}), with:
\begin{equation*}
  \begin{aligned}
  B&=\frac{1}{2}\sum_{k,p,q,\sigma,\nu}
  \eta_k(a^*_{p-k,\sigma}a^*_{q+k,\nu}a_{q,\nu}a_{p,\sigma}-h.c.)\chi_{p-k\notin B^{\sigma}_F}\chi_{q+k\notin B^{\nu}_F}\chi_{p\in B^{\sigma}_F}\chi_{q\in B^{\nu}_F}.
  \end{aligned}
\end{equation*}
The coefficients $\eta_k$ are defined through the 3D scattering equation with Neumann boundary condition
\begin{equation*}
  \left\{\begin{aligned}
  &(-\Delta_{x}+\frac{1}{2}v)f_\ell=\lambda_\ell f_\ell,\quad \vert
  x\vert\leq \frac{\ell}{a},\\
  &\left.\frac{\partial f_\ell}{\partial \mathbf{n}}\right\vert_{\vert
   x\vert=\frac{\ell}{a}}=0,\quad \left.f_\ell\right\vert_{\vert
   x\vert=\frac{\ell}{a}}=1.
  \end{aligned}\right.
\end{equation*}
with $\ell=N^{-\frac{1}{3}-\alpha}$ for some $\alpha>0$ to be determined later.

\par At this point, for comparison purpose, let us recall some facts and previous works regarding the Boson systems. For Boson system, the Hamiltonian $H_N$ defined by (\ref{Hamiltonian}) reads
\begin{equation*}
  H_N=\mathcal{K}+\mathcal{V}=\sum_{k}\vert k\vert^2a_{k}^*a_{k}
  +\frac{1}{2}\sum_{k,p,q}\hat{v}_ka^*_{p-k}a^*_{q+k}a_{q}a_{p}.
\end{equation*}
Due to the existence of Bose-Einstein condensation in the Boson case, a Bosonic intuition is that each $a_0$ contributes to the ground state energy of order $\sqrt{N}$ (since $\Vert a_0\Vert\leq\sqrt{N}$ in the Boson case). This intuition is useful in distinguishing the main part of $H_N$ that may contribute to the ground state energy up to the second order, and the subordinate part which can be considered as small error. Very differently, the anti-symmetry of Fermion systems implies that there may be some \textquotedblleft negative\textquotedblright\, energy which would eliminate other positive energy. This phenomenon leads to the failure of the above Bosonic intuition, and makes it difficult to figure out the order of the energy contribution directly from the forms of the operators.

\par The existence of \textquotedblleft negative\textquotedblright\, energy can be demonstrated more precisely via the \textquotedblleft Bosonization\textquotedblright\, operator $b_{k,\sigma}$. In \cite{meanfieldfermi2019,meanfieldfermi2021} for example, a \textquotedblleft Bosonization\textquotedblright\, method was introduced to deal with the Fermionic renormalizations:
\begin{equation*}
  b_{k,\sigma}^*\sim \frac{1}{\sqrt{N}_\sigma}\sum_{p}a^*_{p-k,\sigma}a_{p,\sigma}
  \chi_{p-k\notin B_F^\sigma}\chi_{p\in B_F^\sigma}.
\end{equation*}
By the anti-symmetry of Fermion systems, we find that
\begin{equation*}
  [b_{k,\sigma},b_{k^\prime,\sigma^\prime}^*]
  =\delta_{k,k\prime}\delta_{\sigma,\sigma^\prime}
  -\delta_{\sigma,\sigma^\prime}\varepsilon_{k,k^\prime}
\end{equation*}
where $\varepsilon_{k,k^\prime}$ is an approximately positive operator. This operator $\varepsilon_{k,k^\prime}$ can not be considered as a small residue when we add them altogether. They eventually become the \textquotedblleft negative\textquotedblright\, energy. This Bosonization method is proved to be indeed useful in the mean-field regime and leads to several nice works like \cite{meanfieldfermi2019,meanfieldfermi2021}, where $\varepsilon_{k,k^\prime}$ can be handled by decomposing paired particle operator $b$ into smaller pieces localized in disjoint patches on the Fermi surface. As far as this paper concerns, we do not adopt this paired particle operator $b$ in Gross-Pitaevskii regime, to define the quadratic renormalization $e^B$ above, and we in fact take a direct approach, in order to demonstrate more explicitly the generative mechanism of the \textquotedblleft negative\textquotedblright\, energy\footnote{We do not know if the Bosonization will work in this paper's setting and we do not know if the unbosonizable parts is small like in \cite{meanfieldfermi2019,meanfieldfermi2021}}.

\par To this end, we are motivated by a comparable, and also difficult but totally different mechanism we dealt with before. Such aforementioned Bosonic intuition also fails during the study \cite{me} of the 2nd order 2D behaviour of 3D Bose gas in the G-P regime (though it is generated by different mechanics \textendash the smallness of the thickness of the region $d\ll L$). More precisely, if we want to reach the correct second order energy expansion, the intuition from the 3D Boson system \cite{2018Bogoliubov} leads to an extra \textbf{cubic renormalizaiton}. On the other hand, the instinct from the 3D-to-2D Boson system \cite{me}, implies this cubic renormalization may generate energy which is large enough to compete with the 2nd order (see Lemmas \ref{lemma V_22,23,1} and \ref{cal [V_3,B'] lemma}, and compare them with \cite[Lemma 15]{bogsimplified} or \cite[Lemma 8.9]{me} in the Boson case). Therefore, after the computation of the quadratic renormalization $e^B$, we further conjugate the Hamiltonian with another unitary operator we observe, the cubic renormalization $e^{B^\prime}$, with
\begin{equation*}
    B^\prime=\sum_{k,p,q,\sigma,\nu}
  \eta_k(a^*_{p-k,\sigma}a^*_{q+k,\nu}a_{q,\nu}a_{p,\sigma}-h.c.)\chi_{p-k\notin B^{\sigma}_F}\chi_{q+k\notin B^{\nu}_F}\chi_{q\in A^\nu_{F,\kappa}}
  \chi_{p\in B^{\sigma}_F}.
\end{equation*}
Here
\begin{equation*}
  A_{F,\kappa}^\nu=\{k\in2\pi\mathbb{Z}^3,\,k_F^\nu<\vert k\vert\leq k_F^\nu+N^\kappa\}.
\end{equation*}
We introduce $\kappa<0$ as a cut-off parameter, to decompose the interaction between low and high momenta. We leave the thorough discussion about $A_{F,\kappa}^\nu$ and the method of decomposition of high-low momenta to Section \ref{number of p ops}.

\par Notice that the $B^\prime$ chosen above can not be written in a cubic form of Bosonization operator $b$, and we refer it as a cubic renormalization since there are three $a$ in its definition whose subscripts (momenta) are out of the corresponding Fermi ball, in corresponding to the cubic renormalization in \cite{bogsimplified,me}. Although it is not clear in the first place whether the designation of $B^\prime$ follows any intuition from \cite{2018Bogoliubov,me,bogsimplified}, $B^\prime$ works appropriately, and helps us reaching the expected result. We sort of guess our way forward and somehow feel our way is general. Despite the definition is in a way, parallel to the Boson case, the calculation is completely Fermionic, distinguishing itself from the method of Bosonizaiton. In fact, the exact size of the energy contribution from the pure Fermionic part of $B^\prime$ is also unclear initially. The only way to justify this definition is a comprehensive evaluation.

\par Another reason we do not apply the Bosonization method but do it directly (and it worked), is that the Bogoliubov coefficients $\xi_{k,q,p}^{\sigma,\nu}$ introduced later in (\ref{later}) are dependent on $p$ and $q$ as well. It would be intricate to handle the relation between the Bosonization operator $b$ and the  Bogoliubov coefficients $\xi_{k,q,p}^{\sigma,\nu}$, in comparison to the direct definition of $\tilde{B}$ below.

\par The above two renormalizations actually capture
the correct correlation structures generating the first order energy and part of the second order energy. This allows us to
apply the Bogoliubov transformation denoted by $e^{\tilde{B}}$ with
\begin{equation*}
  \tilde{B}=\sum_{k,p,q,\sigma,\nu}
  \xi_{k,q,p}^{\nu,\sigma}(a^*_{p-k,\sigma}a^*_{q+k,\nu}a_{q,\nu}a_{p,\sigma}-h.c.)
  \chi_{p-k\notin B^{\sigma}_F}\chi_{q+k\notin B^{\nu}_F}\chi_{p\in B^{\sigma}_F}\chi_{q\in B^{\nu}_F},
\end{equation*}
where $\xi_{k,q,p}^{\nu,\sigma}$ are given by
\begin{equation}\label{later}
  \xi_{k,q,p}^{\nu,\sigma}=\frac{-\big(W_k+\eta_kk(q-p)\big)}
  {\frac{1}{2}\big(\vert q+k\vert^2+\vert p-k\vert^2-\vert q\vert^2-\vert p\vert^2\big)}
  \chi_{p-k\notin B^{\sigma}_F}\chi_{q+k\notin B^{\nu}_F}\chi_{p\in B^{\sigma}_F}\chi_{q\in B^{\nu}_F},
\end{equation}
where
\begin{equation*}
  W_p=\frac{\lambda_{\ell}}{a^2}\left(\widehat{\chi_{\ell}}\left(
  \frac{p}{2\pi}\right)+\eta_p\right)=\left\vert p\right\vert^2\eta_p+\frac{1}{2}
  \sum_{q\in2\pi\mathbb{Z}^3}\hat{v}_{p-q}\eta_q+ \frac{1}{2}\hat{v}_p.
\end{equation*}
More precise analysis of $\eta$ and $\xi$ will be demonstrated in Section \ref{coeff}. The Bogoliubov coefficients $\xi_{k,q,p}^{\nu,\sigma}$ is chosen intentionally so that
\begin{equation*}
  \frac{1}{2}\big(\vert q+k\vert^2+\vert p-k\vert^2-\vert q\vert^2-\vert p\vert^2\big)\xi_{k,q,p}^{\nu,\sigma}+W_k+\eta_kk(q-p)=0.
\end{equation*}
The extra $p,q$\textendash dependence of $\xi_{k,q,p}^{\nu,\sigma}$ (comparing to $\eta_k$ which depends only on $k$) makes it much more technically difficult to evaluate the conjugation. Different from the method of reduction to pair excitations $b^*$ on patches introduced in, for example, \cite{meanfieldfermi2019,meanfieldfermi2021}, an important technique throughout our evaluation is the decomposition of the interaction of different momentum levels (see Section \ref{number of p ops} and Lemma \ref{lemma ineqn K_s} for more details). This idea is analogous to the decomposition of different frequencies when dealing with nonlinear interactions in the terminology of the analysis of nonlinear PDE. In fact, during our computation, the main energy comes from the exchange between
 \begin{equation*}
   [a_{p-k,\sigma}^*a_{q+k,\nu}^*a_{q,\nu}a_{p,\sigma}+h.c.,a_{r-l,\tau}^*a_{s+l,\varpi}^*
   a_{s,\varpi}a_{r,\tau}+h.c.]
 \end{equation*}
 where $p-k,q+k,r-l,s+l$ are momenta outside of Fermi ball, and $p,q,r,s$ are momenta inside of Fermi ball. The commutator has two forms:
 \begin{enumerate}[label=\textbullet]
   \item $a^*a^*aa$: When all four subscripts (momenta) are inside the Fermi ball, (using the PDE's terminology, when it has zero frequency), their interaction contains main energy up to the 2nd order.
   \item $a^*a^*a^*aaa$: Since there are at most four momenta inside or outside of Fermi ball, no zero frequency arises in this form. For momentum that is far away from the boundary of Fermi ball $\partial B_F$, we use the excited kinetic energy $\mathcal{K}_s$ to control it, as in (\ref{ineqn N_h and N_i}). For momentum which is near the surface, we use the usual operator norm to control it, as in (\ref{ineqn N_l and N_s}). Generally speaking, we have to bound all the terms of this form.
 \end{enumerate}
\vspace{1em}
\par We layout this paper as follows. In Section \ref{fock space sub}, for the sake of completeness, we present the formalism of the Fock space $F_N$. In Section \ref{cao}, we give a precise definition of creation and annihilation operators, and in Section \ref{number of p ops}, we introduce various particle number operators, and develop the technique of decomposition of the interaction of different momentum levels. In Section \ref{2ndquanpri}, we collect the prerequisites for renormalizations and the proof of Theorem \ref{core}. In Section \ref{coeff}, we discuss the coefficients $\eta_p$ and $\xi_{k,q,p}^{\nu,\sigma}$ in detail, and provide useful estimates about these coefficients. In Section \ref{mainpropo}, we gather the results of the three steps of renormalizations in Propositions \ref{p1}-\ref{p3}. The detailed proofs of these three propositions are left to Sections \ref{qua}-\ref{bog}. In Section \ref{proofofmain}, we conclude the proof of Theorem \ref{core} by first computing the explicit form of the second order energy $E^{(2)}$ in Section \ref{constant}, and then establishing the excitation estimate in Section \ref{ee}. Matching upper and lower bounds, and hence Theorem \ref{core} follow in Section \ref{proof core}.

\section{Fock Space for Fermions}\label{fock space formalism}
\subsection{Fock Space Formalism}\label{fock space sub}
\
\par For $k\in2\pi\mathbb{Z}^3$ and $\sigma\in\mathcal{S}_\mathbf{q}$, we choose the orthonormal basis of $\mathcal{H}=L^2(\Lambda_L\times\mathcal{S}_\mathbf{q})\cong L^2(\Lambda_L;\mathbb{C}^q)$ by
\begin{equation}\label{orthonormal basis H}
  f_{k,\sigma}(x,\nu)=f_k(x)\chi_{\nu=\sigma}(\nu)=\frac{\chi_{\nu=\sigma}(\nu)}{L^{\frac{3}{2}}}
  e^{ik\cdot x/L},\quad z=(x,\nu)=\Lambda_L\times\mathcal{S}_\mathbf{q}.
\end{equation}
For non-negative integers $n,m$, and $\psi\in L^2\big((\Lambda_L\times\mathcal{S}_\mathbf{q})^n\big)$, $\varphi\in L^2\big((\Lambda_L\times\mathcal{S}_\mathbf{q})^m\big)$, we define the anti-symmetric wedge product between $\psi$ and $\varphi$
\begin{equation*}\label{wedge product}
\begin{aligned}
  \psi\wedge\varphi(z_1,\dots,z_{n+m})=& \frac{1}{\sqrt{n!m!(n+m)!}}\\
&\times\sum_{\varrho\in S_{n+m}}\sgn{\varrho}\,\psi(z_{\varrho(1)},\dots,z_{\varrho(n)})
  \varphi(z_{\varrho(n+1)},\dots,z_{\varrho(n+m)}).
\end{aligned}
\end{equation*}
Under these definitions, it is straight-forward to verify $\psi\wedge\varphi\in L^2_a\big((\Lambda_L\times\mathcal{S}_\mathbf{q})^{n+m}\big)$, and the wedge product is associative and anti-commutative, i.e.
\begin{equation}\label{anti-commutative}
  \psi\wedge\varphi=(-1)^{nm}\varphi\wedge\psi.
\end{equation}
Under the definition of the wedge product, the notation $L_a^2\big((\Lambda_L\times\mathcal{S}_\mathbf{q})^N\big)= \mathcal{H}^{\wedge N}$ holds rigorously and we find an orthonormal basis of $\mathcal{H}^{\wedge N}$ given by
\begin{equation}\label{orthonormal basis H wedge N}
  f_{s_N}=\bigwedge_{i=1}^{N} f_{k_i,\sigma_i}=f_{k_1,\sigma_1}\wedge\cdots\wedge f_{k_N,\sigma_N}.
\end{equation}
Here we set the subscript $s_N$ satisfies the following order condition:
\begin{equation}\label{s_N}
  \begin{aligned}
  &s_N=\{(k_1,\sigma_1),\dots,(k_N,\sigma_N)
  \vert\,k_i\in2\pi\mathbb{Z}^3,\sigma_i\in\mathcal{S}_\mathbf{q}\}\\
  &\vert k_1\vert\leq\vert k_2\vert\leq\cdots\leq\vert k_N\vert\\
  &\text{if $\vert k_i\vert=\vert k_{i+1}\vert$, then $k_i\leq k_{i+1}$ in the lexicographical order sense}\\
  &\text{if $k_i=k_{i+1}$, then for some $j<l$, $k_i=\varsigma_j$ and $k_{i+1}=\varsigma_l$}
  \end{aligned}
\end{equation}
 For the orthonormal basis of the subspace $\mathcal{H}^{\wedge N}(\{N_{\varsigma_i}\})$, we first additionally set
\begin{equation}\label{s_N subspace}
  s_N=\{(k_1,\sigma_1),\dots,(k_N,\sigma_N)
  \vert\,k_i\in2\pi\mathbb{Z}^3,\sigma_i\in\mathcal{S}_\mathbf{q},
  \#\{\sigma_i=\varsigma_j\}\leq N_{\varsigma_j}\}.
\end{equation}
We denote the admissible set of all $s_N$ satisfying (\ref{s_N}) and (\ref{s_N subspace}) by $\mathcal{S}^{(N)}$. Therefore $\{f_{s_N}\,\big\vert\,s_N\in \mathcal{S}^{(N)}\}$ forms an orthonormal basis of the subspace $\mathcal{H}^{\wedge N}(\{N_{\varsigma_i}\})$. With a slight abuse of notation, we may also denote, for $1\leq n\leq N$, the set of all $s_n$ satisfying (\ref{s_N}) and (\ref{s_N subspace}) by $\mathcal{S}^{(n)}$.

\par Recall that in this paper, we normalize $L=1$. For $1\leq n\leq N$, we denote the subset of $\mathcal{S}^{(n)}$ which contains subscripts inside of the Fermi ball by
\begin{equation}\label{define P^(N)}
  \mathcal{P}^{(n)}=\{s_n\in\mathcal{S}^{(n)}\vert\, \text{$k_i\in B^{\sigma}_F$ if $\sigma_i=\sigma$}\}.
\end{equation}
We define additionally $\mathcal{P}^{(0)}=\{s_0\}$, $\mathcal{H}^{\wedge0}=\mathbb{C}$ and $f_{s_0}=1$. Hence we write the one-particle excited space of $\mathcal{H}$ by
\begin{equation}\label{define F}
  \mathcal{F}=\spa\{f_{s_1}\vert\, s_1\in\mathcal{P}^{(1)}\}\subset\mathcal{H}.
\end{equation}
For $\psi\in\mathcal{H}^{\wedge N}(\{N_{\varsigma_i}\})$, we decompose
\begin{equation}\label{rewrite psi for fock space}
  \psi=\sum_{n=0}^{N}\sum_{s_n\in\mathcal{P}^{(n)}}f_{s_n}\wedge\alpha^{(s_n)}
  =\sum_{n=0}^{N}\sum_{s_n\in\mathcal{P}^{(n)}}a^*_{k_1,\sigma_1}\dots a^*_{k_n,\sigma_n}\alpha^{(s_n)}.
\end{equation}
with $\alpha^{(s_n)}\in (\mathcal{F}^{\perp})^{\wedge(N-n)}$. So we define the Fock space for Fermions by
\begin{equation}\label{fock space}
  F_N\coloneqq \bigoplus_{n=0}^{N}\bigoplus_{s_n\in\mathcal{P}^{(n)}}
  (\mathcal{F}^{\perp})^{\wedge(N-n)}.
\end{equation}
Notice that by (\ref{rewrite psi for fock space}), we have
\begin{equation}\label{unitary relation}
  \Vert\psi\Vert_2^2=\sum_{n=0}^{N}\sum_{s_n\in\mathcal{P}^{(n)}}\Vert\alpha^{(s_n)}\Vert_2^2.
\end{equation}
Hence if we let $U_N:\mathcal{H}^{\wedge N}(\{N_{\varsigma_i}\})\to F_N$ with
\begin{equation}\label{define U_N}
  U_N\psi=\bigoplus_{n=0}^{N}\bigoplus_{s_n\in\mathcal{P}^{(n)}}\alpha^{(s_n)},
\end{equation}
we know that $U_N$ is an isometric operator from (\ref{unitary relation}) .

\subsection{Creation and Annihilation Operators}\label{cao}
\
\par We define the creation and annihilation operators on $\mathcal{H}^{\wedge N}(\{N_{\varsigma_i}\})$ by
\begin{equation}\label{define a*(f) & a(g)}
\left.\begin{aligned}
  (a^*(f)\psi)({z}_1,\dots,{z}_{n+1}) &= \frac{1}{\sqrt{n+1}}\sum_{i=1}^{n+1}(-1)^{i+1}f(z_i)
\psi(z_1,\dots,\hat{z_i},\dots,z_{n+1}) \\
  (a(g)\psi)(z_1,\dots,z_{n-1}) &= \sqrt{n}\int_{\Lambda_{L}\times\mathcal{S}_\mathbf{q}}\overline{g}(z)
\psi(z,z_1,\dots,z_{n-1})dz, \quad n\geq 1
\end{aligned}\right.
\end{equation}
It is straight-forward to check the following anti-commutator formula:
\begin{equation}\label{anticommutator}
\begin{aligned}
  &\{a(g),a^*(f)\}=a(g)a^*(f)+a^*(f)a(g)=\langle f,g\rangle,\\
  &\{a(g),a(f)\}=\{a^*(g),a^*(f)\}=0.
\end{aligned}
\end{equation}
Using (\ref{anticommutator}), we can bound the norm of creation and annihilation operators by
\begin{equation}\label{norm bound  of creation and annihilation}
  \Vert a(g)\Vert\leq\Vert g\Vert_2,\quad\Vert a^*(f)\Vert\leq \Vert f\Vert_2.
\end{equation}
Using (\ref{orthonormal basis H}), we denote
\begin{equation}\label{define a_k,sigma a*_k,sigma}
  a_{k,\sigma}=a(f_{k,\sigma}),\quad a^*_{k,\sigma}=a^*(f_{k,\sigma}).
\end{equation}
We can check directly by definition that
\begin{equation}\label{formulae of crea and annihi ops}
  \begin{aligned}
  a^*_{k,\sigma}\psi&=f_{k,\sigma}\wedge\psi\\
  a_{k,\sigma}\psi_{s_N}&=\left\{\begin{aligned}
  &(-1)^{i-1}\psi_{s_N\backslash\{(k,\sigma)\}},\quad \text{if for some $(k,\sigma)=(k_i,\sigma_i)$}\\
  &0,\quad \text{if no such $i$ exists}
  \end{aligned}\right.
  \end{aligned}
\end{equation}
From (\ref{formulae of crea and annihi ops}), we know that $a_{k,\sigma}^2=0$.

\par Using (\ref{anticommutator}), we can deduce some computation results we use multiple times in this paper.
\begin{lemma}\label{lem commutator}\
\begin{enumerate}[label=(\arabic*)]
 \item We have
 \begin{equation}\label{quartic commu}
  \begin{aligned}
    &[a^*_{p_1,\sigma}a^*_{q_1,\nu}a_{q_2,\nu}a_{p_2,\sigma},
    a^*_{r_1,\tau}a^*_{s_1,\pi}a_{s_2,\pi}a_{r_2,\tau}]\\
    =&(\delta_{p_2,r_1}\delta_{\sigma,\tau}\delta_{q_2,s_1}\delta_{\nu,\pi}
    -\delta_{q_2,r_1}\delta_{\nu,\tau}\delta_{p_2,s_1}\delta_{\sigma,\pi})
    a^*_{p_1,\sigma}a^*_{q_1,\nu}a_{s_2,\pi}a_{r_2,\tau}\\
    &+(\delta_{p_1,s_2}\delta_{\sigma,\pi}\delta_{q_1,r_2}\delta_{\nu,\tau}
    -\delta_{p_1,r_2}\delta_{\sigma,\tau}\delta_{q_1,s_2}\delta_{\nu,\pi})
    a^*_{r_1,\tau}a^*_{s_1,\pi}a_{q_2,\nu}a_{p_2,\sigma}\\
   &+a^*_{p_1,\sigma}a^*_{q_1,\nu}
    (\delta_{q_2,r_1}\delta_{\nu,\tau}a^*_{s_1,\pi}a_{p_2,\sigma}
    -\delta_{p_2,r_1}\delta_{\sigma,\tau}a^*_{s_1,\pi}a_{q_2,\nu})a_{s_2,\pi}a_{r_2,\tau}\\
    &+a^*_{p_1,\sigma}a^*_{q_1,\nu}
    (\delta_{p_2,s_1}\delta_{\sigma,\pi}a^*_{r_1,\tau}a_{q_2,\nu}
   -\delta_{q_2,s_1}\delta_{\nu,\pi}a^*_{r_1,\tau}a_{p_2,\sigma})a_{s_2,\pi}a_{r_2,\tau}\\
    &+a^*_{r_1,\tau}a^*_{s_1,\pi}
  (\delta_{p_1,r_2}\delta_{\sigma,\tau}a^*_{q_1,\nu}a_{s_2,\pi}
   -\delta_{p_1,s_2}\delta_{\sigma,\pi}a^*_{q_1,\nu}a_{r_2,\tau})a_{q_2,\nu}a_{p_2,\sigma}\\
    &+a^*_{r_1,\tau}a^*_{s_1,\pi}
    (\delta_{q_1,s_2}\delta_{\nu,\pi}a^*_{p_1,\sigma}a_{r_2,\tau}
    -\delta_{q_1,r_2}\delta_{\nu,\tau}a^*_{p_1,\sigma}a_{s_2,\pi})a_{q_2,\nu}a_{p_2,\sigma}.
  \end{aligned}
  \end{equation}

\item We also have
\begin{equation}\label{quadratic commu}
 \begin{aligned}
 &[a^*_{k,\sigma}a_{k,\sigma},a^*_{r_1,\tau}a^*_{s_1,\pi}a_{s_2,\pi}a_{r_2,\tau}]\\
 =&(\delta_{k,r_1}\delta_{\sigma,\tau}+\delta_{k,s_1}\delta_{\sigma,\pi}
 -\delta_{k,r_2}\delta_{\sigma,\tau}-\delta_{k,s_2}\delta_{\sigma,\pi})
 a^*_{r_1,\tau}a^*_{s_1,\pi}a_{s_2,\pi}a_{r_2,\tau}.
 \end{aligned}
\end{equation}

\item We let for $j=1,2$,
\begin{equation*}
\mathcal{D}_j=\sum_{k,\sigma}d^{(j)}_{k,\sigma}a_{k,\sigma}^*a_{k,\sigma},
\end{equation*}
then $[\mathcal{D}_1,\mathcal{D}_2]=0$.

\item Let $\mathcal{D}=\sum_{k,\sigma}d_{k,\sigma}a_{k,\sigma}^*a_{k,\sigma}$ and $f(z)=\sum_{k,\sigma}
c_{k,\sigma}f_{k,\sigma}(z)\in\mathcal{H}$, then
\begin{equation}\label{com D a(f)}
  [\mathcal{D},a(f)]=-\sum_{k,\sigma}d_{k,\sigma}\overline{c_{k,\sigma}}a_{k,\sigma}.
\end{equation}
In particular, for $S_2\subset S_1\subset2\pi\mathbb{Z}^3\times\mathcal{S}_{\mathbf{q}}$ such that $d_{k,\sigma}=\chi_{(k,\sigma)\in S_1}$ and $c_{k,\sigma}=\chi_{(k,\sigma)\in S_2}c_{k,\sigma}$, we have
\begin{equation}\label{com D a(f) 2}
  [\mathcal{D},a(f)]=-a(f).
\end{equation}

\end{enumerate}
\end{lemma}

\subsection{Particle Number Operators}\label{number of p ops}
\
\par We define for $\sigma\in\mathcal{S}_\mathbf{q}$, the particle number operator
\begin{equation}\label{define op N_sigma}
  \mathcal{N}=\sum_{\sigma\in\mathcal{S}_\mathbf{q}}
  \sum_{k\in2\pi\mathbb{Z}^3}a_{k,\sigma}^{\ast}a_{k,\sigma},\quad
  \mathcal{N}_{\sigma}=\sum_{k\in2\pi\mathbb{Z}^3}a_{k,\sigma}^{\ast}a_{k,\sigma}.
\end{equation}
It is easy to verify on $\mathcal{H}^{\wedge N}$,
\begin{equation}\label{sum N_sigma}
  \sum_{i=1}^{\mathbf{q}}\mathcal{N}_{\varsigma_i}=\mathcal{N}
  =N,
\end{equation}
and in particular, for $\psi\in\mathcal{H}^{\wedge N}(\{N_{\varsigma_i}\})$, we have
\begin{equation}\label{cal N_sigma}
  \mathcal{N}_{\sigma}\psi=N_\sigma\psi.
\end{equation}
We define the number of excitation operators
\begin{equation}\label{number of exicited particles ops}
  \mathcal{N}_{ex}=\sum_{\sigma,k\notin B^{\sigma}_F}a_{k,\sigma}^*a_{k,\sigma},\quad
  \mathcal{N}_{ex,\sigma}=\sum_{k\notin B^{\sigma}_F}a_{k,\sigma}^*a_{k,\sigma}.
\end{equation}
Recall that we demand for each spin $\varsigma_i$, $N_{\varsigma_i}$ precisely fills the Fermi ball $B_F^{\varsigma_i}$. Using (\ref{anticommutator}) and (\ref{number of exicited particles ops}), we find, on $\mathcal{H}^{\wedge N}(\{N_{\varsigma_i}\})$,
\begin{equation}\label{number of exicited particles ops2}
\begin{aligned}
  \sum_{\sigma,k\in B^{\sigma}_F}a_{k,\sigma}a^*_{k,\sigma}
  &=\sum_{\sigma,k\in B^{\sigma}_F}(1-a^*_{k,\sigma}a_{k,\sigma})\\
  &=N-\sum_{\sigma,k\in B^{\sigma}_F}a^*_{k,\sigma}a_{k,\sigma}=\mathcal{N}_{ex}.
\end{aligned}
\end{equation}
Similarly, we also have, on $\mathcal{H}^{\wedge N}(\{N_{\varsigma_i}\})$,
\begin{equation}\label{number of exicited particles ops3}
  \sum_{k\in B^{\sigma}_F}a_{k,\sigma}a^*_{k,\sigma}=\mathcal{N}_{ex,\sigma}.
\end{equation}
For $x\in\Lambda$ and $\sigma\in\mathcal{S}_\mathbf{q}$, we define
\begin{equation}\label{g_x,sigma h_x,sigma}
  g_{x,\sigma}(z)=\sum_{k\in B^{\sigma}_F}e^{ikx}f_{k,\sigma}(z),\quad
  h_{x,\sigma}(z)=\sum_{k\notin B^{\sigma}_F}e^{ikx}f_{k,\sigma}(z),\quad
  z\in \Lambda\times\mathcal{S}_\mathbf{q}.
\end{equation}
Using (\ref{define a*(f) & a(g)}) and (\ref{number of exicited particles ops2}), we reach
\begin{equation}\label{useful eqn1}
  \int_{\Lambda}\sum_{\sigma}a(g_{x,\sigma})a^*(g_{x,\sigma})dx=\mathcal{N}_{ex}.
\end{equation}
With a slightly abuse of notations\footnote{In fact, $a(h_{x,\sigma})$ is well-defined on a dense subset of $\mathcal{H}^{\wedge N}(\{N_{\varsigma_i}\})$, say, the subset of smooth compactly-supported functions.}, we also have
\begin{equation}\label{useful eqn2}
   \int_{\Lambda}\sum_{\sigma}a^*(h_{x,\sigma})a(h_{x,\sigma})dx=\mathcal{N}_{ex}.
\end{equation}

\par Due to technical reasons, for $r<k_F^\sigma<R$ (Here we omit the $\sigma$ dependence of $R$ and $r$), we introduce the following notations
\begin{equation}\label{partition 3D space}
  \begin{aligned}
  P_{F,R}^\sigma&=\{k\in2\pi\mathbb{Z}^3,\,\vert k\vert>R\}\\
  A_{F,R}^\sigma&=\{k\in2\pi\mathbb{Z}^3,\,k_F^\sigma<\vert k\vert\leq R\}\\
  \underline{A}^\sigma_{F,r}&=\{k\in2\pi\mathbb{Z}^3,\,r<\vert k\vert\leq k^\sigma_F\}\\
  \underline{B}^\sigma_{F,r}&=\{k\in2\pi\mathbb{Z}^3,\,\vert k\vert\leq r\}
  \end{aligned}
\end{equation}
With the notations in (\ref{partition 3D space}), we define, for $r<k_F^\sigma<R$ (once again we omit the $\sigma$ dependence of $R$ and $r$), the momentum projections
\begin{equation}\label{tools1}
  \begin{aligned}
  &H_{x,\sigma}[R](z)=\sum_{k\in P^{\sigma}_{F,R}}e^{ikx}f_{k,\sigma}(z),\quad
  L_{x,\sigma}[R](z)=\sum_{k\in A^{\sigma}_{F,R}}e^{ikx}f_{k,\sigma}(z)\\
  &S_{x,\sigma}[R](z)=\sum_{k\in \underline{A}^\sigma_{F,r}}e^{ikx}f_{k,\sigma}(z),\quad
  I_{x,\sigma}[R](z)=\sum_{k\in \underline{B}^\sigma_{F,r}}e^{ikx}f_{k,\sigma}(z)
  \end{aligned}
\end{equation}
and the corresponding particle number operators
\begin{equation}\label{tools2}
  \begin{aligned}
  \mathcal{N}_{high}[R]&\coloneqq \sum_{\sigma,k\in P_{F,R}^{\sigma}}a^*_{k,\sigma}a_{k,\sigma}
  =\int_{\Lambda}\sum_{\sigma}a^*(H_{x,\sigma}[R])a(H_{x,\sigma}[R])dx\\
  \mathcal{N}_{low}[R]&\coloneqq \sum_{\sigma,k\in A_{F,R}^{\sigma}}a^*_{k,\sigma}a_{k,\sigma}
  =\int_{\Lambda}\sum_{\sigma}a^*(L_{x,\sigma}[R])a(L_{x,\sigma}[R])dx\\
   \mathcal{N}_{surface}[r]&\coloneqq \sum_{\sigma,k\in \underline{A}_{F,r}^{\sigma}}a_{k,\sigma}a^*_{k,\sigma}
  =\int_{\Lambda}\sum_{\sigma}a(S_{x,\sigma}[R])a^*(S_{x,\sigma}[R])dx\\
  \mathcal{N}_{inner}[r]&\coloneqq \sum_{\sigma,k\in \underline{B}_{F,r}^{\sigma}}a_{k,\sigma}a^*_{k,\sigma}
  =\int_{\Lambda}\sum_{\sigma}a(I_{x,\sigma}[R])a^*(I_{x,\sigma}[R])dx
  \end{aligned}
\end{equation}
For the sake of succinctness, we may write $\mathcal{N}_h=\mathcal{N}_{high}[R]$ and so on. For $\sharp=h,l,s,i$, the $\mathcal{N}_{\sharp,\sigma}$ version of (\ref{tools2}) can be defined similarly. For some constant $\delta$ such that $R=k_F^\sigma+N^{\delta}$ or $r=k_F^\sigma-N^\delta$, we apply the short hand $\mathcal{N}_{\sharp}[\delta]=\mathcal{N}_\sharp[R]\,\text{or}\,\mathcal{N}_\sharp[r]$ to indicate the $\delta$-dependence, which plays an important role in our discussion. The same notations are also applied to (\ref{partition 3D space}) and (\ref{tools1}). We can easily check that
\begin{equation}\label{tools3}
  \begin{aligned}
    &h_{x,\sigma}=H_{x,\sigma}+L_{x,\sigma},\quad g_{x,\sigma}=S_{x,\sigma}+I_{x,\sigma}\\
    &\mathcal{N}_{ex}=\mathcal{N}_h+\mathcal{N}_l=\mathcal{N}_{s}+\mathcal{N}_i
  \end{aligned}
\end{equation}

\section{Second Quantization and A-Priori Estimates}\label{2ndquanpri}
\par Writing from (\ref{Hamiltonian}), we have
\begin{equation}\label{second quantization H_N}
  H_N=\mathcal{K}+\mathcal{V}=\sum_{k,\sigma}\vert k\vert^2a_{k,\sigma}^*a_{k,\sigma}
  +\frac{1}{2}\sum_{k,p,q,\sigma,\nu}\hat{v}_ka^*_{p-k,\sigma}a^*_{q+k,\nu}a_{q,\nu}a_{p,\sigma},
\end{equation}
where
\begin{equation}\label{define v_k}
  \hat{v}_k=\int_{\Lambda}v_a(x)\overline{f_k}(x)dx=a\int_{\mathbb{R}^3}v(x)e^{-iakx}dx.
\end{equation}
We decompose the potential part $\mathcal{V}$ as
\begin{equation}\label{split V}
  \mathcal{V}=\mathcal{V}_0+\mathcal{V}_1+\mathcal{V}_{21}+\mathcal{V}_{22}+\mathcal{V}_{23}
  +\mathcal{V}_3+\mathcal{V}_4,
\end{equation}
with
\begin{equation}\label{split V detailed}
  \begin{aligned}
  &\mathcal{V}_0=\frac{1}{2}\sum_{k,p,q,\sigma,\nu}
  \hat{v}_ka^*_{p-k,\sigma}a^*_{q+k,\nu}a_{q,\nu}a_{p,\sigma}\chi_{p,p-k\in B^{\sigma}_F}\chi_{q,q+k\in B^{\nu}_F}\\
  &\mathcal{V}_1=\sum_{k,p,q,\sigma,\nu}
  \hat{v}_k(a^*_{p-k,\sigma}a^*_{q+k,\nu}a_{q,\nu}a_{p,\sigma}+h.c.)\chi_{p-k\notin B^{\sigma}_F}\chi_{p\in B^{\sigma}_F}\chi_{q,q+k\in B^{\nu}_F}\\
  &\mathcal{V}_{21}=\frac{1}{2}\sum_{k,p,q,\sigma,\nu}
  \hat{v}_k(a^*_{p-k,\sigma}a^*_{q+k,\nu}a_{q,\nu}a_{p,\sigma}+h.c.)\chi_{p-k\notin B^{\sigma}_F}\chi_{q+k\notin B^{\nu}_F}\chi_{p\in B^{\sigma}_F}\chi_{q\in B^{\nu}_F}\\
  &\mathcal{V}_{22}=\frac{1}{2}\sum_{k,p,q,\sigma,\nu}
  \hat{v}_ka^*_{p-k,\sigma}a^*_{q+k,\nu}a_{q,\nu}a_{p,\sigma}\\
  &\quad\quad\quad\quad\quad\quad\quad\quad
  \times(\chi_{p,p-k\notin B^{\sigma}_F}\chi_{q,q+k\in B^{\nu}_F}+\chi_{p,p-k\in B^{\sigma}_F}\chi_{q,q+k\notin B^{\nu}_F})\\
  &\mathcal{V}_{23}=\frac{1}{2}\sum_{k,p,q,\sigma,\nu}
  \hat{v}_k(a^*_{p-k,\sigma}a^*_{q+k,\nu}a_{q,\nu}a_{p,\sigma}+h.c.)\chi_{p-k\notin B^{\sigma}_F}\chi_{q+k\in B^{\nu}_F}\chi_{p\in B^{\sigma}_F}\chi_{q\notin B^{\nu}_F}\\
  &\mathcal{V}_3=\sum_{k,p,q,\sigma,\nu}
  \hat{v}_k(a^*_{p-k,\sigma}a^*_{q+k,\nu}a_{q,\nu}a_{p,\sigma}+h.c.)\chi_{p-k\notin B^{\sigma}_F}\chi_{p\in B^{\sigma}_F}\chi_{q,q+k\notin B^{\nu}_F}\\
  &\mathcal{V}_4=\frac{1}{2}\sum_{k,p,q,\sigma,\nu}
  \hat{v}_ka^*_{p-k,\sigma}a^*_{q+k,\nu}a_{q,\nu}a_{p,\sigma}\chi_{p,p-k\notin B^{\sigma}_F}\chi_{q,q+k\notin B^{\nu}_F}
  \end{aligned}
\end{equation}
For the kinetic energy, we put down
\begin{equation}\label{K_s}
  {\mathcal{K}}_{s}=\sum_{\sigma,k}\vert k\vert^2a^*_{k,\sigma}a_{k,\sigma}
  -\sum_{\sigma,k\in B_F^{\sigma}}\vert k\vert^2.
\end{equation}
Using (\ref{number of exicited particles ops3}), we can rewrite $\mathcal{K}_s$ on $\mathcal{H}^{\wedge N}(\{N_{\varsigma_i}\})$ that
\begin{equation}\label{K_s re}
  \mathcal{K}_{s}=\sum_{\sigma,k\notin B_F^{\sigma}}\big(\vert k\vert^2 -\vert k_F^{\sigma}\vert^2\big) a^*_{k,\sigma}a_{k,\sigma}+\sum_{\sigma,k\in B_F^{\sigma}}\big(\vert k_F^{\sigma}\vert^2-\vert k\vert^2 \big) a_{k,\sigma}a^*_{k,\sigma}
\end{equation}
The next lemma contains fundamental estimates about $\mathcal{K}_s$.
\begin{lemma}\label{lemma ineqn K_s}\
\begin{enumerate}[label=(\arabic*)]
  \item With previous notation, we have
  \begin{equation}\label{ineqn K_s origin}
    \mathcal{N}_{ex}\lesssim\mathcal{K}_s.
  \end{equation}
   \item For $-\frac{1}{3}\leq\delta_1<\frac{1}{3}$ and $-\frac{1}{3}\leq\delta_2\leq\frac{1}{3}$, we have
   \begin{equation}\label{ineqn N_h and N_i}
     \mathcal{N}_h[\delta_2]\lesssim N^{-\frac{1}{3}-\delta_2}\mathcal{K}_s,\quad
     \mathcal{N}_i[\delta_1]\lesssim N^{-\frac{1}{3}-\delta_1}\mathcal{K}_s
   \end{equation}
   and
   \begin{equation}\label{ineqn N_l and N_s}
     \#A^\sigma_{F,\delta_2}\lesssim N^{\frac{2}{3}+\delta_2+\varepsilon(\delta_2)},\quad
     \#\underline{A}^\sigma_{F,\delta_1}\lesssim N^{\frac{2}{3}+\delta_1+\varepsilon(\delta_1)}
   \end{equation}
   for $r=k_F^\sigma-N^{\delta_1}$ and $R=k_F^\sigma+N^{\delta_2}$ and
   \begin{equation}\label{epsilon delta}
     \varepsilon(\delta)=\left\{
     \begin{aligned}
     &\text{$0$, if $\delta>-\frac{11}{48}$}\\
     &\text{an arbitrary small constant, otherwise}
     \end{aligned}\right.
   \end{equation}
    Moreover, if $\delta_2\geq\frac{1}{3}$, we have
   \begin{equation}\label{ineqn N_h delta>1/3}
     \mathcal{N}_h[\delta_2]\lesssim N^{-2\delta_2}\mathcal{K}_s,
     \quad \#A^\sigma_{F,\delta_2}\lesssim N^{3\delta_2}.
   \end{equation}
  \item For some $-\frac{1}{3}\leq\delta\leq\frac{1}{3}$, the following inequality holds:
 \begin{equation}\label{ineqn K_s}
 \mathcal{N}_{ex}\lesssim N^{-\frac{1}{3}-\delta}\mathcal{K}_s+N^{\frac{2}{3}+\delta+\varepsilon(\delta)}.
 \end{equation}
 \end{enumerate}
 \end{lemma}

\begin{proof}

\par (\ref{ineqn K_s origin}) and (\ref{ineqn N_h and N_i}) comes directly from (\ref{K_s re}), (\ref{lower bound k_F}) and (\ref{tools2}), while (\ref{ineqn N_l and N_s}) follows from (\ref{approximate lattice 3d ball}), (\ref{k_F^sigma}) and (\ref{approximate lattice 2d sphere}).

\par The proof of (\ref{ineqn K_s}) is similar to \cite[Lemma 3.6]{dilutefermiBog}. First we notice that
\begin{equation*}
  \mathcal{N}_l[\delta]=\int_{\Lambda}\sum_{\sigma}a^*(L_{x,\sigma})a(L_{x,\sigma})dx
  \leq \sum_{\sigma}\int_{\Lambda}\Vert a(L_{x,\sigma})\Vert^2dx\leq
  \sum_{\sigma}\#A^\sigma_{F,\delta}
\end{equation*}
where the last inequality comes from (\ref{norm bound  of creation and annihilation}). Using the notations defined in (\ref{tools2}), we take $R=k_F^\sigma+N^{\delta}$. By (\ref{tools3}), we have
\begin{equation*}
  \mathcal{N}_{ex}=\mathcal{N}_{l}[\delta]+\mathcal{N}_h[\delta].
\end{equation*}
By (\ref{ineqn N_h and N_i}) and (\ref{ineqn N_l and N_s}) we deduce (\ref{ineqn K_s}) immediately.
\end{proof}

\par Next we give an a-prior ground energy estimate of $E_{\{N_{\varsigma_i}\},\mathbf{q}}$.
\begin{lemma}\label{0 order approximation lemma}
\begin{equation}\label{0 order approximation}
 E_{\{N_{\varsigma_i}\},\mathbf{q}}=\sum_{\sigma,k\in B_F^{\sigma}}\vert k\vert^2+O(N).
\end{equation}
\end{lemma}

\begin{proof}
\par From (\ref{second quantization H_N}), (\ref{K_s re}) and $v\geq 0$, we know that
\begin{equation*}
  H_N-\sum_{\sigma,k\in B_F^{\sigma}}\vert k\vert^2=\mathcal{K}_s+\mathcal{V}\geq0.
\end{equation*}
On the other hand, we can test against
\begin{equation*}
  \psi=\bigwedge_{\sigma,k\in B_F^{\sigma}}f_{k,\sigma}.
\end{equation*}
that is, $\psi=f_{s_N}$, where $s_N$ is the unique element of $\mathcal{P}^{(N)}$ defined in (\ref{define P^(N)}). Then using (\ref{formulae of crea and annihi ops}), we can calculate directly
\begin{equation*}
  \langle\mathcal{K}\psi,\psi\rangle=\sum_{\sigma,k\in B_F^{\sigma}}\vert k\vert^2.
\end{equation*}
Using Lemma \ref{lemma V_0} below, we can calculate
\begin{equation*}
\begin{aligned}
  \langle\mathcal{V}\psi,\psi\rangle&=\langle\mathcal{V}_0\psi,\psi\rangle
  =\frac{\hat{v}_0}{2}\sum_{\sigma\neq\nu}N_\sigma N_\nu+o(N)=O(N),
\end{aligned}
\end{equation*}
since $\vert\hat{v}_0\vert\lesssim a$. Thence we conclude the proof.
\end{proof}

\par In order to reach our final result, we need an \textquotedblleft Excitations Estimate\textquotedblright (See Proposition \ref{Optimal BEC}). In this section, we give an a-priori excitations estimate in Lemma \ref{lemma prori BEC}. We shall refer to the states which are close enough to the ground state of $H_N$ (in the sense of energy) as approximate ground state. This definition had already arisen in, for example, \cite{dilutefermiBog}.
\begin{definition}[Approximate ground state]\label{def approximate ground state}
  Let $N_j\to\infty$ as $j\to\infty$. If for any $j$, $\psi_j\in\mathcal{H}^{\wedge N_j}(\{N_{\varsigma_i,j}\})$ satisfies $\Vert\psi_j\Vert=1$ and
  \begin{equation}\label{approximate ground state}
    \Big\vert\langle H_{N_j}\psi_j,\psi_j\rangle-\sum_{\sigma,k\in B_F^{\sigma}}\vert k\vert^2\Big\vert\lesssim N_j.
  \end{equation}
  Then we call $\{\psi_j\}$ a family of approximate ground state of order $0$. If, moreover, $\{\psi_j\}$ satisfies
  \begin{equation}\label{approximate ground state 1st order}
    \Big\vert\langle H_{N_j}\psi_j,\psi_j\rangle-\sum_{\sigma,k\in B_F^{\sigma}}\vert k\vert^2-4\pi a\mathfrak{a}_0\sum_{\sigma\neq\nu}N_\sigma N_\nu\Big\vert\lesssim N_j^{\delta}.
  \end{equation}
  for some $\delta<1$, then we call $\{\psi_j\}$ a family of approximate ground state of order $1$.
\end{definition}

That is, Lemma \ref{0 order approximation lemma} says the true ground states of $H_N$ can form a family of approximate ground state of order $0$. We will see later that the true ground state energy can be expanded in the form of (\ref{approximate ground state 1st order}). That is, the above definition is well-defined. With Definition \ref{def approximate ground state}, we can state our a-priori excitation estimate.

\begin{lemma}\label{lemma prori BEC}
Let $\{\psi_j\}$ be a family of approximate ground state of order $0$, then
\begin{equation}\label{prori BEC}
  \langle\mathcal{N}_{ex}\psi_j,\psi_j\rangle\lesssim N^{\frac{2}{3}}.
\end{equation}
\end{lemma}

\begin{proof}
  \par Since $\{\psi_j\}$ is a family of approximate ground state of order $0$, we can deduce from (\ref{approximate ground state}) and the fact that $\mathcal{K}_s,\mathcal{V}\geq0$:
  \begin{equation}\label{K_spsi_j psi_j}
    \langle\mathcal{K}_s\psi_j,\psi_j\rangle\lesssim N.
  \end{equation}
  Insert (\ref{K_spsi_j psi_j}) into (\ref{ineqn K_s}) and optimize it over $\delta$ by choosing $\delta=0$, we reach the bound (\ref{prori BEC}).
\end{proof}

\par Before we get down to the evaluations of three renormalizations, we establish some useful estimates about the terms in (\ref{split V detailed}).

\begin{lemma}\label{lemma V_0}
We can write $\mathcal{V}_0$ defined in (\ref{split V detailed}) by
\begin{equation}\label{est V_0}
  \mathcal{V}_0=\frac{\hat{v}_0}{2}\sum_{\sigma\neq\nu}(N_{\sigma}-\mathcal{N}_{ex,\sigma})
(N_{\nu}-\mathcal{N}_{ex,\nu})+\mathcal{R}_0+\mathcal{E}_{\mathcal{V}_0},
\end{equation}
where
\begin{equation}\label{R_0}
  \mathcal{R}_0=\frac{1}{2}\sum_{k,p,q,\sigma,\nu}
  \hat{v}_ka_{q,\nu}a_{p,\sigma}a^*_{p-k,\sigma}a^*_{q+k,\nu}\chi_{p,p-k\in B^{\sigma}_F}\chi_{q,q+k\in B^{\nu}_F}
\end{equation}
satisfies $\mathcal{R}_0\geq0$ and
\begin{equation}\label{est R_0}
  \pm\mathcal{R}_0\lesssim Na\mathcal{N}_{ex},
\end{equation}
and
\begin{equation}\label{est R_0 true}
  \pm\mathcal{R}_0\lesssim N^{-\frac{1}{6}}\mathcal{N}_{ex}
  +N^{\frac{1}{6}}\mathcal{N}_{i}[0].
\end{equation}
Moreover, the error can be bounded by
\begin{equation}\label{est E_V_0}
  \pm\mathcal{E}_{\mathcal{V}_0}\lesssim a \mathcal{N}_{ex}^2+N^{\frac{8}{3}}a^{3}.
\end{equation}
\end{lemma}

\begin{proof}
\par We rewrite
\begin{equation*}
  \mathcal{V}_0=\mathcal{V}_{01}+\mathcal{V}_{02}+\mathcal{V}_{03},
\end{equation*}
where
\begin{equation*}
  \begin{aligned}
  &\mathcal{V}_{01}=\frac{1}{2}\sum_{k,p,q,\sigma,\nu}
  \hat{v}_ka^*_{p-k,\sigma}a^*_{q+k,\nu}a_{q,\nu}a_{p,\sigma}\chi_{p,p-k\in B^{\sigma}_F}\chi_{q,q+k\in B^{\nu}_F}\chi_{k=0}\\
  &\mathcal{V}_{02}=\frac{1}{2}\sum_{k,p,q,\sigma,\nu}
  \hat{v}_ka^*_{p-k,\sigma}a^*_{q+k,\nu}a_{q,\nu}a_{p,\sigma}\chi_{p,p-k\in B^{\sigma}_F}\chi_{q,q+k\in B^{\nu}_F}\chi_{k\neq0}\chi_{\sigma=\nu}\chi_{p=q+k}\\
  &\mathcal{V}_{03}=\frac{1}{2}\sum_{k,p,q,\sigma,\nu}
  \hat{v}_ka^*_{p-k,\sigma}a^*_{q+k,\nu}a_{q,\nu}a_{p,\sigma}\chi_{p,p-k\in B^{\sigma}_F}\chi_{q,q+k\in B^{\nu}_F}\chi_{k\neq0}\\
  &\quad\quad\quad\quad\quad\quad\quad\quad
  \times(1-\chi_{\sigma=\nu}\chi_{p=q+k})
  \end{aligned}
\end{equation*}

\par Using (\ref{anticommutator}), we have
\begin{equation*}
  a^*_{p,\sigma}a^*_{q,\nu}a_{q,\nu}a_{p,\sigma}=
  a^*_{p,\sigma}a_{p,\sigma}a^*_{q,\nu}a_{q,\nu}-\delta_{p,q}\delta_{\sigma,\nu}
  a^*_{p,\sigma}a_{p,\sigma},
\end{equation*}
which leads to
\begin{equation}\label{cal V_01}
  \begin{aligned}
  \mathcal{V}_{01}&=\frac{1}{2}\sum_{p,q,\sigma,\nu}\hat{v}_0a^*_{p,\sigma}a^*_{q,\nu}
  a_{q,\nu}a_{p,\sigma}\chi_{p\in B^{\sigma}_F}\chi_{q\in B^{\nu}_F}\\
  &=\frac{\hat{v}_0}{2}\sum_{\sigma,\nu}(N_{\sigma}-\mathcal{N}_{ex,\sigma})
  (N_{\nu}-\delta_{\sigma,\nu}-\mathcal{N}_{ex,\nu})\\
  &=\frac{\hat{v}_0}{2}\sum_{\sigma\neq\nu}(N_{\sigma}-\mathcal{N}_{ex,\sigma})
  (N_{\nu}-\mathcal{N}_{ex,\nu})
  +\frac{\hat{v}_0}{2}\sum_{\sigma}(N_{\sigma}-\mathcal{N}_{ex,\sigma})
  (N_{\nu}-1-\mathcal{N}_{ex,\nu}).
  \end{aligned}
\end{equation}

\par For $\mathcal{V}_{02}$, by a change of variables, we have
\begin{equation}\label{cal V_02 1}
\begin{aligned}
  \mathcal{V}_{02}&=\frac{1}{2}\sum_{k,p,\sigma}
  \hat{v}_ka^*_{p-k,\sigma}a^*_{p,\sigma}a_{p-k,\sigma}a_{p,\sigma}\chi_{p,p-k\in B^{\sigma}_F}\chi_{k\neq0}\\
  &=\frac{1}{2}\sum_{p,q,\sigma}
  \hat{v}_{p-q}a^*_{q,\sigma}a^*_{p,\sigma}a_{q,\sigma}a_{p,\sigma}\chi_{p,q\in B^{\sigma}_F}\chi_{p-q\neq0}.
\end{aligned}
\end{equation}
Notice from (\ref{define v_k}) and the fact that $v(-x)=v(x)$, we have, for some $\theta\in(0,1)$, that
\begin{equation}\label{v_k-v_0}
  \begin{aligned}
  \big\vert\hat{v}_k-\hat{v}_0\big\vert
  &=\Big\vert a\int_{\mathbb{R}^3}v(x)\big(e^{-iakx}-1\big)dx\Big\vert\\
  &=\Big\vert a\int_{\mathbb{R}^3}v(x)\big(-iak\cdot x-
  \theta^2a^2e^{-iak\cdot\theta x} k\otimes k(y,y) \big)dx\Big\vert\\
  &\leq a^3\vert k\vert^2\int_{\mathbb{R}^3}\vert x\vert^2v(x)dx\lesssim a^3\vert k\vert^2.
  \end{aligned}
\end{equation}
Then for $\psi\in\mathcal{H}^{\wedge N}(\{N_{\varsigma_i}\})$,
\begin{equation}\label{cal V_02 2}
\begin{aligned}
  &\Big\vert\sum_{p,q,\sigma}
  \big(\hat{v}_{p-q}-\hat{v}_0\big)\chi_{p,q\in B^{\sigma}_F}\chi_{p-q\neq0}
  \langle a^*_{q,\sigma}a^*_{p,\sigma}a_{q,\sigma}a_{p,\sigma}\psi,\psi\rangle\Big\vert\\
  &\lesssim \sum_{p,q,\sigma}
  a^3\vert p-q\vert^2\chi_{p,q\in B^{\sigma}_F}\chi_{p-q\neq0}
  \Vert a_{q,\sigma}a_{p,\sigma}\psi\Vert^2
  \lesssim N^{\frac{8}{3}}a^3\Vert\psi\Vert^2,
\end{aligned}
\end{equation}
where we have used (\ref{norm bound  of creation and annihilation}) and (\ref{k_F^sigma}) in the last inequality. On the other hand, since $a_{p,\sigma}a_{q,\sigma}=0$ when $p=q$, we have
\begin{equation}\label{cal V_02 3}
\begin{aligned}
  \frac{1}{2}\sum_{p,q,\sigma}
  \hat{v}_{0}a^*_{q,\sigma}a^*_{p,\sigma}a_{q,\sigma}a_{p,\sigma}\chi_{p,q\in B^{\sigma}_F}\chi_{p-q\neq0}
  &=\frac{1}{2}\sum_{p,q,\sigma}
  \hat{v}_{0}a^*_{q,\sigma}a^*_{p,\sigma}a_{q,\sigma}a_{p,\sigma}\chi_{p,q\in B^{\sigma}_F}\\
  &=\frac{1}{2}\sum_{p,q,\sigma}
  \hat{v}_{0}a^*_{q,\sigma}(\delta_{p,q}-a_{q,\sigma}a^*_{p,\sigma})a_{p,\sigma}\chi_{p,q\in B^{\sigma}_F}\\
  &=-\frac{\hat{v}_0}{2}\sum_{\sigma}(N_{\sigma}-\mathcal{N}_{ex,\sigma})
  (N_{\sigma}-1-\mathcal{N}_{ex,\sigma}).
\end{aligned}
\end{equation}
Combing (\ref{cal V_02 1}), (\ref{cal V_02 2}) and (\ref{cal V_02 3}), we have
\begin{equation}\label{cal V_02}
  \mathcal{V}_{02}=-\frac{\hat{v}_0}{2}\sum_{\sigma}(N_{\sigma}-\mathcal{N}_{ex,\sigma})
  (N_{\sigma}-1-\mathcal{N}_{ex,\sigma})+\mathcal{E}_{\mathcal{V}_{02}},
\end{equation}
where
\begin{equation}\label{est E_V_02}
  \pm\mathcal{E}_{\mathcal{V}_{02}}\lesssim N^{\frac{8}{3}}a^3.
\end{equation}

\par For $\mathcal{V}_{03}$, we first rewrite it using (\ref{anticommutator})
\begin{equation}\label{cal V_03 1}
\begin{aligned}
  &\mathcal{V}_{03}=\frac{1}{2}\sum_{k,p,q,\sigma,\nu}
  \hat{v}_ka_{q,\nu}a_{p,\sigma}a^*_{p-k,\sigma}a^*_{q+k,\nu}\chi_{p,p-k\in B^{\sigma}_F}\chi_{q,q+k\in B^{\nu}_F}\chi_{k\neq0}\\
  &\quad\quad\quad\quad\quad\quad\quad\quad
  \times(1-\chi_{\sigma=\nu}\chi_{p=q+k}).
\end{aligned}
\end{equation}
Notice that $1=\chi_{k\neq0}(1-\chi_{\sigma=\nu}\chi_{p=q+k})+\chi_{k=0}
  +\chi_{k\neq0}\chi_{\sigma=\nu}\chi_{p=q+k}$. We use (\ref{number of exicited particles ops2}) and (\ref{number of exicited particles ops3}) to calculate
\begin{equation}\label{cal V_03 2}
  \frac{1}{2}\sum_{p,q,\sigma,\nu}
  \hat{v}_0a_{q,\nu}a_{p,\sigma}a^*_{p,\sigma}a^*_{q,\nu}\chi_{p\in B^{\sigma}_F}
  \chi_{q\in B^{\nu}_F}=\frac{\hat{v}_0}{2}\mathcal{N}_{ex}^2,
\end{equation}
and
\begin{equation}\label{cal V_03 3}
\begin{aligned}
  &\frac{1}{2}\sum_{k,p,\sigma}
  \hat{v}_0a_{p-k,\sigma}a_{p,\sigma}a^*_{p-k,\sigma}a^*_{p,\sigma}\chi_{p,p-k\in B^{\sigma}_F}
  \chi_{k\neq0}\\
  =&-\frac{1}{2}\sum_{k,p,\sigma}
  \hat{v}_0a_{p,\sigma}a_{p-k,\sigma}a^*_{p-k,\sigma}a^*_{p,\sigma}\chi_{p,p-k\in B^{\sigma}_F}
  \chi_{k\neq0}\\
  =&-\frac{1}{2}\sum_{k,p,\sigma}
  \hat{v}_0a_{p,\sigma}a_{p-k,\sigma}a^*_{p-k,\sigma}a^*_{p,\sigma}\chi_{p,p-k\in B^{\sigma}_F}\\
  =&-\frac{\hat{v}_0}{2}\sum_{\sigma}\mathcal{N}_{ex,\sigma}^2.
\end{aligned}
\end{equation}
On the other hand, using (\ref{norm bound  of creation and annihilation}) and (\ref{v_k-v_0}), we can bound
\begin{equation}\label{cal V_03 4}
\begin{aligned}
  \pm\frac{1}{2}\sum_{k,p,\sigma}
  (\hat{v}_k-\hat{v}_0)a_{p-k,\sigma}a_{p,\sigma}a^*_{p-k,\sigma}a^*_{p,\sigma}\chi_{p,p-k\in B^{\sigma}_F}\chi_{k\neq0}\lesssim N^{\frac{8}{3}}a^{3}.
\end{aligned}
\end{equation}
Combining (\ref{cal V_03 1}), (\ref{cal V_03 2}), (\ref{cal V_03 3}) and (\ref{cal V_03 4}) we obtain
\begin{equation}\label{cal V_03 5}
  \mathcal{V}_{03}=\frac{1}{2}\sum_{k,p,q,\sigma,\nu}
  \hat{v}_ka_{q,\nu}a_{p,\sigma}a^*_{p-k,\sigma}a^*_{q+k,\nu}\chi_{p,p-k\in B^{\sigma}_F}\chi_{q,q+k\in B^{\nu}_F}+\mathcal{E}_{\mathcal{V}_{03}}.
\end{equation}
where
\begin{equation}\label{cal V_03 6}
  \pm\mathcal{E}_{\mathcal{V}_{03}}\lesssim a \mathcal{N}_{ex}^2+N^{\frac{8}{3}}a^{3}.
\end{equation}
Recall the definition of $g_{x,\sigma}$ and $h_{x,\sigma}$ (\ref{g_x,sigma h_x,sigma}), we rewrite (\ref{cal V_03 5}) away from $\mathcal{E}_{\mathcal{V}_{03}}$ by
\begin{equation}\label{cal V 03 7}
\begin{aligned}\mathcal{R}_0\coloneqq
  &\frac{1}{2}\sum_{k,p,q,\sigma,\nu}
  \hat{v}_ka_{q,\nu}a_{p,\sigma}a^*_{p-k,\sigma}a^*_{q+k,\nu}\chi_{p,p-k\in B^{\sigma}_F}\chi_{q,q+k\in B^{\nu}_F}\\
  =&\frac{1}{2}\sum_{p_1,p_2,q_1,q_2,\sigma,\nu}
  \hat{v}_{p_1-q_1}a_{p_2,\nu}a_{p_1,\sigma}a^*_{q_1,\sigma}a^*_{q_2,\nu}\chi_{p_1,q_1\in B^{\sigma}_F}\chi_{p_2,q_2\in B^{\nu}_F}\delta_{p_1-q_1,q_2-p_2}\\
  =&\frac{1}{2}\sum_{p_1,p_2,q_1,q_2,\sigma,\nu}
  \int_{\Lambda^2}v_a(x)e^{-i(p_1-q_1)x}e^{-i(p_1-q_1-q_2+p_2)y}
  a_{p_2,\nu}a_{p_1,\sigma}a^*_{q_1,\sigma}a^*_{q_2,\nu}\\
  &\quad\quad\quad\quad\quad\quad\quad\quad
  \times\chi_{p_1,q_1\in B^{\sigma}_F}\chi_{p_2,q_2\in B^{\nu}_F}dxdy\\
  =&\frac{1}{2}\sum_{\sigma,\nu}
  \int_{\Lambda^2}v_a(x-y)
  a(g_{y,\nu})a(g_{x,\sigma})a^*(g_{x,\sigma})a^*(g_{y,\nu})dxdy.
  \end{aligned}
\end{equation}
Then (\ref{est V_0}), (\ref{R_0}) and (\ref{est E_V_0}) in Lemma \ref{lemma V_0} follow by combining (\ref{cal V_01}), (\ref{cal V_02}), (\ref{est E_V_02}), (\ref{cal V_03 5}-\ref{cal V 03 7}).
\par Moreover, we know that for any $\psi\in\mathcal{H}^{\wedge N}(\{N_{\varsigma_i}\})$,
\begin{equation}\label{cal V_03 8}
  \begin{aligned}
  \langle\mathcal{R}_0\psi,\psi\rangle&=\frac{1}{2}\sum_{\sigma,\nu}
  \int_{\Lambda^2}v_a(x-y)
  \langle a(g_{y,\nu})a(g_{x,\sigma})a^*(g_{x,\sigma})a^*(g_{y,\nu})\psi,\psi\rangle dxdy\\
  &=\frac{1}{2}\sum_{\sigma,\nu}
  \int_{\Lambda^2}v_a(x-y)
   \Vert a^*(g_{x,\sigma})a^*(g_{y,\nu})\psi\Vert^2 dxdy,
  \end{aligned}
\end{equation}
which immediately yields $\mathcal{R}_0\geq0$. On the other hand, we can bound $\mathcal{R}_0$ using (\ref{norm bound  of creation and annihilation}) and (\ref{useful eqn1}),
\begin{equation}\label{cal V_03 9}
\begin{aligned}
   \vert\langle\mathcal{R}_0\psi,\psi\rangle\vert
   &\leq \sum_{\sigma,\nu}
  \int_{\Lambda^2}v_a(x-y)
   \Vert g_{y,\nu}\Vert \Vert g_{x,\sigma}\Vert \Vert a^*(g_{y,\nu})\psi\Vert\cdot\Vert a^*(g_{x,\sigma})\psi\Vert dxdy\\
  &\lesssim N\Vert v_a\Vert \langle\mathcal{N}_{ex}\psi,\psi\rangle
  \lesssim Na \langle\mathcal{N}_{ex}\psi,\psi\rangle,
  \end{aligned}
\end{equation}
which is (\ref{est R_0}). On the other hand, we decompose $\mathcal{R}_0=\mathcal{R}_{0,1}+\mathcal{R}_{0,2}$ with
\begin{equation}\label{rewrt R_0}
  \begin{aligned}
  \mathcal{R}_{0,1}&=\frac{1}{2}\sum_{\sigma,\nu}
  \int_{\Lambda^2}v_a(x-y)
  a(g_{y,\nu})a(g_{x,\sigma})a^*(S_{x,\sigma}[0])a^*(g_{y,\nu}) dxdy\\
  \mathcal{R}_{0,2}&=\frac{1}{2}\sum_{\sigma,\nu}
  \int_{\Lambda^2}v_a(x-y)
  a(g_{y,\nu})a(g_{x,\sigma})a^*(I_{x,\sigma}[0])a^*(g_{y,\nu}) dxdy
  \end{aligned}
\end{equation}
We bound
\begin{equation}\label{bound R_0,1}
  \begin{aligned}
   \vert\langle\mathcal{R}_{0,1}\psi,\psi\rangle\vert
   &\leq \sum_{\sigma,\nu}
  \int_{\Lambda^2}v_a(x-y)
   \Vert g_{y,\nu}\Vert \Vert S_{x,\sigma}[0]\Vert \Vert a^*(g_{y,\nu})\psi\Vert\cdot\Vert a^*(g_{x,\sigma})\psi\Vert dxdy\\
  &\lesssim N^{\frac{5}{6}}\Vert v_a\Vert \langle\mathcal{N}_{ex}\psi,\psi\rangle
  \lesssim N^{-\frac{1}{6}} \langle\mathcal{N}_{ex}\psi,\psi\rangle,
  \end{aligned}
\end{equation}
and
\begin{equation}\label{bound R_0,2}
  \begin{aligned}
   \vert\langle\mathcal{R}_{0,2}\psi,\psi\rangle\vert
   &\leq \sum_{\sigma,\nu}
  \int_{\Lambda^2}v_a(x-y)
   \Vert g_{y,\nu}\Vert \Vert g_{x,\sigma}\Vert \Vert a^*(g_{y,\nu})\psi\Vert\cdot\Vert a^*(I_{x,\sigma})\psi\Vert dxdy\\
  &\lesssim N\Vert v_a\Vert \langle\mathcal{N}_{ex}\psi,\psi\rangle^\frac{1}{2}
  \langle\mathcal{N}_{i}[0]\psi,\psi\rangle^\frac{1}{2}\\
  &\lesssim N^{-\frac{1}{6}} \langle\mathcal{N}_{ex}\psi,\psi\rangle
  +N^{\frac{1}{6}}\langle\mathcal{N}_{i}[0]\psi,\psi\rangle.
  \end{aligned}
\end{equation}
(\ref{bound R_0,1}) and (\ref{bound R_0,2}) give (\ref{est R_0 true}).
\end{proof}

\begin{lemma}\label{lemma V_22,23,1}
We can directly bound $\mathcal{V}_{1}$ in (\ref{split V detailed}) by
\begin{equation}\label{est V_1}
  \pm\mathcal{V}_{1}\lesssim N^{-\frac{1}{6}}\mathcal{N}_{ex}
  +N^{\frac{1}{6}}\mathcal{N}_{h}[0],
\end{equation}
and rewrite $\mathcal{V}_{22}$ and $\mathcal{V}_{23}$ in (\ref{split V detailed}) by
\begin{equation}\label{est V_22}
  \mathcal{V}_{22}=\sum_{\sigma,\nu}\hat{v}_0N_\nu\mathcal{N}_{ex,\sigma}
  +\mathcal{E}_{\mathcal{V}_{22}},
\end{equation}
\begin{equation}\label{est V_23}
  \mathcal{V}_{23}=-\sum_{\sigma}\hat{v}_0N_{\sigma}\mathcal{N}_{ex,\sigma}
  +\mathcal{E}_{\mathcal{V}_{23}},
\end{equation}
with
\begin{equation}\label{est E_V_22}
  \pm\mathcal{E}_{\mathcal{V}_{22}}\lesssim N^{-\frac{1}{6}}\mathcal{N}_{ex}
  +N^{\frac{1}{6}}\mathcal{N}_{h}[0].
\end{equation}
\begin{equation}\label{est E_V_23}
  \pm\mathcal{E}_{\mathcal{V}_{23}}\lesssim N^{-\frac{1}{6}}\mathcal{N}_{ex}
  +N^{\frac{1}{6}}\mathcal{N}_{h}[0]+N^{-2}\mathcal{K}_s.
\end{equation}
\end{lemma}

\begin{proof}

  \par Using (\ref{g_x,sigma h_x,sigma}), we rewrite the three operators by
\begin{equation}\label{rewrite V_1,22,23}
  \begin{aligned}
  &\mathcal{V}_{1}=\frac{1}{2}\sum_{\sigma,\nu}
  \int_{\Lambda^2}v_a(x-y)
  a^*(h_{x,\sigma})a(g_{y,\nu})a(g_{x,\sigma})a^*(g_{y,\nu})dxdy+h.c.\\
  &\mathcal{V}_{22}=\sum_{\sigma,\nu}
  \int_{\Lambda^2}v_a(x-y)
  a^*(h_{x,\sigma})a^*(g_{y,\nu})a(g_{y,\nu})a(h_{x,\sigma})dxdy\\
  &\mathcal{V}_{23}=-\sum_{\sigma,\nu}
  \int_{\Lambda^2}v_a(x-y)
  a^*(h_{x,\sigma})a^*(g_{y,\nu})a(g_{x,\sigma})a(h_{y,\nu})dxdy
  \end{aligned}
\end{equation}
We decompose $\mathcal{V}_1=\mathcal{V}_{1,1}+\mathcal{V}_{1,2}$ with
\begin{equation}\label{rewrt V_1}
  \begin{aligned}
  \mathcal{V}_{1,1}&=\frac{1}{2}\sum_{\sigma,\nu}
  \int_{\Lambda^2}v_a(x-y)
  a^*(H_{x,\sigma}[0])a(g_{y,\nu})a(g_{x,\sigma})a^*(g_{y,\nu})dxdy+h.c.\\
  \mathcal{V}_{1,2}&=\frac{1}{2}\sum_{\sigma,\nu}
  \int_{\Lambda^2}v_a(x-y)
  a^*(L_{x,\sigma}[0])a(g_{y,\nu})a(g_{x,\sigma})a^*(g_{y,\nu})dxdy+h.c.
  \end{aligned}
\end{equation}
By Lemma \ref{lemma ineqn K_s} and standard Cauchy-Schwarz, we bound
\begin{equation}\label{bound V_1,1 V_1,2}
  \begin{aligned}
  \pm\mathcal{V}_{1,1}&\lesssim N^{-\frac{1}{6}}\mathcal{N}_{ex}
  +N^{\frac{1}{6}}\mathcal{N}_{h}[0]\\
  \pm\mathcal{V}_{1,2}&\lesssim N^{-\frac{1}{6}}\mathcal{N}_{ex}
  \end{aligned}
\end{equation}
which yield (\ref{est V_1}).

\par From (\ref{anticommutator}) we know that
\begin{equation*}
  a(g_{y,\nu})a^*(g_{y,\nu})+a^*(g_{y,\nu})a(g_{y,\nu})=\Vert g_{y,\nu}\Vert^2=N_\nu,
\end{equation*}
therefore we rewrite $\mathcal{V}_{22}$ by
\begin{equation}\label{rewrt V_22}
  \mathcal{V}_{22}=\sum_{\sigma,\nu}\hat{v}_0N_\nu\mathcal{N}_{ex,\sigma}
  +\mathcal{V}_{22,1}+\mathcal{V}_{22,2},
\end{equation}
with
\begin{equation}\label{rewrt V_22 detailed}
  \begin{aligned}
  \mathcal{V}_{22,1}&=-\sum_{\sigma,\nu}
  \int_{\Lambda^2}v_a(x-y)
  a^*(h_{x,\sigma})a(g_{y,\nu})a^*(g_{y,\nu})a(H_{x,\sigma}[0])dxdy\\
  \mathcal{V}_{22,2}&=-\sum_{\sigma,\nu}
  \int_{\Lambda^2}v_a(x-y)
  a^*(h_{x,\sigma})a(g_{y,\nu})a^*(g_{y,\nu})a(L_{x,\sigma}[0])dxdy
  \end{aligned}
\end{equation}
By Lemma \ref{lemma ineqn K_s} and standard Cauchy-Schwarz, we bound
\begin{equation}\label{bound V_22,1 V_22,2}
  \begin{aligned}
  \pm\mathcal{V}_{22,1}&\lesssim N^{-\frac{1}{6}}\mathcal{N}_{ex}
  +N^{\frac{1}{6}}\mathcal{N}_{h}[0]\\
  \pm\mathcal{V}_{22,2}&\lesssim N^{-\frac{1}{6}}\mathcal{N}_{ex}
  \end{aligned}
\end{equation}
which yield (\ref{est V_22}) and (\ref{est E_V_22}).

\par For $\mathcal{V}_{23}$, we switch back to the momentum space and write it as $\mathcal{V}_{23}=\mathcal{V}_{23,1}+\mathcal{V}_{23,2}$ with
\begin{equation}\label{rewrt V_23 1st step}
  \begin{aligned}
  \mathcal{V}_{23,1}&=-\frac{1}{2}\sum_{k,p,q,\sigma,\nu}
  \hat{v}_k\delta_{p,q+k}\delta_{\sigma,\nu}
  (a^*_{p-k,\sigma}a_{q,\nu}+h.c.)\chi_{p-k\notin B^{\sigma}_F}\chi_{q+k\in B^{\nu}_F}\chi_{p\in B^{\sigma}_F}\chi_{q\notin B^{\nu}_F}\\
  \mathcal{V}_{23,2}&=\frac{1}{2}\sum_{k,p,q,\sigma,\nu}
  \hat{v}_k(a^*_{p-k,\sigma}a_{p,\sigma}a^*_{q+k,\nu}a_{q,\nu}+h.c.)\chi_{p-k\notin B^{\sigma}_F}\chi_{q+k\in B^{\nu}_F}\chi_{p\in B^{\sigma}_F}\chi_{q\notin B^{\nu}_F}
  \end{aligned}
\end{equation}
The second term can be bounded on position space. We put it as $\mathcal{V}_{23,2}=\mathcal{V}_{23,21}+\mathcal{V}_{23,22}$ with
\begin{equation}\label{rewrt V_23,2}
\begin{aligned}
  \mathcal{V}_{23,21}&=\sum_{\sigma,\nu}
  \int_{\Lambda^2}v_a(x-y)
  a^*(h_{x,\sigma})a(g_{x,\sigma})a^*(g_{y,\nu})a(H_{y,\nu}[0])dxdy\\
  \mathcal{V}_{23,22}&=\sum_{\sigma,\nu}
  \int_{\Lambda^2}v_a(x-y)
  a^*(h_{x,\sigma})a(g_{x,\sigma})a^*(g_{y,\nu})a(L_{y,\nu}[0])dxdy
\end{aligned}
\end{equation}
By Lemma \ref{lemma ineqn K_s} and standard Cauchy-Schwarz, we bound
\begin{equation}\label{bound V_23,21 V_23,22}
  \begin{aligned}
  \pm\mathcal{V}_{23,21}&\lesssim N^{-\frac{1}{6}}\mathcal{N}_{ex}
  +N^{\frac{1}{6}}\mathcal{N}_{h}[0]\\
  \pm\mathcal{V}_{23,22}&\lesssim N^{-\frac{1}{6}}\mathcal{N}_{ex}
  \end{aligned}
\end{equation}
For the first term in (\ref{rewrt V_23 1st step}), we notice that
\begin{equation}\label{rewrt V_23,1}
\begin{aligned}
  \mathcal{V}_{23,1}&=-\sum_{p,q,\sigma}
  \hat{v}_{p-q}a^*_{q,\sigma}a_{q,\sigma}\chi_{p\in B^{\sigma}_F}\chi_{q\notin B^{\sigma}_F}\\
  &=-\sum_{p,q,\sigma}
  (\hat{v}_{0}+\hat{v}_{p-q}-\hat{v}_{0})a^*_{q,\sigma}a_{q,\sigma}\chi_{p\in B^{\sigma}_F}\chi_{q\notin B^{\sigma}_F}\\
  &=-\sum_{\sigma}\hat{v}_0N_{\sigma}\mathcal{N}_{ex,\sigma}+\sum_{p,q,\sigma}
  (\hat{v}_{p-q}-\hat{v}_{0})a^*_{q,\sigma}a_{q,\sigma}\chi_{p\in B^{\sigma}_F}\chi_{q\notin B^{\sigma}_F}
\end{aligned}
\end{equation}
Using (\ref{v_k-v_0}), we bound, for $\psi\in\mathcal{H}^{\wedge N}(\{N_{\varsigma_i}\})$
\begin{equation}\label{bound V_23,1}
  \begin{aligned}
  &\Big\vert\sum_{p,q,\sigma}
  (\hat{v}_{p-q}-\hat{v}_{0})\chi_{p\in B^{\sigma}_F}\chi_{q\notin B^{\sigma}_F}
  \langle a^*_{q,\sigma}a_{q,\sigma}\psi,\psi\rangle\Big\vert\\
  &\lesssim \sum_{p,q,\sigma}a^3\vert p-q\vert^2\chi_{p\in B^{\sigma}_F}\chi_{q\notin B^{\sigma}_F}
  \Vert a_{q,\sigma}\psi\Vert^2\\
  &\lesssim \sum_{p,q,\sigma}a^3\big(\vert q\vert^2-\vert k_F^\sigma\vert^2+\vert k_F^\sigma\vert^2+\vert p\vert^2\big)\chi_{p\in B^{\sigma}_F}\chi_{q\notin B^{\sigma}_F}
  \Vert a_{q,\sigma}\psi\Vert^2\\
  &\lesssim N^{-2}\langle\mathcal{K}_s\psi,\psi\rangle+
  N^{-\frac{4}{3}}\langle\mathcal{N}_{ex}\psi,\psi\rangle.
  \end{aligned}
\end{equation}
We conclude (\ref{est V_23}) and (\ref{est E_V_23}) by collecting (\ref{rewrt V_23 1st step}-\ref{bound V_23,1}).
\end{proof}

\section{Coefficients of Renormalizations}\label{coeff}
\subsection{Scattering Equation}\label{scattering eqn sec}
\par We first consider the following ground state energy equation with Neumann boundary
condition for some parameter $\ell\in(0,\frac{1}{2})$
\begin{equation}\label{asymptotic energy pde on the ball}
  \left\{\begin{aligned}
  &(-\Delta_{x}+\frac{1}{2}v)f_\ell=\lambda_\ell f_\ell,\quad \vert
  x\vert\leq \frac{\ell}{a},\\
  &\left.\frac{\partial f_\ell}{\partial \mathbf{n}}\right\vert_{\vert
   x\vert=\frac{\ell}{a}}=0,\quad \left.f_\ell\right\vert_{\vert
   x\vert=\frac{\ell}{a}}=1.
  \end{aligned}\right.
\end{equation}
Equation (\ref{asymptotic energy pde on the ball}) has been thoroughly
analysed, and one can consult
\cite{dy2006,2018Bogoliubov,boccatoBrenCena2020optimal} for details. We define
$w_{\ell}=1-f_\ell$, and make constant extensions to both of $f_\ell$
and $w_\ell$ outside of the 3D closed ball $\overline{B}_{\frac{\ell}{a}}$ such
that $f_\ell\in H^2_{loc}(\mathbb{R}^3)$ and $w_\ell\in H^2(\mathbb{R}^3)$. By
scaling we let
\begin{equation}\label{scaling f_l &w_l}
  \widetilde{f}_\ell(x)=f_\ell\Big(\frac{x}{a}\Big),\quad
  \widetilde{w}_\ell(x)=w_\ell\Big(\frac{x}{a}\Big).
\end{equation}
Regarding $\widetilde{w}_\ell$ as a periodic function on the torus $\Lambda$, we
observe that it satisfies the equation
\begin{equation}\label{asymptotic energy pde on the torus}
  \Big(-\Delta_{x}+\frac{1}{2a^2}v\big(\frac{x}{a}\big)\Big)
 \widetilde{w}_\ell(x)
  =\frac{1}{2a^2}v\big(\frac{x}{a}\big)
-\frac{\lambda_\ell}{a^2}\big(1-\widetilde{w}_\ell(x)\big)
  \chi_{\ell}(x),\quad x\in\Lambda.
\end{equation}
Here $\chi_{\ell}$ is the characteristic function of the closed 3D ball
$\overline{B}_{\ell}$, and we choose suitable $\ell\in(0,\frac{1}{2})$ so that
$\overline{B}_{\ell}\subset\Lambda_d$. Standard elliptic equation theory grants the
uniqueness of solution to equation (\ref{asymptotic energy pde on the torus}). By Fourier
transform, (\ref{asymptotic energy pde on the torus}) is equivalent to its discrete version
 \begin{equation}\label{discrete asymptotic energy pde on the torus}
\begin{aligned}
\left\vert p\right\vert^2\widetilde{w}_{\ell,p}+\frac{1}{2}
  \sum_{q\in2\pi\mathbb{Z}^3}\hat{v}_{p-q}\widetilde{w}_{\ell,q}
=\frac{1}{2}
  \hat{v}_p+\frac{\lambda_\ell}{a^2}
  \widetilde{w}_{\ell,p}-\frac{\lambda_\ell}{a^2}\widehat{\chi_{\ell}}
  \left(\frac{p}{2\pi}\right),
\end{aligned}
\end{equation}
where $p$ is an arbitrary 3D vector in $2\pi\mathbb{Z}^3$ and the Fourier coefficients are
given by
\begin{equation*}
  \widetilde{w}_{\ell,p}=\int_{\Lambda}\widetilde{w}_{\ell}(x)
  e^{-ipx}dx,\quad
\widehat{\chi_{\ell}}
  \left(\frac{p}{2\pi}\right)
=\int_{\Lambda}\chi_{\ell}(x)
   e^{-ipx}dx.
\end{equation*}
The required properties of $f_\ell$ and $w_\ell$ are collected in the next lemma. For detailed proof, one can consult, for example, \cite[Lemma 3.1]{me}.
\begin{lemma}\label{fundamental est of v,w,lambda}
Let $v$ be a smooth interaction potential with scattering length $\mathfrak{a}_0$. Recall that
an interaction potential should be a radially-symmetric, compactly supported and non-negative
function. Let $f_\ell$, $\lambda_\ell$, $w_\ell$ and
$\widetilde{w}_{\ell,p}$ be defined as above. Then for parameter $\ell\in
(0,\frac{1}{2})$ satisfying $\frac{a}{\ell}<C$ for a small constant C independent of
$a$ and $\ell$, there exist some constants, also denoted as C, independent of $a$, and $\ell$ such that following estimates hold true for $\frac{a}{\ell}$ small
enough.
\begin{enumerate}[label=(\arabic*)]
  \item The asymptotic estimate of ground state energy $\lambda_\ell$ is
  \begin{equation}\label{est of lambda_l}
          \left\vert\lambda_\ell-\frac{3\mathfrak{a}_0a^3}{\ell^3}
          \left(1+\frac{9}{5}\frac{\mathfrak{a}_0a}{\ell}\right)\right\vert\leq
          \frac{C\mathfrak{a}_0^3a^5}{\ell^5}.
  \end{equation}
  \item $f_\ell$ is radially symmetric and smooth away from the boundary of
      $B_{\frac{\ell}{a}}$ and there is a certain constant $0<c<1$ independent of
      $a$ and $\ell$ such that
  \begin{equation}\label{est of f_l}
    0<c\leq f_\ell(x)\leq1.
  \end{equation}
  Moreover, for any integer $0\leq k\leq3$
  \begin{equation}\label{est of w_l and grad w_l}
        \vert D_{x}^kw_\ell(x)\vert\leq\frac{C}{1+\vert
         x\vert^{k+1}}.
  \end{equation}
  \item We have
  \begin{equation}\label{est of int vf_l}
          \left\vert\int_{\mathbb{R}^3}v(x)f_\ell(x)dx
          -8\pi\mathfrak{a}_0\left(1+\frac{3}{2}
          \frac{\mathfrak{a}_0a}{\ell}\right)\right\vert\leq
          \frac{C\mathfrak{a}_0^3a^2}{\ell^2},
  \end{equation}
  and
        \begin{equation}\label{est of int w_l}
          \left\vert\frac{a^2}{\ell^2}\int_{\mathbb{R}^3}
          w_\ell(x)dx-\frac{2}{5}\pi\mathfrak{a}_0\right\vert\leq
          \frac{C\mathfrak{a}_0^2a}{\ell}.
        \end{equation}
  \item For all $p\in2\pi\mathbb{Z}^3\backslash\{0\}$
  \begin{equation}\label{est of w_l,p}
          \vert\widetilde{w}_{\ell,p}\vert\leq\frac{Ca}{\vert p\vert^2}
  \end{equation}
\end{enumerate}
\end{lemma}
\begin{remark}\label{remark 3d scat eqn}
The construction of $w_\ell$ can not ensure smoothness on the boundary of
$B_{\frac{\ell}{a}}$, but we still use the notation $D_{x}^kw_\ell$ to
represent the $k$-th derivative of $w_\ell$ away from the boundary of
$B_{\frac{\ell}{a}}$. Moreover, since $w_\ell$ is supported on
$B_{\frac{\ell}{a}}$, the integrals concerning $D^k_{x}w_\ell$ always
mean integrating inside of $B_{\frac{\ell}{a}}$ unless otherwise specified.
\end{remark}

\par With Lemma \ref{fundamental est of v,w,lambda}, we thereafter define for all
$p\in2\pi\mathbb{Z}^3$
\begin{equation}\label{eta_p}
  \eta_p= -\widetilde{w}_{\ell,p}.
\end{equation}
(\ref{discrete asymptotic energy pde on the torus}) then reads
\begin{equation}\label{eqn of eta_p}
\begin{aligned}
  \left\vert p\right\vert^2\eta_p+\frac{1}{2}
  \sum_{q\in2\pi\mathbb{Z}^3}\hat{v}_{p-q}\eta_q
=-\frac{1}{2}
  \hat{v}_p+\frac{\lambda_\ell}{a^2}
 \eta_p+\frac{\lambda_\ell}{a^2}\widehat{\chi_{\ell}}
  \left(\frac{p}{2\pi}\right).
\end{aligned}
\end{equation}
Since $w_\ell$ is real-valued and radially symmetric, we have
$\eta_p=\eta_{-p}=\overline{\eta_p}$. Moreover, we let $\eta\in L^2(\Lambda)$ be the
function with Fourier coefficients $\eta_p$, then with (\ref{est of w_l and grad w_l}), we deduce
\begin{equation}\label{est of eta and eta_perp}
\begin{aligned}
  \Vert\eta\Vert_2^2&=\int_{\vert x\vert\leq
  \ell}
  \left\vert\widetilde{w}_{\ell}(x)\right\vert^2dx=a^3
  \int_{\vert y\vert\leq\frac{\ell}{a}}\vert
  w_{\ell}(y)\vert^2y\\
&\leq
  a^3\int_{\vert y\vert\leq\frac{\ell}{a}}\frac{C}{\vert
  y\vert^2} dy
  =Ca^2\ell.
\end{aligned}
\end{equation}
Similarly we have
\begin{align}
  \Vert\nabla_{x}\eta\Vert_2^2
=&\int_{\vert x\vert\leq \ell}
  \left\vert\nabla_{x}\widetilde{w}_{\ell}(x)\right\vert^2dx
=a\int_{\vert y\vert\leq\frac{\ell}{a}}\vert\nabla_{y}
w_{\ell}(y)\vert^2dy\nonumber\\
\leq& Ca\left(\int_{\vert y\vert\leq1}dy+
\int_{1<\vert y\vert\leq\frac{\ell}{a}}\frac{dy}{\vert
y\vert^4}\right)
\leq Ca,
\label{est of grad eta}
\end{align}
as well as
\begin{equation}\label{est of D^2 eta & D^3 eta}
  \Vert D^2_{x}\eta_\perp\Vert^2_2=\Vert
  D^2_{x}\eta\Vert^2_2\leq\frac{C}{a},\quad
  \Vert D^3_{x}\eta_\perp\Vert^2_2=\Vert
  D^3_{x}\eta\Vert^2_2\leq\frac{C}{a^3}.
\end{equation}
We can bound $\eta_p$ for all $p\in2\pi\mathbb{Z}^3$ in the same way
\begin{equation}\label{est of eta_0}
  \vert\eta_p\vert\leq\int_{\vert x\vert\leq
  \ell}\widetilde{w}_{\ell}
  (x)dx={a^3}\int_{\vert
  y\vert\leq\frac{\ell}{a}}w_{\ell}(y)dy
  \leq\frac{a^3}{d}\int_{\vert y\vert\leq\frac{\ell}{a}}\frac{C}{\vert
  y\vert} dy
  \leq Ca\ell^2,
\end{equation}
and in particular, $\Vert\eta\Vert_1\lesssim a\ell^2$. Similarly by (\ref{est of w_l and grad w_l}), we have
\begin{equation}\label{L1 grad eta and laplace eta}
  \Vert\nabla\eta\Vert_1\lesssim a\ell,\quad\Vert D^2\eta\Vert_1\lesssim a\ln(\ell a^{-1}).
\end{equation}
 In addition, since $\widetilde{f}_{\ell}=1-\widetilde{w}_\ell$, we deduce via
Plancherel's equality that
\begin{equation}\label{int vf_l}
  \int_{\mathbb{R}^3}v(x)f_{\ell}(x)dx
=\frac{1}{a}\sum_{p}\hat{v}_p\eta_p
  +\frac{1}{a}\hat{v}_0.
\end{equation}
Combining (\ref{int vf_l}) with (\ref{est of int vf_l}) we find
\begin{equation}\label{有用的屎}
  \frac{1}{a}\sum_{p}\hat{v}_p\eta_p
  +\frac{1}{a}\hat{v}_0=8\pi\mathfrak{a}_0\left(1+\frac{3}{2}
          \frac{\mathfrak{a}_0a}{\ell}\right)+O\left(
          \frac{\mathfrak{a}_0^3a^2}{\ell^2}\right)
\end{equation}

\par For further usage, we let, for $p\in2\pi\mathbb{Z}^3$,
\begin{equation}\label{define W_p}
  W_p=\frac{\lambda_{\ell}}{a^2}\left(\widehat{\chi_{\ell}}\left(
  \frac{p}{2\pi}\right)+\eta_p\right).
\end{equation}
It is also easy to verify that $W_p=W_{-p}=\overline{W_p}$. With $W_p$ defined, equation (\ref{eqn of eta_p}) reads
\begin{equation}\label{eqn of eta_p rewrt}
  \left\vert p\right\vert^2\eta_p+\frac{1}{2}
  \sum_{q\in2\pi\mathbb{Z}^3}\hat{v}_{p-q}\eta_q+ \frac{1}{2}\hat{v}_p=W_p.
\end{equation}
Let $W=\sum W_p\phi_p^{(d)}\in L^2(\Lambda)$ be the function with Fourier coefficients
$W_p$, then
\begin{equation}\label{define W(x)}
  W(x)=\frac{\lambda_{\ell}}{a^2}
(\chi_{\ell}(x)-\widetilde{w}_{\ell}
  (x)).
\end{equation}
From (\ref{est of lambda_l}) and (\ref{est of f_l}) we know that $W$ is supported and smooth
in the 3D ball $B_{d\ell}$, and
\begin{equation}\label{L^infty W}
  0<W(x)\leq\frac{Ca}{\ell^3}.
\end{equation}
Using Lemma \ref{fundamental est of v,w,lambda}, (\ref{est of eta and eta_perp}) and (\ref{est
of eta_0}) we can estimate under the assumptions that $a\to0$ and
$\frac{a}{\ell}<C$
\begin{equation}\label{L1&L2 norm of W 3dscatt}
 \Vert W\Vert_2\leq Ca\ell^{-\frac{3}{2}},\quad
  \Vert W\Vert_1\leq Ca,
\end{equation}
and
\begin{equation}\label{sum_pW_peta_p 3dscatt}
  \vert W_p\vert\leq Ca,\quad\Big\vert\sum_{p\neq0}W_p\eta_p\Big\vert\leq
  Ca^2\ell^{-1}.
\end{equation}
It is also useful to recall that the Fourier transform of the 3D radial symmertic function
$\chi_{\ell}$ is given by
\begin{equation}\label{fourier chi_dl}
  \widehat{\chi_{\ell}}\left(
\frac{p}{2\pi}\right)=\frac{4\pi}{\vert p\vert^2}
\left(\frac{\sin(\ell\vert p\vert)}{\vert p\vert}
-\ell\cos(\ell\vert p\vert)\right)
\end{equation}
Formula (\ref{fourier chi_dl}) together with (\ref{est of w_l,p}) tell us respectively that
for $p\neq0$
\begin{equation}\label{moron2}
  \vert\eta_p\vert\leq Ca\vert p\vert^{-2},
\quad\left\vert\widehat{\chi_{\ell}}\left(
\frac{p}{2\pi}\right)\right\vert\leq
C\ell\vert p\vert^{-2}.
\end{equation}
We combine (\ref{moron2}) with (\ref{define W_p}) and (\ref{est of lambda_l}) to find
\begin{equation}\label{moron3}
  \vert W_p\vert\leq C\frac{a}{\ell^2}\vert p\vert^{-2}.
\end{equation}
\par Moreover, using (\ref{eqn of eta_p rewrt}), we can prove the following useful $\ell^1$
estimate of $\{\eta_p\}$.
\begin{lemma}\label{l1 lemma}
Let $\{\eta_p\}$ be defined in (\ref{eta_p}). Assume that $a$ tends to $0$
and $\frac{1}{a}>\frac{C}{\ell}$ for some universal constant $C$. Then we have, for some
universal constant, also denoted by $C$,
\begin{equation}\label{l1 norm of eta_p lemma}
  \sum_{p\neq0}\vert\eta_p\vert\leq C.
\end{equation}
\end{lemma}
For $\{W_p\}$, Lemma \ref{l1 lemma} is the following:
\begin{lemma}\label{W_p lemma}
Assume that $a$ tends to $0$ and $\frac{1}{a}>\frac{C}{\ell}$ for
some universal constant $C$. Then we have, for some universal constant, also denoted by $C$,
\begin{equation}\label{sum W_p M_dp^-2}
  \sum_{p\neq0}\vert W_p\vert\cdot
\vert p\vert^{-2}\leq \frac{Ca}{\ell}.
\end{equation}
\end{lemma}
The detailed proofs of the above two lemma can be seen in \cite[Lemma 3.3]{me}.

\subsection{Coefficients of the Bogoliubov Transformation}\label{coeffsub}
\
\par Recall that in Section \ref{intro}, we set $\ell=N^{-\frac{1}{3}-\alpha}$ introduced in (\ref{asymptotic energy pde on the ball}) for some $\alpha>0$. We let the coefficients of the Bogoliubov transformation be:
\begin{equation}\label{define xi_k,q,p,nu,sigma 0000}
  \xi_{k,q,p}^{\nu,\sigma}=\frac{-\big(W_k+\eta_kk(q-p)\big)}
  {\frac{1}{2}\big(\vert q+k\vert^2+\vert p-k\vert^2-\vert q\vert^2-\vert p\vert^2\big)}
  \chi_{p-k\notin B^{\sigma}_F}\chi_{q+k\notin B^{\nu}_F}\chi_{p\in B^{\sigma}_F}\chi_{q\in B^{\nu}_F}.
\end{equation}
The collisional energy certainly satisfies the usual hyper plane and parabola parametrization:
\begin{equation}\label{useful formula}
  \frac{1}{2}\big(\vert q+k\vert^2+\vert p-k\vert^2-\vert q\vert^2-\vert p\vert^2\big)
  =\vert k\vert^2+k\cdot(q-p)
  =\Big\vert k+\frac{q-p}{2}\Big\vert^2-\Big\vert\frac{q-p}{2}\Big\vert^2.
\end{equation}
In this section, we collect various of useful estimates of $\xi_{k,q,p}^{\nu,\sigma}$. For simplicity, for $p\in B_F^\sigma$ and $q\in B_F^\nu$, we let
\begin{equation}\label{define xi sec4}
  \xi_{q,p}^{\nu,\sigma}(x)=\sum_{k} \xi_{k,q,p}^{\nu,\sigma}e^{ikx}.
\end{equation}
\begin{lemma}\label{L1 Linfty est lemma}
For $p\in B_F^{\sigma}$ and $q\in B_F^{\nu}$ fixed, and $c_N=4\sup_\sigma k_F^{\sigma}$,
\begin{equation}\label{L1 Linfty est}
\begin{aligned}
&\sum_{k}\Big\vert\frac{W_k}{\vert k\vert^2+k(q-p)}\chi_{p-k\notin B_F^\sigma}
\chi_{q+k\notin B_F^\nu}\Big\vert\lesssim a\ell^{-1}\\
&\sum_{\vert k\vert<c_N}\Big\vert\frac{\eta_kk(q-p)}{\vert k\vert^2+k(q-p)}\chi_{p-k\notin B_F^\sigma}
\chi_{q+k\notin B_F^\nu}\Big\vert\\
&+\sum_{\vert k\vert\geq c_N}\Big\vert\Big(\frac{\eta_kk(q-p)}{\vert k\vert^2+k(q-p)}-\frac{\eta_kk(q-p)}{\vert k\vert^2-k(q-p)}
\Big)\chi_{p-k\notin B_F^\sigma}
\chi_{q+k\notin B_F^\nu}\Big\vert\lesssim aN^{\frac{2}{3}}\ell
\end{aligned}
\end{equation}
Moreover, we have the following naive bound
\begin{equation}\label{too naive}
  \sum_{k}\Big\vert\frac{\eta_kk(q-p)}{\vert k\vert^2+k(q-p)}\chi_{p-k\notin B_F^\sigma}
\chi_{q+k\notin B_F^\nu}\Big\vert\lesssim aN^{\frac{1}{3}}\ln N.
\end{equation}
\end{lemma}
\begin{proof}
  \par We first proof the first inequality in (\ref{L1 Linfty est}) and (\ref{too naive}) altogether. Here we take $R_0\geq\ell^{-1}\gg N^{\frac{1}{3}}\ln N$ as some cut-off parameter to be determined, and denote any universal constant by $C$.
\begin{flushleft}
  \textbf{Case $\mathrm{I}$: $\vert p-k\vert\geq R_0$}
\end{flushleft}
In this case, we use (\ref{L1&L2 norm of W 3dscatt}) to bound
\begin{equation}\label{L1 est case 1 W}
  \begin{aligned}
  \sum_{k}\frac{\vert W_k\vert}{\vert k\vert^2+k(q-p)}&\chi_{p-k\notin B_F^\sigma}
\chi_{q+k\notin B_F^\nu}\lesssim
\sum_{k}\frac{\vert W_k\vert}{\vert p-k\vert^2-\vert p\vert^2}\chi_{p-k\notin B_F^\sigma}
\chi_{q+k\notin B_F^\nu}\\
&\lesssim
\Big(\sum_k\vert W_k\vert^2\Big)^{\frac{1}{2}}
\Big(\sum_{\vert p-k\vert\geq R_0}\frac{1}{\big(\vert p-k\vert^2-\vert p\vert^2\big)^2}\Big)^{\frac{1}{2}}\\
&\lesssim a\ell^{-\frac{3}{2}}\Big(\int_{R_0-C}^{\infty}\frac{(r^2+r+1)dr}{\big(r^2-\vert p\vert^2\big)^2}\Big)^{\frac{1}{2}}
\lesssim a\ell^{-\frac{3}{2}}R_0^{-\frac{1}{2}}.
  \end{aligned}
\end{equation}
On the other hand, we can use (\ref{est of grad eta}) and the fact that $\vert p-q\vert\lesssim N^{\frac{1}{3}}$ to bound
\begin{equation}\label{L1 est case 1 W eta}
  \begin{aligned}
  \sum_{k}\frac{\vert \eta_kk(q-p)\vert}{\vert k\vert^2+k(q-p)}&\chi_{p-k\notin B_F^\sigma}
\chi_{q+k\notin B_F^\nu}\lesssim
\sum_{k}\frac{\vert \eta_kk(q-p)\vert}
{\vert p-k\vert^2-\vert p\vert^2}\chi_{p-k\notin B_F^\sigma}
\chi_{q+k\notin B_F^\nu}\\
&\lesssim N^{\frac{1}{3}}
\Big(\sum_k\vert \eta_kk\vert^2\Big)^{\frac{1}{2}}
\Big(\sum_{\vert p-k\vert\geq R_0}\frac{1}{\big(\vert p-k\vert^2-\vert p\vert^2\big)^2}\Big)^{\frac{1}{2}}\\
&\lesssim a^{\frac{1}{2}}N^{\frac{1}{3}}
\Big(\int_{R_0-C}^{\infty}\frac{(r^2+r+1)dr}{\big(r^2-\vert p\vert^2\big)^2}\Big)^{\frac{1}{2}}
\lesssim a^{\frac{1}{2}}N^{\frac{1}{3}}R_0^{-\frac{1}{2}}.
  \end{aligned}
\end{equation}

\begin{flushleft}
  \textbf{Case $\mathrm{II}$: $k_F^\sigma+1+(12\pi^2)^{1/2}<\vert p-k\vert< R_0$}
\end{flushleft}
In this case, we use (\ref{sum_pW_peta_p 3dscatt}) to bound
\begin{equation}\label{L1 est case 2 W}
  \begin{aligned}
  \sum_{k}\frac{\vert W_k\vert}{\vert k\vert^2+k(q-p)}&\chi_{p-k\notin B_F^\sigma}
\chi_{q+k\notin B_F^\nu}\lesssim
\sum_{k}\frac{\vert W_k\vert}{\vert p-k\vert^2-\vert p\vert^2}\chi_{p-k\notin B_F^\sigma}
\chi_{q+k\notin B_F^\nu}\\
&\lesssim
\sup_k\vert W_k\vert
\sum_{k_F^{\sigma}+C<\vert p-k\vert< R_0}\frac{1}{\vert p-k\vert^2-\vert p\vert^2}\\
&\lesssim a\int_{k_F^\sigma+1}^{R_0}\frac{(r^2+r+1)dr}{r^2-\vert p\vert^2}
\lesssim aR_0.
  \end{aligned}
\end{equation}
On the other hand, since we set $R_0\geq\ell^{-1}$, we can use (\ref{est of w_l,p}), (\ref{eta_p}) and (\ref{est of eta_0}) to bound for $k_F^\sigma+1+(12\pi^2)^{1/2}<\vert p-k\vert< \ell^{-1}$:
\begin{equation}\label{L1 est case 2 W eta}
  \begin{aligned}
  \sum_{k}\frac{\vert \eta_kk(q-p)\vert}{\vert k\vert^2+k(q-p)}&\chi_{p-k\notin B_F^\sigma}
\chi_{q+k\notin B_F^\nu}\lesssim
\sum_{k}\frac{\vert \eta_kk(q-p)\vert}
{\vert p-k\vert^2-\vert p\vert^2}\chi_{p-k\notin B_F^\sigma}
\chi_{q+k\notin B_F^\nu}\\
&\lesssim N^{\frac{1}{3}}
\sup_k\vert \eta_k\vert^{\frac{1}{2}}\sup_k\vert k\vert\vert\eta_k\vert^\frac{1}{2}
\sum_{k_F^{\sigma}+C<\vert p-k\vert< \ell^{-1}}\frac{1}{\vert p-k\vert^2-\vert p\vert^2}\\
&\lesssim aN^{\frac{1}{3}}\ell\int_{k_F^\sigma+1}^{\ell^{-1}}
\frac{(r^2+r+1)dr}{r^2-\vert p\vert^2}
\lesssim aN^{\frac{1}{3}}.
  \end{aligned}
\end{equation}
On the other hand, for $\ell^{-1}\leq\vert p-k\vert<R_0$, we notice that since $\vert p\vert\lesssim N^{\frac{1}{3}}$,
\begin{equation*}
  \ell^{-1}=N^{\frac{1}{3}+\alpha}\lesssim
  \vert p-k\vert-\vert p\vert\leq\vert k\vert\leq \vert p-k\vert+\vert p\vert\lesssim R_0
\end{equation*}
and therefore
\begin{equation*}
  \vert p-k\vert^2-\vert p\vert^2\gtrsim \vert k\vert^2.
\end{equation*}
Hence we use directly (\ref{est of w_l,p}) and (\ref{eta_p}) (i.e. $\vert\eta_k\vert\lesssim  a\vert k\vert^{-2}$) to bound for the case $\ell^{-1}\leq\vert p-k\vert<R_0$,
\begin{equation}\label{L1 est case 2 W eta 2}
  \begin{aligned}
  \sum_{k}\frac{\vert \eta_kk(q-p)\vert}{\vert k\vert^2+k(q-p)}&\chi_{p-k\notin B_F^\sigma}
\chi_{q+k\notin B_F^\nu}\lesssim
\sum_{k}\frac{\vert \eta_kk(q-p)\vert}
{\vert p-k\vert^2-\vert p\vert^2}\chi_{p-k\notin B_F^\sigma}
\chi_{q+k\notin B_F^\nu}\\
&\lesssim aN^{\frac{1}{3}}
\sum_{\ell^{-1}\lesssim\vert k\vert\lesssim R_0 }\frac{1}{\vert k\vert^3}\\
&\lesssim aN^{\frac{1}{3}}\int_{\ell^{-1}}^{R_0}
\frac{(r^2+r+1)dr}{r^3}
\lesssim aN^{\frac{1}{3}}\ln R_0.
  \end{aligned}
\end{equation}

\begin{flushleft}
  \textbf{Case $\mathrm{III}$: $k_F^\sigma<\vert p-k\vert\leq k_F^\sigma+1+(12\pi^2)^{1/2}$}
\end{flushleft}
In this case, we use (\ref{approximate lattice 2d sphere}) and (\ref{sum_pW_peta_p 3dscatt}) to bound
\begin{equation}\label{L1 est case 3 W}
  \begin{aligned}
  \sum_{k}\frac{\vert W_k\vert}{\vert k\vert^2+k(q-p)}&\chi_{p-k\notin B_F^\sigma}
\chi_{q+k\notin B_F^\nu}\lesssim
\sum_{k}\frac{\vert W_k\vert}{\vert p-k\vert^2-\vert p\vert^2}\chi_{p-k\notin B_F^\sigma}
\chi_{q+k\notin B_F^\nu}\\
&\lesssim
N^{\frac{1}{3}+\varepsilon}\sup_k\vert W_k\vert
\sum_{\substack{R^2\in4\pi^2\mathbb{Z}\\k_F^{\sigma}<R\leq k_F^\sigma+C}}\frac{1}{R^2-\vert p\vert^2}\\
&\lesssim aN^{\frac{1}{3}+\varepsilon}\Big(1+
\int_{(k_F^\sigma)^2+1}^{(k_F^\sigma+C)^2}\frac{dr}{r-\vert p\vert^2}\Big)
\lesssim aN^{\frac{1}{3}+\varepsilon^\prime}.
  \end{aligned}
\end{equation}
for any $\varepsilon^\prime>\varepsilon>0$ small but fixed.
We then conclude the first inequality in (\ref{L1 Linfty est}) by choosing $R_0=\ell^{-1}$ and $\varepsilon^\prime<\alpha$. On the other hand, we use (\ref{approximate lattice 2d sphere}) and (\ref{sum_pW_peta_p 3dscatt}) to bound
\begin{equation}\label{L1 est case 3 W eta}
  \begin{aligned}
  \sum_{k}\frac{\vert \eta_kk(q-p)\vert}{\vert k\vert^2+k(q-p)}&\chi_{p-k\notin B_F^\sigma}
\chi_{q+k\notin B_F^\nu}\lesssim
\sum_{k}\frac{\vert \eta_kk(q-p)\vert}
{\vert p-k\vert^2-\vert p\vert^2}\chi_{p-k\notin B_F^\sigma}
\chi_{q+k\notin B_F^\nu}\\
&\lesssim
N^{\frac{2}{3}+\varepsilon}\sup_k\vert \eta_k\vert^{\frac{1}{2}}\sup_k\vert k\vert\vert\eta_k\vert^\frac{1}{2}
\sum_{\substack{R^2\in4\pi^2\mathbb{Z}\\k_F^{\sigma}<R\leq k_F^\sigma+C}}\frac{1}{R^2-\vert p\vert^2}\\
&\lesssim a\ell N^{\frac{2}{3}+\varepsilon}\Big(1+
\int_{(k_F^\sigma)^2+1}^{(k_F^\sigma+C)^2}\frac{dr}{r-\vert p\vert^2}\Big)
\lesssim a\ell N^{\frac{2}{3}+\varepsilon^\prime}.
  \end{aligned}
\end{equation}
for any $\varepsilon^\prime>\varepsilon>0$ small but fixed. We then conclude (\ref{too naive}) by choosing $R_0=a^{-1}=N$ and $\varepsilon^\prime<\alpha$.

\par For the second inequality in (\ref{L1 Linfty est}), we also analyse it case by case.

\begin{flushleft}
  \textbf{Case $\mathrm{IV}$: $\vert k\vert\geq c_N$}
\end{flushleft}
In this case, we first notice that since $\vert k\vert>c_N$ and $\vert p\vert,\vert q\vert\leq\frac{1}{4}c_N\lesssim N^{\frac{1}{3}}$ implies
\begin{equation*}
   \vert k\vert^2\pm k(q-p)\geq \frac{1}{2}\vert k\vert^2.
\end{equation*}
Therefore,
\begin{equation}\label{L1 est case 1 eta}
\begin{aligned}
  &\Big\vert\frac{1}{\vert k\vert^2+k(q-p)}-\frac{1}{\vert k\vert^2-k(q-p)}\Big\vert\\
  &=\Big\vert\frac{2k(q-p)}{\big(\vert k\vert^2+k(q-p)\big)
  \big(\vert k\vert^2-k(q-p)\big)}\Big\vert
  \lesssim \frac{N^{\frac{1}{3}}}{\vert k\vert^3}.
\end{aligned}
\end{equation}
Hence we can use (\ref{est of eta and eta_perp} and (\ref{est of eta_0}) to bound
\begin{equation}\label{L1 est case 1 eta 2}
 \begin{aligned}
 &\sum_{\vert k\vert\geq c_N}\Big\vert\Big(\frac{\eta_kk(q-p)}{\vert k\vert^2+k(q-p)}-\frac{\eta_kk(q-p)}{\vert k\vert^2-k(q-p)}
\Big)\chi_{p-k\notin B_F^\sigma}
\chi_{q+k\notin B_F^\nu}\Big\vert\\
&\lesssim N^{\frac{2}{3}}\sum_{\vert k\vert\geq c_N}\frac{\vert\eta_k\vert}{\vert k\vert^2}
\lesssim aN^{\frac{2}{3}}\ell
 \end{aligned}
\end{equation}
where the last inequality can be deduced in a similar approach like we just did in Cases \textrm{I} and \textrm{II}, by choosing the cut-off parameter $R_0=\ell^{-1}$.

\begin{flushleft}
  \textbf{Case $\mathrm{V}$: $k_F^\sigma+1+(12\pi^2)^{1/2}<\vert p-k\vert\;\textmd{and}\;\vert k\vert<c_N$}
\end{flushleft}
Notice that $\vert p\vert,\vert q\vert<c_N\lesssim N^{\frac{1}{3}}$. Then similar to (\ref{L1 est case 2 W}) in Case \textrm{II}, we can bound using (\ref{est of eta_0}), in this case
\begin{equation}\label{L1 est case 2 eta}
  \sum_{\vert k\vert<c_N}\Big\vert\frac{\eta_kk(q-p)}{\vert k\vert^2+k(q-p)}\chi_{p-k\notin B_F^\sigma}
\chi_{q+k\notin B_F^\nu}\Big\vert\lesssim a\ell^2N\ln N.
\end{equation}

\begin{flushleft}
  \textbf{Case $\mathrm{VI}$: $k_F^\sigma<\vert p-k\vert\leq k_F^\sigma+1+(12\pi^2)^{1/2}$}
\end{flushleft}
Notice that $\vert p\vert,\vert q\vert<c_N\lesssim N^{\frac{1}{3}}$. Then similar to (\ref{L1 est case 3 W}) in Case \textrm{III}, we can bound using (\ref{est of eta_0}), in this case
\begin{equation}\label{L1 est case 3 eta}
  \sum_{\vert k\vert<c_N}\Big\vert\frac{\eta_kk(q-p)}{\vert k\vert^2+k(q-p)}\chi_{p-k\notin B_F^\sigma}
\chi_{q+k\notin B_F^\nu}\Big\vert\lesssim a\ell^2N^{1+\varepsilon}
\end{equation}
for any $\varepsilon>0$ small but fixed.
We then conclude the second inequality in (\ref{L1 Linfty est}) by choosing $\varepsilon<\alpha$.
\end{proof}

\begin{lemma}\label{N^1/3+epsilon lemma}
  Let
 \begin{equation}\label{define S_q,p lemma}
   \mathfrak{S}_{q,p}^{\nu,\sigma}
   \coloneqq\{k\in2\pi\mathbb{Z}^3\,\big\vert\,\vert k\vert^2+k\cdot(q-p)=0\},
 \end{equation}
 then for $p\in B_F^{\sigma}$ and $q\in B_F^{\nu}$ fixed, we have
 \begin{equation}\label{N^1/3+epsilon}
   \sum_{k\in2\pi\mathbb{Z}^3}\Big\vert\frac{\big(1-\chi_{p-k\notin B_F^{\sigma}}
     \chi_{q+k\notin B_F^\nu}\big)\chi_{k\notin \mathfrak{S}_{q,p}^{\nu,\sigma}}}{\vert k\vert^2+k(q-p)}\Big\vert\lesssim N^{\frac{1}{3}+\varepsilon}
 \end{equation}
 for any $\varepsilon>0$ small enough as long as $N$ is correspondingly large enough.
\end{lemma}
\begin{proof}
  \par Notice since we demand $p\in B_F^{\sigma}$ and $q\in B_F^{\nu}$, then $\big(1-\chi_{p-k\notin B_F^{\sigma}}\chi_{q+k\notin B_F^\nu}\big)\neq 0$
  implies $\vert k\vert\lesssim N^{\frac{1}{3}}$. We then analyse (\ref{N^1/3+epsilon}) by cases. We denote $(q-p)/2\coloneqq r_0$ for the time being, apparently $\vert r_0\vert\lesssim N^{\frac{1}{3}}$. We also denote any universal constant by $C$.
\begin{flushleft}
  \textbf{Case $\mathrm{I}$: $\vert k+r_0\vert>\vert r_0\vert+1+(12\pi^2)^{1/2}$}
\end{flushleft}
In this case, we can use Riemann sum to bound
\begin{equation}\label{N^1/3+epsilon case 1}
\begin{aligned}
  \sum_{k}\frac{\big(1-\chi_{p-k\notin B_F^{\sigma}}
     \chi_{q+k\notin B_F^\nu}\big)\chi_{k\notin \mathfrak{S}_{q,p}^{\nu,\sigma}}}
     {\vert k+r_0\vert^2-\vert r_0\vert^2}
     \lesssim\int_{\vert r_0\vert+1}^{CN^{\frac{1}{3}}}
     \frac{(r^2+r+1)dr}{r^2-\vert r_0\vert^2}
     \lesssim N^{\frac{1}{3}}\ln N.
\end{aligned}
\end{equation}
\begin{flushleft}
  \textbf{Case $\mathrm{II}$: $\vert k+r_0\vert<\vert r_0\vert-1-(12\pi^2)^{1/2}$}
\end{flushleft}
If such $k$ exists, we can calculate as in Case $\mathrm{I}$,
\begin{equation}\label{N^1/3+epsilon case 2}
\begin{aligned}
  \sum_{k}\frac{\big(1-\chi_{p-k\notin B_F^{\sigma}}
     \chi_{q+k\notin B_F^\nu}\big)\chi_{k\notin \mathfrak{S}_{q,p}^{\nu,\sigma}}}
     {\vert r_0\vert^2-\vert k+r_0\vert^2}
     &\lesssim\frac{1}{\vert r_0\vert^2+1}+\int_{0}^{\vert r_0\vert-1}
     \frac{(r^2+r+1)dr}{\vert r_0\vert^2-r^2}\\
     &\lesssim N^{\frac{1}{3}}\ln N.
\end{aligned}
\end{equation}
\begin{flushleft}
  \textbf{Case $\mathrm{III}$: $\vert r_0\vert<\vert k+r_0\vert\leq\vert r_0\vert+1+(12\pi^2)^{1/2}$}
\end{flushleft}
We first claim that for $R\lesssim N^{\frac{1}{3}}$ such that $R^2\in \pi^2\mathbb{Z}$,
\begin{equation}\label{N^1/3+epsilon case 3 1}
  \#\{k\in2\pi\mathbb{Z}^3\,\big\vert\,\vert k+r_0\vert=R\}\lesssim N^{\frac{1}{3}+\varepsilon}
\end{equation}
for any $\varepsilon>0$ small enough as long as $N$ is correspondingly large enough. In fact, when $R\gtrsim N^{\frac{1}{9}}$, we can obtain (\ref{N^1/3+epsilon case 3 1}) directly by (\ref{approximate lattice 2d sphere}). On the other hand, when $R\lesssim N^{\frac{1}{9}}$, we have the following naive bound
\begin{equation*}
  \begin{aligned}
  \#\{k\in2\pi\mathbb{Z}^3\,\big\vert\,\vert k+r_0\vert=R\}\lesssim
  \#\{k\in2\pi\mathbb{Z}^3\,\big\vert\,\vert k+r_0\vert\leq CN^{\frac{1}{9}}\}
  \lesssim N^{\frac{1}{3}}.
  \end{aligned}
\end{equation*}
Therefore, in this case we have
\begin{equation}\label{N^1/3+epsilon case 3 2}
  \begin{aligned}
  &\sum_{k}\frac{\big(1-\chi_{p-k\notin B_F^{\sigma}}
     \chi_{q+k\notin B_F^\nu}\big)\chi_{k\notin \mathfrak{S}_{q,p}^{\nu,\sigma}}}
     {\vert k+r_0\vert^2-\vert r_0\vert^2}
\lesssim N^{\frac{1}{3}+\varepsilon}
  \sum_{\substack{R^2\in\pi^2\mathbb{Z}\\ \vert r_0\vert<R\leq\vert r_0\vert+C}}
  \frac{1}{R^2-\vert r_0\vert^2}\\
  &\lesssim N^{\frac{1}{3}+\varepsilon}\Big(1+
  \int_{\vert r_0\vert^2+1}^{(\vert r_0\vert+C)^2}\frac{dr}{r-\vert r_0\vert^2}\Big)
  \lesssim N^{\frac{1}{3}+\varepsilon^\prime}.
\end{aligned}
\end{equation}
for any $\varepsilon^\prime>0$ small enough as long as $N$ is correspondingly large enough.
\begin{flushleft}
  \textbf{Case $\mathrm{IV}$: $\vert r_0\vert-1-(12\pi^2)^{1/2}\leq\vert k+r_0\vert<\vert r_0\vert$}
\end{flushleft}
The proof is similar to Case $\mathrm{III}$, and therefore we omit further details.
\end{proof}

\begin{lemma}\label{energy est bog lem}
  For $p\in B_F^{\sigma}$ and $q\in B_F^{\nu}$ fixed, we have
  \begin{equation}\label{energy est bog}
    \begin{aligned}
    \sum_{k}\Big\vert\frac{W_k^2}{\vert p-k\vert^2-\vert k_F^\sigma\vert^2}\chi_{p-k\notin B_F^\sigma}\Big\vert&\lesssim a^2\ell^{-1}\\
\sum_{k}\Big\vert\frac{\big(\eta_kk(q-p)\big)^2}{\vert p-k\vert^2-\vert k_F^\sigma\vert^2}\chi_{p-k\notin B_F^\sigma}\Big\vert&\lesssim a^2N^{\frac{2}{3}}\ell
    \end{aligned}
  \end{equation}
  and therefore
  \begin{equation}\label{energy est bog true}
    \sum_{k}\big\vert\big(W_k+\eta_kk(q-p)\big)\xi_{k,q,p}^{\nu,\sigma}\big\vert\lesssim
    a^2\ell^{-1}.
  \end{equation}
  Furthermore, for some constant $u>0$ and $\delta=\frac{1}{3}+u$, recall that $P_{F,\delta}^\sigma$ has been defined in (\ref{partition 3D space}), we have
    \begin{equation}\label{energy est bog u}
    \begin{aligned}
    \sum_{k}\Big\vert\frac{W_k^2}{\vert p-k\vert^2-\vert k_F^\sigma\vert^2}\chi_{p-k\in P_{F,\delta}^\sigma}\Big\vert&\lesssim a^2\ell^{-2}N^{-\delta}\\
\sum_{k}\Big\vert\frac{\big(\eta_kk(q-p)\big)^2}{\vert p-k\vert^2-\vert k_F^\sigma\vert^2}\chi_{p-k\in P_{F,\delta}^\sigma}\Big\vert&\lesssim a^2N^{\frac{2}{3}}N^{-\delta}
    \end{aligned}
  \end{equation}
  Moreover,
   \begin{equation}\label{energy est bog laplace eta}
    \begin{aligned}
    \sum_{k}\Big\vert\frac{\eta_k^2\vert k\vert^4}{\vert p-k\vert^2-\vert k_F^\sigma\vert^2}\chi_{p-k\notin B_F^\sigma}\Big\vert\lesssim a
    \end{aligned}
  \end{equation}
\end{lemma}
\begin{proof}
  \par Notice that since
  \begin{equation*}
    \begin{aligned}
    \sum_{k}\big\vert\big(W_k+\eta_kk(q-p)\big)\xi_{k,q,p}^{\nu,\sigma}\big\vert
    &=\sum_k\frac{\big(W_k+\eta_kk(q-p)\big)^2}{\vert k\vert^2+k(q-p)}
     \chi_{p-k\notin B^{\sigma}_F}\chi_{q+k\notin B^{\nu}_F}\\
     &\lesssim\sum_{k}\frac{W_k^2+\big(\eta_kk(q-p)\big)^2}{\vert p-k\vert^2-\vert k_F^\sigma\vert^2}\chi_{p-k\notin B_F^\sigma}.
    \end{aligned}
  \end{equation*}
  Then (\ref{energy est bog true}) follows directly from (\ref{energy est bog}).
  \par The proof of (\ref{energy est bog}) is rather similar to (\ref{L1 Linfty est}), we only need to notice that by (\ref{est of w_l,p}) and (\ref{eta_p}),
  \begin{equation*}
    \big\vert\eta_kk(q-p)\big\vert^2\lesssim aN^{\frac{2}{3}}\vert\eta_k\vert.
  \end{equation*}
  The proof of (\ref{energy est bog laplace eta}) is also of a similar approach as well (see \textbf{Cases $\mathrm{I}$-$\mathrm{III}$} in the proof of Lemma \ref{L1 Linfty est lemma}), but this time we have
  \begin{equation*}
    \eta_k^2\vert k\vert^4\lesssim a^2.
  \end{equation*}
  Also notice that in \textbf{Case $\mathrm{I}$}, i.e. when $\vert p-k\vert>R_0\gg N^\frac{1}{3}\ln N$, we have $\vert k\vert\gtrsim R_0$ and therefore
  \begin{equation*}
    \vert p-k\vert^2-\vert k_F^\sigma\vert^2
   \geq\big(\vert k\vert-\vert p\vert\big)^2-\vert k_F^\sigma\vert^2\gtrsim\vert k\vert^2,
  \end{equation*}
  which together with (\ref{est of grad eta}) leads to
  \begin{equation*}
    \sum_{k,\vert p-k\vert>R_0}\frac{\eta_k^2\vert k\vert^4}{\vert p-k\vert^2
    -\vert k_F^\sigma\vert^2}\lesssim\Vert\nabla\eta\Vert_2^2\lesssim a.
  \end{equation*}
  \par For the proof of (\ref{energy est bog u}), we first notice that $p-k\in P_{F,\delta}^\sigma$ implies $\vert k\vert\gtrsim N^{\delta}$ and moreover
  \begin{equation*}
    \vert p-k\vert^2-\vert k_F^\sigma\vert^2\gtrsim \vert k\vert^2.
  \end{equation*}
  On the other hand, from (\ref{moron2}) and (\ref{moron3}), we know that
  \begin{equation}\label{W_p p^-2}
    \vert W_k\vert^2\lesssim \frac{a^2\ell^{-2}}{\vert k\vert^2},\quad
    \vert \eta_kk(q-p)\vert^2\lesssim \frac{a^2N^{\frac{2}{3}}}{\vert k\vert^2}.
  \end{equation}
  Therefore, for some universal constant $C>0$,
  \begin{equation}\label{proof energy est bog u}
    \begin{aligned}
    \sum_{k}\Big\vert\frac{W_k^2}{\vert p-k\vert^2-\vert k_F^\sigma\vert^2}\chi_{p-k\in P_{F,\delta}^\sigma}\Big\vert&\lesssim\int_{CN^\delta}^{\infty}
    \frac{a^2\ell^{-2}(r^2+r+1)dr}{r^4}
    \lesssim a^2\ell^{-2}N^{-\delta}\\
\sum_{k}\Big\vert\frac{\big(\eta_kk(q-p)\big)^2}{\vert p-k\vert^2-\vert k_F^\sigma\vert^2}\chi_{p-k\in P_{F,\delta}^\sigma}\Big\vert&
\lesssim\int_{CN^\delta}^{\infty}
    \frac{a^2N^{\frac{2}{3}}(r^2+r+1)dr}{r^4}
\lesssim a^2N^{\frac{2}{3}}N^{-\delta}
    \end{aligned}
  \end{equation}
\end{proof}

\begin{lemma}\label{xi l2 lemma}
  For $p\in B_F^\sigma$ and $q\in B_F^\nu$, we have
  \begin{equation}\label{xi l2}
    \sum_{p\in B_F^\sigma}\sum_{k}\vert\xi_{k,q,p}^{\nu,\sigma}\vert^2\lesssim a^2N^{\frac{2}{3}+\varepsilon}
  \end{equation}
  with $\varepsilon>0$ small enough as long as $N$ is correspondingly large enough. On the other hand, recall the discussion we have carried out around (\ref{partition 3D space}), for some $0<\delta<\frac{1}{3}$, we have
  \begin{equation}\label{xi l2 cut off}
    \begin{aligned}
    \sum_{p\in \underline{B}_{F,\delta}^\sigma}\sum_{k}
    \frac{\vert\xi_{k,q,p}^{\nu,\sigma}\vert^2}
    {\vert k_F^\sigma\vert^2-\vert p\vert^2}&\lesssim a^2N^{\frac{1}{3}-\delta}\\
    \sum_{p\in B_F^\sigma}\sum_{k}
    \frac{\vert\xi_{k,q,p}^{\nu,\sigma}\vert^2\chi_{p-k\in P_{F,\delta}^\sigma}}
    {\vert p-k\vert^2-\vert k_F^\sigma\vert^2}
    &\lesssim a^2N^{\frac{1}{3}-\delta}\ln N
    \end{aligned}
  \end{equation}
\end{lemma}
\begin{proof}
  \par Since $p\in B_F^\sigma$ and $q\in B_F^\nu$, recall (\ref{k_F^sigma}), we know that $\vert p\vert, \vert q\vert\lesssim N^{\frac{1}{3}}$. Also recall that $\ell=N^{-\frac{1}{3}-\alpha}$ for some $\alpha>0$. These observations together with (\ref{est of w_l,p}), (\ref{eta_p}), (\ref{est of eta_0}) and (\ref{sum_pW_peta_p 3dscatt}) yield
  \begin{equation}\label{useful bound}
    \big\vert W_k+\eta_kk(q-p)\big\vert\lesssim a.
  \end{equation}
  Therefore, by definition (\ref{define xi_k,q,p,nu,sigma 0000}, we have
  \begin{equation}\label{4.8 1st}
   \vert\xi_{k,q,p}^{\nu,\sigma}\vert^2\lesssim \frac{a^2\chi_{p-k\notin B_F^\sigma}
   \chi_{p\in B_F^\sigma}}{\big(\vert p-k\vert^2-\vert p\vert^2\big)^2}.
  \end{equation}
  Hence in order to prove (\ref{xi l2}), we only need to show that
  \begin{equation}\label{xi l2preposition}
    \sum_{p\in B_F^\sigma}\sum_{k}\frac{\chi_{p-k\notin B_F^\sigma}}
    {\big(\vert p-k\vert^2-\vert p\vert^2\big)^2}\lesssim N^{\frac{2}{3}+\varepsilon}
  \end{equation}
  for any $\varepsilon>0$. We let $0< \delta<\frac{1}{3}$ and denote any universal constant by $C$ and specifically, $C_0=(12\pi^2)^{1/2}$ for short.
  \begin{flushleft}
    \textbf{Case $\mathrm{I}$: $\vert p\vert<k_F^\sigma-1-2(12\pi^2)^{1/2}$}
  \end{flushleft}
  In this case, we have
  \begin{equation}\label{case1 l2}
    \begin{aligned}
    \sum_{p}\sum_{k}\frac{\chi_{p-k\notin B_F^\sigma}}
    {\big(\vert p-k\vert^2-\vert p\vert^2\big)^2}&
    \lesssim\sum_{p}\int_{k_F^\sigma-C_0}^{\infty}\frac{(r^2+r+1)dr}
    {\big(r^2-\vert p\vert^2\big)^2}
    \lesssim \sum_{p}\frac{1}{k_F^\sigma-C_0-\vert p\vert}
    \\
    &\lesssim \int_{0}^{k_F^\sigma-C_0-1}\frac{(r^2+r+1)dr}{k_F^\sigma-C_0-r}
    \lesssim N^{\frac{2}{3}}\ln N.
    \end{aligned}
  \end{equation}

   \begin{flushleft}
    \textbf{Case $\mathrm{II}$: $\vert p\vert\geq k_F^\sigma-1-2(12\pi^2)^{1/2}$,
    $\vert p-k\vert> k_F^\sigma+1+2(12\pi^2)^{1/2}$}
  \end{flushleft}
  In this case, we have
  \begin{equation}\label{case2 l2}
    \begin{aligned}
    \sum_{p}\sum_{k}\frac{\chi_{p-k\notin B_F^\sigma}}
    {\big(\vert p-k\vert^2-\vert p\vert^2\big)^2}&
    \lesssim\sum_{p}\int_{k_F^\sigma+1+C_0}^{\infty}\frac{(r^2+r+1)dr}
    {\big(r^2-\vert p\vert^2\big)^2}
    \lesssim \sum_{p}\frac{1}{k_F^\sigma+1+C_0-\vert p\vert}
    \\
    &\lesssim \int_{k_F^\sigma-1-2C_0}^{k_F^\sigma+C_0}\frac{(r^2+r+1)dr}
    {k_F^\sigma+1+C_0-r}
    \lesssim N^{\frac{2}{3}}.
    \end{aligned}
  \end{equation}

  \begin{flushleft}
    \textbf{Case $\mathrm{III}$: $\vert p\vert\geq k_F^\sigma-1-2(12\pi^2)^{1/2}$,
    $\vert p-k\vert\leq k_F^\sigma+1+2(12\pi^2)^{1/2}$}
  \end{flushleft}
  In this case, we use (\ref{approximate lattice 2d sphere}) to bound
  \begin{equation}\label{case3 l2}
    \begin{aligned}
    &\sum_{p}\sum_{k}\frac{\chi_{p-k\notin B_F^\sigma}}
    {\big(\vert p-k\vert^2-\vert p\vert^2\big)^2}
    \lesssim N^{\frac{1}{3}+\varepsilon}\sum_{p}
    \sum_{\substack{R^2\in 4\pi^2\mathbb{Z}\\k_F^\sigma<R\leq k_F^\sigma+C}}
    \frac{1}{(R^2-\vert p\vert^2)^2}\\
   &\lesssim N^{\frac{1}{3}+\varepsilon}\sum_{p}
    \Big\{\chi_{\vert p\vert<k_F^\sigma}\int_{\vert k_F^\sigma\vert^2}^{
    \vert k_F^\sigma+C\vert^2}\frac{dt}{(t-\vert p\vert^2)^2}+
    \chi_{\vert p\vert=k_F^\sigma}\Big(1+\int_{\vert k_F^\sigma\vert^2+1}^{
    \vert k_F^\sigma+C\vert^2}\frac{dt}{(t-\vert p\vert^2)^2}\Big)\Big\}\\
    &\lesssim N^{\frac{1}{3}+\varepsilon}\sum_{p}\Big(\frac{\chi_{\vert p\vert<k_F^\sigma}}{\vert k_F^\sigma\vert^2-\vert p\vert^2}+\chi_{\vert p\vert=k_F^\sigma}\Big)\\
    &\lesssim N^{\frac{2}{3}+2\varepsilon}\sum_{\substack{R^2\in 4\pi^2\mathbb{Z}\\k_F^\sigma-C\leq R\leq k_F^\sigma}}\Big(\frac{\chi_{R<k_F^\sigma}}{\vert k_F^\sigma\vert^2-R^2}+\chi_{R=k_F^\sigma}\Big)\\
    &
    \lesssim N^{\frac{2}{3}+2\varepsilon}\Big(1+\int_{\vert k_F^\sigma-C\vert^2}^{\vert k_F^\sigma\vert^2-1}\frac{dt}{\vert k_F^\sigma\vert^2-t}\Big)\lesssim N^{\frac{2}{3}+\varepsilon^\prime}
    \end{aligned}
  \end{equation}
  with $0<2\varepsilon<\varepsilon^\prime$ small but fixed.
  \vspace{1em}
  \par Combining (\ref{case1 l2}-\ref{case3 l2}), we reach (\ref{xi l2preposition}) and hence (\ref{xi l2}).
  \par For the proof of the first inequality in (\ref{xi l2 cut off}), we first notice that for $p\in \underline{B}_{F,\delta}^\sigma$,
  \begin{equation*}
    \vert k_F^\sigma\vert^2-\vert p\vert^2\gtrsim N^{\frac{1}{3}}\big(\vert k_F^\sigma\vert-\vert p\vert\big).
  \end{equation*}
  We then use an approach similar to Case $\mathrm{I}$:
  \begin{equation}\label{case4 l2}
    \begin{aligned}
     &\sum_{p\in \underline{B}_{F,\delta}^\sigma}\sum_{k}
    \frac{\vert\xi_{k,q,p}^{\nu,\sigma}\vert^2}
    {\vert k_F^\sigma\vert^2-\vert p\vert^2}\lesssim
    \sum_{p\in \underline{B}_{F,\delta}^\sigma}\sum_{k}
    \frac{a^2N^{-\frac{1}{3}}\chi_{p-k\notin B_F^\sigma}}{\big(\vert p-k\vert^2-\vert p\vert^2\big)^2
    \big(\vert k_F^\sigma\vert-\vert p\vert\big)}\\
    &\lesssim \sum_{p\in \underline{B}_{F,\delta}^\sigma}
    \int_{k_F^\sigma-C_0}^{\infty}\frac{a^2N^{-\frac{1}{3}}(r^2+r+1)dr}
    {\big(r^2-\vert p\vert^2\big)^2 \big(\vert k_F^\sigma\vert-\vert p\vert\big)}
    \lesssim \sum_{p\in \underline{B}_{F,\delta}^\sigma}
    \frac{a^2N^{-\frac{1}{3}}}{\big(k_F^\sigma-C_0-\vert p\vert\big)^2}\\
    &\lesssim a^2N^{-\frac{1}{3}}
    \int_{0}^{k_F^\sigma-CN^\delta}\frac{(r^2+r+1)dr}{(k_F^\sigma-C_0-r)^2}
    \lesssim a^2N^{\frac{1}{3}-\delta}.
    \end{aligned}
  \end{equation}
  \par For the proof of the second inequality in (\ref{xi l2 cut off}), we notice that for $p-k\in P_{F,\delta}^\sigma$,
  \begin{equation*}
    \vert p-k\vert^2-\vert k_F^\sigma\vert^2\gtrsim N^{\frac{1}{3}+\delta}.
  \end{equation*}
  We then use an approach similar to Case $\mathrm{II}$:
   \begin{equation}\label{case5 l2}
    \begin{aligned}
    &\sum_{p}\sum_{k}\frac{\vert\xi_{k,q,p}^{\nu,\sigma}\vert^2
    \chi_{p-k\in P_{F,\delta}^\sigma}}
    {\vert p-k\vert^2 -\vert k_F^\sigma\vert^2}
    \lesssim
    \sum_{p}\sum_{k}\frac{a^2N^{-\frac{1}{3}-\delta}
    \chi_{p-k\in P_{F,\delta}^\sigma}}
    {\big(\vert p-k\vert^2-\vert p\vert^2\big)^2
    \big(\vert p-k\vert^2 -\vert k_F^\sigma\vert^2\big)}
    \\
    &\lesssim
    \sum_{p}\int_{k_F^\sigma+CN^\delta}^{\infty}\frac{a^2N^{-\frac{1}{3}-\delta}(r^2+r+1)dr}
    {\big(r^2-\vert p\vert^2\big)^2}\lesssim \sum_{p}\frac{a^2N^{-\frac{1}{3}-\delta}}
    {k_F^\sigma+CN^\delta-\vert p\vert}\\
    &\lesssim a^2N^{-\frac{1}{3}-\delta}
    \int_{0}^{k_F^\sigma+C_0}\frac{(r^2+r+1)dr}
    {k_F^\sigma+CN^\delta-r}
    \lesssim a^2N^{\frac{1}{3}-\delta}\ln N.
    \end{aligned}
  \end{equation}
\end{proof}

\section{Excitation Hamiltonians and Main Propositions}\label{mainpropo}
\par From here on out, as we have claimed, we always select
\begin{equation}\label{choose l 0}
  \ell=N^{-\frac{1}{3}-\alpha}
\end{equation}
with $0<\alpha<\frac{1}{6}$. As we move forward, further requirement on $\alpha$ is needed.
\subsection{Quadratic Renormalization}\
\par We define the skew-symmetric operator $B$ by
\begin{equation}\label{define B 0}
  B=\frac{1}{2}(A-A^*),
\end{equation}
where
\begin{equation}\label{define A 0}
  A=\sum_{k,p,q,\sigma,\nu}
  \eta_ka^*_{p-k,\sigma}a^*_{q+k,\nu}a_{q,\nu}a_{p,\sigma}\chi_{p-k\notin B^{\sigma}_F}\chi_{q+k\notin B^{\nu}_F}\chi_{p\in B^{\sigma}_F}\chi_{q\in B^{\nu}_F},
\end{equation}
with $\eta_k$ defined in (\ref{eta_p}). $A$ is considered as a quadratic type operator. We consider the following excitation Hamiltonian
\begin{equation}\label{define G_N 0}
  \mathcal{G}_N\coloneqq e^{-B}H_Ne^B.
\end{equation}
The effect of the quadratic renormalization is recorded below.
\begin{proposition}\label{p1}
Recall that $\ell=N^{-\frac{1}{3}-\alpha}$ with $0<\alpha<\frac{1}{6}$. We have
  \begin{equation}\label{define G_N p1}
  \begin{aligned}
  \mathcal{G}_N=C_{\mathcal{G}_N}+Q_{\mathcal{G}_N}
  +\mathcal{K}+\mathcal{V}_{4}+\mathcal{V}_{21}^\prime
  +\Omega+\mathcal{V}_3+\mathcal{E}_{\mathcal{G}_N}
  \end{aligned}
\end{equation}
with
\begin{equation}\label{define C_G_N and Q_G_N p1}
  \begin{aligned}
  C_{\mathcal{G}_N}&=\frac{1}{2}\Big(\hat{v}_0+\sum_{k}\hat{v}_k\eta_k\Big)
  \sum_{\sigma\neq\nu}N_\sigma N_\nu\\
  &+\sum_{k,p,q,\sigma,\nu}W_k\eta_k\chi_{p-k\notin B^{\sigma}_F}\chi_{q+k\notin B^{\nu}_F}\chi_{p\in B^{\sigma}_F}\chi_{q\in B^{\nu}_F}\\
  &-\sum_{k,p,q,\sigma}W_k\eta_{k+q-p}\chi_{p-k,q+k\notin B^{\sigma}_F}\chi_{p,q\in B^{\sigma}_F}\chi_{p\neq q}\\
  &-\sum_{k,p,q,\sigma,\nu}\eta_k^2k(q-p)(1-\chi_{p-k\notin B^{\sigma}_F}\chi_{q+k\notin B^{\nu}_F})\chi_{p\in B^{\sigma}_F}\chi_{q\in B^{\nu}_F}\\
  &-\sum_{k,p,q,\sigma}\eta_{k+q-p}\eta_kk(q-p)\chi_{p-k,q+k\notin B^{\sigma}_F}\chi_{p,q\in B^{\sigma}_F}\chi_{p\neq q}\\
  Q_{\mathcal{G}_N}&=-\sum_{k}\hat{v}_k\eta_k\sum_{\sigma\neq\nu}N_\sigma \mathcal{N}_{ex,\nu}
  \end{aligned}
\end{equation}
and
\begin{equation}\label{define V_21' and Omega p1}
  \begin{aligned}
  \mathcal{V}_{21}^\prime&=\sum_{k,p,q,\sigma,\nu}
  W_k(a^*_{p-k,\sigma}a^*_{q+k,\nu}a_{q,\nu}a_{p,\sigma}+h.c.)\chi_{p-k\notin B^{\sigma}_F}\chi_{q+k\notin B^{\nu}_F}\chi_{p\in B^{\sigma}_F}\chi_{q\in B^{\nu}_F}\\
  \Omega&=\sum_{k,p,q,\sigma,\nu}
  \eta_kk(q-p)(a^*_{p-k,\sigma}a^*_{q+k,\nu}a_{q,\nu}a_{p,\sigma}+h.c.)\\
  &\quad\quad\quad\quad\quad\quad\quad\quad\quad
  \times\chi_{p-k\notin B^{\sigma}_F}\chi_{q+k\notin B^{\nu}_F}\chi_{p\in B^{\sigma}_F}\chi_{q\in B^{\nu}_F}
  \end{aligned}
\end{equation}
Moreover, for $\mu>0$ and $0<\gamma<\alpha$, we have
\begin{equation}\label{bound E_G_N}
 \begin{aligned}
 \pm\mathcal{E}_{\mathcal{G}_N}\lesssim& P_{\mathcal{G}_N}(\mathcal{N}_{ex})
 +\big(N^{-\frac{1}{6}}+N^{-\alpha-\gamma}+N^{-\frac{1}{2}+\alpha+2\gamma}\big)\mathcal{K}_s\\
 &+\big(\mu+N^{-\gamma}\big)\mathcal{V}_4+N^{\frac{1}{3}-\frac{\alpha}{2}}
 \end{aligned}
\end{equation}
with
\begin{equation}\label{bound P_G_N}
\begin{aligned}
  P_{\mathcal{G}_N}(\mathcal{N}_{ex})=&N^{-\alpha}(\mathcal{N}_{ex}+1)
  +\big(N^{-\frac{2}{3}-\frac{\alpha}{2}}+N^{-\frac{7}{9}+\alpha+\frac{7}{3}\gamma}\big)
  (\mathcal{N}_{ex}+1)^2\\
  &+\mu^{-1}N^{-\frac{4}{3}-\alpha}(\mathcal{N}_{ex}+1)^3.
\end{aligned}
\end{equation}
\end{proposition}
We leave the proof of Proposition \ref{p1} to Section \ref{qua}.

\subsection{Cubic Renormalization}\
\par For the cubic renormalization, we let the skew-symmetric operator be
\begin{equation}\label{define B' 0}
  B^\prime=A^\prime-{A^\prime}^*,
\end{equation}
where
\begin{equation}\label{define A' 0}
  A^\prime=\sum_{k,p,q,\sigma,\nu}
  \eta_ka^*_{p-k,\sigma}a^*_{q+k,\nu}a_{q,\nu}a_{p,\sigma}\chi_{p-k\notin B^{\sigma}_F}\chi_{q+k\notin B^{\nu}_F}\chi_{q\in A^\nu_{F,\kappa}}
  \chi_{p\in B^{\sigma}_F}.
\end{equation}
Here, recall (\ref{partition 3D space}),
\begin{equation}\label{A^nu_F,kappa 0}
  A_{F,\kappa}^\nu=\{k\in2\pi\mathbb{Z}^3,\,k_F^\nu<\vert k\vert\leq k_F^\nu+N^\kappa\}
\end{equation}
with $-\frac{11}{48}<\kappa<0$ to be determined. Due to technical reason, we introduce this cut-off parameter $\kappa$ to decompose the interaction between low and high momenta. Different from the Boson case studied, for example, in \cite{2018Bogoliubov}, the cubic renormalization for the Fermion system neutralize the extra quadratic energy $ Q_{\mathcal{G}_N}$ defined in (\ref{define C_G_N and Q_G_N p1}), since Fermion system is anti-symmetry. We consider the following excitation Hamiltonian
\begin{equation}\label{define J_N 0}
  \begin{aligned}
  \mathcal{J}_N\coloneqq e^{-B^\prime}\mathcal{G}_Ne^{B^\prime}
  \end{aligned}
\end{equation}
The effect of the cubic renormalization is recorded below.
\begin{proposition}\label{p2}
Under the same assumptions in Proposition \ref{p1}, we set $\mu=N^{-\gamma}$ and
$\kappa=-2(\alpha+\gamma)$, with the further requirement that $0<\gamma<\alpha<\frac{11}{192}$.
\begin{equation}\label{cal J_N p2}
  \mathcal{J}_N=C_{\mathcal{G}_N}
  +\mathcal{K}+\mathcal{V}_{4}+\mathcal{V}_{21}^\prime
  +\Omega+\mathcal{E}_{\mathcal{J}_N}
\end{equation}
where
\begin{equation}\label{bound E_J_N p2}
  \begin{aligned}
  \pm\mathcal{E}_{\mathcal{J}_N}\lesssim& P_{\mathcal{J}_N}(\mathcal{N}_{ex})
  +N^{-\gamma}\big(\mathcal{V}_4+(\mathcal{N}_{ex}+1)\big)\\
  &+\big(N^{-\alpha-\gamma}+N^{-\frac{1}{3}+2\alpha+3\gamma}
  \big)\mathcal{K}_s+N^{\frac{1}{3}-\frac{\alpha}{2}}
  \end{aligned}
\end{equation}
with
\begin{equation}\label{define P_J_N p2}
  \begin{aligned}
  P_{\mathcal{J}_N}(\mathcal{N}_{ex})=&
  N^{-\frac{4}{3}-\alpha+\gamma}(\mathcal{N}_{ex}+1)^3
  +\big(N^{-\frac{2}{3}-\frac{\alpha}{2}}
  +N^{-\frac{7}{9}+\alpha+\frac{7}{3}\gamma}\big)(\mathcal{N}_{ex}+1)^2\\
  &+N^{-\frac{1}{3}-\frac{3}{2}\alpha-\gamma}(\mathcal{N}_{ex}+1)^{\frac{3}{2}}.
  \end{aligned}
\end{equation}
\end{proposition}
We leave the proof of Proposition \ref{p2} to Section \ref{cub}.

\subsection{Bogoliubov Transformation}\
\par The ending Bogoliubov transformation uses the skew-symmetric opreator
\begin{equation}\label{define B tilde 0}
  \tilde{B}=\frac{1}{2}(\tilde{A}-\tilde{A}^*),
\end{equation}
where
\begin{equation}\label{define A tilde 0}
  \tilde{A}=\sum_{k,p,q,\sigma,\nu}
  \xi_{k,q,p}^{\nu,\sigma}a^*_{p-k,\sigma}a^*_{q+k,\nu}a_{q,\nu}a_{p,\sigma}\chi_{p-k\notin B^{\sigma}_F}\chi_{q+k\notin B^{\nu}_F}\chi_{p\in B^{\sigma}_F}\chi_{q\in B^{\nu}_F},
\end{equation}
with the coefficients given by (\ref{define xi_k,q,p,nu,sigma 0000}) in Section \ref{coeffsub}. We recall here
\begin{equation}\label{define xi_k,q,p,nu,sigma 0}
  \xi_{k,q,p}^{\nu,\sigma}=\frac{-\big(W_k+\eta_kk(q-p)\big)}
  {\frac{1}{2}\big(\vert q+k\vert^2+\vert p-k\vert^2-\vert q\vert^2-\vert p\vert^2\big)}
  \chi_{p-k\notin B^{\sigma}_F}\chi_{q+k\notin B^{\nu}_F}\chi_{p\in B^{\sigma}_F}\chi_{q\in B^{\nu}_F}.
\end{equation}
We consider the following excitation Hamiltonian
\begin{equation}\label{define Z_N 0}
  \begin{aligned}
  \mathcal{Z}_N\coloneqq e^{-\tilde{B}}\mathcal{J}_Ne^{\tilde{B}}
  \end{aligned}
\end{equation}
The effect of the Bogoliubov transformation is recorded below.
\begin{proposition}\label{p3}
  Under the same assumptions in Proposition \ref{p2}, with the further assumption that $0<\gamma<\alpha<\frac{1}{168}$, we have
  \begin{equation}\label{define Z_N p3}
    \mathcal{Z}_N=C_{\mathcal{Z}_N}
  +\mathcal{K}+e^{-\tilde{B}}\mathcal{V}_{4}e^{\tilde{B}}
  +\mathcal{E}_{\mathcal{Z}_N}
  \end{equation}
  with
  \begin{equation}\label{define C_Z_N}
    \begin{aligned}
    C_{\mathcal{Z}_N}=&C_{\mathcal{G}_N}
    +\sum_{k,p,q,\sigma,\nu}\big(W_k+\eta_kk(q-p)\big)\xi_{k,q,p}^{\nu,\sigma}
  \chi_{p-k\notin B^{\sigma}_F}\chi_{q+k\notin B^{\nu}_F}
  \chi_{p\in B^{\sigma}_F}\chi_{q\in B^{\nu}_F}\\
  &-\sum_{k,p,q,\sigma}\big(W_k+\eta_kk(q-p)\big)\xi_{(k+q-p),p,q}^{\sigma,\sigma}
  \chi_{p-k,q+k\notin B^{\sigma}_F}\chi_{p,q\in B^{\sigma}_F}\chi_{p\neq q}
    \end{aligned}
  \end{equation}
  and
  \begin{equation}\label{bound E_Z_N}
  \begin{aligned}
    \pm\mathcal{E}_{\mathcal{Z}_N}\lesssim
  &\big(N^{-\alpha-\gamma}
  +N^{-\frac{1}{3}+2\alpha+3\gamma}+N^{-\frac{25}{231}+\frac{33}{7}\alpha+
  \frac{23}{7}\gamma+\varepsilon}\big)
  \big(\mathcal{K}_s+N^{\frac{1}{3}+\alpha}\big)\\
  &+e^{-\tilde{B}}P_{\mathcal{J}_N}(\mathcal{N}_{ex})e^{\tilde{B}}
  +N^{-\gamma}e^{-\tilde{B}}\mathcal{V}_4e^{\tilde{B}}\\
  &+\big(N^{-\gamma}+N^{-\frac{1}{21}+\frac{33}{7}\alpha+\frac{23}{7}\gamma+\varepsilon}\big)
  (\mathcal{N}_{ex}+1)+N^{\frac{1}{3}-\frac{\alpha}{2}}.
  \end{aligned}
  \end{equation}
  for any $\varepsilon>0$ small enough.
\end{proposition}
We leave the proof of Proposition \ref{p3} to Section \ref{bog}.

\section{Proof of the Main Theorem}\label{proofofmain}
\subsection{On the Ground State Energy}\label{constant}
\begin{lemma}\label{Ground State}
Under the same assumptions in Proposition \ref{p3},
   \begin{equation}\label{computation}
  \begin{aligned}
    C_{\mathcal{Z}_N}=4\pi a\mathfrak{a}_0\sum_{\sigma\neq\nu}
    N_{\sigma}N_{\nu}+E^{(2)}+O(N^{\frac{1}{3}-\alpha+\varepsilon}).
  \end{aligned}
  \end{equation}
  with some $0<\varepsilon<\alpha$ small enough but fixed. For some cut-off constant $c_N=4\sup_{\sigma}k_F^\sigma$ (depending on $N$),
      \begin{equation}\label{rewrt E^(2) lemma}
        \begin{aligned}
         \frac{E^{(2)}}{2(4\pi a\mathfrak{a}_0)^2}&\coloneqq
    \sum_{\substack{\sigma,\nu\in\mathcal{S}_\mathbf{q}\\ \sigma\neq\nu}}
    \sum_{\substack{p,q,k\\0<\vert k\vert\leq c_N}}\Big(\frac{\chi_{p\in B^{\sigma}_F}\chi_{q\in B^{\nu}_F}}{2\vert k\vert^2}-\frac{
    \chi_{p-k\notin B^{\sigma}_F}\chi_{q+k\notin B^{\nu}_F}
    \chi_{p\in B^{\sigma}_F}\chi_{q\in B^{\nu}_F}}{\big(\vert q+k\vert^2+\vert p-k\vert^2-\vert q\vert^2-\vert p\vert^2\big)}\Big)\\
    &+\frac{1}{2}\sum_{\substack{\sigma,\nu\in\mathcal{S}_\mathbf{q}\\ \sigma\neq\nu}}
    \sum_{\substack{p,q,k,\\\vert k\vert> c_N}}
    \Big(\frac{\chi_{p\in B^{\sigma}_F}\chi_{q\in B^{\nu}_F}}{\vert k\vert^2}
    -\frac{
    \chi_{p\in B^{\sigma}_F}\chi_{q\in B^{\nu}_F}}{\big(\vert q+k\vert^2+\vert p-k\vert^2-\vert q\vert^2-\vert p\vert^2\big)}\\
    &\quad\quad\quad\quad\quad\quad\quad\quad
    -\frac{
    \chi_{p\in B^{\sigma}_F}\chi_{q\in B^{\nu}_F}}{\big(\vert q-k\vert^2+\vert p+k\vert^2-\vert q\vert^2-\vert p\vert^2\big)}\Big)
        \end{aligned}
      \end{equation}
      We can formally write (\ref{rewrt E^(2) lemma}) as
  \begin{equation}\label{E^(2) 6}
    E^{(2)}=2(4\pi a\mathfrak{a}_0)^2
    \sum_{\substack{\sigma,\nu\in\mathcal{S}_\mathbf{q}\\ \sigma\neq\nu}}
    \sum_{\substack{p,q\\k\neq0}}\Big(\frac{\chi_{p\in B^{\sigma}_F}\chi_{q\in B^{\nu}_F}}{2\vert k\vert^2}-\frac{
    \chi_{p-k\notin B^{\sigma}_F}\chi_{q+k\notin B^{\nu}_F}
    \chi_{p\in B^{\sigma}_F}\chi_{q\in B^{\nu}_F}}{\big(\vert q+k\vert^2+\vert p-k\vert^2-\vert q\vert^2-\vert p\vert^2\big)}\Big)
  \end{equation}
  Moreover, we have
  \begin{equation}\label{order E^2}
    E^{(2)}=O(N^{\frac{1}{3}+\varepsilon}).
  \end{equation}
\begin{proof}
By Propositions \ref{p1} and \ref{p3}, we write
\begin{equation}\label{rewrt C_Z_N}
  C_{\mathcal{Z}_N}=C_1+C_2+C_3
\end{equation}
with
\begin{equation}\label{details C_Z_N}
  \begin{aligned}
  C_1&=\frac{1}{2}\Big(\hat{v}_0+\sum_{k}\hat{v}_k\eta_k\Big)
  \sum_{\sigma\neq\nu}N_\sigma N_\nu\\
  C_2&=\sum_{k,p,q,\sigma,\nu}W_k\eta_k\chi_{p-k\notin B^{\sigma}_F}\chi_{q+k\notin B^{\nu}_F}\chi_{p\in B^{\sigma}_F}\chi_{q\in B^{\nu}_F}\\
  &-\sum_{k,p,q,\sigma}W_k\eta_{k+q-p}\chi_{p-k,q+k\notin B^{\sigma}_F}\chi_{p,q\in B^{\sigma}_F}\chi_{p\neq q}\\
  &+\sum_{k,p,q,\sigma,\nu}\eta_k^2k(q-p)\chi_{p-k\notin B^{\sigma}_F}\chi_{q+k\notin B^{\nu}_F}\chi_{p\in B^{\sigma}_F}\chi_{q\in B^{\nu}_F}\\
  &-\sum_{k,p,q,\sigma}\eta_{k+q-p}\eta_kk(q-p)\chi_{p-k,q+k\notin B^{\sigma}_F}\chi_{p,q\in B^{\sigma}_F}\chi_{p\neq q}\\
  C_3&=\sum_{k,p,q,\sigma,\nu}\big(W_k+\eta_kk(q-p)\big)\xi_{k,q,p}^{\nu,\sigma}
  \chi_{p-k\notin B^{\sigma}_F}\chi_{q+k\notin B^{\nu}_F}
  \chi_{p\in B^{\sigma}_F}\chi_{q\in B^{\nu}_F}\\
  &-\sum_{k,p,q,\sigma}\big(W_k+\eta_kk(q-p)\big)\xi_{(k+q-p),p,q}^{\sigma,\sigma}
  \chi_{p-k,q+k\notin B^{\sigma}_F}\chi_{p,q\in B^{\sigma}_F}\chi_{p\neq q}
  \end{aligned}
\end{equation}
Notice that in the definition of $C_2$ we have used the relation that
\begin{equation*}
  \begin{aligned}
  &-\sum_{k,p,q,\sigma,\nu}\eta_k^2k(q-p)(1-\chi_{p-k\notin B^{\sigma}_F}\chi_{q+k\notin B^{\nu}_F})\chi_{p\in B^{\sigma}_F}\chi_{q\in B^{\nu}_F}\\
  &=\sum_{k,p,q,\sigma,\nu}\eta_k^2k(q-p)\chi_{p-k\notin B^{\sigma}_F}\chi_{q+k\notin B^{\nu}_F}\chi_{p\in B^{\sigma}_F}\chi_{q\in B^{\nu}_F}.
  \end{aligned}
\end{equation*}
 We then compute (\ref{rewrt C_Z_N}) one by one.
\begin{flushleft}
  \textbf{Analysis of $C_1$:}
\end{flushleft}
\par Using (\ref{有用的屎}), we immediately obtain (recall that $\ell=N^{-\frac{1}{3}-\alpha}$)
\begin{equation}\label{C_11}
  C_1=\big(4\pi a\mathfrak{a}_0+6\pi a^2\mathfrak{a}_0^2\ell^{-1}\big)
  \sum_{\sigma\neq\nu}N_{\sigma}N_{\nu}+O(N^{-\frac{1}{3}+2\alpha}).
\end{equation}

\begin{flushleft}
  \textbf{Analysis of $C_2$:}
\end{flushleft}
\par Using (\ref{est of w_l,p}), (\ref{eta_p}), (\ref{est of eta_0}) and (\ref{sum_pW_peta_p 3dscatt}), we can bound
\begin{equation}\label{analysis of C_2 1st}
  \begin{aligned}
  \Big\vert\sum_{k,q,p,\sigma,\nu}W_k\eta_k\big(1-\chi_{p-k\notin B^{\sigma}_F}
  \chi_{q+k\notin B^{\nu}_F}\big)\chi_{p\in B^{\sigma}_F}\chi_{q\in B^{\nu}_F}\Big\vert
  &\lesssim N^{\frac{1}{3}-2\alpha}\\
  \Big\vert\sum_{k,q,p,\sigma}W_k\eta_{k+q-p}\big(1-\chi_{p-k,q+k\notin B^{\sigma}_F}
  \big)\chi_{p,q\in B^{\sigma}_F}\Big\vert
  &\lesssim N^{\frac{1}{3}-2\alpha}\\
  \Big\vert\sum_{k,q,p,\sigma,\nu}\eta_k^2k(q-p)\big(1-\chi_{p-k\notin B^{\sigma}_F}
  \chi_{q+k\notin B^{\nu}_F}\big)\chi_{p\in B^{\sigma}_F}\chi_{q\in B^{\nu}_F}\Big\vert
  &\lesssim N^{\frac{1}{3}-4\alpha}\\
  \Big\vert\sum_{k,q,p,\sigma}\eta_k\eta_{k+q-p}k(q-p)\big(1-\chi_{p-k,q+k\notin B^{\sigma}_F}
  \big)\chi_{p,q\in B^{\sigma}_F}\Big\vert
  &\lesssim N^{\frac{1}{3}-4\alpha}\\
  \Big\vert\sum_{k,p,q,\sigma}W_k\eta_{k+q-p}\chi_{p-k,q+k\notin B^{\sigma}_F}\chi_{p,q\in B^{\sigma}_F}\chi_{p= q}\Big\vert
  &\lesssim N^{-\frac{2}{3}+\alpha}
  \end{aligned}
\end{equation}
On the other hand, we write for $k\in2\pi\mathbb{Z}^3$,
\begin{equation}\label{define Z_p}
  Z_k\coloneqq \frac{1}{2}
  \sum_{l\in2\pi\mathbb{Z}^3}\hat{v}_{k-l}\eta_l+ \frac{1}{2}\hat{v}_k,
\end{equation}
and also denote $Z(x)$ the function with Fourier coefficients $Z_k$
\begin{equation}\label{define Z(x)}
  Z(x)=\sum_{k}Z_ke^{ikx}=\frac{1}{2}v_a(x)\big(1+\eta(x)\big).
\end{equation}
 Using (\ref{eqn of eta_p rewrt}), we find
 \begin{equation}\label{C_2,1}
   \begin{aligned}
   \big(\vert k\vert^2+k(q-p)\big)\eta_k
   &=\big(W_k+k(q-p)\eta_k\big)-Z_k\\
   \big(\vert k\vert^2+k(q-p)\big)\eta_{k+q-p}
   &=\big(W_{k+q-p}+(k+q-p)(p-q)\eta_k\big)-Z_{k+q-p}
   \end{aligned}
 \end{equation}
 Notice that
 \begin{equation}\label{C_2,2}
   \vert k\vert^2+k\cdot(q-p)
  =\Big\vert k+\frac{q-p}{2}\Big\vert^2-\Big\vert\frac{q-p}{2}\Big\vert^2.
 \end{equation}
 Recall (\ref{define S_q,p lemma}), for fixed $p\in B_F^\sigma$ and $q\in B^{\nu}_F$
 \begin{equation}\label{define S_q,p}
   \mathfrak{S}_{q,p}^{\nu,\sigma}
   \coloneqq\{k\in2\pi\mathbb{Z}^3\,\big\vert\,\vert k\vert^2+k\cdot(q-p)=0\}.
 \end{equation}
 By (\ref{approximate lattice 3d ball}), (\ref{approximate lattice 2d sphere}) and (\ref{C_2,2}), we can bound (see (\ref{N^1/3+epsilon case 3 1}))
 \begin{equation}\label{bound number S_q,p}
   \#\mathfrak{S}_{q,p}^{\nu,\sigma}\lesssim N^{\frac{1}{3}+\varepsilon}
 \end{equation}
 for any $\varepsilon>0$ small enough as long as $N$ is correspondingly large enough.
  Therefore, similar to (\ref{analysis of C_2 1st}), we can bound
  \begin{equation}\label{analysis of C_2 2nd}
  \begin{aligned}
  \Big\vert\sum_{k,q,p,\sigma,\nu}W_k\eta_k\chi_{k\in\mathfrak{S}_{q,p}^{\nu,\sigma}}
  \chi_{p\in B^{\sigma}_F}\chi_{q\in B^{\nu}_F}\Big\vert
  &\lesssim N^{-\frac{1}{3}-2\alpha+\varepsilon}\\
  \Big\vert\sum_{k,q,p,\sigma}W_k\eta_{k+q-p}\chi_{k\in\mathfrak{S}_{q,p}^{\sigma,\sigma}}
  \chi_{p,q\in B^{\sigma}_F}\Big\vert
  &\lesssim N^{-\frac{1}{3}-2\alpha+\varepsilon}\\
  \Big\vert\sum_{k,q,p,\sigma,\nu}\eta_k^2k(q-p)\chi_{k\in\mathfrak{S}_{q,p}^{\nu,\sigma}}
  \chi_{p\in B^{\sigma}_F}\chi_{q\in B^{\nu}_F}\Big\vert
  &\lesssim N^{-\frac{1}{3}-4\alpha+\varepsilon}\\
  \Big\vert\sum_{k,q,p,\sigma}\eta_k\eta_{k+q-p}k(q-p)
  \chi_{k\in\mathfrak{S}_{q,p}^{\sigma,\sigma}}
  \chi_{p,q\in B^{\sigma}_F}\Big\vert
  &\lesssim N^{-\frac{1}{3}-4\alpha+\varepsilon}
  \end{aligned}
\end{equation}
Combing (\ref{analysis of C_2 1st}), (\ref{C_2,1}) and (\ref{analysis of C_2 2nd}), we get to
\begin{equation}\label{rewrt C_2}
  \begin{aligned}
  C_2&=\sum_{k,p,q,\sigma,\nu}\frac{\big( W_k+\eta_kk(q-p)\big)^2}
  {\vert k\vert^2+k(q-p)}\chi_{k\notin\mathfrak{S}_{q,p}^{\nu,\sigma}}
  \chi_{p\in B^{\sigma}_F}\chi_{q\in B^{\nu}_F}\\
  &-\sum_{k,p,q,\sigma}\frac{\big( W_k+\eta_kk(q-p)\big)
  \big(W_{k+q-p}+\eta_{k+q-p}(k+q-p)(p-q)\big)}
  {\vert k\vert^2+k(q-p)}\\
  &\quad\quad\quad\quad
  \times\chi_{k\notin\mathfrak{S}_{q,p}^{\sigma,\sigma}}
  \chi_{p,q\in B^{\sigma}_F}\\
  &-\sum_{k,p,q,\sigma,\nu}\frac{\big( W_k+\eta_kk(q-p)\big)Z_k}
  {\vert k\vert^2+k(q-p)}\chi_{k\notin\mathfrak{S}_{q,p}^{\nu,\sigma}}
  \chi_{p\in B^{\sigma}_F}\chi_{q\in B^{\nu}_F}\\
  &+\sum_{k,p,q,\sigma}\frac{\big( W_k+\eta_kk(q-p)\big)Z_{k+q-p}}
  {\vert k\vert^2+k(q-p)}\chi_{k\notin\mathfrak{S}_{q,p}^{\sigma,\sigma}}
  \chi_{p,q\in B^{\sigma}_F}+O(N^{\frac{1}{3}-2\alpha})
  \end{aligned}
\end{equation}
\par Next, we are going to prove that, for the last two terms in (\ref{rewrt C_2}),
\begin{equation}\label{about Z_k}
  \begin{aligned}
  &-\sum_{k,p,q,\sigma,\nu}\frac{\big( W_k+\eta_kk(q-p)\big)Z_k}
  {\vert k\vert^2+k(q-p)}\chi_{k\notin\mathfrak{S}_{q,p}^{\nu,\sigma}}
  \chi_{p\in B^{\sigma}_F}\chi_{q\in B^{\nu}_F}\\
  &\quad\quad+\sum_{k,p,q,\sigma}\frac{\big( W_k+\eta_kk(q-p)\big)Z_{k+q-p}}
  {\vert k\vert^2+k(q-p)}\chi_{k\notin\mathfrak{S}_{q,p}^{\sigma,\sigma}}
  \chi_{p,q\in B^{\sigma}_F}\\
  &=-\sum_{\sigma\neq\nu}\sum_{p,q}\sum_{k\neq0}\frac{W_kZ_k}{\vert k\vert^2}
  \chi_{p\in B^{\sigma}_F}\chi_{q\in B^{\nu}_F}\\
    &+(4\pi a\mathfrak{a}_0)^2\sum_{\substack{\sigma,\nu\in\mathcal{S}_\mathbf{q}\\ \sigma\neq\nu}}
    \sum_{\substack{p,q,k,\\0<\vert k\vert\leq c_N}}
    \Big(\frac{\chi_{p\in B^{\sigma}_F}\chi_{q\in B^{\nu}_F}}{\vert k\vert^2}
    -\frac{
    2\chi_{k\notin\mathfrak{S}_{q,p}^{\nu,\sigma}}
    \chi_{p\in B^{\sigma}_F}\chi_{q\in B^{\nu}_F}}{\big(\vert q+k\vert^2+\vert p-k\vert^2-\vert q\vert^2-\vert p\vert^2\big)}\Big)\\
  &+(4\pi a\mathfrak{a}_0)^2\sum_{\substack{\sigma,\nu\in\mathcal{S}_\mathbf{q}\\ \sigma\neq\nu}}
    \sum_{\substack{p,q,k,\\\vert k\vert> c_N}}
    \Big(\frac{\chi_{p\in B^{\sigma}_F}\chi_{q\in B^{\nu}_F}}{\vert k\vert^2}
    -\frac{
    \chi_{p\in B^{\sigma}_F}\chi_{q\in B^{\nu}_F}}{\big(\vert q+k\vert^2+\vert p-k\vert^2-\vert q\vert^2-\vert p\vert^2\big)}\\
    &\quad\quad\quad\quad\quad\quad\quad\quad
    -\frac{
    \chi_{p\in B^{\sigma}_F}\chi_{q\in B^{\nu}_F}}{\big(\vert q-k\vert^2+\vert p+k\vert^2-\vert q\vert^2-\vert p\vert^2\big)}\Big)+O(N^{\frac{1}{3}-2\alpha+\varepsilon})
  \end{aligned}
\end{equation}
for any $0<\varepsilon<\alpha$ fixed. First we notice that for $p\in B_F^\sigma$ and $q\in B_F^\nu$,
\begin{equation}\label{Z_k+q-p-Z_k}
\begin{aligned}
  \vert Z_{k+q-p}-Z_k\vert&=
  \Big\vert\int_{\Lambda}\frac{1}{2}v_a(x)\big(1+\eta(x)\big)
  e^{-ikx}\big(e^{-i(q-p)x}-1\big)dx\Big\vert\\
  &\lesssim\int_{\Lambda}v_a(x)\vert q-p\vert\cdot\vert x\vert dx
  \lesssim a^2N^{\frac{1}{3}}.
\end{aligned}
\end{equation}
Also, by (\ref{sum_pW_peta_p 3dscatt}), (\ref{est of w_l,p}), (\ref{eta_p}) and (\ref{est of eta_0}), we have
\begin{equation}\label{sup W_k+k(q-p)eta_k}
  \vert W_k\vert+\vert\eta_kk(q-p)\vert\lesssim a.
\end{equation}
Therefore, we combine (\ref{Z_k+q-p-Z_k}), (\ref{sup W_k+k(q-p)eta_k}) and estimates (\ref{L1 Linfty est}), (\ref{too naive}) in Lemma \ref{L1 Linfty est lemma} and (\ref{N^1/3+epsilon}) in Lemma \ref{N^1/3+epsilon lemma} to find
\begin{equation}\label{finite corr 1}
  \Big\vert\sum_{k,p,q,\sigma}\frac{\big(W_k+\eta_kk(q-p)\big)
  \big(Z_{k+q-p}-Z_k\big)}{\vert k\vert^2+k(q-p)}
  \chi_{k\notin\mathfrak{S}_{q,p}^{\sigma,\sigma}}
  \chi_{p,q\in B^{\sigma}_F}\Big\vert\lesssim N^{-\frac{1}{3}}\ln N.
\end{equation}
Secondly, since $Z_k=Z_{-k}$ and $\vert Z_k\vert\lesssim a$, we can apply (\ref{L1 Linfty est}) in Lemma \ref{L1 Linfty est lemma}, together with (\ref{N^1/3+epsilon}) in Lemma \ref{N^1/3+epsilon lemma} and estimate (\ref{est of eta_0}) to obtain
\begin{equation}\label{finite corr 2}
  \Big\vert\sum_{\sigma\neq\nu}\sum_{k,p,q}\frac{\eta_kk(q-p)Z_k}
  {\vert k\vert^2+k(q-p)}\chi_{k\notin\mathfrak{S}_{q,p}^{\nu,\sigma}}
  \chi_{p\in B^{\sigma}_F}\chi_{q\in B^{\nu}_F}\Big\vert\lesssim N^{\frac{1}{3}-\alpha}.
\end{equation}
Now we have
\begin{equation}\label{about Z_k 2}
  \begin{aligned}
  &-\sum_{k,p,q,\sigma,\nu}\frac{\big( W_k+\eta_kk(q-p)\big)Z_k}
  {\vert k\vert^2+k(q-p)}\chi_{k\notin\mathfrak{S}_{q,p}^{\nu,\sigma}}
  \chi_{p\in B^{\sigma}_F}\chi_{q\in B^{\nu}_F}\\
  &\quad\quad+\sum_{k,p,q,\sigma}\frac{\big( W_k+\eta_kk(q-p)\big)Z_{k+q-p}}
  {\vert k\vert^2+k(q-p)}\chi_{k\notin\mathfrak{S}_{q,p}^{\sigma,\sigma}}
  \chi_{p,q\in B^{\sigma}_F}\\
  &=-\sum_{\sigma\neq\nu}\sum_{k,p,q}\Big(\frac{W_kZ_k}{\vert k\vert^2+k(q-p)}
  -\frac{W_kZ_k}{\vert k\vert^2}\Big)\chi_{k\notin\mathfrak{S}_{q,p}^{\nu,\sigma}}
  \chi_{p\in B^{\sigma}_F}\chi_{q\in B^{\nu}_F}\\
  &\quad\quad-\sum_{\sigma\neq\nu}\sum_{p,q}\sum_{k\neq0}\frac{W_kZ_k}{\vert k\vert^2}
  \chi_{p\in B^{\sigma}_F}\chi_{q\in B^{\nu}_F}\\
  &\quad\quad+\sum_{\sigma\neq\nu}\sum_{p,q}\sum_{k\neq0}\frac{W_kZ_k}{\vert k\vert^2}
  \chi_{k\in\mathfrak{S}_{q,p}^{\nu,\sigma}}
  \chi_{p\in B^{\sigma}_F}\chi_{q\in B^{\nu}_F}
  +O(N^{\frac{1}{3}-\alpha})
  \end{aligned}
\end{equation}
To reach (\ref{about Z_k}), we compute (\ref{about Z_k 2}) one by one. For the first term, we divide it into two cases.
\begin{flushleft}
  \textbf{Case $\mathrm{I}$: $\vert k\vert>c_N$}
\end{flushleft}
In this case, since $W_{k}=W_{-k}$, $Z_k=Z_{-k}$, we can rearrange the sum by
\begin{equation}\label{rearrangement W_kZ_k}
  \begin{aligned}
  &-\sum_{\sigma\neq\nu}\sum_{k,p,q}\Big(\frac{W_kZ_k}{\vert k\vert^2+k(q-p)}
  -\frac{W_kZ_k}{\vert k\vert^2}\Big)\chi_{k\notin\mathfrak{S}_{q,p}^{\nu,\sigma}}
  \chi_{p\in B^{\sigma}_F}\chi_{q\in B^{\nu}_F}\\
  &=-\frac{1}{2}\sum_{\sigma\neq\nu}\sum_{k,p,q}\Big(\frac{W_kZ_k}{\vert k\vert^2+k(q-p)}
  +\frac{W_kZ_k}{\vert k\vert^2-k(q-p)}
  -\frac{2W_kZ_k}{\vert k\vert^2}\Big)\chi_{k\notin\mathfrak{S}_{q,p}^{\nu,\sigma}}
  \chi_{p\in B^{\sigma}_F}\chi_{q\in B^{\nu}_F}\\
  &=-\sum_{\sigma\neq\nu}\sum_{k,p,q}\frac{W_kZ_k\big(k(q-p)\big)^2}
  {\vert k\vert^2\big(\vert k\vert^2+k(q-p)\big)\big(\vert k\vert^2-k(q-p)\big)}
  \chi_{k\notin\mathfrak{S}_{q,p}^{\nu,\sigma}}
  \chi_{p\in B^{\sigma}_F}\chi_{q\in B^{\nu}_F}.
  \end{aligned}
\end{equation}
Now we write
\begin{equation}\label{ob 1}
  W_kZ_k=(4\pi a\mathfrak{a}_0)^2+Z_k(W_k-4\pi a\mathfrak{a}_0)+4\pi a\mathfrak{a}_0
  (Z_k-4\pi a\mathfrak{a}_0).
\end{equation}
Notice that $c_N\sim N^{\frac{1}{3}}$, and since $\vert k\vert>c_N$ implies $k\notin\mathfrak{S}_{q,p}^{\nu,\sigma}$ and
\begin{equation}\label{ob 2}
   \vert k\vert^2\pm k(q-p)\geq \frac{1}{2}\vert k\vert^2.
\end{equation}
Also we have from Section \ref{scattering eqn sec},
\begin{equation}\label{useful bound Z_k, W_k}
  \begin{aligned}
  \vert Z_k\vert\lesssim a,&\quad \vert W_k\vert\lesssim a\\
  \vert Z_k-Z_0\vert\lesssim a^3\vert k\vert^2,&\quad
  \vert W_k-W_0\vert\lesssim a\ell^2\vert k\vert^2\\
  \vert Z_0-4\pi a\mathfrak{a}_0\vert\lesssim a^2\ell^{-1},&\quad
  \vert W_0-4\pi a\mathfrak{a}_0\vert\lesssim a^2\ell^{-1}
  \end{aligned}
\end{equation}
Therefore, we use (\ref{useful bound Z_k, W_k}) to bound for $\vert k\vert>c_N$,
\begin{equation}\label{C_2 case 1 1}
  \begin{aligned}
  &\Big\vert\sum_{\sigma\neq\nu}\sum_{k,p,q}\frac{Z_k(W_0-4\pi a\mathfrak{a}_0)\big(k(q-p)\big)^2}
  {\vert k\vert^2\big(\vert k\vert^2+k(q-p)\big)\big(\vert k\vert^2-k(q-p)\big)}
  \chi_{k\notin\mathfrak{S}_{q,p}^{\nu,\sigma}}
  \chi_{p\in B^{\sigma}_F}\chi_{q\in B^{\nu}_F}\Big\vert
  \lesssim N^{-\frac{1}{3}+\alpha}\\
  &\Big\vert\sum_{\sigma\neq\nu}\sum_{k,p,q}\frac{4\pi a\mathfrak{a}_0(Z_0-4\pi a\mathfrak{a}_0)\big(k(q-p)\big)^2}
  {\vert k\vert^2\big(\vert k\vert^2+k(q-p)\big)\big(\vert k\vert^2-k(q-p)\big)}
  \chi_{k\notin\mathfrak{S}_{q,p}^{\nu,\sigma}}
  \chi_{p\in B^{\sigma}_F}\chi_{q\in B^{\nu}_F}\Big\vert
  \lesssim N^{-\frac{1}{3}+\alpha}
  \end{aligned}
\end{equation}
On the other hand, we choose a cut-off constant $R_0\gg N^{\frac{1}{3}}\ln N$ to be determined, for $\vert k\vert>R_0$, we can use the first line of (\ref{useful bound Z_k, W_k}) to bound
\begin{equation}\label{C_2 case 1 2}
  \begin{aligned}
  &\Big\vert\sum_{\sigma\neq\nu}\sum_{k,p,q}\frac{Z_k(W_k-W_0)\big(k(q-p)\big)^2}
  {\vert k\vert^2\big(\vert k\vert^2+k(q-p)\big)\big(\vert k\vert^2-k(q-p)\big)}
  \chi_{k\notin\mathfrak{S}_{q,p}^{\nu,\sigma}}
  \chi_{p\in B^{\sigma}_F}\chi_{q\in B^{\nu}_F}\Big\vert
  \lesssim N^{\frac{2}{3}}R_0^{-1}\\
  &\Big\vert\sum_{\sigma\neq\nu}\sum_{k,p,q}\frac{4\pi a\mathfrak{a}_0(Z_k-Z_0)\big(k(q-p)\big)^2}
  {\vert k\vert^2\big(\vert k\vert^2+k(q-p)\big)\big(\vert k\vert^2-k(q-p)\big)}
  \chi_{k\notin\mathfrak{S}_{q,p}^{\nu,\sigma}}
  \chi_{p\in B^{\sigma}_F}\chi_{q\in B^{\nu}_F}\Big\vert
  \lesssim N^{\frac{2}{3}}R_0^{-1}
  \end{aligned}
\end{equation}
For $c_N<\vert k\vert<R_0$, we use the second line of (\ref{useful bound Z_k, W_k}) to obtain
\begin{equation}\label{C_2 case 1 3}
  \begin{aligned}
  &\Big\vert\sum_{\sigma\neq\nu}\sum_{k,p,q}\frac{Z_k(W_k-W_0)\big(k(q-p)\big)^2}
  {\vert k\vert^2\big(\vert k\vert^2+k(q-p)\big)\big(\vert k\vert^2-k(q-p)\big)}
  \chi_{k\notin\mathfrak{S}_{q,p}^{\nu,\sigma}}
  \chi_{p\in B^{\sigma}_F}\chi_{q\in B^{\nu}_F}\Big\vert
  \lesssim N^{\frac{2}{3}}\ell^2R_0\\
  &\Big\vert\sum_{\sigma\neq\nu}\sum_{k,p,q}\frac{4\pi a\mathfrak{a}_0(Z_k-Z_0)\big(k(q-p)\big)^2}
  {\vert k\vert^2\big(\vert k\vert^2+k(q-p)\big)\big(\vert k\vert^2-k(q-p)\big)}
  \chi_{k\notin\mathfrak{S}_{q,p}^{\nu,\sigma}}
  \chi_{p\in B^{\sigma}_F}\chi_{q\in B^{\nu}_F}\Big\vert
  \lesssim N^{\frac{2}{3}}a^2R_0
  \end{aligned}
\end{equation}
We then choose $R_0=\ell^{-1}$. Therefore, we conclude that for $\vert k\vert>c_N$
\begin{equation}\label{C_2 case 1 conclusion}
  \begin{aligned}
  &-\sum_{\sigma\neq\nu}\sum_{k,p,q}\Big(\frac{W_kZ_k}{\vert k\vert^2+k(q-p)}
  -\frac{W_kZ_k}{\vert k\vert^2}\Big)\chi_{k\notin\mathfrak{S}_{q,p}^{\nu,\sigma}}
  \chi_{p\in B^{\sigma}_F}\chi_{q\in B^{\nu}_F}\\
  &=\frac{(4\pi a\mathfrak{a}_0)^2}{2}\sum_{\sigma\neq\nu}\sum_{k,p,q}\Big(\frac{2}{\vert k\vert^2}-\frac{1}{\vert k\vert^2+k(q-p)}
  -\frac{1}{\vert k\vert^2-k(q-p)}
  \Big)\\
  &\quad\quad\quad\quad\quad\quad\quad\quad
  \times
  \chi_{p\in B^{\sigma}_F}\chi_{q\in B^{\nu}_F}+O(N^{\frac{1}{3}-\alpha}).
  \end{aligned}
\end{equation}
It is easy to check the order of the leading term on the right hand side is not greater than $N^{\frac{1}{3}}$.
\begin{flushleft}
  \textbf{Case $\mathrm{II}$: $\vert k\vert\leq c_N$}
\end{flushleft}
Using the observation (\ref{ob 1}) and the useful bounds (\ref{useful bound Z_k, W_k}), together with Lemma \ref{N^1/3+epsilon lemma} (the same result still holds for $\vert k\vert\leq c_N$), we obtain the following result for $\vert k\vert\leq c_N\lesssim N^{\frac{1}{3}}$,
\begin{equation}\label{C_2 case 2}
  \begin{aligned}
  &-\sum_{\sigma\neq\nu}\sum_{k,p,q}\Big(\frac{W_kZ_k}{\vert k\vert^2+k(q-p)}
  -\frac{W_kZ_k}{\vert k\vert^2}\Big)\chi_{k\notin\mathfrak{S}_{q,p}^{\nu,\sigma}}
  \chi_{p\in B^{\sigma}_F}\chi_{q\in B^{\nu}_F}\\
  &=(4\pi a\mathfrak{a}_0)^2\sum_{ \sigma\neq\nu}
    \sum_{p,q,k}
    \Big(\frac{1}{\vert k\vert^2}
    -\frac{
    2}{\big(\vert q+k\vert^2+\vert p-k\vert^2-\vert q\vert^2-\vert p\vert^2\big)}\Big)
    \\
    &\quad\quad\quad\quad\quad\quad\quad\quad\times
    \chi_{k\notin\mathfrak{S}_{q,p}^{\nu,\sigma}}
    \chi_{p\in B^{\sigma}_F}\chi_{q\in B^{\nu}_F}+O(N^{\frac{1}{3}-2\alpha+\varepsilon})
  \end{aligned}
\end{equation}
for any $\varepsilon>0$ small but fixed.

\par For the third term in (\ref{about Z_k 2}), notice that $k\in\mathfrak{S}_{q,p}^{\nu,\sigma}$ implies $\vert k\vert\lesssim N^{\frac{1}{3}}$. We then use (\ref{ob 1}), (\ref{useful bound Z_k, W_k}) and (\ref{bound number S_q,p}) to get to
\begin{equation}\label{3rd term about Z_k2}
  \begin{aligned}
  &\sum_{\sigma\neq\nu}\sum_{p,q}\sum_{k\neq0}\frac{W_kZ_k}{\vert k\vert^2}
  \chi_{k\in\mathfrak{S}_{q,p}^{\nu,\sigma}}
  \chi_{p\in B^{\sigma}_F}\chi_{q\in B^{\nu}_F}\\
  &\quad\quad=
  \sum_{\sigma\neq\nu}\sum_{p,q}\sum_{k\neq0}\frac{(4\pi a\mathfrak{a}_0)^2}{\vert k\vert^2}
  \chi_{k\in\mathfrak{S}_{q,p}^{\nu,\sigma}}
  \chi_{p\in B^{\sigma}_F}\chi_{q\in B^{\nu}_F}+O(N^{\frac{1}{3}-\alpha}).
  \end{aligned}
\end{equation}
Plugging (\ref{C_2 case 1 conclusion}), (\ref{C_2 case 2}) and (\ref{3rd term about Z_k2}) into (\ref{about Z_k 2}), we reach (\ref{about Z_k}).

\begin{flushleft}
  \textbf{Analysis of $C_2+C_3$:}
\end{flushleft}
Plugging (\ref{about Z_k}) into (\ref{rewrt C_2}), recall the definition of $\xi_{k,q,p}^{\sigma,\nu}$ in (\ref{define xi_k,q,p,nu,sigma 0000}), and notice the fact that $p-k\notin B_F^\sigma$ and $q+k\notin B_F^\nu$ implies $k\notin\mathfrak{S}_{q,p}^{\nu,\sigma}$, we then get to
\begin{equation}\label{C_2+C_3}
  \begin{aligned}
  &C_2+C_3=-\sum_{\sigma\neq\nu}\sum_{p,q}\sum_{k\neq0}\frac{W_kZ_k}{\vert k\vert^2}
  \chi_{p\in B^{\sigma}_F}\chi_{q\in B^{\nu}_F}\\
  &+\sum_{k,p,q,\sigma,\nu}\frac{\big( W_k+\eta_kk(q-p)\big)^2}
  {\vert k\vert^2+k(q-p)}\chi_{k\notin\mathfrak{S}_{q,p}^{\nu,\sigma}}
  (1-\chi_{p-k\notin B^\sigma_F}\chi_{q+k\notin B^\nu_F})\chi_{p\in B^{\sigma}_F}\chi_{q\in B^{\nu}_F}\\
  &-\sum_{k,p,q,\sigma}\frac{\big( W_k+\eta_kk(q-p)\big)
  \big(W_{k+q-p}+\eta_{k+q-p}(k+q-p)(p-q)\big)}
  {\vert k\vert^2+k(q-p)}\\
  &\quad\quad\quad\quad
  \times\chi_{k\notin\mathfrak{S}_{q,p}^{\sigma,\sigma}}
  (1-\chi_{p-k\notin B^\sigma_F}\chi_{q+k\notin B^\sigma_F})
  \chi_{p,q\in B^{\sigma}_F}\\
    &+(4\pi a\mathfrak{a}_0)^2\sum_{\substack{\sigma,\nu\in\mathcal{S}_\mathbf{q}\\ \sigma\neq\nu}}
    \sum_{\substack{p,q,k,\\0<\vert k\vert\leq c_N}}
    \Big(\frac{\chi_{p\in B^{\sigma}_F}\chi_{q\in B^{\nu}_F}}{\vert k\vert^2}
    -\frac{
    2\chi_{k\notin\mathfrak{S}_{q,p}^{\nu,\sigma}}
    \chi_{p\in B^{\sigma}_F}\chi_{q\in B^{\nu}_F}}{\big(\vert q+k\vert^2+\vert p-k\vert^2-\vert q\vert^2-\vert p\vert^2\big)}\Big)\\
  &+(4\pi a\mathfrak{a}_0)^2\sum_{\substack{\sigma,\nu\in\mathcal{S}_\mathbf{q}\\ \sigma\neq\nu}}
    \sum_{\substack{p,q,k,\\\vert k\vert> c_N}}
    \Big(\frac{\chi_{p\in B^{\sigma}_F}\chi_{q\in B^{\nu}_F}}{\vert k\vert^2}
    -\frac{
    \chi_{p\in B^{\sigma}_F}\chi_{q\in B^{\nu}_F}}{\big(\vert q+k\vert^2+\vert p-k\vert^2-\vert q\vert^2-\vert p\vert^2\big)}\\
    &\quad\quad\quad\quad\quad\quad\quad\quad
    -\frac{
    \chi_{p\in B^{\sigma}_F}\chi_{q\in B^{\nu}_F}}{\big(\vert q-k\vert^2+\vert p+k\vert^2-\vert q\vert^2-\vert p\vert^2\big)}\Big)+O(N^{\frac{1}{3}-2\alpha+\varepsilon})
  \end{aligned}
\end{equation}
As $(1-\chi_{p-k\notin B^\sigma_F}\chi_{q+k\notin B^\nu_F})\neq0$ implies $\vert k\vert\lesssim N^{\frac{1}{3}}$, for the second and third terms of (\ref{C_2+C_3}), we can use Lemma \ref{N^1/3+epsilon lemma} together with estimates (\ref{est of eta_0}) and (\ref{sum_pW_peta_p 3dscatt}) to deduce
\begin{equation}\label{C_2+C_3 part 1}
  \begin{aligned}
  &\sum_{k,p,q,\sigma,\nu}\frac{\big( W_k+\eta_kk(q-p)\big)^2}
  {\vert k\vert^2+k(q-p)}\chi_{k\notin\mathfrak{S}_{q,p}^{\nu,\sigma}}
  (1-\chi_{p-k\notin B^\sigma_F}\chi_{q+k\notin B^\nu_F})\chi_{p\in B^{\sigma}_F}\chi_{q\in B^{\nu}_F}\\
  &-\sum_{k,p,q,\sigma}\frac{\big( W_k+\eta_kk(q-p)\big)
  \big(W_{k+q-p}+\eta_{k+q-p}(k+q-p)(p-q)\big)}
  {\vert k\vert^2+k(q-p)}\\
  &\quad\quad\quad\quad
  \times\chi_{k\notin\mathfrak{S}_{q,p}^{\sigma,\sigma}}
  (1-\chi_{p-k\notin B^\sigma_F}\chi_{q+k\notin B^\sigma_F})
  \chi_{p,q\in B^{\sigma}_F}\\
  &=\sum_{k,p,q,\sigma,\nu}\frac{ W_k^2}
  {\vert k\vert^2+k(q-p)}\chi_{k\notin\mathfrak{S}_{q,p}^{\nu,\sigma}}
  (1-\chi_{p-k\notin B^\sigma_F}\chi_{q+k\notin B^\nu_F})\chi_{p\in B^{\sigma}_F}\chi_{q\in B^{\nu}_F}\\
  &-\sum_{k,p,q,\sigma}\frac{W_k
  W_{k+q-p}}
  {\vert k\vert^2+k(q-p)}
  \chi_{k\notin\mathfrak{S}_{q,p}^{\sigma,\sigma}}
  (1-\chi_{p-k\notin B^\sigma_F}\chi_{q+k\notin B^\sigma_F})
  \chi_{p,q\in B^{\sigma}_F}+O(N^{\frac{1}{3}-2\alpha+\varepsilon})
  \end{aligned}
\end{equation}
for any $\varepsilon>0$ small but fixed. Moreover, using Lemma \ref{N^1/3+epsilon lemma} together with the fact that
\begin{equation}\label{W_k diff}
  \vert W_{k+q-p}-W_k\vert\lesssim a\ell\vert q-p\vert\lesssim a\ell N^{\frac{1}{3}},
\end{equation}
we attain from \ref{C_2+C_3 part 1} that
\begin{equation}\label{C_2+C_3 part2}
  \begin{aligned}
  &\sum_{k,p,q,\sigma,\nu}\frac{ W_k^2}
  {\vert k\vert^2+k(q-p)}\chi_{k\notin\mathfrak{S}_{q,p}^{\nu,\sigma}}
  (1-\chi_{p-k\notin B^\sigma_F}\chi_{q+k\notin B^\nu_F})\chi_{p\in B^{\sigma}_F}\chi_{q\in B^{\nu}_F}\\
  &-\sum_{k,p,q,\sigma}\frac{W_k
  W_{k+q-p}}
  {\vert k\vert^2+k(q-p)}
  \chi_{k\notin\mathfrak{S}_{q,p}^{\sigma,\sigma}}
  (1-\chi_{p-k\notin B^\sigma_F}\chi_{q+k\notin B^\sigma_F})
  \chi_{p,q\in B^{\sigma}_F}\\
  &=\sum_{\sigma\neq\nu}\sum_{k,p,q}\frac{ W_k^2}
  {\vert k\vert^2+k(q-p)}\chi_{k\notin\mathfrak{S}_{q,p}^{\nu,\sigma}}
  (1-\chi_{p-k\notin B^\sigma_F}\chi_{q+k\notin B^\nu_F})\chi_{p\in B^{\sigma}_F}\chi_{q\in B^{\nu}_F}\\
  &\quad\quad\quad\quad+O(N^{\frac{1}{3}-\alpha+\varepsilon})
  \end{aligned}
\end{equation}
for any $\varepsilon>0$ small but fixed. Once again we use (\ref{useful bound Z_k, W_k}) to get to
\begin{equation}\label{C_2+C_3 part 3}
  \begin{aligned}
  &\sum_{\sigma\neq\nu}\sum_{k,p,q}\frac{ W_k^2}
  {\vert k\vert^2+k(q-p)}\chi_{k\notin\mathfrak{S}_{q,p}^{\nu,\sigma}}
  (1-\chi_{p-k\notin B^\sigma_F}\chi_{q+k\notin B^\nu_F})\chi_{p\in B^{\sigma}_F}\chi_{q\in B^{\nu}_F}\\
  &=\sum_{\sigma\neq\nu}\sum_{k,p,q}\frac{ (4\pi a\mathfrak{a}_0)^2}
  {\vert k\vert^2+k(q-p)}\chi_{k\notin\mathfrak{S}_{q,p}^{\nu,\sigma}}
  (1-\chi_{p-k\notin B^\sigma_F}\chi_{q+k\notin B^\nu_F})\chi_{p\in B^{\sigma}_F}\chi_{q\in B^{\nu}_F}\\
  &\quad\quad\quad\quad+O(N^{\frac{1}{3}-2\alpha+\varepsilon})
  \end{aligned}
\end{equation}
Combining (\ref{C_2+C_3}-\ref{C_2+C_3 part 3}), we deduce
\begin{equation}\label{C_2+C_3 conclusion}
   \begin{aligned}
  &C_2+C_3=-\sum_{\sigma\neq\nu}\sum_{p,q}\sum_{k\neq0}\frac{W_kZ_k}{\vert k\vert^2}
  \chi_{p\in B^{\sigma}_F}\chi_{q\in B^{\nu}_F}\\
    &+(4\pi a\mathfrak{a}_0)^2\sum_{\substack{\sigma,\nu\in\mathcal{S}_\mathbf{q}\\ \sigma\neq\nu}}
    \sum_{\substack{p,q,k,\\0<\vert k\vert\leq c_N}}
    \Big(\frac{\chi_{p\in B^{\sigma}_F}\chi_{q\in B^{\nu}_F}}{\vert k\vert^2}
    -\frac{
    2\chi_{p-k\notin B^\sigma_F}\chi_{q+k\notin B^\nu_F}
    \chi_{p\in B^{\sigma}_F}\chi_{q\in B^{\nu}_F}}{\big(\vert q+k\vert^2+\vert p-k\vert^2-\vert q\vert^2-\vert p\vert^2\big)}\Big)\\
  &+(4\pi a\mathfrak{a}_0)^2\sum_{\substack{\sigma,\nu\in\mathcal{S}_\mathbf{q}\\ \sigma\neq\nu}}
    \sum_{\substack{p,q,k,\\\vert k\vert> c_N}}
    \Big(\frac{\chi_{p\in B^{\sigma}_F}\chi_{q\in B^{\nu}_F}}{\vert k\vert^2}
    -\frac{
    \chi_{p\in B^{\sigma}_F}\chi_{q\in B^{\nu}_F}}{\big(\vert q+k\vert^2+\vert p-k\vert^2-\vert q\vert^2-\vert p\vert^2\big)}\\
    &\quad\quad\quad\quad\quad\quad\quad\quad
    -\frac{
    \chi_{p\in B^{\sigma}_F}\chi_{q\in B^{\nu}_F}}{\big(\vert q-k\vert^2+\vert p+k\vert^2-\vert q\vert^2-\vert p\vert^2\big)}\Big)+O(N^{\frac{1}{3}-\alpha+\varepsilon})
  \end{aligned}
\end{equation}
where the order of the term on the second line of (\ref{C_2+C_3 conclusion}) is no greater than $N^{\frac{1}{3}+\varepsilon}$ according to an argument by combining (\ref{L1 est case 2 W}) and (\ref{L1 est case 3 W}).

\begin{flushleft}
  \textbf{Conclusion of Lemma \ref{Ground State}:}
\end{flushleft}
Combining (\ref{C_11}) and (\ref{C_2+C_3 conclusion}), and recall the definition of $E^{(2)}$ in (\ref{rewrt E^(2) lemma}), we have
\begin{equation}\label{conclusion constant}
\begin{aligned}
  C_{\mathcal{Z}_N}=&\big(4\pi a\mathfrak{a}_0+6\pi a^2\mathfrak{a}_0^2\ell^{-1}\big)
  \sum_{\sigma\neq\nu}N_{\sigma}N_{\nu}
  -\sum_{\sigma\neq\nu}\sum_{p,q}\sum_{k\neq0}\frac{W_kZ_k}{\vert k\vert^2}
  \chi_{p\in B^{\sigma}_F}\chi_{q\in B^{\nu}_F}\\
  &+E^{(2)}+O(N^{\frac{1}{3}-\alpha+\varepsilon})
\end{aligned}
\end{equation}
for any $\varepsilon>0$ small but fixed. We claim that
\begin{equation}\label{bose}
\begin{aligned}
  6\pi a^2\mathfrak{a}_0^2\ell^{-1}
  \sum_{\sigma\neq\nu}N_{\sigma}N_{\nu}
  -\sum_{\sigma\neq\nu}\sum_{p,q}\sum_{k\neq0}\frac{W_kZ_k}{\vert k\vert^2}
  \chi_{p\in B^{\sigma}_F}\chi_{q\in B^{\nu}_F}=\mathfrak{e}+O(N^{-\frac{1}{3}+2\alpha})
\end{aligned}
\end{equation}
where $\mathfrak{e}=O(1)$ independent of $\ell$ is given by
\begin{equation}\label{define e_0}
    \mathfrak{e}=\Big(2a^2\mathfrak{a}_0^2
-\lim_{M\to\infty}\sum_{\substack{p\in\mathbb{Z}^3\backslash\{0\}\\
\vert p_1\vert,\vert p_2\vert,\vert p_3\vert\leq
M}}\frac{4a^2\mathfrak{a}_0^2\cos(\vert p\vert)}{\vert p\vert^2}\Big)
\sum_{\sigma\neq\nu}N_{\sigma}N_{\nu}
\end{equation}
This term also appears in the second order behaviour of Boson systems in the Gross-Pitaevskii regime, for detailed proof one can check \cite[Lemma 5.4]{2018Bogoliubov} or \cite[Lemma 5.1]{me}.
\end{proof}

\end{lemma}
\subsection{Excitation Estimate}\label{ee}
\
\par Putting together Proposition \ref{p3} and Lemma \ref{Ground State}, we let
\begin{equation}\label{U*H_NU}
 \mathcal{Z}_N=\mathcal{U}^*H_N\mathcal{U}
=E_0+\mathcal{K}_s+e^{-\tilde{B}}\mathcal{V}_4e^{\tilde{B}}+
\mathcal{E}
\end{equation}
where
\begin{equation}\label{define U}
  \mathcal{U}=e^Be^{B^\prime}e^{\tilde{B}}
\end{equation}
and
\begin{equation}\label{define E_0}
  E_0=\sum_{\sigma,k\in B_F^{\sigma}}\vert k\vert^2+4\pi a\mathfrak{a}_0\sum_{\sigma\neq\nu}N_{\sigma}N_{\nu}+E^{(2)}
\end{equation}
We have the following estimate regarding the number of excitation states (sort of an optimal BEC estimate as in the Boson case, see for example, \cite{me}):
\begin{proposition}\label{Optimal BEC}
  \begin{equation}\label{Exci Esti Optimal BEC}
    H_N- E_0\gtrsim\mathcal{N}_{ex}-N^{\frac{1}{3}-\frac{1}{360}}.
  \end{equation}
\end{proposition}
\begin{proof}
  \par We let $f,g:\mathbb{R}\to[0,1]$ be smooth functions such that $f(s)^2+g(s)^2=1$ for all
$s\in\mathbb{R}$, and $f(s)=1$ for $s\leq\frac{1}{2}$, $f(s)=0$ for $s\geq1$. For
\begin{equation}\label{define M}
  M=CN^{\frac{2}{3}}
\end{equation}
with some large but universal constant $C>0$ to be determined, we define
\begin{equation}\label{f_M g_M}
  f_M(s)=f(s/M),\quad g_M(s)=g(s/M).
\end{equation}
Then we can calculate directly
\begin{equation}\label{H_N region I}
  H_N=f_M(\mathcal{N}_{ex})H_Nf_M(\mathcal{N}_{ex})+
g_M(\mathcal{N}_{ex})H_Ng_M(\mathcal{N}_{ex})+\mathcal{E}_M,
\end{equation}
where
\begin{equation}\label{E_M define}
  \mathcal{E}_M=\frac{1}{2}
\left([f_M(\mathcal{N}_{ex}),[f_M(\mathcal{N}_{ex}),H_N]]
+[g_M(\mathcal{N}_{ex}),[g_M(\mathcal{N}_{ex}),H_N]]\right).
\end{equation}
We claim that
\begin{equation}\label{bound E_M}
  \pm\mathcal{E}_M\lesssim M^{-2}(\mathcal{V}_4+N).
\end{equation}
We leave the proof of (\ref{bound E_M}) to Lemma \ref{lemma bound E_M} below.

\par Now we deal with the first term on the right hand side of (\ref{H_N region I}). Using
(\ref{U*H_NU}), we can obtain (we write $f_M=f_M(\mathcal{N}_{ex})$ for short)
\begin{equation}\label{f_MH_Nf_M}
\begin{aligned}
  f_MH_Nf_M=&f_M\mathcal{U}\mathcal{U}^*H_N\mathcal{U}\mathcal{U}^*f_M\\
  =&f_M\mathcal{U}\big(E_0+\mathcal{K}_s+e^{-\tilde{B}}\mathcal{V}_4e^{\tilde{B}}+
\mathcal{E}\big)\mathcal{U}^*f_M.
\end{aligned}
\end{equation}
Let $\psi\in\mathcal{H}^{\wedge N}(\{N_{\varsigma_i}\})$, by the definition
of $f$, we can verify on the Fock space that for all $n\in\frac{1}{2}\mathbb{N}$ and $n>1$
\begin{equation}\label{N_ex<M}
  \langle\mathcal{N}_{ex}^nf_M\psi,f_M\psi\rangle
\leq M^{n-1}\langle \mathcal{N}_{ex}f_M\psi,
f_M\psi\rangle.
\end{equation}
By (\ref{N_ex<M}) and together with Lemmas \ref{control eB on N_ex lemma}, \ref{control eB' on N_ex lemma} and \ref{control eBtilde on N_ex lemma}, we obtain
\begin{equation}\label{N_ex<M f_M}
  \begin{aligned}
  &f_M\mathcal{U}e^{-\tilde{B}}(\mathcal{N}_{ex}+1)^ne^{\tilde{B}}\mathcal{U}^*f_M
  =f_Me^Be^{{B}^\prime}(\mathcal{N}_{ex}+1)^ne^{-{B}^\prime}e^{-B}f_M\\
  &\quad\quad\quad\quad\lesssim
  f_M(\mathcal{N}_{ex}+1)^nf_M\lesssim M^{n-1}f_M(\mathcal{N}_{ex}+1)f_M\\
  &\quad\quad\quad\quad=
  M^{n-1}f_M\mathcal{U}\mathcal{U}^*(\mathcal{N}_{ex}+1)\mathcal{U}\mathcal{U}^*f_M\\
  &\quad\quad\quad\quad\lesssim
  M^{n-1}f_M\mathcal{U}\big[(\mathcal{N}_{ex}+1)+N^{-\frac{2}{33}}
  \big(\mathcal{K}_s+N^{\frac{1}{3}+\alpha}\big)\big]\mathcal{U}^*f_M.
  \end{aligned}
\end{equation}
With the help of (\ref{N_ex<M f_M}) and (\ref{ineqn K_s origin}) i.e. $\mathcal{N}_{ex}\lesssim\mathcal{K}_s$,  we apply Proposition \ref{p3} and Lemma \ref{Ground State} with $\varepsilon>0$ small enough but fixed, and
\begin{equation}\label{choose alpha}
   \alpha=2\gamma=\frac{1}{180}
\end{equation}
to reach
\begin{equation}\label{f_ME_f_M}
\begin{aligned}
  f_M\mathcal{U}\mathcal{E}\mathcal{U}^*f_M
  \gtrsim -N^{-\frac{1}{360}}f_M\mathcal{U}\big\{
  \mathcal{K}_s+e^{-\tilde{B}}\mathcal{V}_4e^{\tilde{B}}+N^{\frac{1}{3}}\big\}
  \mathcal{U}^*f_M.
\end{aligned}
\end{equation}
Since $\mathcal{K}_s\gtrsim\mathcal{N}_{ex}$ and $\mathcal{V}_4\geq0$, we then combine (\ref{f_MH_Nf_M}) and (\ref{f_ME_f_M}) to attain
\begin{equation}\label{mechante1}
  f_M(H_N-E_0)f_M
\gtrsim f_M(\mathcal{N}_+)^2(\mathcal{N}_{ex}-N^{\frac{1}{3}-\frac{1}{360}}).
\end{equation}

\par Now we turn to the second term on the right hand side of (\ref{H_N region I}). We are
going to prove by contradiction, that (we write $g_M=g_M(\mathcal{N}_{ex})$ for short)
\begin{equation}\label{mechante2}
   g_M(H_N-E_0)g_M\gtrsim Ng_M(\mathcal{N}_+)^2
\geq \mathcal{N}_{ex}g_M(\mathcal{N}_{ex})^2.
\end{equation}
By the definition of $g_M$, we observe that
\begin{align*}
g_M(\mathcal{N}_+)(H_N-E_0)g_M(\mathcal{N}_+)
\geq\left(\inf_{\psi\in Y,\Vert\psi\Vert_2=1}\frac{1}{N}
\langle(H_N-E_0)\psi,\psi\rangle\right)Ng_M(\mathcal{N}_+)^2.
\end{align*}
Here $Y\subset \mathcal{H}^{\wedge N}(\{N_{\varsigma_i}\})$ is explicitly given by
\begin{equation}\label{Y}
  Y=U_N^*\bigoplus_{n=0}^{N-\frac{M}{2}}\bigoplus_{s_n\in\mathcal{P}^{(n)}}
  (\mathcal{F}^{\perp})^{\wedge(N-n)}.
\end{equation}
 Then to prove (\ref{mechante2}) it is sufficient to prove
\begin{equation}\label{mechante2.1}
  \inf_{\psi\in Y,\Vert\psi\Vert_2=1}\frac{1}{N}
\langle(H_N-E_0)\psi,\psi\rangle\gtrsim1.
\end{equation}
As we have already attained from Lemmas \ref{0 order approximation lemma} and \ref{Ground State} that
\begin{equation}\label{mechante2222}
  \inf_{\psi\in Y,\Vert\psi\Vert_2=1}\frac{1}{N}
\langle(H_N-E_0)\psi,\psi\rangle
\geq \inf_{\substack{\psi\in \mathcal{H}^{\wedge N}(\{N_{\varsigma_i}\})\\\Vert\psi\Vert_2=1}}\frac{1}{N}
\langle (H_N-E_0)\psi,\psi\rangle\gtrsim -1,
\end{equation}
assuming (\ref{mechante2.1}) is not true, we can find a family of $\{N_j\}$
in the Gross-Piaevskii regime, and $\psi_j\in Y_j$ (the subscript $j$ implies this
space $Y$ depends on $j$) with $\Vert\psi_j\Vert_2=1$, such that
\begin{align*}
   \frac{1}{N_j}\Big\vert\langle H_{N_j}\psi_j,\psi_j\rangle-\sum_{\sigma,k\in B_F^{\sigma}}\vert k\vert^2\Big\vert\lesssim 1.
\end{align*}
That is to say, $\{\psi_j\}$ is a family of approximate ground state of order $0$. Then by Lemma \ref{lemma prori BEC} we infer that
\begin{equation}\label{mechante2.2}
  \langle\mathcal{N}_{ex}\psi_j,\psi_j\rangle\leq
cN_j^{\frac{2}{3}}.
\end{equation}
On the other hand, since $\psi_j\in Y_j$, and recall the definition of $M$ (\ref{define M}), we have
\begin{equation}\label{mechante2.3}
\frac{1}{N_j}\langle\mathcal{N}_{ex}\psi_j,\psi_j\rangle\geq
\frac{M_j}{2N_j}=\frac{1}{2}CN_j^{-\frac{1}{3}}
> cN_j^{-\frac{1}{3}}.
\end{equation}
which contradicts (\ref{mechante2.2}) since we can choose $C>2c$. Hence (\ref{mechante2.1}) and thus (\ref{mechante2}) hold.

\par To conclude (\ref{Exci Esti Optimal BEC}), we first combine (\ref{H_N region I}), (\ref{bound E_M}), (\ref{mechante1}) and (\ref{mechante2}) to attain
\begin{equation}\label{11111}
  H_N-E_0\gtrsim \mathcal{N}_{ex}-N^{\frac{1}{3}-\frac{1}{360}}-N^{-\frac{4}{3}}\mathcal{V}_4.
\end{equation}
By Lemmas \ref{control eB on N_ex lemma} and \ref{corollary control eB on V_4}, we know
\begin{equation}\label{22222}
  \mathcal{G}_N-E_0\gtrsim \mathcal{N}_{ex}-N^{\frac{1}{3}-\frac{1}{360}}-N^{-\frac{4}{3}}\mathcal{V}_4.
\end{equation}
On the other hand, we apply Proposition \ref{p3} and Lemma \ref{Ground State} again with (\ref{choose alpha}) and the naive fact that, for any $n>1$ and $n\in\frac{1}{2}\mathbb{N}$
\begin{equation}\label{naive}
  \begin{aligned}
  (\mathcal{N}_{ex}+1)^n\lesssim N^{n-1}(\mathcal{N}_{ex}+1),\quad
  0\leq\mathcal{N}_{ex}\lesssim\mathcal{K}_s
  \end{aligned}
\end{equation}
we reach
\begin{equation}\label{33333}
  \mathcal{Z}_N-E_0\gtrsim e^{-\tilde{B}}\mathcal{V}_4e^{\tilde{B}}
  -N^{\frac{2}{3}-\frac{1}{360}}e^{-\tilde{B}}(\mathcal{N}_{ex}+1)e^{\tilde{B}}.
\end{equation}
By Lemma \ref{control eB' on N_ex lemma} and Corollary \ref{control eB' on K,V_4}, and using again (\ref{naive}), we attain
\begin{equation}\label{44444}
   \mathcal{G}_N-E_0\gtrsim \mathcal{V}_4
  -N^{\frac{2}{3}-\frac{1}{360}}(\mathcal{N}_{ex}+1).
\end{equation}
Combining (\ref{22222}) and (\ref{44444}), we obtain
\begin{equation}\label{55555}
  \mathcal{G}_N-E_0\gtrsim \mathcal{N}_{ex}-N^{\frac{1}{3}-\frac{1}{360}}.
\end{equation}
Thus we conclude the proof by applying Lemma \ref{control eB on N_ex lemma} on (\ref{55555}).
\end{proof}

\begin{lemma}\label{lemma bound E_M}
Let $\mathcal{E}_M$ be defined in (\ref{E_M define}), we have (\ref{bound E_M}), which is
  \begin{equation}\label{bound E_M lemma}
  \pm\mathcal{E}_M\lesssim M^{-2}(\mathcal{V}_4+N).
\end{equation}
\end{lemma}
\begin{proof}
  Using the fact that
\begin{equation}\label{com N_ex a_k,sigma}
  [\mathcal{N}_{ex},a_{k,\sigma}]=-\chi_{k\notin B_F^\sigma}a_{k,\sigma},
\end{equation}
and together with (\ref{split V detailed}), we can calculate for any bounded real function $h$ point-wisely defined on $\mathbb{R}$,
\begin{equation}\label{cal com [h,[h,H_N]]}
  \begin{aligned}
  [h(\mathcal{N}_{ex}),[h(\mathcal{N}_{ex}),H_N]]=\mathcal{E}_{h,1}+\mathcal{E}_{h,21}
  +\mathcal{E}_{h,3}
  \end{aligned}
\end{equation}
with
\begin{equation}\label{split com [h,[h,H_N]]}
  \begin{aligned}
  &\mathcal{E}_{h,1}=\sum_{k,p,q,\sigma,\nu}
  \hat{v}_k\big\{a^*_{p-k,\sigma}a^*_{q+k,\nu}a_{q,\nu}a_{p,\sigma}
  \big(h(\mathcal{N}_{ex}+1)-h(\mathcal{N}_{ex})\big)^2+h.c.\big\}\\
  &\quad\quad\quad\quad\quad\quad\quad\quad\times
  \chi_{p-k\notin B^{\sigma}_F}\chi_{p\in B^{\sigma}_F}\chi_{q,q+k\in B^{\nu}_F}\\
  &\mathcal{E}_{h,21}=\frac{1}{2}\sum_{k,p,q,\sigma,\nu}
  \hat{v}_k\big\{a^*_{p-k,\sigma}a^*_{q+k,\nu}a_{q,\nu}a_{p,\sigma}
  \big(h(\mathcal{N}_{ex}+2)-h(\mathcal{N}_{ex})\big)^2+h.c.\big\}\\
  &\quad\quad\quad\quad\quad\quad\quad\quad\times
  \chi_{p-k\notin B^{\sigma}_F}\chi_{q+k\notin B^{\nu}_F}\chi_{p\in B^{\sigma}_F}\chi_{q\in B^{\nu}_F}\\
  &\mathcal{E}_{h,3}=\sum_{k,p,q,\sigma,\nu}
  \hat{v}_k\big\{a^*_{p-k,\sigma}a^*_{q+k,\nu}a_{q,\nu}a_{p,\sigma}
  \big(h(\mathcal{N}_{ex}+1)-h(\mathcal{N}_{ex})\big)^2+h.c.\big\}\\
  &\quad\quad\quad\quad\quad\quad\quad\quad\times
  \chi_{p-k\notin B^{\sigma}_F}\chi_{p\in B^{\sigma}_F}\chi_{q,q+k\notin B^{\nu}_F}
  \end{aligned}
\end{equation}
Then (\ref{bound E_M lemma}) is deduced by controlling these there terms individually, for $h=f_M,\,g_M$. We treat $\mathcal{E}_{f_M,21}$ as a typical example, other terms can be bounded similarly. Let $\psi\in\mathcal{H}^{\wedge N}(\{N_{\varsigma_i}\})$, estimating on position space, we know that, since
\begin{equation}\label{E_M 1}
  \begin{aligned}
  \mathcal{E}_{f_M,21}=&\frac{1}{2}\sum_{\sigma,\nu}\int_{\Lambda^2}
  v_a(x-y)a^*(h_{x,\sigma})a^*(h_{y,\nu})a(g_{y,\nu})a(g_{x,\sigma})\\
  &\quad\quad\quad\quad\times
  \big(h(\mathcal{N}_{ex}+1)-h(\mathcal{N}_{ex})\big)^2dxdy+h.c.
  \end{aligned}
\end{equation}
then
\begin{equation}\label{E_M_2}
  \begin{aligned}
  &\vert\langle \mathcal{E}_{f_M,21},\psi,\psi\rangle\vert
  \lesssim \Big(\sum_{\sigma,\nu}\int_{\Lambda^2}
  v_a(x-y)\langle a^*(h_{x,\sigma})a^*(h_{y,\nu})a(h_{y,\nu})a(h_{x,\sigma})\psi,
  \psi\rangle\Big)^{\frac{1}{2}}\\
  &\quad\quad\quad\quad\times
  \Big(\sum_{\sigma,\nu}\int_{\Lambda^2}
  v_a(x-y) N^2\langle\big(h(\mathcal{N}_{ex}+2)-h(\mathcal{N}_{ex})\big)^4 \psi,
  \psi\rangle\Big)^{\frac{1}{2}}
  \end{aligned}
\end{equation}
Notice that for $h=f_M,\,g_M$ and $t>0$,
\begin{equation}\label{mechante4}
\begin{aligned}
\vert h(s+t)-h(s)\vert&=0,\quad \text{$s<\frac{M}{2}-t$ or $s>M$}\\
 \vert h(s+t)-h(s)\vert&\lesssim\frac{t}{M},\quad\text{$\frac{M}{2}-t\leq s\leq M$}
\end{aligned}
\end{equation}
Therefore,
\begin{equation}\label{mechante5}
  \pm \big(h(\mathcal{N}_{ex}+2)-h(\mathcal{N}_{ex})\big)^4\lesssim M^{-4},
\end{equation}
which yields
\begin{equation}\label{mechante6}
  \pm\mathcal{E}_{f_M,21}\lesssim M^{-2}(\mathcal{V}_4+N).
\end{equation}
\end{proof}

\subsection{Proof of Theorem \ref{core}}\label{proof core}
\begin{proof}[Upper bound]
  \par We can choose the test function
\begin{equation*}
  \psi=e^Be^{B^\prime}e^{\tilde{B}}\bigwedge_{\sigma,k\in B_F^{\sigma}}f_{k,\sigma}.
\end{equation*}
In other words, $\psi=e^Be^{B^\prime}e^{\tilde{B}}f_{s_N}$, where $s_N$ is the unique element of $\mathcal{P}^{(N)}$. Notice that $\Vert\psi\Vert=1$. By definition,
\begin{equation*}
  \langle H_N\psi,\psi\rangle=\langle\mathcal{Z}_Nf_{s_N},f_{s_N}\rangle.
\end{equation*}
We then apply Proposition \ref{p3} and Lemma \ref{Ground State} (In other words, using equation (\ref{U*H_NU})) with (\ref{choose alpha}). This time we claim that for $\vert t\vert\leq 1$ and $n\in\frac{1}{2}\mathbb{N}$,
\begin{equation}\label{claim appB}
  \langle\mathcal{N}_{ex}^ne^{t\tilde{B}}f_{s_N},e^{t\tilde{B}}f_{s_N}\rangle\lesssim 1.
\end{equation}
Moreover we use Lemmas \ref{control eBtiled on K_s lemma} and \ref{control eBtilde on V_4 lemma}. On the other hand, we can calculate directly using (\ref{formulae of crea and annihi ops}) to obtain
\begin{equation*}
  \mathcal{V}_4f_{s_N},\,\mathcal{N}_{ex}f_{s_N},\,\mathcal{K}_sf_{s_N}=0.
\end{equation*}
Thus we conclude that
\begin{equation}\label{uppper}
  \langle\mathcal{Z}_Nf_{s_N},f_{s_N}\rangle-E_0\lesssim N^{\frac{1}{3}-\frac{1}{360}},
\end{equation}
which gives the upper bound of $E_{\{N_{\varsigma_i}\},\mathbf{q}}$. We detour the proof of (\ref{claim appB}) to Appendix \ref{claim}.
\end{proof}

\begin{proof}[Lower bound]
  \par The lower bound of $E_{\{N_{\varsigma_i}\},\mathbf{q}}$ is a direct consequence of the excitations estimate (\ref{Exci Esti Optimal BEC}) from Proposition \ref{Optimal BEC} since obviously $\mathcal{N}_{ex}\geq0$.
\end{proof}

\section{Quadratic Renormalization}\label{qua}
\par In this section, we analyse the excitation Hamiltonian $\mathcal{G}_N$ defined in (\ref{define G_N 0}) and prove Proposition \ref{p1}. Recall that we always set
\begin{equation}\label{choose l}
  \ell=N^{-\frac{1}{3}-\alpha}
\end{equation}
with $0<\alpha<\frac{1}{6}$. Recall (\ref{define B 0}) and (\ref{define A 0})
\begin{equation}\label{define B}
  B=\frac{1}{2}(A-A^*)
\end{equation}
and
\begin{equation}\label{define A}
  A=\sum_{k,p,q,\sigma,\nu}
  \eta_ka^*_{p-k,\sigma}a^*_{q+k,\nu}a_{q,\nu}a_{p,\sigma}\chi_{p-k\notin B^{\sigma}_F}\chi_{q+k\notin B^{\nu}_F}\chi_{p\in B^{\sigma}_F}\chi_{q\in B^{\nu}_F}.
\end{equation}

\par Using (\ref{second quantization H_N}) and Newton-Leibniz formula,  we rewrite $\mathcal{G}_N$ by
\begin{equation}\label{define G_N}
  \begin{aligned}
  \mathcal{G}_N\coloneqq e^{-B}H_Ne^B=&\mathcal{K}+\mathcal{V}_{4}+\mathcal{V}_{21}^\prime
  +\Omega+e^{-B}\big(\mathcal{V}_{0}
  +\mathcal{V}_{1}+\mathcal{V}_{22}+\mathcal{V}_{23}+\mathcal{V}_{3}\big)e^B\\
  &+\int_{0}^{1}e^{-tB}\Gamma e^{tB}dt
  +\int_{0}^{1}\int_{0}^{t}e^{-sB}[\mathcal{V}_{21}^\prime+\Omega,B]e^{sB}dsdt\\
  &+\int_{0}^{1}\int_{t}^{1}e^{-sB}[\mathcal{V}_{21},B]e^{sB}dsdt
  \end{aligned}
\end{equation}
with
\begin{equation}\label{define V_21' and Omega}
  \begin{aligned}
  \mathcal{V}_{21}^\prime&=\sum_{k,p,q,\sigma,\nu}
  W_k(a^*_{p-k,\sigma}a^*_{q+k,\nu}a_{q,\nu}a_{p,\sigma}+h.c.)\chi_{p-k\notin B^{\sigma}_F}\chi_{q+k\notin B^{\nu}_F}\chi_{p\in B^{\sigma}_F}\chi_{q\in B^{\nu}_F}\\
  \Omega&=\sum_{k,p,q,\sigma,\nu}
  \eta_kk(q-p)(a^*_{p-k,\sigma}a^*_{q+k,\nu}a_{q,\nu}a_{p,\sigma}+h.c.)\\
  &\quad\quad\quad\quad\quad\quad\quad\quad\quad
  \times\chi_{p-k\notin B^{\sigma}_F}\chi_{q+k\notin B^{\nu}_F}\chi_{p\in B^{\sigma}_F}\chi_{q\in B^{\nu}_F}
  \end{aligned}
\end{equation}
and
\begin{equation}\label{define Gamma}
  \Gamma=[\mathcal{K}+\mathcal{V}_4,B]+\mathcal{V}_{21}-\mathcal{V}_{21}^\prime-\Omega.
\end{equation}
We then prove Proposition \ref{p1} by analysing each term on the right hand side of (\ref{define G_N}) up to some error terms that may be discarded in the large $N$ limit. Corresponding results and detailed proofs are collected in Lemmas \ref{cal eB on V_0,1,22,23 lemma}-\ref{cal com V_21 lemma}. To control some of the error terms that may not contribute to the ground state energy up to the second order, we first Lemmas \ref{control eB on N_ex lemma}, \ref{variant lem 7.1 lemma}, \ref{corollary control eB on V_4} and Corollary \ref{corollary control eB on K_s} that control the action of $e^B$ on $\mathcal{N}_{ex}$, $\mathcal{K}_s$ and $\mathcal{V}_4$, which are fairly useful in the up-coming discussion. Lemma \ref{control eB on N_ex lemma} collects estimates that control the action of $e^B$ on several kinds of particle number operators defined in Section \ref{number of p ops}.

\begin{lemma}\label{control eB on N_ex lemma}
  Let $\ell$ be chosen as in (\ref{choose l}). Then for $n\in\frac{1}{2}\mathbb{Z}$, $\vert t\vert\leq1$, we have
  \begin{align}
    e^{-tB}(\mathcal{N}_{ex}+1)^ne^{tB} & \lesssim (\mathcal{N}_{ex}+1)^n \label{control eB on N_ex}\\
    \pm\big(e^{-tB}(\mathcal{N}_{ex}+1)^ne^{tB}-(\mathcal{N}_{ex}+1)^n\big) & \lesssim
    \ell^{\frac{1}{2}}(\mathcal{N}_{ex}+1)^n\label{control diff eB on N_ex}
  \end{align}
  Moreover, (\ref{control eB on     N_ex}) and (\ref{control diff eB on N_ex}) still hold if we replace the $\mathcal{N}_{ex}$ on the left hand side by $\mathcal{N}_{ex,\sigma}$. On the other hand, for $-\frac{1}{3}\leq\delta_1<\frac{1}{3}$ and $\delta_2\geq-\frac{1}{3}$,
  \begin{align}
    e^{-tB}\mathcal{N}_{h}[\delta_2]e^{tB}&\lesssim \mathcal{N}_{h}[\delta_2]+\ell(\mathcal{N}_{ex}+1)\label{control eB on N_h}\\
    e^{-tB}\mathcal{N}_{i}[\delta_1]e^{tB}&\lesssim \mathcal{N}_{i}[\delta_1]+ \mathcal{N}_{h}[\delta_1]+\ell(\mathcal{N}_{ex}+1)\label{control eB on N_i}\\
    e^{-tB}\mathcal{N}_{h}[\delta_2]\mathcal{N}_{ex}e^{tB}&\lesssim
    \mathcal{N}_{h}[\delta_2]\mathcal{N}_{ex}+\ell(\mathcal{N}_{ex}+1)^2\label{jjbound}
  \end{align}
\end{lemma}
\begin{proof}
  \par Calculating directly using Lemma \ref{lem commutator} yields $[\mathcal{N}_{ex},A]=2A$ and therefore $[\mathcal{N}_{ex}, B]=(A+A^*)$. For $\psi\in\mathcal{H}^{\wedge N}(\{N_{\varsigma_i}\})$, we can bound
  \begin{equation}\label{Lemma7.1 1}
  \begin{aligned}
    &\vert\langle A\psi,\psi\rangle\vert
    =\Big\vert\sum_{\sigma,\nu}\int_{\Lambda^2}\eta(x-y)\langle a^*(h_{x,\sigma})a^*(h_{y,\nu})
    a(g_{y,\nu})a(g_{x,\sigma})\psi,\psi\rangle dxdy\Big\vert\\
    &\lesssim\sum_{\sigma,\nu}\int_{\Lambda^2}\vert\eta(x-y)\vert
    \Vert (\mathcal{N}_{ex}+3)^{-\frac{1}{2}}a(h_{y,\nu})a(h_{x,\sigma})\psi\Vert
    \Vert (\mathcal{N}_{ex}+3)^{\frac{1}{2}}a(g_{y,\nu})a(g_{x,\sigma})\psi\Vert\\
    &\lesssim \Big(\sum_{\sigma,\nu}\int_{\Lambda^2}\langle
    a^*(h_{x,\sigma})a^*(h_{y,\nu})a(h_{y,\nu})a(h_{x,\sigma})
     (\mathcal{N}_{ex}+1)^{-\frac{1}{2}}\psi,(\mathcal{N}_{ex}+1)^{-\frac{1}{2}}\psi
     \rangle\Big)^{\frac{1}{2}}\\
     &\quad\quad\quad\quad
     \times\Big(\sum_{\sigma,\nu}\int_{\Lambda^2}\vert\eta(x-y)\vert^2N^2
     \langle(\mathcal{N}_{ex}+3)^{\frac{1}{2}}\psi,\psi\rangle\Big)^{\frac{1}{2}}\\
     &\lesssim N\Vert\eta\Vert_2\langle(\mathcal{N}_{ex}+1)\psi,\psi\rangle
     \lesssim \ell^{\frac{1}{2}}\langle(\mathcal{N}_{ex}+1)\psi,\psi\rangle
    \end{aligned}
  \end{equation}
  where we have used (\ref{est of eta and eta_perp}) in the last inequality. By the fact that
  \begin{equation*}
    \frac{d}{dt}\big(e^{-tB}(\mathcal{N}_{ex}+1)e^{tB}\big)
    =e^{-tB}[\mathcal{N}_{ex},B]e^{tB},
  \end{equation*}
  we reach (\ref{control eB on     N_ex}) and (\ref{control diff eB on N_ex}) for $n=1$ by Gronwall's inequality and Newton-Leibniz formula. For the proof of (\ref{control eB on     N_ex}) and (\ref{control diff eB on N_ex}) for arbitrary $n\in\frac{1}{2}\mathbb{Z}$, one can see the proof of \cite[Lemma 7.1]{me} for details. Moreover, (\ref{control eB on     N_ex}) and (\ref{control diff eB on N_ex}) still hold if we replace the $\mathcal{N}_{ex}$ on the left hand side by $\mathcal{N}_{ex,\sigma}$. Since by Lemma \ref{lem commutator}, for fixed $\sigma\in\mathcal{S}_{\mathbf{q}}$, we have
  \begin{equation*}
    [\mathcal{N}_{ex,\sigma},A]=2\sum_{k,p,q,\nu}
  \eta_ka^*_{p-k,\sigma}a^*_{q+k,\nu}a_{q,\nu}a_{p,\sigma}\chi_{p-k\notin B^{\sigma}_F}\chi_{q+k\notin B^{\nu}_F}\chi_{p\in B^{\sigma}_F}\chi_{q\in B^{\nu}_F},
  \end{equation*}
  thus we can argue similar to (\ref{Lemma7.1 1}).
  \par The proofs of (\ref{control eB on N_h}) and (\ref{control eB on N_i}) are also similar. Take the proof of (\ref{control eB on N_h}) for instance, by Lemma \ref{lem commutator}, we have
  \begin{equation*}
    [\mathcal{N}_{h}[\delta_2],A]=2\sum_{k,p,q,\sigma,\nu}
  \eta_ka^*_{p-k,\sigma}a^*_{q+k,\nu}a_{q,\nu}a_{p,\sigma}\chi_{p-k\in P^{\sigma}_{F,\delta_2}}\chi_{q+k\notin B^{\nu}_F}\chi_{p\in B^{\sigma}_F}\chi_{q\in B^{\nu}_F}.
  \end{equation*}
  Following the evaluation in (\ref{Lemma7.1 1}), we can bound
  \begin{equation}\label{Lemma7.1 2}
  \begin{aligned}
    &\frac{1}{2}\big\vert\langle [\mathcal{N}_{h}[\delta_2],A]\psi,\psi\rangle\big\vert
    =\Big\vert\sum_{\sigma,\nu}\int_{\Lambda^2}\eta(x-y)\langle a^*(H_{x,\sigma}[\delta_2])a^*(h_{y,\nu})
    a(g_{y,\nu})a(g_{x,\sigma})\psi,\psi\rangle\Big\vert\\
    &\lesssim\sum_{\sigma,\nu}\int_{\Lambda^2}\vert\eta(x-y)\vert
    \Vert (\mathcal{N}_{ex}+3)^{-\frac{1}{2}}a(h_{y,\nu})a(H_{x,\sigma}[\delta_2])\psi\Vert\\
    &\quad\quad\quad\quad\times
    \Vert (\mathcal{N}_{ex}+3)^{\frac{1}{2}}a(g_{y,\nu})a(g_{x,\sigma})\psi\Vert\\
    &\lesssim \Big(\sum_{\sigma,\nu}\int_{\Lambda^2}\langle
    a^*(H_{x,\sigma}[\delta_2])a^*(h_{y,\nu})a(h_{y,\nu})(\mathcal{N}_{ex}+2)^{-1}
    a(H_{x,\sigma}[\delta_2])
     \psi,\psi
     \rangle\Big)^{\frac{1}{2}}\\
     &\quad\quad\quad\quad
     \times\Big(\sum_{\sigma,\nu}\int_{\Lambda^2}\vert\eta(x-y)\vert^2N^2
     \langle(\mathcal{N}_{ex}+3)^{\frac{1}{2}}\psi,\psi\rangle\Big)^{\frac{1}{2}}\\
     &\lesssim N\Vert\eta\Vert_2
     \langle\mathcal{N}_h[\delta_2]\psi,\psi\rangle^{\frac{1}{2}}
     \langle(\mathcal{N}_{ex}+1)\psi,\psi\rangle^{\frac{1}{2}}
     \lesssim\langle\mathcal{N}_h[\delta_2]\psi,\psi\rangle
     + \ell\langle(\mathcal{N}_{ex}+1)\psi,\psi\rangle.
    \end{aligned}
  \end{equation}
  Combining (\ref{Lemma7.1 2}) with (\ref{control eB on     N_ex}), we reach (\ref{control eB on N_h}) via Gronwall's inequality. On the other hand, we also have
  \begin{equation*}
    [\mathcal{N}_{i}[\delta_1],A]=2\sum_{k,p,q,\sigma,\nu}
  \eta_ka^*_{p-k,\sigma}a^*_{q+k,\nu}a_{q,\nu}a_{p,\sigma}\chi_{p-k\notin B^{\sigma}_{F}}\chi_{q+k\notin B^{\nu}_F}\chi_{p\in \underline{B}^{\sigma}_{F,\delta_1}}\chi_{q\in B^{\nu}_F}.
  \end{equation*}
  Switching to the position space, we rewrite $[\mathcal{N}_{i}[\delta_1],A]=A_{\delta_1,1}+A_{\delta_1,2}$ with
  \begin{equation}\label{A_delta_1 1and2}
    \begin{aligned}
    A_{\delta_1,1}&=2\sum_{\sigma,\nu}\int_{\Lambda^2}\eta(x-y) a^*(H_{x,\sigma}[\delta_1])a^*(h_{y,\nu})
    a(g_{y,\nu})a(I_{x,\sigma}[\delta_1])dxdy\\
     A_{\delta_1,2}&=2\sum_{\sigma,\nu}\int_{\Lambda^2}\eta(x-y) a^*(L_{x,\sigma}[\delta_1])a^*(h_{y,\nu})
    a(g_{y,\nu})a(I_{x,\sigma}[\delta_1])dxdy
    \end{aligned}
  \end{equation}
  Similar to (\ref{Lemma7.1 2}), we can bound $A_{\delta_1,1}$ by
  \begin{equation}\label{bound A_delta_1,1}
   \pm A_{\delta_1,1}\lesssim \mathcal{N}_h[\delta_1]+\ell(\mathcal{N}_{ex}+1),
  \end{equation}
  since $\Vert a(I_{x,\sigma}[\delta_1])\Vert^2\leq N$. For the second term, we can bound
  \begin{equation}\label{bound A_delta_1,2}
    \begin{aligned}
    \vert\langle A_{\delta_1,2}\psi,\psi\rangle\vert
    &\lesssim\sum_{\sigma,\nu}\int_{\Lambda^2}
    \vert\eta(x-y)\vert\Vert (\mathcal{N}_{ex}+3)^{-\frac{1}{2}}a^*(I_{x,\sigma}[\delta_1])
    a(h_{y,\nu})a(L_{x,\sigma}[\delta_1])\psi\Vert\\
    &\quad\quad\quad\quad
    \times\Vert(\mathcal{N}_{ex}+3)^{\frac{1}{2}}a(g_{y,\nu})\psi\Vert \\
    &\lesssim\langle\mathcal{N}_i[\delta_1]\psi,\psi\rangle
    +N^{-\frac{1}{3}+\delta_1+\varepsilon(\delta_1)}\ell
    \langle(\mathcal{N}_{ex}+1)\psi,\psi\rangle
    \end{aligned}
  \end{equation}
  where we have used (\ref{ineqn N_l and N_s}) to bound $\Vert a^*(L_{x,\sigma}[\delta_1])\Vert$. Since $\delta_1<\frac{1}{3}$, we can combine (\ref{bound A_delta_1,1}), (\ref{bound A_delta_1,2}) and (\ref{control eB on N_h}) to reach (\ref{control eB on N_i}) via Gronwall's inequality.

  \par For the proof of (\ref{jjbound}), we notice that
  \begin{equation*}
    [\mathcal{N}_{h}[\delta_2]\mathcal{N}_{ex},A]=[\mathcal{N}_h[\delta_2],A]\mathcal{N}_{ex}
    +\mathcal{N}_h[\delta_2][\mathcal{N}_{ex},A].
  \end{equation*}
  Therefore, via (\ref{Lemma7.1 1}), (\ref{Lemma7.1 2}), and  the fact that $\mathcal{N}_h[\delta_2]\leq\mathcal{N}_{ex}$, we can bound
  \begin{equation}\label{Lemma7.1 3}
    \pm[\mathcal{N}_{h}[\delta_2]\mathcal{N}_{ex},A]\lesssim
    \mathcal{N}_{h}[\delta_2]\mathcal{N}_{ex}+\ell(\mathcal{N}_{ex}+1)^2.
  \end{equation}
  Thence (\ref{jjbound}) can be deduced via Gronwall's inequality.
\end{proof}

\par The next lemma is a useful variant of Lemma \ref{control eB on N_ex lemma}. We will use it in the proofs of Lemmas \ref{control eB on V_3 lemma}-\ref{cal com V_21 lemma}.
\begin{lemma}\label{variant lem 7.1 lemma}
 For $\varphi\in L^1\cap L^\infty(\Lambda)$ and $\varphi\geq0$, we define
 \begin{equation}\label{define N[phi_ex]}
   \mathcal{N}_{ex}[\varphi]\coloneqq\sum_{\sigma}
   \int_{\Lambda}\varphi(x)a^*(h_{x,\sigma})a(h_{x,\sigma})dx.
 \end{equation}
 Then $\mathcal{N}_{ex}[\varphi]$ is at least well-defined on the subset of smooth compactly-supported functions. For any $\vert t\vert\leq1$, we have

 \begin{equation}\label{variant lem7.1}
   e^{-tB}\mathcal{N}_{ex}[\varphi]e^{tB}\lesssim \mathcal{N}_{ex}[\varphi]
   +\ell\Vert\varphi\Vert_1(\mathcal{N}_{ex}+1).
 \end{equation}
 Moreover, for any $n\in\frac{1}{2}\mathbb{Z}$,
 \begin{equation}\label{variant lem important}
    e^{-tB}\mathcal{N}_{ex}[\varphi](\mathcal{N}_{ex}+1)^ne^{tB}\lesssim \mathcal{N}_{ex}[\varphi](\mathcal{N}_{ex}+1)^n
   +\ell\Vert\varphi\Vert_1(\mathcal{N}_{ex}+1)^{n+1}.
 \end{equation}
\end{lemma}
\begin{proof}
  \par Notice that
  \begin{equation*}
    \mathcal{N}_{ex}[\varphi]=\sum_{\tau,r,r^\prime}\int_{\Lambda}
    \varphi(x)e^{i(r-r^\prime)x}a^*_{r,\tau}a_{r^\prime,\tau}dx.
  \end{equation*}
  Formal calculation yields
  \begin{equation*}
  \begin{aligned}
    [&a^*(h_{x,\sigma})a(h_{x,\sigma}),A]\\
    &=2\sum_{k,p,q,\nu}
    \eta_ke^{-i(p-k) x}a^*(h_{x,\sigma})
    a^*_{q+k,\nu}a_{q,\nu}a_{p,\sigma}\chi_{p-k\notin B^{\sigma}_F}\chi_{q+k\notin B^{\nu}_F}\chi_{p\in B^{\sigma}_F}\chi_{q\in B^{\nu}_F}.
    \end{aligned}
  \end{equation*}
  Adopting the notation
  \begin{equation*}
    k_{y_1,y_2,\sigma}(z)\coloneqq\sum_{q\in B^{\sigma}_F}\Big(
    \sum_{p\notin B_F^\sigma}\int_{\Lambda}\eta(x-y_1)
    e^{iqx}e^{ip(y_2-x)}dx\Big)f_{q,\sigma}(z),
  \end{equation*}
  one can check that for $j=1,2$
  \begin{equation*}
    \int_{\Lambda}\Vert k_{y_1,y_2,\sigma}\Vert^2_2dy_j\leq\sum_{q\in B_F^\sigma}
    \Vert\eta\Vert_2^2.
  \end{equation*}
  Therefore, evaluating in the position space, we find
  \begin{equation}\label{[N_ex[varphi],A]}
    \begin{aligned}
    [\mathcal{N}_{ex}[\varphi],A]&=2\sum_{\sigma,\nu}\int_{\Lambda^3}
    \eta(x_1-x_2)\varphi(x)\Big(\sum_{p_1\notin B_F^\sigma}e^{ip_1(x_1-x)}\Big)\\
    &\quad\quad\quad\quad\times
    a^*(h_{x,\sigma})a^*(h_{x_2,\nu})a(g_{x_2,\nu})a(g_{x_1,\sigma})dx_1dx_2dx\\
    &=2\sum_{\sigma,\nu}\int_{\Lambda^2}
    \varphi(x)
    a^*(h_{x,\sigma})a^*(h_{x_2,\nu})a(g_{x_2,\nu})a(k_{x_2,x,\sigma})dx_2dx
    \end{aligned}
  \end{equation}
  and hence we have
  \begin{equation}\label{[N_ex[varphi],A] bound}
    \begin{aligned}
    \pm[\mathcal{N}_{ex}[\varphi],B]&\lesssim
    \mathcal{N}_{ex}[\varphi]
   +N^2\Vert\varphi\Vert_1\Vert\eta\Vert_2^2(\mathcal{N}_{ex}+1)\\
   &\lesssim\mathcal{N}_{ex}[\varphi]
   +\ell\Vert\varphi\Vert_1(\mathcal{N}_{ex}+1).
    \end{aligned}
  \end{equation}
  We then conclude (\ref{variant lem7.1}) in Lemma \ref{variant lem 7.1 lemma} by combining (\ref{[N_ex[varphi],A] bound}), (\ref{control eB on     N_ex}) and Gronwall's inequality.
  \par For the proof of (\ref{variant lem important}), we notice that
  \begin{equation*}
    [\mathcal{N}_{ex}[\varphi](\mathcal{N}_{ex}+1)^n,A]=
    [\mathcal{N}_{ex}[\varphi],A](\mathcal{N}_{ex}+1)^n
    +\mathcal{N}_{ex}[\varphi][(\mathcal{N}_{ex}+1)^n,A].
  \end{equation*}
  Similar to the proof of (\ref{control eB on     N_ex}) for arbitrary $n\in\frac{1}{2}\mathbb{Z}$, we can show
  \begin{equation*}
    \pm[\mathcal{N}_{ex}[\varphi](\mathcal{N}_{ex}+1)^n,A]\lesssim
     \mathcal{N}_{ex}[\varphi](\mathcal{N}_{ex}+1)^n
   +\ell\Vert\varphi\Vert_1(\mathcal{N}_{ex}+1)^{n+1}.
  \end{equation*}
  Then (\ref{variant lem important}) follows by Gronwall's inequality. We omit further details.
\end{proof}

\par The next Lemma \ref{contol V_21'Omega [K,B]} is an important preliminary which allows us to bound the action of $e^B$ on the excited kinetic energy operator $\mathcal{K}_s$ defined in (\ref{K_s}). As we will see, it is also useful in proving Lemma \ref{control eBtiled on K_s lemma}.

\begin{lemma}\label{contol V_21'Omega [K,B]}
Let $0<\alpha<\frac{1}{6}$. For $\beta\geq0$ and $2\alpha+5\beta<\frac{1}{3}$, we have,
\begin{align}
  \pm\mathcal{V}_{21}^\prime &\lesssim \big(N^{-\beta}+N^{2\alpha+
  5\beta-\frac{1}{3}}\big)\mathcal{K}_s+\ell^{-1}N^{\beta}\label{est V_21'} \\
  \pm\Omega &\lesssim \big(N^{-\beta}+N^{5\beta-4\alpha-\frac{1}{3}}\big)
  \mathcal{K}_s+N^{\frac{2}{3}+\beta}\ell\label{est Omega}\\
  \pm [\mathcal{K},B]&\lesssim \mathcal{K}_s+N(\mathcal{N}_{ex}+1)\label{est [K,B]}
\end{align}
\end{lemma}
\begin{proof}
  \par We first prove (\ref{est V_21'}). We adopt only within this proof that the notation $k_{p,x,\sigma}$
  \begin{equation*}
    k_{p,x,\sigma}(z)\coloneqq\sum_{q\in B_F^\sigma}W_{q-p}e^{i(q-p)x}f_{q,\sigma}(z).
  \end{equation*}
  Notice that by (\ref{norm bound  of creation and annihilation}), we have
  \begin{equation*}
    \Vert a(k_{p,x,\sigma})\Vert^2\leq\Vert k_{p,x,\sigma}\Vert^2
    = \sum_{q\in B_F^\sigma}\vert W_{q-p}\vert^2.
  \end{equation*}
  Let $-\frac{11}{48}<\delta_1,\delta_2<\frac{1}{3}$, we can rewrite $\mathcal{V}_{21}^\prime=\mathcal{V}_{21,1}^\prime
  +\mathcal{V}_{21,2}^\prime+\mathcal{V}_{21,3}^\prime$ with
  \begin{equation}\label{rewrt V_21'}
    \begin{aligned}
    \mathcal{V}_{21,1}^\prime&=\sum_{\sigma,\nu,p_1}\chi_{p_1\notin B_F^{\sigma}}
    \int_{\Lambda}a^*_{p_1,\sigma}a^*(H_{x,\nu}[\delta_2])a(g_{x,\nu})a(k_{p_1,x,\sigma})dx
    +h.c.\\
    \mathcal{V}_{21,2}^\prime&=\sum_{\sigma,\nu,p_1}\chi_{p_1\notin B_F^{\sigma}}
    \int_{\Lambda}a^*_{p_1,\sigma}a^*(L_{x,\nu}[\delta_2])a(I_{x,\nu}[\delta_1])
    a(k_{p_1,x,\sigma})dx+h.c.\\
    \mathcal{V}_{21,3}^\prime&=\sum_{\sigma,\nu,p_1}\chi_{p_1\notin B_F^{\sigma}}
    \int_{\Lambda}a^*_{p_1,\sigma}a^*(L_{x,\nu}[\delta_2])a(S_{x,\nu}[\delta_1])
    a(k_{p_1,x,\sigma})dx+h.c.
    \end{aligned}
  \end{equation}
 For $\psi\in\mathcal{H}^{\wedge N}(\{N_{\varsigma_i}\})$, we can bound the first term $\mathcal{V}_{21,1}^\prime$ by
 \begin{equation}\label{V_21,1' bound}
 \begin{aligned}
   &\vert\langle\mathcal{V}_{21,1}^\prime\psi,\psi\rangle\vert\\
   &\lesssim \Big(\sum_{\sigma,\nu,p_1}\chi_{p_1\notin B_F^{\sigma}}
   \big(\vert p_1\vert^2-\vert k_F^\sigma\vert^2\big)
    \\
    &\quad\quad\int_{\Lambda}\langle a^*_{p_1,\sigma}a^*(H_{x,\nu}[\delta_2])
    (\mathcal{N}_h[\delta_2]+2)^{-1}a(H_{x,\nu}[\delta_2])a_{p_1,\sigma}\psi,\psi\rangle
    dx\Big)^{\frac{1}{2}}\\
    &\quad\quad\times\Big(\sum_{\sigma,\nu,p_1}\chi_{p_1\notin B_F^{\sigma}}
   \big(\vert p_1\vert^2-\vert k_F^\sigma\vert^2\big)^{-1}\int_{\Lambda} \Vert a(g_{x,\nu})a(k_{p_1,x,\sigma})
    (\mathcal{N}_h[\delta_2]+2)^{\frac{1}{2}}\psi\Vert^2 dx\Big)^{\frac{1}{2}}\\
    &\lesssim \Big(N\sum_{\sigma,p_1,q_1}\frac{\chi_{p_1\notin B_F^{\sigma}}
    \chi_{q_1\in B_F^{\sigma}}W_{q_1-p_1}^2}
    {\vert p_1\vert^2-\vert k_F^\sigma\vert^2}\Big)^{\frac{1}{2}}
    \langle\mathcal{K}_s\psi,\psi\rangle^{\frac{1}{2}}
    \langle(\mathcal{N}_h[\delta_2]+1)\psi,\psi\rangle^{\frac{1}{2}}\\
    &\lesssim \Big(N^{-\beta}+\ell^{-1}N^{\beta-\frac{1}{3}-\delta_2}\Big)
    \langle\mathcal{K}_s\psi,\psi\rangle
    +\ell^{-1}N^\beta\langle\psi,\psi\rangle
 \end{aligned}
 \end{equation}
 where in the last inequality, we have used (\ref{ineqn N_h and N_i}) and (\ref{energy est bog}). Recall that $\ell=N^{-\frac{1}{3}-\alpha}$, to optimize, we choose $\delta_2=\alpha+2\beta$.
 \par For the second term in (\ref{rewrt V_21'}), we can bound similarly
 \begin{equation}\label{V_21,2' bound}
    \begin{aligned}
   &\vert\langle\mathcal{V}_{21,2}^\prime\psi,\psi\rangle\vert\\
   &\lesssim \Big(\sum_{\sigma,\nu,p_1}\chi_{p_1\notin B_F^{\sigma}}
   \big(\vert p_1\vert^2-\vert k_F^\sigma\vert^2\big)
    \\
    &\quad\quad\int_{\Lambda}\langle a^*_{p_1,\sigma}a(I_{x,\nu}[\delta_1])
    (\mathcal{N}_i[\delta_1]+2)^{-1}a^*(I_{x,\nu}[\delta_1])a_{p_1,\sigma}\psi,\psi\rangle
    dx\Big)^{\frac{1}{2}}\\
    &\times\Big(\sum_{\sigma,\nu,p_1}\chi_{p_1\notin B_F^{\sigma}}
   \big(\vert p_1\vert^2-\vert k_F^\sigma\vert^2\big)^{-1}\int_{\Lambda} \Vert a^*(L_{x,\nu}[\delta_2])a(k_{p_1,x,\sigma})
    (\mathcal{N}_i[\delta_1]+2)^{\frac{1}{2}}\psi\Vert^2 dx\Big)^{\frac{1}{2}}\\
    &\lesssim \Big(N^{\frac{2}{3}+\delta_2}\sum_{\sigma,p_1,q_1}\frac{\chi_{p_1\notin B_F^{\sigma}}
    \chi_{q_1\in B_F^{\sigma}}W_{q_1-p_1}^2}
    {\vert p_1\vert^2-\vert k_F^\sigma\vert^2}\Big)^{\frac{1}{2}}
    \langle\mathcal{K}_s\psi,\psi\rangle^{\frac{1}{2}}
    \langle(\mathcal{N}_i[\delta_1]+1)\psi,\psi\rangle^{\frac{1}{2}}\\
    &\lesssim \Big(N^{-\beta}+N^{2\alpha+3\beta-\frac{1}{3}-\delta_1}\Big)
    \langle\mathcal{K}_s\psi,\psi\rangle
    +N^{2\alpha+3\beta}\langle\psi,\psi\rangle
 \end{aligned}
 \end{equation}
 For the third term in (\ref{rewrt V_21'}), we write
 \begin{equation}\label{V_21,3' bound}
   \begin{aligned}
   \vert\langle\mathcal{V}_{21,3}^\prime\psi,\psi\rangle\vert
   &=\Big\vert\sum_{\sigma,\nu}\int_{\Lambda^2}
   W(x-y)a^*(h_{x,\sigma})a^*(L_{y,\nu}[\delta_2])a(S_{y,\nu}[\delta_1])
    a(g_{x,\sigma})dxdy\Big\vert\\
    &\lesssim\sum_{\sigma,\nu}\int_{\Lambda^2}
     N^{\frac{7}{6}+\frac{\delta_1}{2}+\frac{\delta_2}{2}}
     W(x-y)\Vert a(h_{x,\sigma})\psi\Vert\Vert\psi\Vert dxdy\\
     &\lesssim N^{\frac{7}{6}+\frac{\delta_1}{2}+\frac{\delta_2}{2}}\Vert W\Vert_1
     \langle\mathcal{N}_{ex}\psi,\psi\rangle^{\frac{1}{2}}
     \langle\psi,\psi\rangle^{\frac{1}{2}}\\
     &\lesssim N^{-\beta}\langle\mathcal{K}_s\psi,\psi\rangle
     +N^{\frac{1}{3}+\alpha+3\beta+\delta_1}\langle\psi,\psi\rangle
   \end{aligned}
 \end{equation}
 We then choose $\delta_1=-2\beta$. Combining (\ref{V_21,1' bound}-\ref{V_21,3' bound}), we reach (\ref{est V_21'}).
 \par The proof of (\ref{est Omega}) is similar to (\ref{est V_21'}), we only need to notice that this time we use the second inequality in (\ref{energy est bog}) rather than the first one, and replace the bound $\Vert W\Vert_1\lesssim a$ given in (\ref{L1&L2 norm of W 3dscatt}) by $N^{\frac{1}{3}}\Vert\nabla\eta\Vert_1\lesssim aN^{\frac{1}{3}}\ell$ given in (\ref{L1 grad eta and laplace eta}). We also choose $\delta_2=-\alpha+2\beta$ and $\delta_1=2\alpha-2\beta$ for (\ref{est Omega}).
 \par For the proof of (\ref{est [K,B]}), we first notice by Lemma \ref{lem commutator} that
 \begin{equation}\label{[K,A]}
 \begin{aligned}
   [\mathcal{K},B]&=\sum_{k,p,q,\sigma,\nu}
  \eta_k\big(\vert k\vert^2+k(q-p)\big)(a^*_{p-k,\sigma}a^*_{q+k,\nu}a_{q,\nu}a_{p,\sigma}+h.c.)\\
  &\quad\quad\quad\quad\quad\quad\quad\quad\quad
  \times\chi_{p-k\notin B^{\sigma}_F}\chi_{q+k\notin B^{\nu}_F}\chi_{p\in B^{\sigma}_F}\chi_{q\in B^{\nu}_F}\\
  &\eqqcolon\Omega_{m}+\Omega.
 \end{aligned}
 \end{equation}
 We have bounded $\Omega$ in (\ref{est Omega}) by choosing $\beta=0$,
  \begin{equation}\label{Omega bound}
    \pm\Omega\lesssim\mathcal{K}_s+N^{\frac{1}{3}-\alpha},
  \end{equation}
  while the estimate of $\Omega_m$ is similar to (\ref{V_21,1' bound}). In fact, letting
  \begin{equation*}
    k_{p,x,\sigma}(z)\coloneqq\sum_{q\in B_F^\sigma}\eta_{q-p}\vert
    q-p\vert^2e^{iqx}f_{q,\sigma}(z),
  \end{equation*}
 we have
 \begin{equation}\label{Omeg_m bound}
 \begin{aligned}
  &\vert\langle\Omega_m\psi,\psi\rangle\vert
  =\Big\vert\sum_{\sigma,\nu,p_1}\chi_{p_1\notin B_F^{\sigma}}
    \int_{\Lambda}\langle
    a^*_{p_1,\sigma}a^*(h_{x,\nu})a(g_{x,\nu})a(k_{p_1,x,\sigma})\psi,\psi\rangle
    dx+h.c.\Big\vert\\
   &\lesssim \Big(\sum_{\sigma,\nu,p_1}\chi_{p_1\notin B_F^{\sigma}}
   \big(\vert p_1\vert^2-\vert k_F^\sigma\vert^2\big)
    \\
    &\quad\quad\int_{\Lambda}\langle a^*_{p_1,\sigma}a^*(h_{x,\nu})
    (\mathcal{N}_{ex}+2)^{-1}a(h_{x,\nu})a_{p_1,\sigma}\psi,\psi\rangle
    dx\Big)^{\frac{1}{2}}\\
    &\quad\quad\times\Big(\sum_{\sigma,\nu,p_1}\chi_{p_1\notin B_F^{\sigma}}
   \big(\vert p_1\vert^2-\vert k_F^\sigma\vert^2\big)^{-1}\int_{\Lambda} \Vert a(g_{x,\nu})a(k_{p_1,x,\sigma})
    (\mathcal{N}_{ex}+2)^{\frac{1}{2}}\psi\Vert^2 dx\Big)^{\frac{1}{2}}\\
    &\lesssim \Big(N\sum_{\sigma,p_1,q_1}\frac{\chi_{p_1\notin B_F^{\sigma}}
    \chi_{q_1\in B_F^{\sigma}}\eta_{q_1-p_1}^2\vert q_1-p_1\vert^4}
    {\vert p_1\vert^2-\vert k_F^\sigma\vert^2}\Big)^{\frac{1}{2}}
    \langle\mathcal{K}_s\psi,\psi\rangle^{\frac{1}{2}}
    \langle(\mathcal{N}_{ex}+1)\psi,\psi\rangle^{\frac{1}{2}}\\
    &\lesssim \langle\mathcal{K}_s\psi,\psi\rangle
    +N\langle(\mathcal{N}_{ex}+1)\psi,\psi\rangle
 \end{aligned}
 \end{equation}
 where we have used (\ref{energy est bog laplace eta}) in the last inequality. Combining (\ref{[K,A]}-\ref{Omeg_m bound}) we reach (\ref{est [K,B]}).
\end{proof}

\par As a direct consequence of (\ref{est [K,B]}) in Lemma \ref{contol V_21'Omega [K,B]} (together with (\ref{control eB on     N_ex}) in Lemma \ref{control eB on N_ex lemma} and Gronwall's inequality), we have the following estimate of the action of $e^B$ on $\mathcal{K}_s$, since $[\mathcal{K}_s,B]=[\mathcal{K},B]$.
\begin{corollary}\label{corollary control eB on K_s}
  For $\vert t\vert\leq1$, we have
  \begin{equation}\label{control eB on K_s}
    e^{-tB}\mathcal{K}_se^{tB}\lesssim \mathcal{K}_s+N(\mathcal{N}_{ex}+1).
  \end{equation}
\end{corollary}

\par The next lemma controls the action of $e^B$ on the excited potential energy operator $\mathcal{V}_4$. It is also useful in the proof of Proposition \ref{Optimal BEC}.
\begin{lemma}\label{corollary control eB on V_4}
  For $\vert t\vert\leq1$, we have
  \begin{equation}\label{control eB on V_4}
    e^{-tB}\mathcal{V}_4e^{tB}\lesssim \mathcal{V}_4+N.
  \end{equation}
\end{lemma}
\begin{proof}
  \par By Lemma \ref{lem commutator}, we have
  \begin{equation}\label{[V_4,B]}
    \begin{aligned}
    [\mathcal{V}_4,B]=\Theta_m+\Theta_d+\Theta_r
    \end{aligned}
  \end{equation}
  with
  \begin{equation}\label{Theta}
    \begin{aligned}
    \Theta_m&=\frac{1}{2}\sum_{k,p,q,\sigma,\nu}\Big(\sum_{l}\hat{v}_{k-l}\eta_l\Big)
    (a^*_{p-k,\sigma}a^*_{q+k,\nu}a_{q,\nu}a_{p,\sigma}+h.c.)\\
    &\quad\quad\times
    \chi_{p-k\notin B^{\sigma}_F}\chi_{q+k\notin B^{\nu}_F}
    \chi_{p\in B^{\sigma}_F}\chi_{q\in B^{\nu}_F}\\
    \Theta_d&=\frac{1}{2}\sum_{l,k,p,q,\sigma,\nu}\hat{v}_{k-l}\eta_l
    (a^*_{p-k,\sigma}a^*_{q+k,\nu}a_{q,\nu}a_{p,\sigma}+h.c.)\\
    &\quad\quad\times
    \chi_{p-k\notin B^{\sigma}_F}\chi_{q+k\notin B^{\nu}_F}
    \chi_{p\in B^{\sigma}_F}\chi_{q\in B^{\nu}_F}
    (\chi_{p-l\notin B^{\sigma}_F}\chi_{q+l\notin B^{\nu}_F}-1)\\
    \Theta_r&=-\sum_{k,p,q,l,r,s,\sigma,\nu,\varpi}
    \hat{v}_{k}\eta_l
    (a^*_{p-k,\sigma}a^*_{q+k,\nu}a^*_{s+l,\varpi}a_{q,\nu}
    a_{s,\varpi}a_{r,\sigma}+h.c.)\\
    &\quad\quad\times
    \chi_{p-k,p\notin B^{\sigma}_F}\chi_{q+k,q\notin B^{\nu}_F}
    \chi_{r\in B_F^\sigma}
    \chi_{s+l\notin B_F^\varpi}\chi_{s\in B_F^\varpi}\delta_{p,r-l}
    \end{aligned}
  \end{equation}
  Notice that switching to position space, we have
  \begin{equation}\label{V_4position}
    \mathcal{V}_4=\frac{1}{2}\sum_{\sigma,\nu}
    \int_{\Lambda^2}v_a(x-y)a^*(h_{x,\sigma})a^*(h_{y,\nu})
    a(h_{y,\nu})a(h_{x,\sigma})dxdy.
  \end{equation}
  For the first term (main part) in (\ref{Theta}), since $\Vert v_a\Vert_1\lesssim a$, we can easily bound
  \begin{equation}\label{Theta_m bound}
    \begin{aligned}
    \pm\Theta_m&=\pm\frac{1}{2}\sum_{\sigma,\nu}\int_{\Lambda^2}
    v_a(x-y)\eta(x-y)a^*(h_{x,\sigma})a^*(h_{y,\nu})
    a(g_{y,\nu})a(g_{x,\sigma})dxdy+h.c.\\
    &\lesssim \mathcal{V}_4+N^2\Vert v_a\Vert_1\Vert\eta\Vert_{\infty}^2\lesssim
    \mathcal{V}_4+N
    \end{aligned}
  \end{equation}
 \par For the second term (difference part) in (\ref{Theta}), notice that for any $p\in B_F^\sigma$ and $q\in B_F^\nu$, $(\chi_{p-l\notin B^{\sigma}_F}\chi_{q+l\notin B^{\nu}_F}-1)\neq0$ implies $\vert l\vert\lesssim N^{\frac{1}{3}}$. Using
   \begin{equation*}
     (\chi_{p-l\notin B^{\sigma}_F}\chi_{q+l\notin B^{\nu}_F}-1)
     =-\chi_{p-l\in B^\sigma_F}-\chi_{q+l\in B^{\nu}_F}+
     \chi_{p-l\in B^\sigma_F}\chi_{q+l\in B^{\nu}_F}
   \end{equation*}
   we decompose $\Theta_d=\Theta_{d,1}+\Theta_{d,2}+\Theta_{d,3}$. For fixed $l$, we let
   \begin{equation*}
     g_{l,x,\sigma}(z)\coloneqq\sum_{p\in B_F^\sigma}-\chi_{p-l\in B^\sigma_F}e^{ipx}f_{p,\sigma}(z).
   \end{equation*}
   Obviously, $\Vert g_{l,x,\sigma}\Vert_2^2\leq N_\sigma$. Therefore,
  \begin{equation}\label{Theta_d bound}
    \begin{aligned}
  \Theta_{d,1}
    =\sum_{l,\sigma,\nu}\frac{\eta_l}{2}
    \int_{\Lambda^2}v_a(x-y)e^{-il(x-y)}
    a^*(h_{x,\sigma})a^*(h_{y,\nu})
    a(g_{y,\nu})a(g_{l,x,\sigma}) dxdy+h.c.
    \end{aligned}
  \end{equation}
  and we can bound it by
  \begin{equation}\label{Theta d bound true}
    \pm\Theta_{d,1}
    \lesssim\sum_{\sigma,\nu}\sum_{\vert l\vert\lesssim N^{\frac{1}{3}}}
    \vert\eta_l\vert(\mathcal{V}_4+N)
    \lesssim \ell^2(\mathcal{V}_4+N).
  \end{equation}
  $\Theta_{d,2}$ and $\Theta_{d,3}$ can be bounded similarly.
  \par For the third term (residue term) in (\ref{Theta}), we let
  \begin{equation*}
    k_{y_1,y_2,\sigma}(z)\coloneqq\sum_{q\in B^{\sigma}_F}\Big(
    \sum_{p\notin B_F^\sigma}\int_{\Lambda}\eta(x-y_1)
    e^{iqx}e^{ip(y_2-x)}dx\Big)f_{q,\sigma}(z),
  \end{equation*}
  which has the property that, for $j=1,2$
  \begin{equation*}
    \int_{\Lambda}\Vert k_{y_1,y_2,\sigma}\Vert^2_2dy_j\leq\sum_{q\in B_F^\sigma}
    \Vert\eta\Vert_2^2.
  \end{equation*}
  Thus, evaluating in the position space, we have
  \begin{equation}\label{Theta_r bound}
    \begin{aligned}
    \pm\Theta_r
    &=\mp\sum_{\sigma,\nu,\varpi}
    \int_{\Lambda^3}v_a(x_1-x_3)a^*(h_{x_3,\sigma})a^*(h_{x_1,\nu})
    a^*(h_{x_2,\varpi})\\
    &\quad \quad \quad \quad \times
    a(h_{x_1,\nu})a(g_{x_2,\varpi})a(k_{x_2,x_3,\sigma})dx_1dx_2dx_3+h.c.\\
    &\lesssim \mathcal{V}_4+N^2\Vert v_a\Vert_1\Vert\eta\Vert^2_2
    (\mathcal{N}_{ex}+1)^2\lesssim\mathcal{V}_4+N\ell.
    \end{aligned}
  \end{equation}
  \par We can therefore conclude (\ref{control eB on V_4}), by combining (\ref{[V_4,B]}-\ref{Theta_r bound}) and Gronwall's inequality.
\end{proof}

\par The next four lemmas collect the results about the right hand side of (\ref{define G_N}).
\begin{lemma}\label{cal eB on V_0,1,22,23 lemma}
  \begin{equation}\label{cal eB on V_0,1,22,23}
    e^{-B}\big(\mathcal{V}_{0}
  +\mathcal{V}_{1}+\mathcal{V}_{22}+\mathcal{V}_{23}\big)e^B
  =\frac{\hat{v}_0}{2}\sum_{\sigma\neq\nu}N_{\sigma}N_\nu+\mathcal{E}_0,
  \end{equation}
  with
  \begin{equation}\label{bound E_0}
    \pm\mathcal{E}_0\lesssim N^{-\frac{1}{6}}(\mathcal{K}_s+1).
  \end{equation}
\end{lemma}
\begin{proof}
  \par Lemmas \ref{lemma ineqn K_s}, \ref{lemma V_0}, \ref{lemma V_22,23,1}, \ref{control eB on N_ex lemma} and Corollary \ref{corollary control eB on K_s} together yield Lemma \ref{cal eB on V_0,1,22,23 lemma}.
\end{proof}

\begin{lemma}\label{control eB on V_3 lemma}
\begin{equation}\label{contron eB on V_3}
  e^{-B}\mathcal{V}_3e^{B}=\mathcal{V}_3+\mathcal{E}_{\mathcal{V}_3},
\end{equation}
where
\begin{equation}\label{est E_V_3}
\begin{aligned}
  \pm\mathcal{E}_{\mathcal{V}_3}&\lesssim (N^{-\frac{1}{6}-\frac{\alpha}{2}}+\mu)\mathcal{V}_4
  +\mu^{-1}N^{-\frac{4}{3}-\alpha}(\mathcal{N}_{ex}+1)^3\\
  &+N^{-2\alpha}(\mathcal{K}_s+N^{\frac{1}{3}})+N^{-\alpha}
    (\mathcal{N}_{ex}+1)+N^{-\frac{2}{3}-\alpha}(\mathcal{N}_{ex}+1)^2
\end{aligned}
\end{equation}
for some $\mu>0$.
\end{lemma}
\begin{proof}
  \par Using the fact that
  \begin{equation*}
  \begin{aligned}
    e^{-B}a^*_{p-k,\sigma}a^*_{q+k,\nu}e^B
    &=a^*_{p-k,\sigma}a^*_{q+k,\nu}+\int_{0}^{1}
    e^{-tB}[a^*_{p-k,\sigma}a^*_{q+k,\nu},B]e^{tB}dt\\
    e^{-B}a_{q,\nu}a_{p,\sigma}e^B
    &=a_{q,\nu}a_{p,\sigma}+\int_{0}^{1}
    e^{-tB}[a_{q,\nu}a_{p,\sigma},B]e^{tB}dt
  \end{aligned}
  \end{equation*}
  we rewrite $e^{-B}\mathcal{V}_3e^B=\mathcal{V}_3+\mathcal{E}_{\mathcal{V}_3}$ with
  $\mathcal{E}_{\mathcal{V}_3}=\mathcal{E}_{\mathcal{V}_3,1}+\mathcal{E}_{\mathcal{V}_3,2}$ given by
  \begin{equation}\label{split E_V_3}
  \begin{aligned}
    \mathcal{E}_{\mathcal{V}_3,1}&=\sum_{k,p,q,\sigma,\nu}
  \hat{v}_ka^*_{p-k,\sigma}a^*_{q+k,\nu}\int_{0}^{1}e^{-tB}
  [a_{q,\nu}a_{p,\sigma},B]e^{tB}dt\\
  &\quad\quad\quad\quad\times
  \chi_{p-k\notin B^{\sigma}_F}\chi_{p\in B^{\sigma}_F}\chi_{q,q+k\notin B^{\nu}_F}
  +h.c.\\
    \mathcal{E}_{\mathcal{V}_3,2}&=\sum_{k,p,q,\sigma,\nu}
  \hat{v}_k\int_{0}^{1}e^{-tB}
  [a^*_{p-k,\sigma}a^*_{q+k,\nu},B]e^{(t-1)B}
  a_{q,\nu}a_{p,\sigma}e^Bdt\\
  &\quad\quad\quad\quad\times
  \chi_{p-k\notin B^{\sigma}_F}\chi_{p\in B^{\sigma}_F}\chi_{q,q+k\notin B^{\nu}_F}
  +h.c.
  \end{aligned}
  \end{equation}
   During the following evaluations, we apply the notation:
   \begin{equation*}
    k_{y_1,y_2,\sigma}(z)\coloneqq\sum_{q\in B^{\sigma}_F}\Big(
    \sum_{p\notin B_F^\sigma}\int_{\Lambda}\eta(x-y_1)
    e^{iqx}e^{ip(y_2-x)}dx\Big)f_{q,\sigma}(z).
  \end{equation*}
  One can check that for $j=1,2$
  \begin{equation*}
    \int_{\Lambda}\Vert k_{y_1,y_2,\sigma}\Vert^2_2dy_j\leq\sum_{q\in B_F^\sigma}
    \Vert\eta\Vert_2^2.
  \end{equation*}
  For the first term in (\ref{split E_V_3}), by direct calculation, we can write $\mathcal{E}_{\mathcal{V}_3,1}=\sum_{j=1}^{2}\mathcal{E}_{\mathcal{V}_3,1j}$ with
  \begin{equation}\label{E_V_3_1 split}
    \begin{aligned}
    \mathcal{E}_{\mathcal{V}_3,11}&=\sum_{\substack{\sigma,\nu\\k,p,q}}
    \sum_{\substack{\tau,\varpi\\l,r,s}}\hat{v}_k\eta_l\delta_{q,r-l}\delta_{\nu,\tau}
    a^*_{p-k,\sigma}a^*_{q+k,\nu}\int_{0}^{1}e^{-tB}a^*_{s+l,\varpi}a_{s,\varpi}a_{r,\tau}
    a_{p,\sigma}e^{tB}dt\\
    &\quad\quad\times
    \chi_{p-k\notin B^{\sigma}_F}\chi_{p\in B^{\sigma}_F}\chi_{q,q+k\notin B^{\nu}_F}
    \chi_{r-l\notin B^{\tau}_F}\chi_{r\in B^{\tau}_F}\chi_{s+l\notin B^{\varpi}_F}
    \chi_{s\in B^{\varpi}_F}+h.c.\\
    \mathcal{E}_{\mathcal{V}_3,12}&=\sum_{\substack{\sigma,\nu\\k,p,q}}
    \sum_{\substack{\tau,\varpi\\l,r,s}}\hat{v}_k\eta_l\delta_{p,r}\delta_{\sigma,\tau}
    a^*_{p-k,\sigma}a^*_{q+k,\nu}\int_{0}^{1}e^{-tB}a^*_{s,\varpi}a_{q,\nu}a_{s+l,\varpi}
    a_{r-l,\tau}e^{tB}dt\\
    &\quad\quad\times
    \chi_{p-k\notin B^{\sigma}_F}\chi_{p\in B^{\sigma}_F}\chi_{q,q+k\notin B^{\nu}_F}
    \chi_{r-l\notin B^{\tau}_F}\chi_{r\in B^{\tau}_F}\chi_{s+l\notin B^{\varpi}_F}
    \chi_{s\in B^{\varpi}_F}+h.c.
    \end{aligned}
  \end{equation}
  Switching to position space, with the help of (\ref{V_4position}) and Lemma \ref{control eB on N_ex lemma}, we have
  \begin{equation}\label{E_V_3,11 bound}
    \begin{aligned}
    \pm \mathcal{E}_{\mathcal{V}_3,11}&=\pm
    \sum_{\sigma,\nu,\varpi}\int_{\Lambda^3}\int_{0}^{1}v_a(x_1-x_3)
    a^*(h_{x_1,\sigma})a^*(h_{x_3,\nu})e^{-tB}a^*(h_{x_4,\varpi})\\
    &\quad\quad\quad\quad
     \times a(g_{x_4,\varpi})a(k_{x_4,x_3,\nu})a(g_{x_1,\sigma})e^{tB}
     dtdx_1dx_3dx_4+h.c.\\
     &\lesssim N^{\frac{3}{2}}\Vert v_a\Vert_1^{\frac{1}{2}}\Vert\eta\Vert_2
     (\mathcal{V}_4+\mathcal{N}_{ex}+1)
     \lesssim\ell^{\frac{1}{2}}(\mathcal{V}_4+\mathcal{N}_{ex}+1),
    \end{aligned}
  \end{equation}
  since $\Vert v_a\Vert_1\lesssim a$. On the other hand, we can bound
  \begin{equation}\label{E_V_3,12 bound}
    \begin{aligned}
    &\pm\mathcal{E}_{\mathcal{V}_3,12}\\
    &=\pm
    \sum_{\sigma,\nu,\varpi}\int_{\Lambda^4}\int_{0}^{1}v_a(x_1-x_3)\eta(x_2-x_4)
    \Big(\sum_{q_1\in B^\sigma_F}e^{-iq_1(x_1-x_2)}\Big)a^*(h_{x_1,\sigma})a^*(h_{x_3,\nu})\\
    &\quad\quad\quad\quad\times
    e^{-tB}a^*(g_{x_4,\varpi})a(h_{x_3,\nu})a(h_{x_4,\varpi})a(h_{x_2,\sigma})e^{tB}
    dtdx_1dx_2dx_3dx_4+h.c.\\
    &\lesssim\mu\mathcal{V}_4+\mu^{-1}N^2\Vert v_a\Vert_1\Vert\eta\Vert_2^2
    (\mathcal{N}_{ex}+1)^3
    \lesssim\mu\mathcal{V}_4+\mu^{-1}N^{-1}\ell
    (\mathcal{N}_{ex}+1)^3.
    \end{aligned}
  \end{equation}
  \par For the second term in (\ref{split E_V_3}), we rewrite $\mathcal{E}_{\mathcal{V}_3,1}=\sum_{j=1}^{3}\mathcal{E}_{\mathcal{V}_3,2j}$ with
  \begin{equation}\label{E_V_3,2 split}
    \begin{aligned}
     \mathcal{E}_{\mathcal{V}_3,21}&=\sum_{\substack{\sigma,\nu\\k,p,q}}
    \sum_{\substack{\tau,\varpi\\l,r,s}}\hat{v}_k\eta_l
    \delta_{p-k,r-l}\delta_{\sigma,\tau}\delta_{q+k,s+l}\delta_{\nu,\varpi}
    \int_{0}^{1}e^{-tB}a^*_{r,\tau}a^*_{s,\varpi}e^{(t-1)B}
    a_{q,\nu}a_{p,\sigma}e^{B}dt\\
    &\quad\quad\times
    \chi_{p-k\notin B^{\sigma}_F}\chi_{p\in B^{\sigma}_F}\chi_{q,q+k\notin B^{\nu}_F}
    \chi_{r-l\notin B^{\tau}_F}\chi_{r\in B^{\tau}_F}\chi_{s+l\notin B^{\varpi}_F}
    \chi_{s\in B^{\varpi}_F}+h.c.\\
    \mathcal{E}_{\mathcal{V}_3,22}&=\sum_{\substack{\sigma,\nu\\k,p,q}}
    \sum_{\substack{\tau,\varpi\\l,r,s}}\hat{v}_k\eta_l
    \delta_{p-k,r-l}\delta_{\sigma,\tau}
    \int_{0}^{1}e^{-tB}a^*_{r,\tau}a^*_{s,\varpi}a^*_{q+k,\nu}a_{s+l,\varpi}e^{(t-1)B}
    a_{q,\nu}a_{p,\sigma}e^{B}dt\\
    &\quad\quad\times(-
    \chi_{p-k\notin B^{\sigma}_F}\chi_{p\in B^{\sigma}_F}\chi_{q,q+k\notin B^{\nu}_F}
    \chi_{r-l\notin B^{\tau}_F}\chi_{r\in B^{\tau}_F}\chi_{s+l\notin B^{\varpi}_F}
    \chi_{s\in B^{\varpi}_F})+h.c.\\
    \mathcal{E}_{\mathcal{V}_3,23}&=\sum_{\substack{\sigma,\nu\\k,p,q}}
    \sum_{\substack{\tau,\varpi\\l,r,s}}\hat{v}_k\eta_l
    \delta_{q+k,s+l}\delta_{\nu,\varpi}
    \int_{0}^{1}e^{-tB}a^*_{r,\tau}a^*_{s,\varpi}a^*_{p-k,\sigma}a_{r-l,\tau}e^{(t-1)B}
    a_{q,\nu}a_{p,\sigma}e^{B}dt\\
    &\quad\quad\times(-
    \chi_{p-k\notin B^{\sigma}_F}\chi_{p\in B^{\sigma}_F}\chi_{q,q+k\notin B^{\nu}_F}
    \chi_{r-l\notin B^{\tau}_F}\chi_{r\in B^{\tau}_F}\chi_{s+l\notin B^{\varpi}_F}
    \chi_{s\in B^{\varpi}_F})+h.c.
    \end{aligned}
  \end{equation}
  Switching to the position space, with the help of Lemma \ref{control eB on N_ex lemma}, we have
  \begin{equation}\label{E_V_3,23 bound}
    \begin{aligned}
    \pm\mathcal{E}_{\mathcal{V}_3,23}&=\pm
    \sum_{\sigma,\nu,\tau}\int_{\Lambda^3}\int_{0}^{1}v_a(x_1-x_3)e^{-tB}
    a^*(g_{x_2,\tau})a^*(k_{x_2,x_3,\nu})a^*(h_{x_1,\sigma})\\
    &\quad\quad\quad\quad\times
    a(h_{x_2,\tau})e^{(t-1)B}a(h_{x_3,\nu})a(g_{x_1,\sigma})e^Bdtdx_1dx_2dx_3+h.c.\\
    &\lesssim N^{\frac{3}{2}}\Vert v_a\Vert_1\Vert\eta\Vert_2
    (\mathcal{N}_{ex}+1)^{\frac{3}{2}}
    \lesssim \ell^{\frac{1}{2}}(\mathcal{N}_{ex}+1),
    \end{aligned}
  \end{equation}
  since $(\mathcal{N}_{ex}+1)^{\frac{1}{2}}\lesssim N^{\frac{1}{2}}$. Similarly,
  \begin{equation}\label{E_V_3,22 bound}
     \begin{aligned}
    \pm\mathcal{E}_{\mathcal{V}_3,22}&=\pm
    \sum_{\sigma,\nu,\varpi}\int_{\Lambda^3}\int_{0}^{1}v_a(x_1-x_3)e^{-tB}
    a^*(k_{x_4,x_1,\sigma})a^*(g_{x_4,\varpi})a^*(h_{x_3,\nu})\\
    &\quad\quad\quad\quad\times
    a(h_{x_4,\varpi})e^{(t-1)B}a(h_{x_3,\nu})a(g_{x_1,\sigma})e^Bdtdx_1dx_3dx_4+h.c.\\
    &\lesssim N^{\frac{3}{2}}\Vert v_a\Vert_1\Vert\eta\Vert_2
    (\mathcal{N}_{ex}+1)^{\frac{3}{2}}
    \lesssim \ell^{\frac{1}{2}}(\mathcal{N}_{ex}+1).
    \end{aligned}
  \end{equation}
  The estimate of $\mathcal{E}_{\mathcal{V}_3,21}$ is more subtle. Switching to the position space, we write
  \begin{equation}\label{E_V_3,21 bound}
     \begin{aligned}
    &\mathcal{E}_{\mathcal{V}_3,21}=
    \sum_{\sigma,\nu}\int_{\Lambda^4}\int_{0}^{1}v_a(x_1-x_3)
    \eta(x_2-x_4)\Big(\sum_{p_1\notin B_F^\sigma}e^{ip_1(x_1-x_2)}\Big)
    \Big(\sum_{p_2\notin B_F^\nu}e^{ip_2(x_3-x_4)}\Big)\\
    &\quad\quad\times
    e^{-tB}a^*(g_{x_2,\sigma})a^*(g_{x_4,\nu})
    e^{(t-1)B}a(h_{x_3,\nu})a(g_{x_1,\sigma})e^Bdtdx_1dx_2dx_3dx_4+h.c.
    \end{aligned}
  \end{equation}
  Using the fact that
  \begin{equation*}
    \sum_{p_1\notin B_F^\sigma}\sum_{p_2\notin B_F^\nu}
    =\sum_{p_1,p_2}-\sum_{p_1}\sum_{p_2\in B_F^\nu}
    -\sum_{p_2}\sum_{p_1\in B_F^\sigma}+
     \sum_{p_1\in B_F^\sigma}\sum_{p_2\in B_F^\nu},
  \end{equation*}
  we rewrite $\mathcal{E}_{\mathcal{V}_3,21}=\sum_{j=1}^{4}\mathcal{E}_{\mathcal{V}_3,21j}$ with
  \begin{equation}\label{E_V_3,21 split}
    \begin{aligned}
    &\mathcal{E}_{\mathcal{V}_3,211}=
    \sum_{\sigma,\nu}\int_{\Lambda^2}\int_{0}^{1}v_a(x_1-x_3)
    \eta(x_1-x_3)\\
    &\quad\quad\times
    e^{-tB}a^*(g_{x_1,\sigma})a^*(g_{x_3,\nu})
    e^{(t-1)B}a(h_{x_3,\nu})a(g_{x_1,\sigma})e^Bdtdx_1dx_3+h.c.\\
    &\mathcal{E}_{\mathcal{V}_3,212}=-
    \sum_{\sigma,\nu}\int_{\Lambda^3}\int_{0}^{1}v_a(x_1-x_3)
    \eta(x_1-x_4)
    \Big(\sum_{p_2\in B_F^\nu}e^{ip_2(x_3-x_4)}\Big)\\
    &\quad\quad\times
    e^{-tB}a^*(g_{x_1,\sigma})a^*(g_{x_4,\nu})
    e^{(t-1)B}a(h_{x_3,\nu})a(g_{x_1,\sigma})e^Bdtdx_1dx_3dx_4+h.c.\\
    &\mathcal{E}_{\mathcal{V}_3,213}=-
    \sum_{\sigma,\nu}\int_{\Lambda^3}\int_{0}^{1}v_a(x_1-x_3)
    \eta(x_2-x_3)\Big(\sum_{p_1\in B_F^\sigma}e^{ip_1(x_1-x_2)}\Big)
    \\
    &\quad\quad\times
    e^{-tB}a^*(g_{x_2,\sigma})a^*(g_{x_3,\nu})
    e^{(t-1)B}a(h_{x_3,\nu})a(g_{x_1,\sigma})e^Bdtdx_1dx_2dx_3+h.c.\\
    &\mathcal{E}_{\mathcal{V}_3,214}=
    \sum_{\sigma,\nu}\int_{\Lambda^4}\int_{0}^{1}v_a(x_1-x_3)
    \eta(x_2-x_4)\Big(\sum_{p_1\in B_F^\sigma}e^{ip_1(x_1-x_2)}\Big)
    \Big(\sum_{p_2\in B_F^\nu}e^{ip_2(x_3-x_4)}\Big)\\
    &\quad\quad\times
    e^{-tB}a^*(g_{x_2,\sigma})a^*(g_{x_4,\nu})
    e^{(t-1)B}a(h_{x_3,\nu})a(g_{x_1,\sigma})e^Bdtdx_1dx_2dx_3dx_4+h.c.
    \end{aligned}
  \end{equation}
  Notice that for $j=3,4$
  \begin{equation*}
    \int_{\Lambda}\Big\vert\sum_{p_2\in B_F^\nu}e^{ip_2(x_3-x_4)}\Big\vert^2 dx_j
    \leq N_\nu.
  \end{equation*}
  Therefore, we can bound with the help of Lemma \ref{control eB on N_ex lemma} that
   \begin{equation}\label{E_V_3,212}
    \begin{aligned}
    \pm\mathcal{E}_{\mathcal{V}_3,212}\lesssim
    N^2\Vert v_a\Vert_1\Vert\eta\Vert_2(\mathcal{N}_{ex}+1)
    \lesssim\ell^{\frac{1}{2}}(\mathcal{N}_{ex}+1).
    \end{aligned}
  \end{equation}
  Similarly
  \begin{equation}\label{E_V_3,213,214}
     \begin{aligned}
    \pm\mathcal{E}_{\mathcal{V}_3,213}&\lesssim
    N^2\Vert v_a\Vert_1\Vert\eta\Vert_2(\mathcal{N}_{ex}+1)
    \lesssim\ell^{\frac{1}{2}}(\mathcal{N}_{ex}+1)\\
    \pm\mathcal{E}_{\mathcal{V}_3,214}&\lesssim
    N^\frac{5}{2}\Vert v_a\Vert_1\Vert\eta\Vert_1(\mathcal{N}_{ex}+1)
    \lesssim N^{\frac{1}{2}}\ell^{2}(\mathcal{N}_{ex}+1)
    \end{aligned}
  \end{equation}
  For $\mathcal{E}_{\mathcal{V}_3,211}$, we first switch back to the momentum space to find
  \begin{equation}\label{E_V_3,211}
    \begin{aligned}
    \mathcal{E}_{\mathcal{V}_3,211}&=\int_{0}^{1}
    \sum_{k,p,q,\sigma,\nu}\Big(\sum_{m}\hat{v}_{k-m}\eta_m\Big)
    e^{-tB}a^*_{p-k,\sigma}a^*_{q+k,\nu}e^{(t-1)B}a_{q,\nu}a_{p,\sigma}e^Bdt\\
    &\quad\quad\quad\quad\times
    \chi_{p-k,p\in B_F^\sigma}\chi_{q+k\in B_F^\nu}\chi_{q\notin B_F^\nu}+h.c.
    \end{aligned}
  \end{equation}
  We once again use the fact
  \begin{equation*}
    e^{(t-1)B}a_{q,\nu}a_{p,\sigma}e^{(1-t)B}
    =a_{q,\nu}a_{p,\sigma}+\int_{0}^{1-t}
    e^{-sB}[a_{q,\nu}a_{p,\sigma},B]e^{sB}ds
  \end{equation*}
  to rewrite $\mathcal{E}_{\mathcal{V}_3,211}=\sum_{j=1}^{3}\mathcal{E}_{\mathcal{V}_3,211j}$ with
  \begin{equation}\label{E_V_3,211 split}
  \begin{aligned}
    &\mathcal{E}_{\mathcal{V}_3,2111}=\int_{0}^{1}
    \sum_{k,p,q,\sigma,\nu}\Big(\sum_{m}\hat{v}_{k-m}\eta_m\Big)
    e^{-tB}a^*_{p-k,\sigma}a^*_{q+k,\nu}a_{q,\nu}a_{p,\sigma}e^{tB}dt\\
    &\quad\quad\quad\quad\times
    \chi_{p-k,p\in B_F^\sigma}\chi_{q+k\in B_F^\nu}\chi_{q\notin B_F^\nu}+h.c.\\
   &\mathcal{E}_{\mathcal{V}_3,2112}=\int_{0}^{1}\int_{0}^{1-t}
   \sum_{\substack{\sigma,\nu\\k,p,q}}
    \sum_{\substack{\tau,\varpi\\l,r,s}}\Big(\sum_{m}\hat{v}_{k-m}\eta_m\Big)
    \eta_l\delta_{q,r-l}\delta_{\nu,\tau}\\
    &\quad\quad\times
    e^{-tB}a^*_{p-k,\sigma}a^*_{q+k,\nu}
    e^{-sB}a^*_{s+l,\varpi}a_{s,\varpi}a_{r,\tau}
    a_{p,\sigma}e^{(s+t)B}dsdt\\
    &\quad\quad\times
    \chi_{p-k,p\in B^{\sigma}_F}\chi_{q+k\in B_F^\nu}\chi_{q\notin B^{\nu}_F}
    \chi_{r-l\notin B^{\tau}_F}\chi_{r\in B^{\tau}_F}\chi_{s+l\notin B^{\varpi}_F}
    \chi_{s\in B^{\varpi}_F}+h.c.\\
   & \mathcal{E}_{\mathcal{V}_3,2113}=\int_{0}^{1}\int_{0}^{1-t}
   \sum_{\substack{\sigma,\nu\\k,p,q}}
    \sum_{\substack{\tau,\varpi\\l,r,s}}\Big(\sum_{m}\hat{v}_{k-m}\eta_m\Big)
    \eta_l\delta_{p,r}\delta_{\sigma,\tau}\\
    &\quad\quad\times
    e^{-tB}a^*_{p-k,\sigma}a^*_{q+k,\nu}e^{-sB}a^*_{s,\varpi}a_{q,\nu}a_{s+l,\varpi}
    a_{r-l,\tau}e^{(s+t)B}dsdt\\
    &\quad\quad\times
     \chi_{p-k,p\in B^{\sigma}_F}\chi_{q+k\in B_F^\nu}\chi_{q\notin B^{\nu}_F}
    \chi_{r-l\notin B^{\tau}_F}\chi_{r\in B^{\tau}_F}\chi_{s+l\notin B^{\varpi}_F}
    \chi_{s\in B^{\varpi}_F}+h.c.
  \end{aligned}
  \end{equation}
  Since $\Vert\eta\Vert_\infty\leq1$, $\mathcal{E}_{\mathcal{V}_3,2111}$ can be bounded similar to bounding $\mathcal{V}_1$ in Lemma \ref{lemma V_22,23,1}, with the help of Lemma \ref{control eB on N_ex lemma}. We state the result directly without further details:
  \begin{equation}\label{E_V_3,2111}
    \pm\mathcal{E}_{\mathcal{V}_3,2111}\lesssim N^{-\frac{1}{6}}(\mathcal{K}_s+1).
  \end{equation}
  For $\mathcal{E}_{\mathcal{V}_3,2112}$, we write $\mathcal{E}_{\mathcal{V}_3,2112}=
  \sum_{j=1}^{3}\mathcal{E}_{\mathcal{V}_3,2112j}$ in the position space with
  \begin{equation}\label{E_V_3,2112 split}
    \begin{aligned}
    &\mathcal{E}_{\mathcal{V}_3,21121}=
    \sum_{\sigma,\nu,\varpi}\int_{\Lambda^3}\int_{0}^{1}\int_{0}^{1-t}
    v_a(x_1-x_3)\eta(x_1-x_3)\eta(x_3-x_4)\\
    &\quad\quad\quad\quad\times e^{-tB}a^*(g_{x_1,\sigma})a^*(g_{x_3,\nu})e^{-sB}
    a^*(h_{x_4,\varpi})\\
    &\quad\quad\quad\quad\times
    a(g_{x_4,\varpi})a(g_{x_3,\nu})a(g_{x_1,\sigma})e^{(s+t)B}dsdtdx_1dx_3dx_4+h.c.\\
     &\mathcal{E}_{\mathcal{V}_3,21122}=
    \sum_{\sigma,\nu,\varpi}\int_{\Lambda^4}\int_{0}^{1}\int_{0}^{1-t}
    v_a(x_1-x_3)\eta(x_1-x_3)\eta(x_2-x_4)\Big(\sum_{p_3\in B^\nu_F}e^{ip_3(x_2-x_3)}\Big)\\
    &\quad\quad\quad\quad\times e^{-tB}a^*(g_{x_1,\sigma})a^*(g_{x_3,\nu})e^{-sB}
    a^*(H_{x_4,\varpi}[0])\\
    &\quad\quad\quad\quad\times
    a(g_{x_4,\varpi})a(g_{x_2,\nu})a(g_{x_1,\sigma})e^{(s+t)B}dsdtdx_1dx_2dx_3dx_4+h.c.\\
     &\mathcal{E}_{\mathcal{V}_3,21123}=
    \sum_{\sigma,\nu,\varpi}\int_{\Lambda^4}\int_{0}^{1}\int_{0}^{1-t}
    v_a(x_1-x_3)\eta(x_1-x_3)\eta(x_2-x_4)\Big(\sum_{p_3\in B^\nu_F}e^{ip_3(x_2-x_3)}\Big)\\
    &\quad\quad\quad\quad\times e^{-tB}a^*(g_{x_1,\sigma})a^*(g_{x_3,\nu})e^{-sB}
    a^*(L_{x_4,\varpi}[0])a(g_{x_4,\varpi})\\
    &\quad\quad\quad\quad\times
    a(g_{x_2,\nu})a(g_{x_1,\sigma})e^{(s+t)B}dsdtdx_1dx_2dx_3dx_4+h.c.
    \end{aligned}
  \end{equation}
  Using Lemmas \ref{control eB on N_ex lemma} and \ref{variant lem 7.1 lemma}, we can bound
  \begin{equation}\label{E_V_3, 21121}
    \pm\mathcal{E}_{\mathcal{V}_3,21121}\lesssim
    N^{\frac{5}{2}}\Vert v_a\Vert_1\Vert\eta\Vert_1\Vert\eta\Vert_\infty(\mathcal{N}_{ex}+1)
    \lesssim N^{\frac{1}{2}}\ell^2(\mathcal{N}_{ex}+1).
  \end{equation}
  Using Lemmas \ref{lemma ineqn K_s} and \ref{control eB on N_ex lemma}, we can bound
  \begin{equation}\label{E_V,21122}
    \begin{aligned}
    \pm\mathcal{E}_{\mathcal{V}_3,21122}\lesssim
    N^{3}\Vert v_a\Vert_1\Vert\eta\Vert_1\Vert\eta\Vert_\infty
    \big(\mathcal{N}_h[0]+1+\ell\mathcal{N}_{ex}\big)
    \lesssim N^{-2\alpha}(\mathcal{K}_s+N^{\frac{1}{3}}),
    \end{aligned}
  \end{equation}
  and
  \begin{equation}\label{E_V,21123}
    \begin{aligned}
    \pm\mathcal{E}_{\mathcal{V}_3,21123}\lesssim
    N^{\frac{5}{2}}N^{\frac{1}{3}}\Vert v_a\Vert_1\Vert\eta\Vert_1\Vert\eta\Vert_\infty
    \big(N^{-\frac{1}{6}}\mathcal{N}_{ex}+N^{\frac{1}{6}}\big)
    \lesssim N^{-2\alpha}(\mathcal{N}_{ex}+N^{\frac{1}{3}}).
    \end{aligned}
  \end{equation}
  For $\mathcal{E}_{\mathcal{V}_3,2113}$, we write $\mathcal{E}_{\mathcal{V}_3,2112}=
  \sum_{j=1}^{2}\mathcal{E}_{\mathcal{V}_3,2113j}$ in the position space with
  \begin{equation}\label{E_V_3,2113 split}
    \begin{aligned}
    &\mathcal{E}_{\mathcal{V}_3,21131}
    =\sum_{\sigma,\nu,\varpi}\int_{\Lambda^4}\int_{0}^{1}\int_{0}^{1-t}
    v_a(x_1-x_3)\eta(x_1-x_3)\eta(x_2-x_4)\Big(\sum_{q_1\in
    B_F^\sigma}e^{iq_1(x_2-x_1)}\Big)\\
    &\quad\quad\quad\quad
    \times e^{-tB}a^*(g_{x_1,\sigma})a^*(g_{x_3,\nu})e^{-sB}a^*(g_{x_4,\varpi})\\
    &\quad\quad\quad\quad
    \times a(h_{x_3,\nu})a(H_{x_4,\varpi}[2\alpha])a(h_{x_2,\sigma})e^{(s+t)B}
    dsdtdx_1dx_2dx_3dx_4+h.c.\\
    &\mathcal{E}_{\mathcal{V}_3,21132}
    =\sum_{\sigma,\nu,\varpi}\int_{\Lambda^4}\int_{0}^{1}\int_{0}^{1-t}
    v_a(x_1-x_3)\eta(x_1-x_3)\eta(x_2-x_4)\Big(\sum_{q_1\in
    B_F^\sigma}e^{iq_1(x_2-x_1)}\Big)\\
    &\quad\quad\quad\quad
    \times e^{-tB}a^*(g_{x_1,\sigma})a^*(g_{x_3,\nu})e^{-sB}a^*(g_{x_4,\varpi})\\
    &\quad\quad\quad\quad
    \times a(h_{x_3,\nu})a(L_{x_4,\varpi}[2\alpha])a(h_{x_2,\sigma})e^{(s+t)B}
    dsdtdx_1dx_2dx_3dx_4+h.c.
    \end{aligned}
  \end{equation}
  Using Lemmas \ref{lemma ineqn K_s} and \ref{control eB on N_ex lemma}, we can bound
  \begin{equation}\label{E_V_3 21131}
    \begin{aligned}
    \pm\mathcal{E}_{\mathcal{V}_3,21131}&\lesssim N^2
    \Vert v_a\Vert_1\Vert\eta\Vert_2\Vert\eta\Vert_\infty
    \big(\theta\mathcal{N}_h[2\alpha]+\theta\ell(\mathcal{N}_{ex}+1)+\theta^{-1}
    (\mathcal{N}_{ex}+1)^2\big)\\
    &\lesssim N^{-2\alpha}\mathcal{K}_s+N^{-\alpha}
    (\mathcal{N}_{ex}+1)+N^{-\frac{2}{3}-\alpha}(\mathcal{N}_{ex}+1)^2
    \end{aligned}
  \end{equation}
  where we choose $\theta=N^{\frac{1}{2}+\frac{\alpha}{2}}$. Similarly,
  \begin{equation}\label{E_V_3 21132}
    \begin{aligned}
    \pm\mathcal{E}_{\mathcal{V}_3,21132}&\lesssim N^2N^{\frac{1}{3}+\alpha}
    \Vert v_a\Vert_1\Vert\eta\Vert_1\Vert\eta\Vert_\infty
    (\mathcal{N}_{ex}+1)\\
    &\lesssim N^{\frac{1}{3}+\alpha}\ell^2(\mathcal{N}_{ex}+1).
    \end{aligned}
  \end{equation}
  Combining (\ref{split E_V_3}-\ref{E_V_3 21132}), we conclude Lemma \ref{control eB on V_3 lemma}.
\end{proof}

\begin{lemma}\label{cal com V_21' lemma}
  \begin{equation}\label{cal com V_21'}
    \begin{aligned}
    &\int_{0}^{1}\int_{0}^{t}e^{-sB}[\mathcal{V}_{21}^\prime,B]e^{sB}dsdt=
    \sum_{k,p,q,\sigma,\nu}W_k\eta_k\chi_{p-k\notin B^{\sigma}_F}\chi_{q+k\notin B^{\nu}_F}\chi_{p\in B^{\sigma}_F}\chi_{q\in B^{\nu}_F}\\
    &\quad\quad\quad\quad\quad\quad\quad\quad\quad
    -\sum_{k,p,q,\sigma}W_k\eta_{k+q-p}\chi_{p-k,q+k\notin B^{\sigma}_F}\chi_{p,q\in B^{\sigma}_F}\chi_{p\neq q}+\mathcal{E}_{\mathcal{V}_{21}^\prime},
    \end{aligned}
  \end{equation}
  and
  \begin{equation}\label{cal com Omega}
    \begin{aligned}
    &\int_{0}^{1}\int_{0}^{t}e^{-sB}[\Omega,B]e^{sB}dsdt\\
    &\quad\quad\quad\quad\quad
    =-\sum_{k,p,q,\sigma,\nu}\eta_k^2k(q-p)(1-\chi_{p-k\notin B^{\sigma}_F}\chi_{q+k\notin B^{\nu}_F})\chi_{p\in B^{\sigma}_F}\chi_{q\in B^{\nu}_F}\\
    &\quad\quad\quad\quad\quad
    -\sum_{k,p,q,\sigma}\eta_{k+q-p}\eta_kk(q-p)\chi_{p-k,q+k\notin B^{\sigma}_F}\chi_{p,q\in B^{\sigma}_F}\chi_{p\neq q}+\mathcal{E}_{\Omega},
    \end{aligned}
  \end{equation}
  where
  \begin{equation}\label{est E_V_21'}
  \begin{aligned}
    \pm(\mathcal{E}_{\mathcal{V}_{21}^\prime}+\mathcal{E}_{\Omega})\lesssim
    &\big(N^{-\frac{2}{3}+3\gamma}+N^{-\frac{3}{2}\alpha}\big)(\mathcal{N}_{ex}+1)\\
    +&\big(N^{-\frac{2}{3}-2\alpha}+N^{-\frac{7}{9}+\alpha+\frac{7}{3}\gamma}\big)
   (\mathcal{N}_{ex}+1)^2\\
    +&\big(N^{-\alpha-\gamma}+N^{-\frac{1}{2}+\alpha+2\gamma}\big)
    \mathcal{K}_s+N^{\frac{1}{3}-\frac{\alpha}{2}}.
  \end{aligned}
  \end{equation}
  with $0<\alpha<\frac{1}{6}$ and $0<\gamma<\alpha$.
\end{lemma}
\begin{proof}
  \par We first prove (\ref{cal com V_21'}) and bound $\mathcal{E}_{\mathcal{V}_{21}^\prime}$. We will see later that (\ref{cal com Omega}) follows similarly. By Lemma \ref{lem commutator}, we can write
  \begin{equation*}
    [\mathcal{V}_{21}^\prime,B]=\Xi_{1}+\Xi_2+\Xi_3
  \end{equation*}
  with
  \begin{equation}\label{[V_21' B] XI}
   \begin{aligned}
    &\Xi_1=\sum_{\substack{\sigma,\nu\\k,p,q}}\sum_{l,r,s}W_k\eta_l
    \delta_{p-k,r-l}\delta_{q+k,s+l}a^*_{p,\sigma}a^*_{q,\nu}a_{s,\nu}a_{r,\sigma}\\
    &\quad\quad\quad\quad\times
    \chi_{p-k\notin B_F^\sigma}\chi_{p,r\in B_F^\sigma}\chi_{q+k\notin B_F^\nu}
    \chi_{q,s\in B_F^\nu}+h.c.\\
    &\Xi_2=\sum_{\substack{\sigma,\nu\\k,p,q}}\sum_{\tau,l,r,s}2W_k\eta_l
    \delta_{p-k,s+l}a^*_{p,\sigma}a^*_{q,\nu}a^*_{r-l,\tau}
    a_{q+k,\nu}a_{s,\sigma}a_{r,\tau}\\
    &\quad\quad\quad\quad\times
    \chi_{p-k\notin B_F^\sigma}\chi_{p,s\in B_F^\sigma}\chi_{q+k\notin B_F^\nu}
    \chi_{q\in B_F^\nu}\chi_{r-l\notin B^\tau_F}\chi_{r\in B_F^\tau}+h.c.\\
    &\Xi_3=\sum_{\substack{\sigma,\nu\\k,p,q}}\sum_{\varpi,l,r,s}W_k\eta_l
    \delta_{p,r}a^*_{r-l,\sigma}a^*_{s+l,\varpi}(a^*_{q,\nu}a_{s,\varpi}-a_{s,\varpi}a^*_{q,\nu})
    a_{q+k,\nu}a_{p-k,\sigma}\\
    &\quad\quad\quad\quad\times
    \chi_{p-k,r-l\notin B_F^\sigma}\chi_{p\in B_F^\sigma}\chi_{q+k\notin B_F^\nu}
    \chi_{q\in B_F^\nu}\chi_{s+l\notin B^\varpi_F}\chi_{s\in B^\varpi_F}+h.c.
   \end{aligned}
  \end{equation}
  We then analyse (\ref{[V_21' B] XI}) term by term.
  \begin{flushleft}
    \textbf{Analysis of $\Xi_1$:}
  \end{flushleft}
  We first split $\Xi_1=\Xi_{11}+\Xi_{12}$ using $1=\chi_{k=l}+\chi_{k\neq l}$. For $\Xi_{11}$, we furthermore split it using
  \begin{equation*}
  \begin{aligned}
    &a^*_{p,\sigma}a^*_{q,\nu}a_{q,\nu}a_{p,\sigma}\\
    =&1-(a_{p,\sigma}a_{p,\sigma}^*+a_{q,\nu}a^*_{q,\nu})+a_{p,\sigma}a_{q,\nu}
    a^*_{q,\nu}a^*_{p,\sigma}+\delta_{p,q}\delta_{\sigma,\nu}(a_{q,\nu}a^*_{p,\sigma}
    +a_{p,\sigma}a^*_{q,\nu}-1)
  \end{aligned}
  \end{equation*}
  That is, we write $\Xi_{11}=\sum_{j=1}^{4}\Xi_{11j}$ with
  \begin{equation}\label{Xi_11j}
    \begin{aligned}
    &\Xi_{111}=\sum_{k,p,q,\sigma,\nu}2W_k\eta_k\chi_{p-k\notin B^{\sigma}_F}\chi_{q+k\notin B^{\nu}_F}\chi_{p\in B^{\sigma}_F}\chi_{q\in B^{\nu}_F}\\
    &\Xi_{112}=-\sum_{k,p,q,\sigma,\nu}4W_k\eta_ka_{p,\sigma}a_{p,\sigma}^*
    \chi_{p-k\notin B^{\sigma}_F}\chi_{q+k\notin B^{\nu}_F}\chi_{p\in B^{\sigma}_F}
    \chi_{q\in B^{\nu}_F}\\
    &\Xi_{113}=\sum_{k,p,q,\sigma,\nu}2W_k\eta_ka_{p,\sigma}a_{q,\nu}
    a^*_{q,\nu}a^*_{p,\sigma}
    \chi_{p-k\notin B^{\sigma}_F}\chi_{q+k\notin B^{\nu}_F}\chi_{p\in B^{\sigma}_F}
    \chi_{q\in B^{\nu}_F}\\
    &\Xi_{114}=\sum_{k,p,\sigma}2W_k\eta_k(2a_{p,\sigma}a^*_{p,\sigma}-1)
    \chi_{p-k\notin B^{\sigma}_F}\chi_{q+k\notin B^{\nu}_F}\chi_{p\in B^{\sigma}_F}
    \chi_{q\in B^{\nu}_F}
    \end{aligned}
      \end{equation}
    Notice that for $p\in B^{\sigma}_F$ and $q\in B^{\nu}_F$, $\chi_{p-k\notin B^{\sigma}_F}\chi_{q+k\notin B^{\nu}_F}\neq 1$ implies $\vert k\vert\lesssim N^{\frac{1}{3}}$. Therefore for $p\in B^{\sigma}_F$ and $q\in B^{\nu}_F$ fixed, using (\ref{est of eta_0}) and (\ref{sum_pW_peta_p 3dscatt}), we have
    \begin{equation}\label{part sumW_keta_k}
      \Big\vert\sum_{k}W_k\eta_k(\chi_{p-k\notin B^{\sigma}_F}\chi_{q+k\notin B^{\nu}_F}-1)\Big\vert\lesssim Na^2\ell^2\ll a^2\ell^{-1}.
    \end{equation}
    Combining (\ref{part sumW_keta_k}) with (\ref{sum_pW_peta_p 3dscatt}), we have
    \begin{equation}\label{part sumW_keta_k 2}
      \Big\vert\sum_{k}W_k\eta_k\chi_{p-k\notin B^{\sigma}_F}\chi_{q+k\notin B^{\nu}_F}\Big\vert\lesssim  a^2\ell^{-1}.
    \end{equation}
    Using (\ref{number of exicited particles ops2}) and (\ref{part sumW_keta_k 2}), we can bound
    \begin{equation}\label{Xi_112-114 bound}
      \begin{aligned}
      &\pm\Xi_{112}\lesssim Na^2\ell^{-1}\mathcal{N}_{ex}
      =N^{-\frac{2}{3}+\alpha}\mathcal{N}_{ex}\\
       &\pm\Xi_{113}\lesssim Na^2\ell^{-1}\mathcal{N}_{ex}
      =N^{-\frac{2}{3}+\alpha}\mathcal{N}_{ex}\\
       &\pm\Xi_{114}\lesssim Na^2\ell^{-1}+a^2\ell^{-1}\mathcal{N}_{ex}
      \lesssim N^{-\frac{2}{3}+\alpha}
      \end{aligned}
    \end{equation}
    \par For $\Xi_{12}$, we similarly rewrite it by $\Xi_{12}=\sum_{j=1}^{3}\Xi_{12j}$ with
    \begin{equation}\label{Xi_12j}
      \begin{aligned}
      &\Xi_{121}=-\sum_{k,p,q,\sigma}2W_k\eta_{k+q-p}\chi_{p-k,q+k\notin B^{\sigma}_F}\chi_{p,q\in B^{\sigma}_F}\chi_{p\neq q}\\
      &\Xi_{122}=\sum_{k,p,q,\sigma}4W_k\eta_{k+q-p}a_{p,\sigma}a_{p,\sigma}^*
      \chi_{p-k,q+k\notin B^{\sigma}_F}\chi_{p,q\in B^{\sigma}_F}\chi_{p\neq q}\\
      &\Xi_{123}=\sum_{k,p,q,\sigma,\nu}\sum_{l,r,s}W_k\eta_l
      a_{r,\sigma}a_{s,\nu}a^*_{q,\nu}a^*_{p,\sigma}
      \delta_{p-k,r-l}\delta_{q+k,s+l}\chi_{k\neq l}\\
    &\quad\quad\quad\quad\times
    \chi_{p-k\notin B_F^\sigma}\chi_{p,r\in B_F^\sigma}\chi_{q+k\notin B_F^\nu}
    \chi_{q,s\in B_F^\nu}+h.c.
      \end{aligned}
    \end{equation}
    Notice that similar to (\ref{part sumW_keta_k 2}), we can bound
    \begin{equation}\label{par sum W_keta_k3}
       \Big\vert\sum_{k}W_k\eta_{k+q-p}\chi_{p-k,q+k\notin B^{\sigma}_F}\Big\vert\lesssim  a^2\ell^{-1}.
    \end{equation}
    Therefore, we have
    \begin{equation}\label{Xi_122}
      \pm\Xi_{122}\lesssim Na^2\ell^{-1}\mathcal{N}_{ex}
      =N^{-\frac{2}{3}+\alpha}\mathcal{N}_{ex}.
    \end{equation}
    For $\Xi_{123}$, we further decompose it by $\Xi_{123}=\sum_{j=1}^{3}\Xi_{123j}$, where
    \begin{equation}\label{Xi_123j}
      \begin{aligned}
      &\Xi_{1231}=-\sum_{k,p,q,\sigma,\nu}2W_k\eta_ka_{p,\sigma}a_{q,\nu}
    a^*_{q,\nu}a^*_{p,\sigma}
    \chi_{p-k\notin B^{\sigma}_F}\chi_{q+k\notin B^{\nu}_F}\chi_{p\in B^{\sigma}_F}
    \chi_{q\in B^{\nu}_F}\\
    &\Xi_{1232}=\sum_{k,p,q,\sigma,\nu}\sum_{l,r,s}W_k\eta_l
      a_{r,\sigma}a_{s,\nu}a^*_{q,\nu}a^*_{p,\sigma}
      \delta_{p-k,r-l}\delta_{q+k,s+l}
   \chi_{p,r\in B_F^\sigma}\chi_{q,s\in B_F^\nu}+h.c.\\
   &\Xi_{1233}=\sum_{k,p,q,\sigma,\nu}\sum_{l,r,s}W_k\eta_l
      a_{r,\sigma}a_{s,\nu}a^*_{q,\nu}a^*_{p,\sigma}
      \delta_{p-k,r-l}\delta_{q+k,s+l}\\
    &\quad\quad\quad\quad\times
    (\chi_{p-k\notin B_F^\sigma}\chi_{q+k\notin B_F^\nu}-1)\chi_{p,r\in B_F^\sigma}
    \chi_{q,s\in B_F^\nu}+h.c.
      \end{aligned}
    \end{equation}
    Similar to (\ref{Xi_112-114 bound}), we have
    \begin{equation}\label{Xi_1231}
      \pm\Xi_{1231}\lesssim Na^2\ell^{-1}\mathcal{N}_{ex}
      =N^{-\frac{2}{3}+\alpha}\mathcal{N}_{ex}.
    \end{equation}
    Switching to the position space, we attain
    \begin{equation}\label{Xi_1232}
      \begin{aligned}
      \pm\Xi_{1232}&=\pm\sum_{\sigma,\nu}\int_{\Lambda^2}\eta(x-y)W(x-y)
      a(g_{x,\sigma})a(g_{y,\nu})a^*(g_{y,\nu})a^*(g_{x,\sigma})dxdy+h.c.\\
      &\lesssim N\Vert\eta\Vert_2^2\Vert W\Vert_2^2\mathcal{N}_{ex}
      \lesssim N^{-\frac{2}{3}+\alpha}\mathcal{N}_{ex}.
      \end{aligned}
    \end{equation}
    Notice that for $p\in B^{\sigma}_F$ and $q\in B^{\nu}_F$, $\chi_{p-k\notin B^{\sigma}_F}\chi_{q+k\notin B^{\nu}_F}\neq 1$ implies $\vert k\vert\lesssim N^{\frac{1}{3}}$. We bound it for $\psi\in\mathcal{H}^{\wedge N}(\{N_{\varsigma_i}\})$,
    \begin{equation}\label{Xi_1233}
      \begin{aligned}
      \vert\langle\Xi_{1233}\psi,\psi\rangle\vert&\lesssim
      \sum_{\substack{p,q,r,\sigma,\nu\\ \vert k\vert\lesssim N^{\frac{1}{3}}}}
      \vert W_k\vert \vert\eta_{k+r-p}\vert \Vert a^*_{q+p-r,\nu}a^*_{r,\sigma}\psi\Vert
      \Vert a^*_{q,\nu}a^*_{p,\sigma}\psi\Vert\\
      &\quad\quad\quad\quad\quad\quad\quad\quad
      \times\chi_{p,r\in B_F^\sigma}\chi_{q,q+p-r\in B_F^\nu}\\
      &\lesssim N^2a^2\ell^2 (\mathcal{N}_{ex}+1)^2\lesssim N^{-\frac{2}{3}-2\alpha}
      (\mathcal{N}_{ex}+1)^2.
      \end{aligned}
    \end{equation}
    Combining (\ref{Xi_11j}-\ref{Xi_1233}), we conclude that
    \begin{equation}\label{Xi_1}
      \begin{aligned}
    &\int_{0}^{1}\int_{0}^{t}e^{-sB}\Xi_1e^{sB}dsdt=
    \sum_{k,p,q,\sigma,\nu}W_k\eta_k\chi_{p-k\notin B^{\sigma}_F}\chi_{q+k\notin B^{\nu}_F}\chi_{p\in B^{\sigma}_F}\chi_{q\in B^{\nu}_F}\\
    &\quad\quad\quad\quad\quad\quad\quad\quad\quad
    -\sum_{k,p,q,\sigma}W_k\eta_{k+q-p}\chi_{p-k,q+k\notin B^{\sigma}_F}\chi_{p,q\in B^{\sigma}_F}\chi_{p\neq q}+\mathcal{E}_{\Xi_1},
      \end{aligned}
    \end{equation}
    and by Lemma \ref{control eB on N_ex lemma}, we have
    \begin{equation}\label{E_XI_1}
      \pm\mathcal{E}_{\Xi_1}\lesssim N^{-\frac{2}{3}+\alpha}(\mathcal{N}_{ex}+1)
      +N^{-\frac{2}{3}-2\alpha}
      (\mathcal{N}_{ex}+1)^2.
    \end{equation}
    \begin{flushleft}
    \textbf{Analysis of $\Xi_2$:}
  \end{flushleft}
  Switching to the position space, we have
  \begin{equation*}
    \begin{aligned}
    \Xi_{2}&=2\sum_{\sigma,\nu,\tau}\int_{\Lambda^4}W(x_1-x_3)\eta(x_2-x_4)
    \Big(\sum_{p_1\notin B_F^\sigma }e^{ip_1(x_1-x_4)}\Big)\\
    &\quad\quad\times a^*(h_{x_2,\tau})
    a^*(g_{x_1,\sigma})a^*(g_{x_3,\nu})
    a(g_{x_4,\sigma})a(g_{x_2,\tau})a(h_{x_3,\nu})dx_1dx_2dx_3dx_4+h.c.
    \end{aligned}
  \end{equation*}
  We further split it into two parts $\Xi_2=\Xi_{21}+\Xi_{22}$ with
  \begin{equation}\label{Xi_2 split}
    \begin{aligned}
    \Xi_{21}&=2\sum_{\sigma,\nu,\tau}\int_{\Lambda^4}W(x_1-x_3)\eta(x_2-x_1)
     a^*(h_{x_2,\tau})
    a^*(g_{x_1,\sigma})a^*(g_{x_3,\nu})\\
    &\quad\quad\times
    a(g_{x_1,\sigma})a(g_{x_2,\tau})a(h_{x_3,\nu})dx_1dx_2dx_3+h.c.\\
     \Xi_{22}&=2\sum_{\sigma,\nu,\tau}\int_{\Lambda^4}W(x_1-x_3)\eta(x_2-x_4)
    \Big(\sum_{p_1\in B_F^\sigma }e^{ip_1(x_1-x_4)}\Big)\\
    &\quad\quad\times a^*(h_{x_2,\tau})
    a^*(g_{x_1,\sigma})a^*(g_{x_3,\nu})
    a(g_{x_4,\sigma})a(g_{x_2,\tau})a(h_{x_3,\nu})dx_1dx_2dx_3dx_4+h.c.
    \end{aligned}
  \end{equation}
  Then we bound them respectively by
  \begin{equation}\label{Xi_21 22 bound}
    \begin{aligned}
  &\pm\Xi_{21}\lesssim N^2\Vert W\Vert_1\Vert\eta\Vert_1\mathcal{N}_{ex}\lesssim
  N^{-\frac{2}{3}-2\alpha}\mathcal{N}_{ex}\\
  &\pm\Xi_{21}\lesssim N^{\frac{5}{2}}\Vert W\Vert_1\Vert\eta\Vert_1\mathcal{N}_{ex}\lesssim
  N^{-\frac{1}{6}-2\alpha}\mathcal{N}_{ex}
  \end{aligned}
  \end{equation}
  By Lemma \ref{control eB on N_ex lemma}, we have
  \begin{equation}\label{Xi_2}
    \pm\int_{0}^{1}\int_{0}^{t}e^{-sB}\Xi_2e^{sB}dsdt\lesssim
    N^{-\frac{1}{6}-2\alpha}(\mathcal{N}_{ex}+1).
  \end{equation}
   \begin{flushleft}
    \textbf{Analysis of $\Xi_3$:}
  \end{flushleft}
  Using the fact that
  \begin{equation*}
    \begin{aligned}
    e^{-sB}a^*_{r-l,\sigma}a^*_{s+l,\varpi}e^{sB}&=
    a^*_{r-l,\sigma}a^*_{s+l,\varpi}+\int_{0}^{s}e^{-mB}
    [a^*_{r-l,\sigma}a^*_{s+l,\varpi},B]e^{mB}dm\\
    e^{-sB}a_{q+k,\nu}a_{p-k,\sigma}e^{sB}&=
    a_{q+k,\nu}a_{p-k,\sigma}+\int_{0}^{s}e^{-mB}
    [a_{q+k,\nu}a_{p-k,\sigma},B]e^{mB}dm
    \end{aligned}
  \end{equation*}
  We write $ e^{-sB}\Xi_3e^{sB}=M_1+M_2+M_3$ with
  \begin{equation}\label{M split}
    \begin{aligned}
    &M_1=\sum_{\substack{\sigma,\nu\\k,p,q}}\sum_{\varpi,l,r,s}W_k\eta_l
    \delta_{p,r}a^*_{r-l,\sigma}a^*_{s+l,\varpi}\\
    &\quad\quad\quad\quad\times
    e^{-sB}(a^*_{q,\nu}a_{s,\varpi}-a_{s,\varpi}a^*_{q,\nu})e^{sB}
    a_{q+k,\nu}a_{p-k,\sigma}\\
    &\quad\quad\quad\quad\times
    \chi_{p-k,r-l\notin B_F^\sigma}\chi_{p\in B_F^\sigma}\chi_{q+k\notin B_F^\nu}
    \chi_{q\in B_F^\nu}\chi_{s+l\notin B^\varpi_F}\chi_{s\in B^\varpi_F}+h.c.\\
     &M_2=\sum_{\substack{\sigma,\nu\\k,p,q}}\sum_{\varpi,l,r,s}W_k\eta_l
    \delta_{p,r}\int_{0}^{s}e^{-sB}a^*_{r-l,\sigma}a^*_{s+l,\varpi}\\
    &\quad\quad\quad\quad\times
    (a^*_{q,\nu}a_{s,\varpi}-a_{s,\varpi}a^*_{q,\nu})e^{(s-m)B}
    [a_{q+k,\nu}a_{p-k,\sigma},B]e^{mB}dm\\
    &\quad\quad\quad\quad\times
    \chi_{p-k,r-l\notin B_F^\sigma}\chi_{p\in B_F^\sigma}\chi_{q+k\notin B_F^\nu}
    \chi_{q\in B_F^\nu}\chi_{s+l\notin B^\varpi_F}\chi_{s\in B^\varpi_F}+h.c.\\
     &M_3=\sum_{\substack{\sigma,\nu\\k,p,q}}\sum_{\varpi,l,r,s}W_k\eta_l
    \delta_{p,r}\int_{0}^{s}e^{-mB}[a^*_{r-l,\sigma}a^*_{s+l,\varpi},B]\\
    &\quad\quad\quad\quad\times
    e^{(m-s)B}(a^*_{q,\nu}a_{s,\varpi}-a_{s,\varpi}a^*_{q,\nu})e^{sB}
    a_{q+k,\nu}a_{p-k,\sigma}dm\\
    &\quad\quad\quad\quad\times
    \chi_{p-k,r-l\notin B_F^\sigma}\chi_{p\in B_F^\sigma}\chi_{q+k\notin B_F^\nu}
    \chi_{q\in B_F^\nu}\chi_{s+l\notin B^\varpi_F}\chi_{s\in B^\varpi_F}+h.c.
    \end{aligned}
  \end{equation}
  We then analyse (\ref{M split}) term by term.
   \begin{flushleft}
    \textbf{Analysis of $M_1$:}
  \end{flushleft}
  We can write $M_1$ by
  \begin{equation}\label{M_1 rewrt}
    \begin{aligned}
    M_1&=\sum_{\substack{\sigma,\nu,\varpi\\p_2,q_2}}\int_{\Lambda^3}W_{p_2-q_2}e^{i(q_2-p_2)x_3}
   \chi_{p_2\notin B_F^\nu}\chi_{q_2\in B_F^\nu}\eta(x_2-x_4)
    \Big(\sum_{q_1\in B_F^\sigma}e^{-iq_1(x_2-x_3)}\Big)\\
    &\quad\quad\times
    a^*(h_{x_2,\sigma})a^*(h_{x_4,\varpi})e^{-sB}
    (a^*_{q_2,\nu}a(g_{x_4,\varpi})-a(g_{x_4,\varpi})a^*_{q_2,\nu})
    e^{sB}\\
    &\quad\quad\times
    a_{p_2,\nu}a(h_{x_3,\sigma})dx_2dx_3dx_4+h.c.
    \end{aligned}
  \end{equation}
  For some $u>0$ to be determined and $\delta=\frac{1}{3}+u$, we decompose $M_1=\sum_{j=1}^{4}M_{1j}$ with
  \begin{equation}\label{M_1 rewrt 2}
    \begin{aligned}
     M_{11}&=\sum_{\substack{\sigma,\nu,\varpi\\p_2,q_2}}\int_{\Lambda^3}
     W_{p_2-q_2}e^{i(q_2-p_2)x_3}
   \chi_{p_2\in P_{F,\delta}^\nu}\chi_{q_2\in B_F^\nu}\eta(x_2-x_4)
    \Big(\sum_{q_1\in B_F^\sigma}e^{-iq_1(x_2-x_3)}\Big)\\
    &\quad\quad\times
    a^*(h_{x_2,\sigma})a^*(H_{x_4,\varpi}[\delta])e^{-sB}
    (a^*_{q_2,\nu}a(g_{x_4,\varpi})-a(g_{x_4,\varpi})a^*_{q_2,\nu})
    e^{sB}\\
    &\quad\quad\times
    a_{p_2,\nu}a(h_{x_3,\sigma})dx_2dx_3dx_4+h.c.\\
     M_{12}&=\sum_{\substack{\sigma,\nu,\varpi\\p_2,q_2}}\int_{\Lambda^3}
     W_{p_2-q_2}e^{i(q_2-p_2)x_3}
   \chi_{p_2\in P_{F,\delta}^\nu}\chi_{q_2\in B_F^\nu}\eta(x_2-x_4)
    \Big(\sum_{q_1\in B_F^\sigma}e^{-iq_1(x_2-x_3)}\Big)\\
    &\quad\quad\times
    a^*(h_{x_2,\sigma})a^*(L_{x_4,\varpi}[\delta])e^{-sB}
    (a^*_{q_2,\nu}a(g_{x_4,\varpi})-a(g_{x_4,\varpi})a^*_{q_2,\nu})
    e^{sB}\\
    &\quad\quad\times
    a_{p_2,\nu}a(h_{x_3,\sigma})dx_2dx_3dx_4+h.c.\\
     M_{13}&=\sum_{\substack{\sigma,\nu,\varpi\\p_2,q_2}}\int_{\Lambda^3}
     W_{p_2-q_2}e^{i(q_2-p_2)x_3}
   \chi_{p_2\in A_{F,\delta}^\nu}\chi_{q_2\in B_F^\nu}\eta(x_2-x_4)
    \Big(\sum_{q_1\in B_F^\sigma}e^{-iq_1(x_2-x_3)}\Big)\\
    &\quad\quad\times
    a^*(h_{x_2,\sigma})a^*(H_{x_4,\varpi}[\delta])e^{-sB}
    (a^*_{q_2,\nu}a(g_{x_4,\varpi})-a(g_{x_4,\varpi})a^*_{q_2,\nu})
    e^{sB}\\
    &\quad\quad\times
    a_{p_2,\nu}a(h_{x_3,\sigma})dx_2dx_3dx_4+h.c.\\
    M_{14}&=\sum_{\substack{\sigma,\nu,\varpi\\p_2,q_2}}\int_{\Lambda^3}
     W_{p_2-q_2}e^{i(q_2-p_2)x_3}
   \chi_{p_2\in A_{F,\delta}^\nu}\chi_{q_2\in B_F^\nu}\eta(x_2-x_4)
    \Big(\sum_{q_1\in B_F^\sigma}e^{-iq_1(x_2-x_3)}\Big)\\
    &\quad\quad\times
    a^*(h_{x_2,\sigma})a^*(L_{x_4,\varpi}[\delta])e^{-sB}
    (a^*_{q_2,\nu}a(g_{x_4,\varpi})-a(g_{x_4,\varpi})a^*_{q_2,\nu})
    e^{sB}\\
    &\quad\quad\times
    a_{p_2,\nu}a(h_{x_3,\sigma})dx_2dx_3dx_4+h.c.
    \end{aligned}
  \end{equation}
  We adopt the notation $k_{p,x,\sigma}$:
  \begin{equation*}
    k_{p,x,\sigma}(z)\coloneqq\sum_{q\in B_F^\sigma}W_{q-p}e^{i(q-p)x}f_{q,\sigma}(z),
  \end{equation*}
  on which by (\ref{norm bound  of creation and annihilation}), we have
  \begin{equation*}
    \Vert a(k_{p,x,\sigma})\Vert^2\leq\Vert k_{p,x,\sigma}\Vert^2
    = \sum_{q\in B_F^\sigma}\vert W_{q-p}\vert^2.
  \end{equation*}
  Similar to the proof of Lemma \ref{contol V_21'Omega [K,B]}, we can bound
  \begin{equation}\label{M_11 bound}
    \begin{aligned}
    &\pm M_{11}\\
    &=\pm\sum_{\substack{\sigma,\nu,\varpi,p_2}}\int_{\Lambda^3}
   \chi_{p_2\in P_{F,\delta}^\nu}\eta(x_2-x_4)
    \Big(\sum_{q_1\in B_F^\sigma}e^{-iq_1(x_2-x_3)}\Big)
    \\
    &\quad\quad\times
   a^*(h_{x_2,\sigma})a^*(H_{x_4,\varpi}[\delta])
    e^{-sB}(a^*(k_{p_2,x_3,\nu})a(g_{x_4,\varpi})-a(g_{x_4,\varpi})a^*(k_{p_2,x_3,\nu}))
    \\
    &\quad\quad\times e^{sB} a_{p_2,\nu}a(h_{x_3,\sigma}) dx_2dx_3dx_4+h.c.\\
    &\lesssim N\Vert\eta\Vert_2\Big(\sum_{\nu,q_2\in B_F^\nu}\sum_{p_2\in P_{F,\delta}^\nu}
    \frac{\vert W_{q_2-p_2}\vert^2}{\vert p_2\vert^2-\vert k_F^\nu\vert^2}\Big)^\frac{1}{2}
    \big(\theta(\mathcal{N}_{ex}+1)\mathcal{N}_h[\delta]
    +\theta^{-1}(\mathcal{N}_{ex}+1)\mathcal{K}_s\big)\\
    &\lesssim N^{\frac{1}{6}+\frac{\alpha}{2}-\frac{3}{2}u}\mathcal{K}_s
    \end{aligned}
  \end{equation}
  where we have used (\ref{energy est bog u}) in Lemma \ref{energy est bog lem} and we choose
  $\theta=N^{\frac{1}{3}+u}$. We then choose for some $0<\gamma<\alpha$,
  \begin{equation}\label{choose u}
    u=\frac{1}{9}+\alpha+\frac{2}{3}\gamma,
  \end{equation}
  such that
  \begin{equation}\label{M_11 bound true}
    \pm M_{11}\lesssim N^{-\alpha-\gamma}\mathcal{K}_s.
  \end{equation}
  For $M_{12}$, using Lemma \ref{control eB on N_ex lemma}, we bound
  \begin{equation}\label{M_12 bound}
    \begin{aligned}
    &\pm M_{12}\\
    &=\pm\sum_{\substack{\sigma,\nu,\varpi,p_2}}\int_{\Lambda^3}
   \chi_{p_2\in P_{F,\delta}^\nu}\eta(x_2-x_4)
    \Big(\sum_{q_1\in B_F^\sigma}e^{-iq_1(x_2-x_3)}\Big)
    \\
    &\quad\quad\times
   a^*(h_{x_2,\sigma})a^*(L_{x_4,\varpi}[\delta])
    e^{-sB}(a^*(k_{p_2,x_3,\nu})a(g_{x_4,\varpi})-a(g_{x_4,\varpi})a^*(k_{p_2,x_3,\nu}))
    \\
    &\quad\quad\times e^{sB} a_{p_2,\nu}a(h_{x_3,\sigma}) dx_2dx_3dx_4+h.c.\\
    &\lesssim N^{\frac{3}{2}+\frac{3}{2}u}\Vert\eta\Vert_1\Big(\sum_{\nu,q_2\in B_F^\nu}\sum_{p_2\in P_{F,\delta}^\nu}
    \frac{\vert W_{q_2-p_2}\vert^2}{\vert p_2\vert^2-\vert k_F^\nu\vert^2}\Big)^\frac{1}{2}
    \big(\theta(\mathcal{N}_{ex}+1)^2
    +\theta^{-1}\mathcal{K}_s\big)\\
    &\lesssim N^{-\alpha-\gamma}\mathcal{K}_s+N^{-\frac{7}{9}+\alpha+\frac{7}{3}\gamma}
    (\mathcal{N}_{ex}+1)^2
    \end{aligned}
  \end{equation}
  where we choose $\theta=N^{-\frac{7}{18}+\alpha+\frac{5}{3}\gamma}$.
   \par For $M_{13}$, we evaluate it in the position space,
  \begin{equation}\label{M_13 bound}
    \begin{aligned}
    &\pm M_{13}\\
    &=\pm\sum_{\sigma,\nu,\varpi}\int_{\Lambda^4}W(x_1-x_3)\eta(x_2-x_4)
    \Big(\sum_{q_1\in B_F^\sigma}e^{-iq_1(x_2-x_3)}\Big)
    \\
    &\quad\quad\times
   a^*(h_{x_2,\sigma})a^*(H_{x_4,\varpi}[\delta])
    e^{-sB}(a^*(g_{x_1,\nu})a(g_{x_4,\varpi})-a(g_{x_4,\varpi})a^*(g_{x_1,\nu}))
    \\
    &\quad\quad\times e^{sB} a(L_{x_1,\nu}[\delta])a(h_{x_3,\sigma})dx_1dx_2dx_3dx_4+h.c.\\
    &\lesssim N^{2+\frac{3}{2}u}\Vert W\Vert_1\Vert\eta\Vert_2
    \big(\theta(\mathcal{N}_{ex}+1)^2
    +\theta^{-1}\mathcal{N}_h[\delta]\big)\\
    &\lesssim N^{-\alpha-\gamma}\mathcal{K}_s+N^{-\frac{8}{9}+\alpha+\frac{5}{3}\gamma}
    (\mathcal{N}_{ex}+1)^2
    \end{aligned}
  \end{equation}
  where we choose $\theta=N^{-\frac{13}{18}+\frac{3}{2}\alpha+\frac{5}{3}\gamma-\frac{3}{2}u}$. For $M_{14}$, we furthermore split it by $M_{14}=\sum_{j=1}^{4}M_{14j}$, with
  \begin{equation}\label{M_14 split}
    \begin{aligned}
    & M_{141}=\sum_{\sigma,\nu,\varpi}\int_{\Lambda^4}W(x_1-x_3)\eta(x_2-x_4)
    \Big(\sum_{q_1\in B_F^\sigma}e^{-iq_1(x_2-x_3)}\Big)
    \\
    &\quad\quad\times
   a^*(H_{x_2,\sigma}[1/3])a^*(L_{x_4,\varpi}[\delta])
    e^{-sB}(a^*(g_{x_1,\nu})a(g_{x_4,\varpi})-a(g_{x_4,\varpi})a^*(g_{x_1,\nu}))
    \\
    &\quad\quad\times e^{sB} a(L_{x_1,\nu}[\delta])a(H_{x_3,\sigma}[1/3])
    dx_1dx_2dx_3dx_4+h.c.\\
     & M_{142}=\sum_{\sigma,\nu,\varpi}\int_{\Lambda^4}W(x_1-x_3)\eta(x_2-x_4)
    \Big(\sum_{q_1\in B_F^\sigma}e^{-iq_1(x_2-x_3)}\Big)
    \\
    &\quad\quad\times
   a^*(L_{x_2,\sigma}[1/3])a^*(L_{x_4,\varpi}[\delta])
    e^{-sB}(a^*(g_{x_1,\nu})a(g_{x_4,\varpi})-a(g_{x_4,\varpi})a^*(g_{x_1,\nu}))
    \\
    &\quad\quad\times e^{sB} a(L_{x_1,\nu}[\delta])a(H_{x_3,\sigma}[1/3])
    dx_1dx_2dx_3dx_4+h.c.\\
     & M_{143}=\sum_{\sigma,\nu,\varpi}\int_{\Lambda^4}W(x_1-x_3)\eta(x_2-x_4)
    \Big(\sum_{q_1\in B_F^\sigma}e^{-iq_1(x_2-x_3)}\Big)
    \\
    &\quad\quad\times
   a^*(H_{x_2,\sigma}[1/3])a^*(L_{x_4,\varpi}[\delta])
    e^{-sB}(a^*(g_{x_1,\nu})a(g_{x_4,\varpi})-a(g_{x_4,\varpi})a^*(g_{x_1,\nu}))
    \\
    &\quad\quad\times e^{sB} a(L_{x_1,\nu}[\delta])a(L_{x_3,\sigma}[1/3])
    dx_1dx_2dx_3dx_4+h.c.\\
     & M_{144}=\sum_{\sigma,\nu,\varpi}\int_{\Lambda^4}W(x_1-x_3)\eta(x_2-x_4)
    \Big(\sum_{q_1\in B_F^\sigma}e^{-iq_1(x_2-x_3)}\Big)
    \\
    &\quad\quad\times
   a^*(L_{x_2,\sigma}[1/3])a^*(L_{x_4,\varpi}[\delta])
    e^{-sB}(a^*(g_{x_1,\nu})a(g_{x_4,\varpi})-a(g_{x_4,\varpi})a^*(g_{x_1,\nu}))
    \\
    &\quad\quad\times e^{sB} a(L_{x_1,\nu}[\delta])a(L_{x_3,\sigma}[1/3])
    dx_1dx_2dx_3dx_4+h.c.
    \end{aligned}
  \end{equation}
  Notice the fact that $\mathcal{N}_{l}[\delta]\leq\mathcal{N}_{ex}$,  we bound (\ref{M_14 split}) respectively by
  \begin{equation}\label{M_14 bound}
    \begin{aligned}
    &\pm M_{141}\lesssim N^{\frac{5}{2}+3u}\Vert W\Vert_1\Vert\eta\Vert_1
    \mathcal{N}_{h}[1/3]\lesssim N^{-\frac{1}{2}+\alpha+2\gamma}\mathcal{K}_s\\
    &\pm M_{142}\lesssim N^{\frac{5}{2}+\frac{3}{2}u}\Vert W\Vert_1\Vert\eta\Vert_1
    (\theta\mathcal{N}_l[\delta]+\theta^{-1}\mathcal{N}_h[1/3])
    \lesssim N^{-\alpha-\gamma}\mathcal{K}_s+N^{-\frac{2}{3}+3\gamma}\mathcal{N}_{ex}\\
    &\pm M_{143}\lesssim N^{\frac{5}{2}+\frac{3}{2}u}\Vert W\Vert_1\Vert\eta\Vert_1
    (\theta\mathcal{N}_l[\delta]+\theta^{-1}\mathcal{N}_h[1/3])
    \lesssim N^{-\alpha-\gamma}\mathcal{K}_s+N^{-\frac{2}{3}+3\gamma}\mathcal{N}_{ex}\\
    &\pm M_{144}\lesssim N^{\frac{5}{2}}\Vert W\Vert_1\Vert\eta\Vert_1
    \mathcal{N}_l[\delta]
    \lesssim N^{-\frac{1}{6}-2\alpha}\mathcal{N}_{ex}
    \end{aligned}
  \end{equation}
  where we choose $\theta=N^{-\frac{2}{3}+\frac{\alpha}{2}+2\gamma}$. Combining (\ref{M_1 rewrt}-\ref{M_14 bound}), we conclude that
  \begin{equation}\label{M_1 bound}
    \begin{aligned}
    \pm M_1\lesssim& \big(N^{-\frac{2}{3}+3\gamma}+N^{-\frac{1}{6}-2\alpha}\big)
    \mathcal{N}_{ex}+\big(N^{-\alpha-\gamma}+N^{-\frac{1}{2}+\alpha+2\gamma}\big)
    \mathcal{K}_s\\
    &+N^{-\frac{7}{9}+\alpha+\frac{7}{3}\gamma}(\mathcal{N}_{ex}+1)^2
    \end{aligned}
  \end{equation}
  with $0<\gamma<\alpha<\frac{1}{6}$.
   \begin{flushleft}
    \textbf{Analysis of $M_2$:}
  \end{flushleft}
  By a direct calculation, we can write $M_2=\sum_{j=1}^{3}M_{2j}$. For $M_{21}$, we define it by
  \begin{equation}\label{M_2 split}
    \begin{aligned}
    &M_{21}=\sum_{\substack{\sigma,\nu\\k,p,q}}\sum_{\varpi,l,r,s}\sum_{l_1,r_1,s_1}
    W_k\eta_l\eta_{l_1}
    \delta_{p,r}\delta_{p-k,r_1-l_1}\delta_{q+k,s_1+l_1}
    \\
    &\quad\quad\times\int_{0}^{s}e^{-sB}a^*_{r-l,\sigma}a^*_{s+l,\varpi}
    (a^*_{q,\nu}a_{s,\varpi}-a_{s,\varpi}a^*_{q,\nu})e^{(s-m)B}
    a_{s_1,\nu}a_{r_1,\sigma}e^{mB}dm\\
    &\quad\quad\times
    \chi_{p-k,r-l\notin B_F^\sigma}\chi_{p,r_1\in B_F^\sigma}\chi_{q+k\notin B_F^\nu}
    \chi_{q,s_1\in B_F^\nu}\chi_{s+l\notin B^\varpi_F}\chi_{s\in B^\varpi_F}+h.c.
    \end{aligned}
  \end{equation}
  Switching to the position space, we rewrite $M_{21}=\sum_{j=1}^{4}M_{21j}$ with
  \begin{equation}\label{M_21 split}
    \begin{aligned}
    &M_{211}=\sum_{\sigma,\nu,\varpi}\int_{0}^{s}\int_{\Lambda^4}
    W(x_1-x_4)\eta(x_2-x_5)\eta(x_1-x_4)\Big(\sum_{q_1\in B_F^\sigma}
    e^{-iq_1(x_2-x_1)}\Big)\\
    &\quad\quad\times
    e^{-sB}a^*(h_{x_2,\sigma})a^*(h_{x_5,\varpi})
    (a^*(g_{x_4,\nu})a(g_{x_5,\varpi})-a(g_{x_5,\varpi})a^*(g_{x_4,\nu}))\\
    &\quad\quad\times e^{(s-m)B}a(g_{x_4,\nu})a(g_{x_1,\sigma})e^{mB}dx_1dx_2dx_4dx_5dm+h.c.\\
    &M_{212}=-\sum_{\sigma,\nu,\varpi}\int_{0}^{s}\int_{\Lambda^5}
    W(x_1-x_4)\eta(x_2-x_5)\eta(x_1-x_6)\\
    &\quad\quad\times \Big(\sum_{q_1\in B_F^\sigma}
    e^{-iq_1(x_2-x_1)}\Big)\Big(\sum_{p_2\in B_F^\nu}
    e^{-ip_2(x_4-x_6)}\Big)\\
    &\quad\quad\times
    e^{-sB}a^*(h_{x_2,\sigma})a^*(h_{x_5,\varpi})
    (a^*(g_{x_4,\nu})a(g_{x_5,\varpi})-a(g_{x_5,\varpi})a^*(g_{x_4,\nu}))\\
    &\quad\quad\times e^{(s-m)B}a(g_{x_6,\nu})a(g_{x_1,\sigma})e^{mB}
    dx_1dx_2dx_4dx_5dx_6dm+h.c.\\
    &M_{213}=-\sum_{\sigma,\nu,\varpi}\int_{0}^{s}\int_{\Lambda^5}
    W(x_1-x_4)\eta(x_2-x_5)\eta(x_3-x_4)\\
    &\quad\quad\times \Big(\sum_{q_1\in B_F^\sigma}
    e^{-iq_1(x_2-x_1)}\Big)\Big(\sum_{p_1\in B_F^\sigma}
    e^{-ip_1(x_1-x_3)}\Big)\\
    &\quad\quad\times
    e^{-sB}a^*(h_{x_2,\sigma})a^*(h_{x_5,\varpi})
    (a^*(g_{x_4,\nu})a(g_{x_5,\varpi})-a(g_{x_5,\varpi})a^*(g_{x_4,\nu}))\\
    &\quad\quad\times e^{(s-m)B}a(g_{x_4,\nu})a(g_{x_3,\sigma})e^{mB}
    dx_1dx_2dx_3dx_4dx_5dm+h.c.\\
    &M_{214}=\sum_{\sigma,\nu,\varpi}\int_{0}^{s}\int_{\Lambda^6}
    W(x_1-x_4)\eta(x_2-x_5)\eta(x_3-x_6)\\
    &\quad\quad\times \Big(\sum_{q_1\in B_F^\sigma}
    e^{-iq_1(x_2-x_1)}\Big)\Big(\sum_{p_1\in B_F^\sigma}
    e^{-ip_1(x_1-x_3)}\Big)\Big(\sum_{p_2\in B_F^\nu}
    e^{-ip_2(x_4-x_6)}\Big)\\
    &\quad\quad\times
    e^{-sB}a^*(h_{x_2,\sigma})a^*(h_{x_5,\varpi})
    (a^*(g_{x_4,\nu})a(g_{x_5,\varpi})-a(g_{x_5,\varpi})a^*(g_{x_4,\nu}))\\
    &\quad\quad\times e^{(s-m)B}a(g_{x_6,\nu})a(g_{x_3,\sigma})e^{mB}
    dx_1dx_2dx_3dx_4dx_5dx_6dm+h.c.
    \end{aligned}
  \end{equation}
  By Lemma \ref{control eB on N_ex lemma}, we can bound
  \begin{equation}\label{M_212-214 bound}
    \begin{aligned}
    &\pm M_{212}\lesssim N^3\Vert W\Vert_1\Vert\eta\Vert_2^2(\mathcal{N}_{ex}+1)
    \lesssim N^{-\frac{1}{3}-\alpha}(\mathcal{N}_{ex}+1)\\
    &\pm M_{213}\lesssim N^3\Vert W\Vert_1\Vert\eta\Vert_2^2(\mathcal{N}_{ex}+1)
    \lesssim N^{-\frac{1}{3}-\alpha}(\mathcal{N}_{ex}+1)\\
    &\pm M_{214}\lesssim N^{\frac{7}{2}}\Vert W\Vert_1\Vert\eta\Vert_2\Vert\eta\Vert_1(\mathcal{N}_{ex}+1)
    \lesssim N^{-\frac{1}{3}-\frac{5}{2}\alpha}(\mathcal{N}_{ex}+1)
    \end{aligned}
  \end{equation}
  For $M_{211}$, we choose $\delta=\frac{\alpha}{2}+\gamma$ such that $0<\delta<\frac{1}{3}$, and further split it by $M_{211}=\sum_{j=1}^{4}M_{211j}$ using the fact
  \begin{equation}\label{M_211  split}
    \begin{aligned}
    a^*(h_{x_2,\sigma})a^*(h_{x_5,\varpi})=&
    a^*(H_{x_2,\sigma}[\delta])a^*(H_{x_5,\varpi}[\delta])+
    a^*(H_{x_2,\sigma}[\delta])a^*(L_{x_5,\varpi}[\delta])\\&+
    a^*(L_{x_2,\sigma}[\delta])a^*(H_{x_5,\varpi}[\delta])+
    a^*(L_{x_2,\sigma}[\delta])a^*(L_{x_5,\varpi}[\delta])
    \end{aligned}
  \end{equation}
  Notice the direct inequality $\mathcal{N}_h[\delta]\leq\mathcal{N}_{ex}$, we can use Lemma \ref{control eB on N_ex lemma} to bound for $j=2,3,4$
  \begin{equation}\label{M_211 bound}
    \begin{aligned}
    \pm M_{2111}&\lesssim N^{\frac{5}{2}}\Vert W\Vert_1\Vert\eta\Vert_2
    \Vert\eta\Vert_\infty
    \big(\mathcal{N}_h[\delta]+1+\ell(\mathcal{N}_{ex}+1)\big)\\
    &\lesssim
    N^{-\alpha-\gamma}\mathcal{K}_s+N^{-\frac{3}{2}\alpha}\mathcal{N}_{ex}
    +N^{\frac{1}{3}-\frac{\alpha}{2}}\\
    \pm M_{211j}&\lesssim N^{\frac{5}{2}+\frac{1}{3}+\frac{\delta}{2}}\Vert
    W\Vert_1\Vert\eta\Vert_1
    \Vert\eta\Vert_\infty
    \big(\theta(\mathcal{N}_{ex}+1)+\theta^{-1}\big)\\
    &\lesssim
    N^{-3\alpha+\gamma}(\mathcal{N}_{ex}+1)
    +N^{\frac{1}{3}-\frac{\alpha}{2}}
    \end{aligned}
  \end{equation}
  where we choose $\theta=N^{-\frac{1}{6}-\frac{5}{4}\alpha+\frac{\gamma}{2}}$.
  \par For $M_{22}$ and $M_{23}$, we have
  \begin{equation}\label{M_22 M_23 write}
    \begin{aligned}
    &M_{22}=-\sum_{\substack{\sigma,\nu\\k,p,q}}\sum_{\varpi,l,r,s}\sum_{\varpi_1,l_1,r_1,s_1}
    W_k\eta_l\eta_{l_1}
    \delta_{p,r}\delta_{p-k,r_1-l_1}
    \\
    &\quad\quad\times\int_{0}^{s}e^{-sB}a^*_{r-l,\sigma}a^*_{s+l,\varpi}
    (a^*_{q,\nu}a_{s,\varpi}-a_{s,\varpi}a^*_{q,\nu})e^{(s-m)B}\\
    &\quad\quad\times a^*_{s_1+l_1,\varpi_1}a_{q+k,\nu}
    a_{s_1,\varpi_1}a_{r_1,\sigma}e^{mB}dm
    \chi_{p-k,r-l\notin B_F^\sigma}\chi_{p,r_1\in B_F^\sigma}\\
    &\quad\quad\times
    \chi_{q+k\notin B_F^\nu}
    \chi_{q\in B_F^\nu}\chi_{s+l\notin B^\varpi_F}\chi_{s\in B^\varpi_F}
    \chi_{s_1+l_1\notin B_F^{\varpi_1}}\chi_{s_1\in B_F^{\varpi_1}}+h.c.\\
    &M_{23}=-\sum_{\substack{\sigma,\nu\\k,p,q}}\sum_{\varpi,l,r,s}\sum_{\tau_1,l_1,r_1,s_1}
    W_k\eta_l\eta_{l_1}
    \delta_{p,r}\delta_{q+k,s_1+l_1}
    \\
    &\quad\quad\times\int_{0}^{s}e^{-sB}a^*_{r-l,\sigma}a^*_{s+l,\varpi}
    (a^*_{q,\nu}a_{s,\varpi}-a_{s,\varpi}a^*_{q,\nu})e^{(s-m)B}\\
    &\quad\quad\times a^*_{r_1-l_1,\tau_1}a_{p-k,\sigma}
    a_{s_1,\nu}a_{r_1,\tau_1}e^{mB}dm
    \chi_{p-k,r-l\notin B_F^\sigma}\chi_{p\in B_F^\sigma}\\
    &\quad\quad\times
    \chi_{q+k\notin B_F^\nu}
    \chi_{q,s_1\in B_F^\nu}\chi_{s+l\notin B^\varpi_F}\chi_{s\in B^\varpi_F}
    \chi_{r_1-l_1\notin B_F^{\tau_1}}\chi_{r_1\in B_F^{\tau_1}}+h.c.
    \end{aligned}
  \end{equation}
  Switching to the position space, we have
  \begin{equation}\label{M_22 M_23}
    \begin{aligned}
    &M_{22}=-\sum_{\sigma,\nu,\varpi,\varpi_1}\int_{0}^{s}\int_{\Lambda^6}
    W(x_1-x_4)\eta(x_2-x_5)\eta(x_3-x_6)\\
    &\quad\quad\times \Big(\sum_{q_1\in B_F^\sigma}
    e^{-iq_1(x_2-x_1)}\Big)\Big(\sum_{p_1\notin B_F^\sigma}
    e^{-ip_1(x_1-x_3)}\Big)\\
    &\quad\quad\times
    e^{-sB}a^*(h_{x_2,\sigma})a^*(h_{x_5,\varpi})
    (a^*(g_{x_4,\nu})a(g_{x_5,\varpi})-a(g_{x_5,\varpi})a^*(g_{x_4,\nu}))\\
    &\quad\quad\times
    e^{(s-m)B}a^*(h_{x_6,\varpi_1})a(h_{x_4,\nu})a(g_{x_{6},\varpi_1})a(g_{x_3,\sigma})
    e^{mB}+h.c.\\
    &M_{23}=-\sum_{\sigma,\nu,\varpi,\tau_1}\int_{0}^{s}\int_{\Lambda^6}
    W(x_1-x_4)\eta(x_2-x_5)\eta(x_3-x_6)\\
    &\quad\quad\times \Big(\sum_{q_1\in B_F^\sigma}
    e^{-iq_1(x_2-x_1)}\Big)\Big(\sum_{p_1\notin B_F^\sigma}
    e^{-ip_2(x_4-x_6)}\Big)\\
    &\quad\quad\times
    e^{-sB}a^*(h_{x_2,\sigma})a^*(h_{x_5,\varpi})
    (a^*(g_{x_4,\nu})a(g_{x_5,\varpi})-a(g_{x_5,\varpi})a^*(g_{x_4,\nu}))\\
    &\quad\quad\times
    e^{(s-m)B}a^*(h_{x_3,\tau_1})a(h_{x_1,\sigma})a(g_{x_{6},\nu})a(g_{x_3,\tau_1})
    e^{mB}+h.c.
    \end{aligned}
  \end{equation}
  We can use the equality $\sum_{p_j\notin B_F^\sharp}=\sum_{p_j}-\sum_{p_j\in B_F^\sharp}$ for suitable subscript $j$ and spin $\sharp$ to split $M_{22}$ and $M_{23}$. Using Lemmas \ref{control eB on N_ex lemma} and \ref{variant lem 7.1 lemma}, we can bound them by
  \begin{equation}\label{M_22,M23 bound}
    \pm M_{22},\pm M_{23}\lesssim N^3\Vert W\Vert_1\Vert\eta\Vert_2\Vert\eta\Vert_1
    (\mathcal{N}_{ex}+1)^2\lesssim N^{-\frac{5}{6}-\frac{5}{2}\alpha}(\mathcal{N}_{ex}+1)^2.
  \end{equation}
  Combining (\ref{M_2 split}-\ref{M_22,M23 bound}), we conclude that
  \begin{equation}\label{M_2 bound}
    \begin{aligned}
    \pm M_2&\lesssim N^{-\alpha-\gamma}\mathcal{K}_s+N^{-\frac{3}{2}\alpha}\mathcal{N}_{ex}
    +N^{\frac{1}{3}-\frac{\alpha}{2}}\\
    &+N^{-\frac{5}{6}-\frac{5}{2}\alpha}(\mathcal{N}_{ex}+1)^2.
    \end{aligned}
  \end{equation}
  \begin{flushleft}
    \textbf{Analysis of $M_3$:}
  \end{flushleft}
  Via a direct calculation, we can write $M_3=\sum_{j=1}^{3}M_{3j}$, where
  \begin{equation}\label{M_31}
    \begin{aligned}
    &M_{31}=\sum_{\substack{\sigma,\nu\\k,p,q}}\sum_{\varpi,l,r,s}
    \sum_{l_1,r_1,s_1}W_k\eta_l\eta_{l_1}\delta_{p,r}\delta_{r-l,r_1-l_1}\delta_{s+l,s_1+l_1}\\
    &\quad\quad\times
    \int_{0}^{s}e^{-mB}a^*_{r_1,\sigma}a^*_{s_1,\varpi}e^{(m-s)B}
    (a^*_{q,\nu}a_{s,\varpi}-a_{s,\varpi}a^*_{q,\nu})e^{sB}a_{q+k,\nu}a_{p-k,\sigma}dm\\
    &\quad\quad\times \chi_{p-k\notin B_F^\sigma}\chi_{p,r_1\in B_F^\sigma}
    \chi_{q+k\notin B_F^\nu}
    \chi_{q\in B_F^\nu}\chi_{r-l\notin B^\sigma_F}\chi_{s+l\notin B_F^\varpi}
    \chi_{s,s_1\in B_F^\varpi}+h.c.\\
    &M_{32}=-\sum_{\substack{\sigma,\nu\\k,p,q}}\sum_{\varpi,l,r,s}
    \sum_{\varpi_1,l_1,r_1,s_1}
    W_k\eta_l\eta_{l_1}\delta_{p,r}\delta_{r-l,r_1-l_1}
    \int_{0}^{s}e^{-mB}a^*_{r_1,\sigma}a^*_{s_1,\varpi_1}
    \\
    &\quad\quad\times a^*_{s+l,\varpi}a_{s_1+l_1,\varpi}
    e^{(m-s)B}
    (a^*_{q,\nu}a_{s,\varpi}-a_{s,\varpi}a^*_{q,\nu})e^{sB}a_{q+k,\nu}a_{p-k,\sigma}dm
    \\
    &\quad\quad\times \chi_{p-k\notin B_F^\sigma}\chi_{p,r_1\in B_F^\sigma}\chi_{q+k\notin B_F^\nu}
    \chi_{q\in B_F^\nu}\chi_{r-l\notin B^\sigma_F}\chi_{s+l\notin B_F^\varpi}
    \chi_{s\in B_F^\varpi}\\
    &\quad\quad\times\chi_{s_1+l_1\notin B_F^{\varpi_1}}
    \chi_{s_1\in B_F^{\varpi_1}}+h.c.\\
    &M_{33}=-\sum_{\substack{\sigma,\nu\\k,p,q}}\sum_{\varpi,l,r,s}
    \sum_{\tau_1,l_1,r_1,s_1}
    W_k\eta_l\eta_{l_1}\delta_{p,r}\delta_{s+l,s_1+l_1}
    \int_{0}^{s}e^{-mB}a^*_{r_1,\tau_1}a^*_{s_1,\varpi}
    \\
    &\quad\quad\times a^*_{r-l,\sigma}a_{r_1-l_1,\tau_1}
    e^{(m-s)B}
    (a^*_{q,\nu}a_{s,\varpi}-a_{s,\varpi}a^*_{q,\nu})e^{sB}a_{q+k,\nu}a_{p-k,\sigma}dm
    \\
    &\quad\quad\times \chi_{p-k\notin B_F^\sigma}\chi_{p\in B_F^\sigma}\chi_{q+k\notin B_F^\nu}
    \chi_{q\in B_F^\nu}\chi_{r-l\notin B^\sigma_F}\chi_{s+l\notin B_F^\varpi}
    \chi_{s,s_1\in B_F^\varpi}\\
    &\quad\quad\times\chi_{r_1-l_1\notin B_F^{\tau_1}}
    \chi_{r_1\in B_F^{\tau_1}}+h.c.
    \end{aligned}
  \end{equation}
  For $M_{31}$, switching to the position space, we have
  \begin{equation}\label{M_31 pos}
    \begin{aligned}
    &M_{31}=\sum_{\sigma,\nu,\varpi,p_1}
    \chi_{p_1\notin B_F^\sigma}\int_{0}^{s}\int_{\Lambda^5}
    \eta(x_2-x_5)\eta(x_3-x_6)
    \Big(\sum_{q_1\in B_F^\sigma}W_{q_1-p_1}e^{-iq_1(x_2-x_4)}\Big)\\
    &\quad\quad\times\Big(\sum_{p_3\notin B_F^\sigma}e^{ip_3(x_2-x_3)}\Big)
    \Big(\sum_{p_4\notin B_F^\varpi}e^{ip_4(x_5-x_6)}\Big)\\
    &\quad\quad\times e^{-mB}a^*(g_{x_3,\sigma})a^*(g_{x_6,\varpi})e^{(m-s)B}
    (a^*(g_{x_4,\nu})a(g_{x_5,\varpi})-a(g_{x_5,\varpi})a^*(g_{x_4,\nu}))\\
    &\quad\quad\times e^{sB}a(h_{x_4,\nu})a_{p_1,\sigma}e^{-ip_1x_4}
    dx_2dx_3dx_4dx_5dx_6dm+h.c.
    \end{aligned}
  \end{equation}
  Using
  \begin{equation*}
     \sum_{p_3\notin B_F^\sigma}\sum_{p_4\notin B_F^\varpi}=\sum_{p_3,p_4}-
  \sum_{p_4}\sum_{p_3\in B_F^\sigma}-\sum_{p_3}\sum_{p_4\in B_F^\varpi}
  +\sum_{p_3\in B_F^\sigma}\sum_{p_4\in B_F^\varpi},
  \end{equation*}
 we can split $M_{31}=\sum_{j=1}^{4}M_{31j}$. Also notice that for $j=2,4$
 \begin{equation*}
   \int_{\Lambda}\Big\vert\sum_{q_1\in B_F^\sigma}W_{q_1-p_1}e^{-iq_1(x_2-x_4)}\Big\vert^2
   dx_j=\sum_{q_1\in B_F^\sigma}\vert W_{q_1-p_1}\vert^2.
 \end{equation*}
 Therefore, we bound
 \begin{equation}\label{M_31 bound}
   \begin{aligned}
   \pm M_{311}&\lesssim N^2\Vert\eta\Vert_2^2\Big(
   \sum_{\sigma,q_1\in B_F^\sigma}\sum_{p_1\notin B_F^\sigma}\frac{\vert W_{q_1-p_1}\vert^2}
   {\vert p_1\vert^2-\vert k_F^\sigma\vert^2}\Big)^{\frac{1}{2}}
   (\mathcal{K}_s+\mathcal{N}_{ex}+1)\\
   &\lesssim N^{-\frac{2}{3}-\frac{\alpha}{2}}(\mathcal{K}_s+1)\\
   \pm M_{312},\pm M_{313}&\lesssim N^{\frac{5}{2}}\Vert\eta\Vert_2\Vert\eta\Vert_1
   \Big(\sum_{\sigma,q_1\in B_F^\sigma}\sum_{p_1\notin B_F^\sigma}\frac{\vert W_{q_1-p_1}\vert^2}
   {\vert p_1\vert^2-\vert k_F^\sigma\vert^2}\Big)^{\frac{1}{2}}
   (\mathcal{K}_s+\mathcal{N}_{ex}+1)\\
   &\lesssim N^{-\frac{2}{3}-2\alpha}(\mathcal{K}_s+1)\\
   \pm M_{314}&\lesssim N^3\Vert\eta\Vert_1^2\Big(
   \sum_{\sigma,q_1\in B_F^\sigma}\sum_{p_1\notin B_F^\sigma}\frac{\vert W_{q_1-p_1}\vert^2}
   {\vert p_1\vert^2-\vert k_F^\sigma\vert^2}\Big)^{\frac{1}{2}}
   (\mathcal{K}_s+\mathcal{N}_{ex}+1)\\
   &\lesssim N^{-\frac{2}{3}-\frac{7}{2}\alpha}(\mathcal{K}_s+1)
   \end{aligned}
 \end{equation}
 where we have used (\ref{energy est bog}) in Lemma \ref{energy est bog lem}.
 \par For $M_{32}$ and $M_{33}$, switching to the position space, we have
  %

 \begin{equation}\label{M_32 M_33 pos}
   \begin{aligned}
   &M_{32}=\sum_{\sigma,\nu,\varpi,\varpi_1,p_1}
    \chi_{p_1\notin B_F^\sigma}\int_{0}^{s}\int_{\Lambda^5}
    \eta(x_2-x_5)\eta(x_3-x_6)
    \Big(\sum_{q_1\in B_F^\sigma}W_{q_1-p_1}e^{-iq_1(x_2-x_4)}\Big)\\
    &\quad\quad\times
    \Big(\sum_{p_3\notin B_F^\sigma}e^{ip_3(x_2-x_3)}\Big)
    e^{-mB}a^*(g_{x_3,\sigma})a^*(g_{x_6,\varpi_1})a^*(h_{x_5,\varpi})a(h_{x_6,\varpi_1})\\
    &\quad\quad\times e^{(m-s)B}
    (a(g_{x_5,\varpi})a^*(g_{x_4,\nu})-a^*(g_{x_4,\nu})a(g_{x_5,\varpi}))\\
    &\quad\quad\times e^{sB}a(h_{x_4,\nu})a_{p_1,\sigma}e^{-ip_1x_4}
    dx_2dx_3dx_4dx_5dx_6dm+h.c.\\
    &M_{33}=\sum_{\sigma,\nu,\varpi,\tau_1,p_1}
    \chi_{p_1\notin B_F^\sigma}\int_{0}^{s}\int_{\Lambda^5}
    \eta(x_2-x_5)\eta(x_3-x_6)
    \Big(\sum_{q_1\in B_F^\sigma}W_{q_1-p_1}e^{-iq_1(x_2-x_4)}\Big)\\
    &\quad\quad\times
    \Big(\sum_{p_4\notin B_F^\varpi}e^{ip_4(x_5-x_6)}\Big)
    e^{-mB}a^*(g_{x_3,\tau_1})a^*(g_{x_6,\varpi})a^*(h_{x_2,\sigma})a(h_{x_3,\tau_1})\\
    &\quad\quad\times e^{(m-s)B}
    (a(g_{x_5,\varpi})a^*(g_{x_4,\nu})-a^*(g_{x_4,\nu})a(g_{x_5,\varpi}))\\
    &\quad\quad\times e^{sB}a(h_{x_4,\nu})a_{p_1,\sigma}e^{-ip_1x_4}
    dx_2dx_3dx_4dx_5dx_6dm+h.c.
   \end{aligned}
 \end{equation}
 We use the equality $\sum_{p_j\notin B_F^\sharp}=\sum_{p_j}-\sum_{p_j\in B_F^\sharp}$ for suitable subscript $j$ and spin $\sharp$ to decompose $M_{32}$ and $M_{33}$. Using Lemmas \ref{control eB on N_ex lemma} and \ref{variant lem 7.1 lemma}, we can bound them by
 \begin{equation}\label{M_32 M_33 bound}
 \begin{aligned}
   \pm M_{32},\pm M_{33}&\lesssim N^{\frac{5}{2}}\Vert\eta\Vert_1^2\Big(
   \sum_{\sigma,q_1\in B_F^\sigma}\sum_{p_1\notin B_F^\sigma}\frac{\vert W_{q_1-p_1}\vert^2}
   {\vert p_1\vert^2-\vert k_F^\sigma\vert^2}\Big)^{\frac{1}{2}}
   \big(\theta^{-1} \mathcal{K}_s+\theta(\mathcal{N}_{ex}+1)^3\big)\\
   &\lesssim N^{-\alpha-\gamma}\mathcal{K}_s+N^{-\frac{7}{3}-6\alpha+\gamma}
   (\mathcal{N}_{ex}+1)^3
 \end{aligned}
 \end{equation}
 where we choose $\theta=N^{-\frac{7}{6}-\frac{5}{2}\alpha+\gamma}$. Combining (\ref{M_31}-\ref{M_32 M_33 bound}), we conclude that
 \begin{equation}\label{M_3 bound}
   \begin{aligned}
   \pm M_3\lesssim N^{-\alpha-\gamma}\mathcal{K}_s+N^{-\frac{7}{3}-6\alpha+\gamma}
   (\mathcal{N}_{ex}+1)^3+N^{-\frac{2}{3}-\frac{\alpha}{2}}.
   \end{aligned}
 \end{equation}

 \par Combining (\ref{M split}), (\ref{M_1 bound}), (\ref{M_2 bound}) and (\ref{M_3 bound}), we conclude that
 \begin{equation}\label{Xi_3 bound}
   \begin{aligned}
   \pm\int_{0}^{1}\int_{0}^{t}e^{-sB}\Xi_3e^{eB}dsdt&\lesssim
   \big(N^{-\frac{2}{3}+2\gamma}+N^{-\frac{3}{2}\alpha}\big)(\mathcal{N}_{ex}+1)\\
   &+\big(N^{-\alpha-\gamma}+N^{-\frac{1}{2}+\alpha+2\gamma}\big)\mathcal{K}_s\\
   &+N^{-\frac{7}{9}+\alpha+\frac{7}{3}\gamma}(\mathcal{N}_{ex}+1)^2+N^{\frac{1}{3}
   -\frac{\alpha}{2}}.
   \end{aligned}
 \end{equation}
 \begin{flushleft}
   \textbf{Conclusion of Lemma \ref{cal com V_21' lemma}:}
 \end{flushleft}
 Combining (\ref{Xi_1}), (\ref{Xi_2}) and (\ref{Xi_3 bound}), we reach (\ref{cal com V_21'}), with $\mathcal{E}_{\mathcal{V}_{21}^\prime}$ satisfies the bound (\ref{est E_V_21'}).
 \par The proof of (\ref{cal com Omega}) is almost the same as the proof of (\ref{cal com V_21'}). We can similarly rewrite
 \begin{equation*}
   [\Omega,B]=\tilde{\Xi}_1+\tilde{\Xi}_2+\tilde{\Xi}_3,
 \end{equation*}
 where each of $\tilde{\Xi}_j$ is defined by replacing the coefficient $W_k$ in $\Xi_j$ defined in (\ref{[V_21' B] XI}) by $\eta_kk(q-p)$. The estimates to $\tilde{\Xi}_2$ and $\tilde{\Xi}_3$ follow directly from above discussions. One only needs to notice that since $\vert q-p\vert\lesssim N^{\frac{1}{3}}$ and $\ell=N^{-\frac{1}{3}-\alpha}$, we have
 \begin{equation*}
   \begin{aligned}
   \Vert W\Vert_1\lesssim a,&\quad N^{\frac{1}{3}}\Vert\nabla\eta\Vert_1\lesssim
   aN^{\frac{1}{3}}\ell\ll a\\
   \sum_{p_1\notin B_F^\sigma}\frac{\vert W_{q_1-p_1}\vert^2}
   {\vert p_1\vert^2-\vert k_F^\sigma\vert^2}\lesssim a^2\ell^{-1},&\quad
   N^{\frac{2}{3}}\sum_{p_1\notin B_F^\sigma}\frac{\vert \eta_{q_1-p_1}(q_1-p_1)\vert^2}
   {\vert p_1\vert^2-\vert k_F^\sigma\vert^2}
   \lesssim a^2N^{\frac{2}{3}}\ell\ll a^2\ell^{-1}
   \end{aligned}
 \end{equation*}
 and for some $\delta>\frac{1}{3}$,
 \begin{equation*}
 \begin{aligned}
   &\sum_{p_1\in P_{F,\delta}^\sigma}\frac{\vert W_{q_1-p_1}\vert^2}
   {\vert p_1\vert^2-\vert k_F^\sigma\vert^2}\lesssim a^2\ell^{-2}N^{-\delta}\\
   &N^{\frac{2}{3}}\sum_{p_1\in P_{F,\delta}^\sigma}\frac{\vert \eta_{q_1-p_1}(q_1-p_1)\vert^2}
   {\vert p_1\vert^2-\vert k_F^\sigma\vert^2}
   \lesssim a^2N^{\frac{2}{3}}N^{-\delta}\ll a^2\ell^{-2}N^{-\delta}
 \end{aligned}
 \end{equation*}
 \par We need to deal with $\tilde{\Xi}_{1}$. For the evaluation of $\tilde{\Xi}_{11}$, we additionally notice that for $p\in B_F^\sigma$ and $q\in B_F^\nu$,
 \begin{equation*}
   \begin{aligned}
   \sum_{k}\eta_k^2k(q-p)\chi_{p-k\notin B_F^\sigma}\chi_{q+k\notin B_F^\nu}=
   -\sum_{k}\eta_k^2k(q-p)(1-\chi_{p-k\notin B_F^\sigma}\chi_{q+k\notin B_F^\nu}).
   \end{aligned}
 \end{equation*}
 By (\ref{est of eta_0}), we can bound
 \begin{equation}\label{eta_peta_p}
   \begin{aligned}
   \Big\vert\sum_{k}\eta_k^2k(q-p)(1-\chi_{p-k\notin B_F^\sigma}\chi_{q+k\notin B_F^\nu})\Big\vert
   \lesssim N^{\frac{2}{3}}\sum_{\vert k\vert\lesssim N^{\frac{1}{3}}} a^2\ell^4
   \lesssim N^{\frac{5}{3}}a^2\ell^4.
   \end{aligned}
 \end{equation}
 Comparing (\ref{eta_peta_p}) with (\ref{part sumW_keta_k 2}), it is easy to find we can evaluate $\tilde{\Xi}_{11}$ in the way we bound $\Xi_{11}$.
 \par For the evaluation of $\tilde{\Xi}_{12}$, we similarly rewrite it by $\tilde{\Xi}_{12}=\sum_{j=1}^{3}\tilde{\Xi}_{12j}$ with
    \begin{equation}\label{Xi_12j tilde}
      \begin{aligned}
      &\tilde{\Xi}_{121}=-\sum_{k,p,q,\sigma}2\eta_k\eta_{k+q-p}k(q-p)\chi_{p-k,q+k\notin B^{\sigma}_F}\chi_{p,q\in B^{\sigma}_F}\chi_{p\neq q}\\
      &\tilde{\Xi}_{122}=\sum_{k,p,q,\sigma}4\eta_k\eta_{k+q-p}k(q-p)a_{p,\sigma}a_{p,\sigma}^*
      \chi_{p-k,q+k\notin B^{\sigma}_F}\chi_{p,q\in B^{\sigma}_F}\chi_{p\neq q}\\
      &\tilde{\Xi}_{123}=\sum_{k,p,q,\sigma,\nu}\sum_{l,r,s}\eta_k\eta_lk(q-p)
      a_{r,\sigma}a_{s,\nu}a^*_{q,\nu}a^*_{p,\sigma}
      \delta_{p-k,r-l}\delta_{q+k,s+l}\chi_{k\neq l}\\
    &\quad\quad\quad\quad\times
    \chi_{p-k\notin B_F^\sigma}\chi_{p,r\in B_F^\sigma}\chi_{q+k\notin B_F^\nu}
    \chi_{q,s\in B_F^\nu}+h.c.
      \end{aligned}
    \end{equation}
    Notice that for $p,q\in B_F^\sigma$,
 \begin{equation*}
   \begin{aligned}
   &\sum_{k}\eta_k\eta_{k+q-p}k(q-p)\chi_{p-k,q+k\notin B_F^\sigma}\\
   =&\sum_{k}\eta_k\eta_{k+q-p}(k+q-p)(q-p)\chi_{p-k,q+k\notin B_F^\sigma}
   -\sum_{k}\eta_k\eta_{k+q-p}(q-p)^2\chi_{p-k,q+k\notin B_F^\sigma}\\
   =&-\sum_{\tilde{k}}\eta_{\tilde{k}+q-p}
   \eta_{\tilde{k}}\tilde{k}(q-p)\chi_{p-\tilde{k},q+\tilde{k}\notin B_F^\sigma}
   -\sum_{k}\eta_k\eta_{k+q-p}(q-p)^2\chi_{p-k,q+k\notin B_F^\sigma}
   \end{aligned}
 \end{equation*}
 where we let $\tilde{k}=-(k+q-p)$ in the last equality.
 \par Therefore, since $\vert q-p\vert\lesssim N^{\frac{1}{3}}$, we can bound
 \begin{equation}\label{etaeta2}
   \Big\vert\sum_{k}\eta_k\eta_{k+q-p}(q-p)^2\chi_{p-k,q+k\notin B_F^\sigma}\Big\vert
   \lesssim N^{\frac{2}{3}}\Vert\eta\Vert^2_2\lesssim N^{\frac{2}{3}}a^2\ell.
 \end{equation}
 Hence we have
    \begin{equation}\label{Xi_122 tilde}
      \pm\tilde{\Xi}_{122}\lesssim N^{\frac{5}{3}}a^2\ell\mathcal{N}_{ex}
      =N^{-\frac{2}{3}-\alpha}\mathcal{N}_{ex}.
    \end{equation}
    For $\tilde{\Xi}_{123}$, we continue by rewriting it by $\tilde{\Xi}_{123}=\sum_{j=1}^{3}\tilde{\Xi}_{123j}$ with
    \begin{equation}\label{Xi_123j tilde}
      \begin{aligned}
      &\tilde{\Xi}_{1231}=-\sum_{k,p,q,\sigma,\nu}2\eta_k^2k(q-p)a_{p,\sigma}a_{q,\nu}
    a^*_{q,\nu}a^*_{p,\sigma}
    \chi_{p-k\notin B^{\sigma}_F}\chi_{q+k\notin B^{\nu}_F}\chi_{p\in B^{\sigma}_F}
    \chi_{q\in B^{\nu}_F}\\
    &\tilde{\Xi}_{1232}=\sum_{k,p,q,\sigma,\nu}\sum_{l,r,s}\eta_k\eta_lk(q-p)
      a_{r,\sigma}a_{s,\nu}a^*_{q,\nu}a^*_{p,\sigma}
      \delta_{p-k,r-l}\delta_{q+k,s+l}\\
    &\quad\quad\quad\quad\times
   \chi_{p,r\in B_F^\sigma}\chi_{q,s\in B_F^\nu}+h.c.\\
   &\tilde{\Xi}_{1233}=\sum_{k,p,q,\sigma,\nu}\sum_{l,r,s}\eta_k\eta_lk(q-p)
      a_{r,\sigma}a_{s,\nu}a^*_{q,\nu}a^*_{p,\sigma}
      \delta_{p-k,r-l}\delta_{q+k,s+l}\\
    &\quad\quad\quad\quad\times
    (\chi_{p-k\notin B_F^\sigma}\chi_{q+k\notin B_F^\nu}-1)\chi_{p,r\in B_F^\sigma}
    \chi_{q,s\in B_F^\nu}+h.c.
      \end{aligned}
    \end{equation}
    Similar to $\tilde{\Xi}_{11}$, we have
    \begin{equation}\label{Xi_1231 tilde}
      \pm\tilde{\Xi}_{1231}\lesssim N^{\frac{8}{3}}a^2\ell^{4}\mathcal{N}_{ex}
      =N^{-\frac{2}{3}-4\alpha}\mathcal{N}_{ex}.
    \end{equation}
    Switching to the position space, we know that for some universal constant $C$,
    \begin{equation}\label{Xi_1232 tilde}
      \begin{aligned}
      &\tilde{\Xi}_{1232}=C\sum_{\sigma,\nu}\int_{\Lambda^2}\eta(x-y)\nabla_{x}\eta(x-y)
      a(g_{x,\sigma})a(g_{y,\nu})a^*(\nabla_{y}g_{y,\nu})a^*(g_{x,\sigma})dxdy+h.c.\\
      &=-C\sum_{\sigma,\nu}\int_{\Lambda^2}\nabla_{x}\eta(x-y)\eta(x-y)
      a(g_{x,\sigma})a(g_{y,\nu})a^*(\nabla_{y}g_{y,\nu})a^*(g_{x,\sigma})dxdy+h.c.\\
      &-C\sum_{\sigma,\nu}\int_{\Lambda^2}\eta(x-y)\eta(x-y)
      a(\nabla_{x}g_{x,\sigma})a(g_{y,\nu})a^*(\nabla_{y}g_{y,\nu})a^*(g_{x,\sigma})dxdy+h.c.\\
      &-C\sum_{\sigma,\nu}\int_{\Lambda^2}\eta(x-y)\eta(x-y)
      a(g_{x,\sigma})a(g_{y,\nu})a^*(\nabla_{y}g_{y,\nu})a^*(\nabla_{x}g_{x,\sigma})dxdy+h.c.
      \end{aligned}
    \end{equation}
    Therefore we bound
    \begin{equation}\label{Xi_1232 tilde bound}
      \pm\tilde{\Xi}_{1232}\lesssim N^{\frac{5}{3}}\Vert\eta\Vert_2^2\mathcal{N}_{ex}
      \lesssim N^{-\frac{2}{3}-\alpha}\mathcal{N}_{ex}
    \end{equation}
    since we have
    \begin{equation}\label{nable a(g_x)}
      \sum_\sigma\int_{\Lambda}a(\nabla_{x}g_{x,\sigma})a^*(\nabla_{x}g_{x,\sigma})dx
      =\sum_{\sigma,k\in B_F^\sigma}\vert k\vert^2 a_{k,\sigma}a^*_{k,\sigma}\lesssim
      N^{\frac{2}{3}}\mathcal{N}_{ex}.
    \end{equation}
    Notice that for $p\in B^{\sigma}_F$ and $q\in B^{\nu}_F$, $\chi_{p-k\notin B^{\sigma}_F}\chi_{q+k\notin B^{\nu}_F}\neq 1$ implies $\vert k\vert\lesssim N^{\frac{1}{3}}$. We can bound for $\psi\in\mathcal{H}^{\wedge N}(\{N_{\varsigma_i}\})$,
    \begin{equation}\label{Xi_1233 tilde}
      \begin{aligned}
      \vert\langle\tilde{\Xi}_{1233}\psi,\psi\rangle\vert&\lesssim N^{\frac{2}{3}}
      \sum_{\substack{p,q,r,\sigma,\nu\\ \vert k\vert\lesssim N^{\frac{1}{3}}}}
      \vert \eta_k\vert \vert\eta_{k+r-p}\vert \Vert a^*_{q+p-r,\nu}a^*_{r,\sigma}\psi\Vert
      \Vert a^*_{q,\nu}a^*_{p,\sigma}\psi\Vert\\
      &\quad\quad\quad\quad\quad\quad\quad\quad
      \times\chi_{p,r\in B_F^\sigma}\chi_{q,q+p-r\in B_F^\nu}\\
      &\lesssim N^{\frac{8}{3}}a^2\ell^4 (\mathcal{N}_{ex}+1)^2\lesssim N^{-\frac{2}{3}-4\alpha}
      (\mathcal{N}_{ex}+1)^2.
      \end{aligned}
    \end{equation}
    Therefore, we conclude the proof of Lemma \ref{cal com V_21' lemma}.
\end{proof}

\begin{lemma}\label{cal com V_21 lemma}
   \begin{equation}\label{cal com V_21}
    \begin{aligned}
    \int_{0}^{1}\int_{t}^{1}e^{-sB}[\mathcal{V}_{21},B]e^{sB}dsdt=&
    \frac{1}{2}\sum_{k}\hat{v}_k\eta_k\sum_{\sigma\neq\nu}N_\sigma N_\nu\\
    &-\sum_{k}\hat{v}_k\eta_k\sum_{\sigma\neq\nu}N_\sigma \mathcal{N}_{ex,\nu}
    +\mathcal{E}_{\mathcal{V}_{21}},
    \end{aligned}
  \end{equation}
  and
  \begin{equation}\label{control Gamma}
    \int_{0}^{1}e^{-tB}\Gamma e^{tB}dt=\mathcal{E}_{\Gamma},
  \end{equation}
  where
  \begin{equation}\label{est E_V_21}
  \begin{aligned}
    \pm(\mathcal{E}_{\mathcal{V}_{21}}+\mathcal{E}_{\Gamma})\lesssim&
    N^{-\frac{3}{2}\alpha}(\mathcal{N}_{ex}+1)
    +N^{-\frac{2}{3}-\frac{\alpha}{2}}(\mathcal{N}_{ex}+1)^2\\
    +&\big(N^{-\alpha-\gamma}+N^{-\frac{1}{6}}\big)
    \mathcal{K}_s+N^{-\gamma}\mathcal{V}_4+N^{\frac{1}{3}-\frac{\alpha}{2}}.
  \end{aligned}
  \end{equation}
  with $0<\alpha<\frac{1}{6}$ and $0<\gamma<\alpha$.
\end{lemma}
\begin{proof}
  \par We first prove (\ref{cal com V_21}) and bound $\mathcal{E}_{\mathcal{V}_{21}}$. Similar to the evaluation of (\ref{cal com V_21'}), we have
   \begin{equation*}
    [\mathcal{V}_{21}^\prime,B]=\hat{\Xi}_{1}+\hat{\Xi}_2+\hat{\Xi}_3
  \end{equation*}
  with
  \begin{equation}\label{[V_21 B] XI}
   \begin{aligned}
    &\hat{\Xi}_1=\sum_{\substack{\sigma,\nu\\k,p,q}}\sum_{l,r,s}\frac{1}{2}\hat{v}_k\eta_l
    \delta_{p-k,r-l}\delta_{q+k,s+l}a^*_{p,\sigma}a^*_{q,\nu}a_{s,\nu}a_{r,\sigma}\\
    &\quad\quad\quad\quad\times
    \chi_{p-k\notin B_F^\sigma}\chi_{p,r\in B_F^\sigma}\chi_{q+k\notin B_F^\nu}
    \chi_{q,s\in B_F^\nu}+h.c.\\
    &\hat{\Xi}_2=\sum_{\substack{\sigma,\nu\\k,p,q}}\sum_{\tau,l,r,s}\hat{v}_k\eta_l
    \delta_{p-k,s+l}a^*_{p,\sigma}a^*_{q,\nu}a^*_{r-l,\tau}
    a_{q+k,\nu}a_{s,\sigma}a_{r,\tau}\\
    &\quad\quad\quad\quad\times
    \chi_{p-k\notin B_F^\sigma}\chi_{p,s\in B_F^\sigma}\chi_{q+k\notin B_F^\nu}
    \chi_{q\in B_F^\nu}\chi_{r-l\notin B^\tau_F}\chi_{r\in B_F^\tau}+h.c.\\
    &\hat{\Xi}_3=\sum_{\substack{\sigma,\nu\\k,p,q}}\sum_{\varpi,l,r,s}\frac{1}{2}\hat{v}_k\eta_l
    \delta_{p,r}a^*_{r-l,\sigma}a^*_{s+l,\varpi}(a^*_{q,\nu}a_{s,\varpi}-a_{s,\varpi}a^*_{q,\nu})
    a_{q+k,\nu}a_{p-k,\sigma}\\
    &\quad\quad\quad\quad\times
    \chi_{p-k,r-l\notin B_F^\sigma}\chi_{p\in B_F^\sigma}\chi_{q+k\notin B_F^\nu}
    \chi_{q\in B_F^\nu}\chi_{s+l\notin B^\varpi_F}\chi_{s\in B^\varpi_F}+h.c.
   \end{aligned}
  \end{equation}
  \begin{flushleft}
    \textbf{Analysis of $\hat{\Xi}_1$:}
  \end{flushleft}
  Similar to the analysis of $\Xi_1$ in Lemma \ref{cal com V_21' lemma}, we split $\hat{\Xi}_1=\hat{\Xi}_{11}+\hat{\Xi}_{12}$ using $1=\chi_{k=l}+\chi_{k\neq l}$, and similar to (\ref{Xi_11j}), we write $\hat{\Xi}_{11}=\sum_{j=1}^{4}\hat{\Xi}_{11j}$ with
  \begin{equation}\label{hat Xi_11j}
    \begin{aligned}
    &\hat{\Xi}_{111}=\sum_{k,p,q,\sigma,\nu}\hat{v}_k\eta_k\chi_{p-k\notin B^{\sigma}_F}\chi_{q+k\notin B^{\nu}_F}\chi_{p\in B^{\sigma}_F}\chi_{q\in B^{\nu}_F}\\
    &\hat{\Xi}_{112}=-\sum_{k,p,q,\sigma,\nu}2\hat{v}_k\eta_ka_{p,\sigma}a_{p,\sigma}^*
    \chi_{p-k\notin B^{\sigma}_F}\chi_{q+k\notin B^{\nu}_F}\chi_{p\in B^{\sigma}_F}
    \chi_{q\in B^{\nu}_F}\\
    &\hat{\Xi}_{113}=\sum_{k,p,q,\sigma,\nu}\hat{v}_k\eta_ka_{p,\sigma}a_{q,\nu}
    a^*_{q,\nu}a^*_{p,\sigma}
    \chi_{p-k\notin B^{\sigma}_F}\chi_{q+k\notin B^{\nu}_F}\chi_{p\in B^{\sigma}_F}
    \chi_{q\in B^{\nu}_F}\\
    &\hat{\Xi}_{114}=\sum_{k,p,\sigma}\hat{v}_k\eta_k(2a_{p,\sigma}a^*_{p,\sigma}-1)
    \chi_{p-k\notin B^{\sigma}_F}\chi_{q+k\notin B^{\nu}_F}\chi_{p\in B^{\sigma}_F}
    \chi_{q\in B^{\nu}_F}
    \end{aligned}
      \end{equation}
 Notice that for $p\in B^{\sigma}_F$ and $q\in B^{\nu}_F$, $\chi_{p-k\notin B^{\sigma}_F}\chi_{q+k\notin B^{\nu}_F}\neq 1$ implies $\vert k\vert\lesssim N^{\frac{1}{3}}$. Therefore for $p\in B^{\sigma}_F$ and $q\in B^{\nu}_F$ fixed, similar to (\ref{part sumW_keta_k}), we have
    \begin{equation}\label{part sumv_keta_k}
      \Big\vert\sum_{k}\hat{v}_k\eta_k(\chi_{p-k\notin B^{\sigma}_F}\chi_{q+k\notin B^{\nu}_F}-1)\Big\vert\lesssim Na^2\ell^2.
    \end{equation}
Also recall from (\ref{有用的屎}), we have $\vert\sum_{k}\hat{v}_k\eta_k\vert\lesssim
 a$, then through estimates similar to $\Xi_{11}$, we have
\begin{equation}\label{hat Xi_11 bound}
  \hat{\Xi}_{11}=\sum_{k}\hat{v}_k\eta_k\sum_{\sigma,\nu}N_\sigma N_\nu
    -2\sum_{k}\hat{v}_k\eta_k\sum_{\sigma,\nu}N_\sigma \mathcal{N}_{ex,\nu}
    +\mathcal{E}_{\hat{\Xi}_{11}},
\end{equation}
with
\begin{equation}\label{hat E_Xi_11 bound}
  \pm\mathcal{E}_{\hat{\Xi}_{11}}\lesssim N^{-1}\mathcal{N}_{ex}^2+N^{\frac{1}{3}-2\alpha}.
\end{equation}
\par For $\hat{\Xi}_{12}$, we write it by $\hat{\Xi}_{12}=\sum_{j=1}^{3}\hat{\Xi}_{12j}$ similar to (\ref{Xi_12j}) with
    \begin{equation}\label{hat Xi_12j}
      \begin{aligned}
      &\hat{\Xi}_{121}=-\sum_{k,p,q,\sigma}\hat{v}_k\eta_{k+q-p}\chi_{p-k,q+k\notin B^{\sigma}_F}\chi_{p,q\in B^{\sigma}_F}\chi_{p\neq q}\\
      &\hat{\Xi}_{122}=\sum_{k,p,q,\sigma}2\hat{v}_k\eta_{k+q-p}a_{p,\sigma}a_{p,\sigma}^*
      \chi_{p-k,q+k\notin B^{\sigma}_F}\chi_{p,q\in B^{\sigma}_F}\chi_{p\neq q}\\
      &\hat{\Xi}_{123}=\sum_{k,p,q,\sigma,\nu}\sum_{l,r,s}\frac{1}{2}\hat{v}_k\eta_l
      a_{r,\sigma}a_{s,\nu}a^*_{q,\nu}a^*_{p,\sigma}
      \delta_{p-k,r-l}\delta_{q+k,s+l}\chi_{k\neq l}\\
    &\quad\quad\quad\quad\times
    \chi_{p-k\notin B_F^\sigma}\chi_{p,r\in B_F^\sigma}\chi_{q+k\notin B_F^\nu}
    \chi_{q,s\in B_F^\nu}+h.c.
      \end{aligned}
    \end{equation}
    Notice that similar to (\ref{v_k-v_0}), we can bound for $p,q\in B^{\sigma}_F$
    \begin{equation}\label{sum_k v_keta_k+q-p-eta_k}
      \begin{aligned}
      \Big\vert\sum_{k}\hat{v}_k(\eta_{k+q-p}-\eta_k)\Big\vert
      &=\Big\vert\int_{\Lambda}v_a(x)\eta(x)\big(e^{-i(q-p)x}-1\big)dx\Big\vert\\
      &\lesssim \vert q-p\vert^2\int_{B_{a}}\vert v_a(x)\eta(x)\vert\vert x\vert^2dx\\
      &\lesssim N^{\frac{2}{3}}a^3,
      \end{aligned}
    \end{equation}
    and similar to (\ref{part sumv_keta_k}), we have
     \begin{equation}\label{part sumv_keta_k 2}
      \Big\vert\sum_{k}\hat{v}_k\eta_{k+q-p}(\chi_{p-k,q+k\notin B^{\sigma}_F}-1)\Big\vert\lesssim Na^2\ell^2.
    \end{equation}
    By (\ref{sum_k v_keta_k+q-p-eta_k}) and (\ref{part sumv_keta_k 2}), we have
    \begin{equation}\label{Xi_12 part 1}
      \hat{\Xi}_{121}+\hat{\Xi}_{122}=-\sum_{k}\hat{v}_k\eta_k\sum_{\sigma}N_\sigma^2
    +2\sum_{k}\hat{v}_k\eta_k\sum_{\sigma}N_\sigma \mathcal{N}_{ex,\sigma}
    +\mathcal{E}_{\hat{\Xi}_{12}},
    \end{equation}
    with
    \begin{equation}\label{hat E_Xi_12 bound}
  \pm\mathcal{E}_{\hat{\Xi}_{12}}\lesssim N^{\frac{1}{3}-2\alpha}.
\end{equation}
 As for $\hat{\Xi}_{123}$, we decompose it by $\hat{\Xi}_{123}=\sum_{j=1}^{3}\hat{\Xi}_{123j}$ similar to (\ref{Xi_123j}) with
    \begin{equation}\label{hat Xi_123j}
      \begin{aligned}
      &\hat{\Xi}_{1231}=-\sum_{k,p,q,\sigma,\nu}\hat{v}_k\eta_ka_{p,\sigma}a_{q,\nu}
    a^*_{q,\nu}a^*_{p,\sigma}
    \chi_{p-k\notin B^{\sigma}_F}\chi_{q+k\notin B^{\nu}_F}\chi_{p\in B^{\sigma}_F}
    \chi_{q\in B^{\nu}_F}\\
    &\hat{\Xi}_{1232}=\sum_{k,p,q,\sigma,\nu}\sum_{l,r,s}\frac{1}{2}\hat{v}_k\eta_l
      a_{r,\sigma}a_{s,\nu}a^*_{q,\nu}a^*_{p,\sigma}
      \delta_{p-k,r-l}\delta_{q+k,s+l}
   \chi_{p,r\in B_F^\sigma}\chi_{q,s\in B_F^\nu}+h.c.\\
   &\hat{\Xi}_{1233}=\sum_{k,p,q,\sigma,\nu}\sum_{l,r,s}\frac{1}{2}\hat{v}_k\eta_l
      a_{r,\sigma}a_{s,\nu}a^*_{q,\nu}a^*_{p,\sigma}
      \delta_{p-k,r-l}\delta_{q+k,s+l}\\
    &\quad\quad\quad\quad\times
    (\chi_{p-k\notin B_F^\sigma}\chi_{q+k\notin B_F^\nu}-1)\chi_{p,r\in B_F^\sigma}
    \chi_{q,s\in B_F^\nu}+h.c.
      \end{aligned}
    \end{equation}
    By $\vert\sum_{k}\hat{v}_k\eta_k\vert\lesssim a$ and (\ref{part sumv_keta_k}), we bound
    \begin{equation}\label{hat Xi_1231}
      \pm\hat{\Xi}_{1231}\lesssim a\mathcal{N}_{ex}^2.
    \end{equation}
    Similar to (\ref{Xi_1233}), we have
    \begin{equation}\label{hat Xi_1233}
      \pm\hat{\Xi}_{1233}\lesssim N^{-\frac{2}{3}-2\alpha}(\mathcal{N}_{ex}+1)^2.
    \end{equation}
    Switching to the position space, since $\Vert\eta\Vert_{\infty}\leq 1$, $\hat{\Xi}_{1232}$ can be bounded similar to (\ref{est R_0 true})
    \begin{equation}\label{hat Xi_1232}
      \begin{aligned}
      \pm\Xi_{1232}&=\pm\sum_{\sigma,\nu}\int_{\Lambda^2}\eta(x-y)v_a(x-y)
      a(g_{x,\sigma})a(g_{y,\nu})a^*(g_{y,\nu})a^*(g_{x,\sigma})dxdy\\
      &
      \lesssim N^{-\frac{1}{6}}\mathcal{N}_{ex}+N^{\frac{1}{6}}\mathcal{N}_i[0].
      \end{aligned}
    \end{equation}
    Collecting (\ref{hat Xi_11 bound}), (\ref{hat E_Xi_11 bound}), (\ref{Xi_12 part 1}), (\ref{hat E_Xi_12 bound}) and (\ref{hat Xi_1231}-\ref{hat Xi_1233}), with the help of Lemmas \ref{lemma ineqn K_s} and \ref{control eB on N_ex lemma}, we conclude that
    \begin{equation}\label{hat Xi_1}
      \begin{aligned}
    \int_{0}^{1}\int_{t}^{1}e^{-sB}\hat{\Xi}_1e^{sB}dsdt=&
    \frac{1}{2}\sum_{k}\hat{v}_k\eta_k\sum_{\sigma\neq\nu}N_\sigma N_\nu\\
    &-\sum_{k}\hat{v}_k\eta_k\sum_{\sigma\neq\nu}N_\sigma \mathcal{N}_{ex,\nu}
    +\mathcal{E}_{\hat{\Xi}_1},
    \end{aligned}
    \end{equation}
    with
    \begin{equation}\label{hat E_Xi_1}
      \pm\mathcal{E}_{\hat{\Xi}_1}\lesssim N^{-\frac{1}{6}}\mathcal{K}_s+N^{-\frac{2}{3}
      -2\alpha}(\mathcal{N}_{ex}+1)^2+N^{\frac{1}{3}-2\alpha}.
    \end{equation}

   \begin{flushleft}
    \textbf{Analysis of $\hat{\Xi}_2$:}
  \end{flushleft}
  The estimate to $\hat{\Xi}_2$ is mostly the same as (\ref{Xi_2}), we only need to replace the estimate $\Vert W\Vert_1\lesssim a$ by $\Vert v_a\Vert\lesssim a$. Thus we have
  \begin{equation}\label{hat Xi_2}
    \pm\int_{0}^{1}\int_{t}^{1}e^{-sB}\hat{\Xi}_2e^{sB}dsdt\lesssim
    N^{-\frac{1}{6}-2\alpha}(\mathcal{N}_{ex}+1).
  \end{equation}

  \begin{flushleft}
    \textbf{Analysis of $\hat{\Xi}_3$:}
  \end{flushleft}
    Using the fact that
  \begin{equation*}
    \begin{aligned}
    e^{-sB}a_{q+k,\nu}a_{p-k,\sigma}e^{sB}=
    a_{q+k,\nu}a_{p-k,\sigma}+\int_{0}^{s}e^{-mB}
    [a_{q+k,\nu}a_{p-k,\sigma},B]e^{mB}dm
    \end{aligned}
  \end{equation*}
  We write $ e^{-sB}\hat{\Xi}_3e^{sB}=\hat{M}_1+\hat{M}_2$ with
  \begin{equation}\label{hat M split}
    \begin{aligned}
    &\hat{M}_1=\sum_{\substack{\sigma,\nu\\k,p,q}}\sum_{\varpi,l,r,s}\frac{1}{2}
    \hat{v}_k\eta_l
    \delta_{p,r}e^{-sB}a^*_{r-l,\sigma}a^*_{s+l,\varpi}\\
    &\quad\quad\quad\quad\times
    (a^*_{q,\nu}a_{s,\varpi}-a_{s,\varpi}a^*_{q,\nu})e^{sB}
    a_{q+k,\nu}a_{p-k,\sigma}\\
    &\quad\quad\quad\quad\times
    \chi_{p-k,r-l\notin B_F^\sigma}\chi_{p\in B_F^\sigma}\chi_{q+k\notin B_F^\nu}
    \chi_{q\in B_F^\nu}\chi_{s+l\notin B^\varpi_F}\chi_{s\in B^\varpi_F}+h.c.\\
     &\hat{M}_2=\sum_{\substack{\sigma,\nu\\k,p,q}}\sum_{\varpi,l,r,s}\frac{1}{2}
    \hat{v}_k\eta_l
    \delta_{p,r}\int_{0}^{s}e^{-sB}a^*_{r-l,\sigma}a^*_{s+l,\varpi}\\
    &\quad\quad\quad\quad\times
    (a^*_{q,\nu}a_{s,\varpi}-a_{s,\varpi}a^*_{q,\nu})e^{(s-m)B}
    [a_{q+k,\nu}a_{p-k,\sigma},B]e^{mB}dm\\
    &\quad\quad\quad\quad\times
    \chi_{p-k,r-l\notin B_F^\sigma}\chi_{p\in B_F^\sigma}\chi_{q+k\notin B_F^\nu}
    \chi_{q\in B_F^\nu}\chi_{s+l\notin B^\varpi_F}\chi_{s\in B^\varpi_F}+h.c.
    \end{aligned}
  \end{equation}
  The estimate to $\hat{M}_2$ is almost the same as (\ref{M_2 bound}), we merely need to replace the estimate $\Vert W\Vert_1\lesssim a$ by $\Vert v_a\Vert\lesssim a$. Thus we have
  \begin{equation}\label{hat M_2 bound}
    \begin{aligned}
    \pm \hat{M}_2&\lesssim N^{-\alpha-\gamma}\mathcal{K}_s+N^{-\frac{3}{2}\alpha}\mathcal{N}_{ex}
    +N^{\frac{1}{3}-\frac{\alpha}{2}}\\
    &+N^{-\frac{5}{6}-\frac{5}{2}\alpha}(\mathcal{N}_{ex}+1)^2.
    \end{aligned}
  \end{equation}
  for some $0<\gamma<\alpha<\frac{1}{6}$.
  \par For the estimate of $\hat{M}_1$, we rewrite it in the position space by $\hat{M}_{1}=\hat{M}_{11}+\hat{M}_{12}$ with
  \begin{equation}\label{hat M_1j}
    \begin{aligned}
    &\hat{M}_{11}=\sum_{\sigma,\nu,\varpi}\frac{1}{2}\int_{\Lambda^4}v_a(x_1-x_3)\eta(x_2-x_4)
    \Big(\sum_{q_1\in B_F^\sigma}e^{-iq(x_2-x_3)}\Big)\\
    &\quad\quad\times e^{-sB}a^*(h_{x_2,\sigma})a^*(H_{x_4,\varpi}[\delta])
    (a^*(g_{x_1,\nu})a(g_{x_4,\varpi})-a(g_{x_4,\varpi})a^*(g_{x_1,\nu}))\\
    &\quad\quad\times e^{sB}a(h_{x_1,\nu})a(h_{x_3,\sigma})dx_1dx_2dx_3dx_4+h.c.\\
    &\hat{M}_{12}=\sum_{\sigma,\nu,\varpi}\frac{1}{2}\int_{\Lambda^4}v_a(x_1-x_3)\eta(x_2-x_4)
    \Big(\sum_{q_1\in B_F^\sigma}e^{-iq(x_2-x_3)}\Big)\\
    &\quad\quad\times e^{-sB}a^*(h_{x_2,\sigma})a^*(L_{x_4,\varpi}[\delta])
    (a^*(g_{x_1,\nu})a(g_{x_4,\varpi})-a(g_{x_4,\varpi})a^*(g_{x_1,\nu}))\\
    &\quad\quad\times e^{sB}a(h_{x_1,\nu})a(h_{x_3,\sigma})dx_1dx_2dx_3dx_4+h.c.
    \end{aligned}
  \end{equation}
  where we choose $\delta=\frac{1}{3}+\gamma$. Using (\ref{ineqn N_h delta>1/3}), (\ref{jjbound}) and (\ref{V_4position}), we bound
  \begin{equation}\label{hat M_1j bound}
    \begin{aligned}
    \pm\hat{M}_{11}&\lesssim N^{\frac{3}{2}}\Vert v_a\Vert_1^{\frac{1}{2}}\Vert\eta\Vert_2
    \big(\theta_1\mathcal{V}_4+\theta_1^{-1}\mathcal{N}_{h}[\delta_2](\mathcal{N}_{ex}+1)
    +\theta_1^{-1}\ell(\mathcal{N}_{ex}+1)^2\big)\\
    &\lesssim N^{-\gamma}\mathcal{V}_4+N^{-\alpha-\gamma}\mathcal{K}_s
    +N^{-\frac{2}{3}-2\alpha+\gamma}(\mathcal{N}_{ex}+1)^2\\
    \pm\hat{M}_{12}&\lesssim N^{2+\frac{3}{2}\gamma}\Vert v_a\Vert_1^{\frac{1}{2}}\Vert\eta\Vert_1
    \big(\theta_2\mathcal{V}_4+\theta_2^{-1}(\mathcal{N}_{ex}+1)\big)\\
    &\lesssim N^{-\gamma}\mathcal{V}_4+N^{-\frac{1}{3}+4\gamma-4\alpha}(\mathcal{N}_{ex}+1)
    \end{aligned}
  \end{equation}
  where we choose $\theta_1=N^{\frac{1}{6}+\frac{1}{2}\alpha-\gamma}$ and $\theta_2=N^{\frac{1}{6}-\frac{5}{2}\gamma+2\alpha}$. We collect (\ref{hat M_1j bound}) and (\ref{hat M_2 bound}) to reach
  \begin{equation}\label{hat Xi_3}
  \begin{aligned}
    \pm\int_{0}^{1}\int_{t}^{1}e^{-sB}\hat{\Xi}_3e^{sB}dsdt&\lesssim
    N^{-\gamma}\mathcal{V}_4+N^{-\alpha-\gamma}\mathcal{K}_s
    +N^{-\frac{3}{2}\alpha}(\mathcal{N}_{ex}+1)\\
    &
    +N^{-\frac{2}{3}-2\alpha+\gamma}(\mathcal{N}_{ex}+1)^2+N^{\frac{1}{3}-\frac{\alpha}{2}}
  \end{aligned}
  \end{equation}
  for some $0<\gamma<\alpha<\frac{1}{6}$.
  \begin{flushleft}
    \textbf{Conclusion of Lemma \ref{cal com V_21 lemma}:}
  \end{flushleft}
  \par With (\ref{hat Xi_1}), (\ref{hat E_Xi_1}), (\ref{hat Xi_2}) and (\ref{hat Xi_3}), we reach (\ref{cal com V_21}) and $\mathcal{E}_{\mathcal{V}_{21}}$ satisfies the bound (\ref{est E_V_21}).
  \par For the estimate of (\ref{control Gamma}), we collect (\ref{eqn of eta_p rewrt}), (\ref{define V_21' and Omega}), (\ref{define Gamma}), (\ref{[K,A]}) and (\ref{[V_4,B]}) to reach
  \begin{equation}\label{Gamma}
    \Gamma=\Theta_d+\Theta_r
  \end{equation}
  where $\Theta_d$ and $\Theta_r$ are defined in (\ref{Theta}). From (\ref{Theta d bound true}) and (\ref{Theta_r bound}), we have
  \begin{equation}\label{Gamma bound}
    \pm\Gamma\lesssim N^{-\frac{2}{3}-\frac{\alpha}{2}}\big(\mathcal{V}_4+
    (\mathcal{N}_{ex}+1)^2\big)+N^{\frac{1}{3}-2\alpha}.
  \end{equation}
  By Lemmas \ref{control eB on N_ex lemma} and \ref{corollary control eB on V_4}, we have
  \begin{equation}\label{Gamma bound true}
    \pm\int_{0}^{1}e^{-tB}\Gamma e^{tB}dt\lesssim N^{-\frac{2}{3}-\frac{\alpha}{2}}\big(\mathcal{V}_4+
    (\mathcal{N}_{ex}+1)^2\big)+N^{\frac{1}{3}-\frac{\alpha}{2}}.
  \end{equation}
  Therefore we conclude (\ref{est E_V_21}).
\end{proof}

\begin{proof}[Proof of Proposition \ref{p1}]
  \par Proposition \ref{p1} is deduced by combining Lemmas \ref{cal eB on V_0,1,22,23 lemma}-\ref{cal com V_21 lemma} and equation (\ref{define G_N}).
\end{proof}

\section{Cubic Renormalization}\label{cub}
\par In this section, we analyse the excitation Hamiltonian $\mathcal{J}_N$ defined in (\ref{define J_N 0}) and prove Proposition \ref{p2}. Recall (\ref{define B' 0}-\ref{A^nu_F,kappa 0}),
\begin{equation}\label{define B'}
  B^\prime=A^\prime-{A^\prime}^*
\end{equation}
with
\begin{equation}\label{define A'}
  A^\prime=\sum_{k,p,q,\sigma,\nu}
  \eta_ka^*_{p-k,\sigma}a^*_{q+k,\nu}a_{q,\nu}a_{p,\sigma}\chi_{p-k\notin B^{\sigma}_F}\chi_{q+k\notin B^{\nu}_F}\chi_{q\in A^\nu_{F,\kappa}}
  \chi_{p\in B^{\sigma}_F},
\end{equation}
and
\begin{equation}\label{A^nu_F,kappa}
  A_{F,\kappa}^\nu=\{k\in2\pi\mathbb{Z}^3,\,k_F^\nu<\vert k\vert\leq k_F^\nu+N^\kappa\}
\end{equation}
where $-\frac{11}{48}<\kappa<0$ is a cut-off parameter to be determined. Here we introduce $\kappa$ due to technical reasons, it helps us separate the energy exchanges between different momentum levels. Using (\ref{define G_N p1}) and Newton-Leibniz formula,  we rewrite $\mathcal{J}_N$ by
\begin{equation}\label{define J_N}
  \begin{aligned}
  \mathcal{J}_N\coloneqq e^{-B^\prime}\mathcal{G}_Ne^{B^\prime}=&C_{\mathcal{G}_N}+\mathcal{K}
  +\mathcal{V}_4+e^{-B^\prime}\big(Q_{\mathcal{G}_N}+\mathcal{V}_{21}^\prime
  +\Omega+\mathcal{E}_{\mathcal{G}_N}\big)e^{B^\prime}\\
  &+\int_{0}^{1}e^{-tB\prime}\Gamma^\prime e^{tB^\prime}dt
  +\int_{0}^{1}\int_{t}^{1}e^{-sB^\prime}[\mathcal{V}_3,B^\prime]e^{sB^\prime}dsdt.
  \end{aligned}
\end{equation}
with
\begin{equation}\label{define Gamma'}
  \Gamma^\prime=[\mathcal{K}+\mathcal{V}_4,B^\prime]+\mathcal{V}_3.
\end{equation}
We prove Proposition \ref{p2} by analysing each term on the right hand side of (\ref{define J_N}) up to some error terms that may be discarded in the limit of large $N$. We collect corresponding results and detailed proofs in Lemmas \ref{cal eB' on Q_G_N and E_G_N}-\ref{cal [V_3,B'] lemma}.
\par To control some of these error terms that may not contribute to the ground state energy up to the second order, we first prove Lemma \ref{control eB' on N_ex lemma} and Corollary \ref{control eB' on K,V_4} which control the action of $e^{B^\prime}$ on $\mathcal{N}_{ex}$, $\mathcal{K}_s$ and $\mathcal{V}_4$, and are fairly useful in the up-coming discussion. The next lemma, Lemma \ref{control eB' on N_ex lemma}, collects estimates that control the action of $e^{B^\prime}$ on the particle number operator defined in Section \ref{number of p ops}.
\begin{lemma}\label{control eB' on N_ex lemma}
   Let $\ell$ be chosen as in (\ref{choose l}). Then for $n\in\frac{1}{2}\mathbb{Z}$, $\vert t\vert\leq1$, we have
  \begin{align}
    e^{-tB^\prime}(\mathcal{N}_{ex}+1)^ne^{tB^\prime} & \lesssim (\mathcal{N}_{ex}+1)^n \label{control eB' on N_ex}\\
    \pm\big(e^{-tB^\prime}(\mathcal{N}_{ex}+1)^ne^{tB^\prime}
    -(\mathcal{N}_{ex}+1)^n\big) & \lesssim
    N^{-\frac{1}{3}+\frac{\kappa}{2}-\frac{\alpha}{2}}
    (\mathcal{N}_{ex}+1)^n\label{control diff eB' on N_ex}
  \end{align}
  (\ref{control eB' on     N_ex}) and (\ref{control diff eB' on N_ex}) still hold if we replace the $\mathcal{N}_{ex}$ on the left hand side by $\mathcal{N}_{ex,\sigma}$.
\end{lemma}
\begin{proof}
  \par We prove the lemma for $n=1$. For the rest of the proof, one can see \cite[Lemma 8.1]{me}. By Lemma \ref{lem commutator}, we have $[\mathcal{N}_{ex},A^\prime]=A^\prime$. Switching to the position space, we can bound
  \begin{equation}\label{A' bound}
    \begin{aligned}
    \pm A^\prime&=\pm\sum_{\sigma,\nu}\int_{\Lambda^2}
    \eta(x-y)a^*(h_{x,\sigma})a^*(h_{y,\nu})a(L_{y,\nu}[\kappa])a(g_{x,\sigma})dxdy\\
    &\lesssim \Vert\eta\Vert_2 N^{\frac{5}{6}+\frac{\kappa}{2}}(\mathcal{N}_{ex}+1)
    \lesssim N^{-\frac{1}{3}+\frac{\kappa}{2}-\frac{\alpha}{2}}(\mathcal{N}_{ex}+1),
    \end{aligned}
  \end{equation}
  where we have used the fact $\Vert a(L_{y,\nu}[\kappa])\Vert^2\lesssim N^{\frac{2}{3}+\kappa}$ from (\ref{ineqn N_l and N_s}). Thence (\ref{control eB' on N_ex}) and (\ref{control diff eB' on N_ex}) follows by Newton-Leibniz fomula and Gronwall's inequality.
\end{proof}

\par The next lemma is an important preliminary that allows us to bound the action of $e^{B^\prime}$ on the excited kinetic energy operator $\mathcal{K}_s$ defined in (\ref{K_s}) and the excited potential energy operator $\mathcal{V}_4$ defined in (\ref{split V detailed}). 
\begin{lemma}\label{cal [K,V_4,B'] lemma}
We have
  \begin{equation}\label{cal [K,B']}
    \begin{aligned}
    [\mathcal{K},B^\prime]=&\sum_{k,p,q,\sigma,\nu}
  2\vert k\vert^2\eta_k(a^*_{p-k,\sigma}a^*_{q+k,\nu}a_{q,\nu}a_{p,\sigma}+h.c.)\\
  &\quad\quad\quad\quad
  \times\chi_{p-k\notin B^{\sigma}_F}\chi_{q+k\notin B^{\nu}_F}\chi_{q\in A^\nu_{F,\kappa}}
  \chi_{p\in B^{\sigma}_F}+\mathcal{E}_{[\mathcal{K},B^\prime]}
    \end{aligned}
  \end{equation}
  where
  \begin{equation}\label{E_[K,B']}
    \pm\mathcal{E}_{[\mathcal{K},B^\prime]}\lesssim
    N^{-\frac{1}{6}-\frac{\alpha}{2}+\frac{\kappa}{2}}\mathcal{K}_s
    +N^{-\frac{1}{6}-\alpha}(\mathcal{N}_{ex}+1).
  \end{equation}
  Moreover,
  \begin{equation}\label{control [K,B']}
    \pm [\mathcal{K},B^\prime]\lesssim \mathcal{K}_s+N^{\frac{\kappa}{2}}\ln N
    (\mathcal{N}_{ex}+1).
  \end{equation}
  On the other hand,
  \begin{equation}\label{cal [V_4,B']}
  \begin{aligned}
    [\mathcal{V}_4,B^\prime]=&\sum_{k,p,q,\sigma,\nu}
 \Big(\sum_{l}\hat{v}_{k-l}\eta_l\Big)
 (a^*_{p-k,\sigma}a^*_{q+k,\nu}a_{q,\nu}a_{p,\sigma}+h.c.)\\
  &\quad\quad\quad\quad
  \times\chi_{p-k\notin B^{\sigma}_F}\chi_{q+k\notin B^{\nu}_F}\chi_{q\in A^\nu_{F,\kappa}}
  \chi_{p\in B^{\sigma}_F}+\mathcal{E}_{[\mathcal{V}_4,B^\prime]}
  \end{aligned}
  \end{equation}
  where
  \begin{equation}\label{E_[V_4,B']}
    \pm\mathcal{E}_{[\mathcal{V}_4,B^\prime]}\lesssim N^{-\gamma}\mathcal{V}_4
    +N^{-\frac{2}{3}-2{\alpha}}(\mathcal{N}_{ex}+1)
    +N^{-\frac{5}{3}+\kappa-\alpha+\gamma}(\mathcal{N}_{ex}+1)^3
  \end{equation}
  for $0<\gamma<\alpha<\frac{1}{6}$. Moreover,
  \begin{equation}\label{control [V_4,B']}
    \pm[\mathcal{V}_4,B^\prime]\lesssim \mathcal{V}_4+(\mathcal{N}_{ex}+1)
    +N^{-\frac{5}{3}+\kappa-\alpha+\gamma}(\mathcal{N}_{ex}+1)^3.
  \end{equation}
\end{lemma}
\begin{proof}
\par By Lemma \ref{lem commutator}, we reach (\ref{cal [V_4,B']}) with
\begin{equation}\label{cal [V_4,B']define}
  \begin{aligned}
  [\mathcal{V}_4,B^\prime]\coloneqq\Theta_m^\prime+\Theta_d^\prime+\sum_{j=1}^{3}
  \Theta_{r,j}^\prime
  \end{aligned}
\end{equation}
where we let $\mathcal{E}_{[\mathcal{V}_4,B^\prime]}\coloneqq\Theta_d^\prime+\sum_{j=1}^{3}
  \Theta_{r,j}^\prime$, and
  \begin{equation}\label{cal [V_4,B'] detailed}
    \begin{aligned}
    \Theta_m^\prime&=\sum_{k,p,q,\sigma,\nu}
 \Big(\sum_{l}\hat{v}_{k-l}\eta_l\Big)
 (a^*_{p-k,\sigma}a^*_{q+k,\nu}a_{q,\nu}a_{p,\sigma}+h.c.)\\
  &\quad\quad\quad\quad
  \times\chi_{p-k\notin B^{\sigma}_F}\chi_{q+k\notin B^{\nu}_F}\chi_{q\in A^\nu_{F,\kappa}}
  \chi_{p\in B^{\sigma}_F}\\
   \Theta_d^\prime&=\sum_{l,k,p,q,\sigma,\nu}
 \hat{v}_{k-l}\eta_l
 (a^*_{p-k,\sigma}a^*_{q+k,\nu}a_{q,\nu}a_{p,\sigma}+h.c.)\\
  &\quad\quad\quad\quad
  \times\chi_{p-k\notin B^{\sigma}_F}\chi_{q+k\notin B^{\nu}_F}\chi_{q\in A^\nu_{F,\kappa}}
  \chi_{p\in B^{\sigma}_F}(\chi_{p-l\notin B^{\sigma}_F}\chi_{q+l\notin B^{\nu}_F}-1)\\
  \Theta_{r,1}^\prime&=\sum_{\substack{k,p,q\\\sigma,\nu}}\sum_{\varpi,l,r,s}
  -\hat{v}_k\eta_l\delta_{p,r-l}(a^*_{p-k,\sigma}a^*_{q+k,\nu}a^*_{s+l,\varpi}
  a_{q,\nu}a_{s,\varpi}a_{r,\sigma}+h.c.)\\
  &\quad\quad\quad\quad\times
  \chi_{p-k,p\notin B^{\sigma}_F}\chi_{q+k,q\notin B^{\nu}_F}\chi_{r\in B_F^\sigma}
  \chi_{s+l\notin B_F^\varpi}\chi_{s\in A^\varpi_{F,\kappa}}\\
  \Theta_{r,2}^\prime&=\sum_{\substack{k,p,q\\\sigma,\nu}}\sum_{\tau,l,r,s}
  \hat{v}_k\eta_l\delta_{p,s+l}(a^*_{p-k,\sigma}a^*_{q+k,\nu}a^*_{r-l,\tau}
  a_{q,\nu}a_{s,\sigma}a_{r,\tau}+h.c.)\\
  &\quad\quad\quad\quad\times
  \chi_{p-k,p\notin B^{\sigma}_F}\chi_{q+k,q\notin B^{\nu}_F}\chi_{r-l\notin B_F^\tau}
  \chi_{r\in B_F^\tau}
  \chi_{s\in A^\sigma_{F,\kappa}}\\
  \Theta_{r,3}^\prime&=\sum_{\substack{k,p,q\\\sigma,\nu}}\sum_{\tau,l,r,s}
  -\hat{v}_k\eta_l\delta_{p-k,s}(a^*_{r-l,\tau}a^*_{s+l,\sigma}a^*_{q+k,\nu}
  a_{q,\nu}a_{p,\sigma}a_{r,\tau}+h.c.)\\
  &\quad\quad\quad\quad\times \chi_{p-k\in A^\sigma_{F,\kappa}}
  \chi_{p\notin B^{\sigma}_F}\chi_{q+k,q\notin B^{\nu}_F}\chi_{r-l\notin B_F^\tau}
  \chi_{r\in B_F^\tau}\chi_{s+l\notin B_F^\sigma}
    \end{aligned}
  \end{equation}
  Switching to the position space, we bound, from (\ref{V_4position}), that
  \begin{equation}\label{Theta'_m bound}
    \begin{aligned}
    \pm\Theta_m^\prime&=\pm\sum_{\sigma,\nu}\int_{\Lambda^2}
    v_a(x-y)\eta(x-y)a^*(h_{x,\sigma})a^*(h_{y,\nu})\\
    &\quad\quad\quad\quad\times a(L_{y,\nu}[\kappa])a(g_{x,\sigma})dxdy
    +h.c.\\
    &\lesssim N^{\frac{1}{2}}\Vert v_a\Vert_1^{\frac{1}{2}}\Vert\eta\Vert_\infty
    (\mathcal{V}_4+\mathcal{N}_{l}[\kappa])\lesssim\mathcal{V}_4+\mathcal{N}_{ex}.
    \end{aligned}
  \end{equation}
  For the second term (the difference part), notice that for any $p\in B_F^\sigma$ and $q\in A_{F,\kappa}^\nu$, we have $\vert p\vert, \vert q\vert\lesssim N^{\frac{1}{3}}$. Therefore $(\chi_{p-l\notin B^{\sigma}_F}\chi_{q+l\notin B^{\nu}_F}-1)\neq0$ implies $\vert l\vert\lesssim N^{\frac{1}{3}}$. We bound it, as in (\ref{Theta d bound true}), that
  \begin{equation}\label{Theta' d bound}
    \pm\Theta_{d}^\prime
    \lesssim\sum_{\sigma,\nu}\sum_{\vert l\vert\lesssim N^{\frac{1}{3}}}
    \vert\eta_l\vert(\mathcal{V}_4+\mathcal{N}_{ex})
    \lesssim N^{-\frac{2}{3}-2\alpha}(\mathcal{V}_4+\mathcal{N}_{ex}).
  \end{equation}
  To bound $\Theta_{r,1}^\prime$ and $\Theta_{r,2}^\prime$, we adopt the notation
 \begin{equation*}
 \begin{aligned}
    k_{y_1,y_2,\sigma}(z)&\coloneqq\sum_{q\in B^{\sigma}_F}\Big(
    \sum_{p\notin B_F^\sigma}\int_{\Lambda}\eta(x-y_1)
    e^{iqx}e^{ip(y_2-x)}dx\Big)f_{q,\sigma}(z)\\
     k_{y_1,y_2,\sigma}^{(\kappa)}(z)&\coloneqq\sum_{q\in A^{\sigma}_{F,\kappa}}\Big(
    \sum_{p\notin B_F^\sigma}\int_{\Lambda}\eta(x-y_1)
    e^{iqx}e^{ip(y_2-x)}dx\Big)f_{q,\sigma}(z).
 \end{aligned}
  \end{equation*}
  One can check that for $j=1,2$
  \begin{equation*}
  \begin{aligned}
    \int_{\Lambda}\Vert k_{y_1,y_2,\sigma}\Vert^2_2dy_j\leq\sum_{q\in B_F^\sigma}
    \Vert\eta\Vert_2^2,\quad
    \int_{\Lambda}\Vert k_{y_1,y_2,\sigma}^{(\kappa)}\Vert^2_2dy_j\leq
    \sum_{q\in A^{\sigma}_{F,\kappa}}
    \Vert\eta\Vert_2^2
  \end{aligned}
  \end{equation*}
  Thus we have
  \begin{equation}\label{Theta'_r,1 and 2}
    \begin{aligned}
    &\Theta_{r,1}^\prime=-\sum_{\sigma,\nu,\varpi}\int_{\Lambda^3}v_{a}(x_1-x_3)
    a^*(h_{x_1,\sigma})a^*(h_{x_3,\nu})a^*(h_{x_4,\varpi})\\
    &\quad\quad\quad\quad
    \times a(h_{x_3,\nu})a(L_{x_4,\varpi}[\kappa])a(k_{x_4,x_1,\sigma})
    dx_1dx_3dx_4+h.c.\\
    &\Theta_{r,2}^\prime=\sum_{\sigma,\nu,\tau}\int_{\Lambda^3}v_{a}(x_1-x_3)
    a^*(h_{x_1,\sigma})a^*(h_{x_3,\nu})a^*(h_{x_2,\tau})\\
    &\quad\quad\quad\quad
    \times a(h_{x_3,\nu})a(k^{(\kappa)}_{x_2,x_1,\sigma})a(g_{x_2,\tau})
    dx_1dx_2dx_3+h.c.
    \end{aligned}
  \end{equation}
  Hence we bound them by
  \begin{equation}\label{Theta'_r,1 and 2 bound}
  \begin{aligned}
    \pm\Theta_{r,1}^\prime,\pm\Theta_{r,2}^\prime
    &\lesssim \theta\mathcal{V}_4+\theta^{-1}N^{\frac{5}{3}+\kappa}\Vert v_a\Vert_1
    \Vert\eta\Vert_2^2(\mathcal{N}_{ex}+1)^2\\
    &\lesssim
    N^{-\gamma}\mathcal{V}_4+N^{-\frac{5}{3}+\kappa-\alpha+\gamma}(\mathcal{N}_{ex}+1)^2,
  \end{aligned}
  \end{equation}
  where we choose $\theta=N^{-\gamma}$ for $0<\gamma<\alpha<\frac{1}{6}$. For $\Theta_{r,3}^\prime$, we have
  \begin{equation}\label{Theta'_r 3}
    \begin{aligned}
    \Theta_{r,3}^\prime&=-\sum_{\sigma,\nu,\tau}\int_{\Lambda^4}v_a(x_1-x_3)\eta(x_2-x_4)
    \Big(\sum_{p_1\in A^\sigma_{F,\kappa}}e^{-ip_1(x_4-x_1)}\Big)
    a^*(h_{x_2,\tau})a^*(h_{x_4,\sigma})\\
    &\quad\quad\quad\quad\times
    a^*(h_{x_3,\nu})a(h_{x_3,\nu})a(h_{x_1,\sigma})a(g_{x_2,\tau})dx_1dx_2dx_3dx_4+h.c.
    \end{aligned}
  \end{equation}
  and we bound it by
   \begin{equation}\label{Theta'_r,3 bound}
  \begin{aligned}
    \pm\Theta_{r,3}^\prime
    &\lesssim \theta\mathcal{V}_4+\theta^{-1}N^{\frac{5}{3}+\kappa}\Vert v_a\Vert_1
    \Vert\eta\Vert_2^2(\mathcal{N}_{ex}+1)^3\\
    &\lesssim
    N^{-\gamma}\mathcal{V}_4+N^{-\frac{5}{3}+\kappa-\alpha+\gamma}(\mathcal{N}_{ex}+1)^3,
  \end{aligned}
  \end{equation}
  where we choose $\theta=N^{-\gamma}$ for $0<\gamma<\alpha<\frac{1}{6}$. Collecting (\ref{cal [V_4,B']define}), (\ref{Theta'_m bound}), (\ref{Theta' d bound}), (\ref{Theta'_r,1 and 2 bound}) and (\ref{Theta'_r,3 bound}), we reach (\ref{E_[V_4,B']}) and (\ref{control [V_4,B']}).

  \par By Lemma \ref{lem commutator}, we reach (\ref{cal [K,B']}) with
  \begin{equation}\label{cal [K,B'] define}
    [\mathcal{K},B^\prime]\coloneqq \Omega_m^\prime+\Omega^\prime
  \end{equation}
  where we let $\mathcal{E}_{[\mathcal{K},B^\prime]}\coloneqq \Omega^\prime$, and
  \begin{equation}\label{E_[K,B'] define}
    \begin{aligned}
    \Omega_m^\prime=&\sum_{k,p,q,\sigma,\nu}
  2\vert k\vert^2\eta_k(a^*_{p-k,\sigma}a^*_{q+k,\nu}a_{q,\nu}a_{p,\sigma}+h.c.)\\
  &\quad\quad\quad\quad
  \times\chi_{p-k\notin B^{\sigma}_F}\chi_{q+k\notin B^{\nu}_F}\chi_{q\in A^\nu_{F,\kappa}}
  \chi_{p\in B^{\sigma}_F}\\
  \Omega^\prime=&\sum_{k,p,q,\sigma,\nu}
  2 k(q-p)\eta_k(a^*_{p-k,\sigma}a^*_{q+k,\nu}a_{q,\nu}a_{p,\sigma}+h.c.)\\
  &\quad\quad\quad\quad
  \times\chi_{p-k\notin B^{\sigma}_F}\chi_{q+k\notin B^{\nu}_F}\chi_{q\in A^\nu_{F,\kappa}}
  \chi_{p\in B^{\sigma}_F}
    \end{aligned}
  \end{equation}
  The bounds on $\Omega_m^\prime$ and $\Omega^\prime$ are more subtle. In fact, for the bound of $\Omega^\prime$, we first use the relation $\chi_{q+k\notin B^{\nu}_F}=\chi_{q+k\in P^{\nu}_{F,0}}+\chi_{q+k\in A^{\nu}_{F,0}}$, to rewrite it by $\Omega^\prime=\Omega^\prime_1+\Omega^\prime_2$. We adopt the notation for $p\notin B_F^\sigma$:
  \begin{equation*}
    \begin{aligned}
    k^{(1)}_{p,x,\sigma}(z)&\coloneqq\sum_{q\in B_F^\sigma}(q-p)\eta_{q-p}\nabla_x
    e^{i(q-p)x}f_{q,\sigma}(z)\\
    k^{(2)}_{p,x,\sigma}(z)&\coloneqq\sum_{q\in B_F^\sigma}(q-p)\eta_{q-p}
    e^{i(q-p)x}f_{q,\sigma}(z)
    \end{aligned}
  \end{equation*}
  and denote any universal constant by $C$. Switching to the position space, we have
  \begin{equation}\label{Omega'_1}
    \begin{aligned}
    \Omega^\prime_1&=C\sum_{\sigma,\nu,p_1}\chi_{p_1\notin B_F^\sigma}\int_{\Lambda}
    a^*_{p_1,\sigma}a^*(H_{x,\nu}[0])a(L_{x,\nu}[\kappa])a(k^{(1)}_{p_1,x,\sigma})dx+h.c.\\
    &+C\sum_{\sigma,\nu,p_1}\chi_{p_1\notin B_F^\sigma}\int_{\Lambda}
    a^*_{p_1,\sigma}a^*(H_{x,\nu}[0])a(\nabla_{x}L_{x,\nu}[\kappa])a(k^{(2)}_{p_1,x,\sigma})dx
    +h.c.\\
    &\eqqcolon \Omega^\prime_{11}+\Omega^\prime_{12}.
    \end{aligned}
  \end{equation}
  Notice that
  \begin{equation*}
    \begin{aligned}
    \Vert a(k^{(1)}_{p_1,x,\sigma})\Vert^2&=\sum_{q_1\in B_F^\sigma}\vert q_1-p_1\vert^2
    \eta_{q_1-p_1}^2\vert q_1\vert^2\lesssim N^{\frac{2}{3}}\sum_{q_1\in B_F^\sigma}\vert q_1-p_1\vert^2 \eta_{q_1-p_1}^2\\
     \Vert a(k^{(2)}_{p_1,x,\sigma})\Vert^2&=\sum_{q_1\in B_F^\sigma}\vert q_1-p_1\vert^2
    \eta_{q_1-p_1}^2
    \end{aligned}
  \end{equation*}
  For $\psi\in\mathcal{H}^{\wedge N}(\{N_{\varsigma_i}\})$, we can bound $\Omega_{11}^\prime$ and $\Omega_{12}^\prime$ similarly to (\ref{V_21,1' bound}). For $\Omega_{11}^\prime$,
  \begin{equation}\label{Omega_11' bound 1st}
    \begin{aligned}
   &\vert\langle\Omega_{11}^\prime\psi,\psi\rangle\vert\\
   &\lesssim \Big(\sum_{\sigma,\nu,p_1}\chi_{p_1\notin B_F^{\sigma}}
   \big(\vert p_1\vert^2-\vert k_F^\sigma\vert^2\big)
    \\
    &\quad\quad\int_{\Lambda}\langle a^*(H_{x,\nu}[0])
    (a(H_{x,\nu}[0])(\mathcal{N}_h[0]+1)^{-1}a_{p_1,\sigma}\psi,a_{p_1,\sigma}\psi\rangle
    dx\Big)^{\frac{1}{2}}\\
    &\times\Big(\sum_{\sigma,\nu,p_1}\chi_{p_1\notin B_F^{\sigma}}
   \big(\vert p_1\vert^2-\vert k_F^\sigma\vert^2\big)^{-1}\int_{\Lambda} \Vert (\mathcal{N}_h[0]+2)^{\frac{1}{2}}a(L_{x,\nu}[\kappa])a(k_{p_1,x,\sigma})
    \psi\Vert^2 dx\Big)^{\frac{1}{2}}
 \end{aligned}
  \end{equation}
  Since $\kappa<0$, we have $[\mathcal{N}_h[0],a(L_{x,\nu}[\kappa])]=0$, and hence
  \begin{equation}\label{Omega_11' bound 2nd}
    \begin{aligned}
    &\sum_{\sigma,\nu,p_1}\chi_{p_1\notin B_F^{\sigma}}
   \big(\vert p_1\vert^2-\vert k_F^\sigma\vert^2\big)^{-1}\int_{\Lambda} \Vert (\mathcal{N}_h[0]+2)^{\frac{1}{2}}a(L_{x,\nu}[\kappa])a(k_{p_1,x,\sigma})
    \psi\Vert^2 dx\\
    &\lesssim\sum_{\sigma,\nu,p_1}\chi_{p_1\notin B_F^{\sigma}}
   \big(\vert p_1\vert^2-\vert k_F^\sigma\vert^2\big)^{-1}\int_{\Lambda}
    \Vert a(k_{p_1,x,\sigma})\Vert^2
    \Vert (\mathcal{N}_h[0]+2)^{\frac{1}{2}}a(L_{x,\nu}[\kappa])
    \psi\Vert^2 dx\\
    &\lesssim\sum_{\sigma,\nu}\sum_{p_1\notin B_F^{\sigma},q_1\in B_F^\sigma}
    \frac{\vert q_1\vert^2\vert q_1-p_1\vert^2\eta_{q_1-p_1}^2}
    {\vert p_1\vert^2-\vert k_F^\sigma\vert^2}
    \int_{\Lambda}
    \Vert (\mathcal{N}_h[0]+2)^{\frac{1}{2}}a(L_{x,\nu}[\kappa])
    \psi\Vert^2 dx\\
    &\lesssim a^2N^{\frac{5}{3}}\ell\sum_{\nu}
    \int_{\Lambda}\big(\Vert a(L_{x,\nu}[\kappa])
    \mathcal{N}_{h}[0]^\frac{1}{2}\psi\Vert^2+\Vert a(L_{x,\nu}[\kappa])\psi\Vert^2\big)dx\\
    &\lesssim a^2N^{\frac{5}{3}}\ell\big\langle\big(N^{\frac{1}{3}+\kappa}\mathcal{K}_s
    +\mathcal{N}_l[\kappa]\big)\psi,\psi\rangle\lesssim N^{-\frac{1}{3}-\alpha+\kappa}
    \langle\mathcal{K}_s\psi,\psi\rangle
    \end{aligned}
  \end{equation}
   where we have used (\ref{energy est bog}) in Lemma \ref{energy est bog lem}. Combining (\ref{Omega_11' bound 1st}) and (\ref{Omega_11' bound 2nd}), we conclude that
   \begin{equation}\label{Omega_11' bound}
     \pm\Omega_{11}^\prime\lesssim N^{-\frac{1}{6}-\frac{\alpha}{2}+\frac{\kappa}{2}}
    \mathcal{K}_s.
   \end{equation}
   Notice that $q\in A^\nu_{F,\kappa}$ also implies $\vert q\vert\lesssim N^{\frac{1}{3}}$. Therefore
  \begin{equation*}
    \Vert a(\nabla_{x}L_{x,\nu}[\kappa])\Vert^2\lesssim N^{\frac{2}{3}}\#A^\nu_{F,\kappa}
    \lesssim N^{\frac{4}{3}+\kappa}.
  \end{equation*}
  Moreover,
  \begin{equation*}
    \int_{\Lambda}a^*(\nabla_{x}L_{x,\nu}[\kappa])a(\nabla_{x}L_{x,\nu}[\kappa])dx
    =\sum_{q\in A^\nu_{F,\kappa}}\vert q\vert^2 a_{q,\nu}^*a_{q,\nu}\lesssim N^{\frac{2}{3}}\mathcal{N}_{ex}.
  \end{equation*}
  Therefore, $\Omega_{12}^\prime$ also satisfies the bound (\ref{Omega_11' bound}). Thence we attain
  \begin{equation}\label{Omega_1' bound}
    \begin{aligned}
    \pm\Omega^\prime_1\lesssim N^{-\frac{1}{6}-\frac{\alpha}{2}+\frac{\kappa}{2}}
    \mathcal{K}_s.
    \end{aligned}
  \end{equation}
  For $\Omega_2^\prime$, we write it in the position space
  \begin{equation}\label{Omega_2'pos}
    \begin{aligned}
    \Omega_2^\prime&=C\sum_{\sigma,\nu}\int_{\Lambda^2}\nabla\eta(x-y)
    a^*(h_{x,\sigma})a^*(L_{y,\nu}[0])a(L_{y,\nu}[\kappa])a(\nabla_xg_{x,\sigma})dxdy\\
    &+C\sum_{\sigma,\nu}\int_{\Lambda^2}\nabla\eta(x-y)
    a^*(h_{x,\sigma})a^*(L_{y,\nu}[0])a(\nabla_yL_{y,\nu}[\kappa])a(g_{x,\sigma})dxdy
    +h.c.
    \end{aligned}
  \end{equation}
  and bound it directly
  \begin{equation}\label{Omega_2' bound}
    \pm\Omega_2^\prime\lesssim N^{\frac{1}{3}}\Vert\nabla\eta\Vert_1N^{\frac{5}{6}}
    (\mathcal{N}_{ex}+\mathcal{N}_l[\kappa])\lesssim N^{-\frac{1}{6}-\alpha}\mathcal{N}_{ex}.
  \end{equation}
  Therefore, combining (\ref{Omega_1' bound}) and (\ref{Omega_2' bound}), we conclude that
  \begin{equation}\label{Omega' bound}
    \pm\Omega^\prime\lesssim N^{-\frac{1}{6}-\frac{\alpha}{2}+\frac{\kappa}{2}}
    \mathcal{K}_s+N^{-\frac{1}{6}-\alpha}\mathcal{N}_{ex}
  \end{equation}
  which is (\ref{E_[K,B']}).
  \par For the bound of $\Omega_m^\prime$, we also use the relation $\chi_{q+k\notin B^{\nu}_F}=\chi_{q+k\in P^{\nu}_{F,\delta}}+\chi_{q+k\in A^{\nu}_{F,\delta}}$ with $\delta=\frac{1}{3}+\kappa$ (recall that $-\frac{11}{48}<\kappa<0$), to rewrite it by $\Omega_m^\prime=\Omega^\prime_{m,1}+\Omega^\prime_{m,2}$. Similar to (\ref{Omega_11' bound 1st}) and (\ref{Omega_11' bound 2nd}), we can bound $\Omega_{m,1}^\prime$, using (\ref{energy est bog laplace eta}) in Lemma \ref{energy est bog lem}, by
  \begin{equation}\label{Omega_m 1 bound}
    \pm\Omega_{m,1}^\prime\lesssim \mathcal{K}_s.
  \end{equation}
  Using (\ref{L1 grad eta and laplace eta}), we can bound $\Omega_{m,2}^\prime$ similar to (\ref{Omega_2' bound}):
  \begin{equation}\label{Omega_m 2 bound}
    \pm\Omega_{m,2}^\prime\lesssim N^{1+\frac{\kappa}{2}}\Vert\Delta\eta\Vert_1\mathcal{N}_{ex}
     \lesssim N^{\frac{\kappa}{2}}\ln N\mathcal{N}_{ex}.
  \end{equation}
Combining (\ref{cal [K,B'] define}), (\ref{Omega' bound}), (\ref{Omega_m 1 bound}) and (\ref{Omega_m 2 bound}), we reach (\ref{control [K,B']}).
\end{proof}

\par Combining Lemmas \ref{cal [K,V_4,B'] lemma} and \ref{control eB' on N_ex lemma}, we deduce the following corollary via Gronwall's inequality.
\begin{corollary}\label{control eB' on K,V_4}
For $\vert t\vert\leq 1$, we have
  \begin{align}
  e^{-tB^\prime}\mathcal{K}_se^{tB^\prime} & \lesssim \mathcal{K}_s+N^{\frac{\kappa}{2}}\ln N
    (\mathcal{N}_{ex}+1) \label{control eB' on K_s}\\
  e^{-tB^\prime}\mathcal{V}_4e^{tB^\prime} & \lesssim  \mathcal{V}_4+(\mathcal{N}_{ex}+1)
    +N^{-\frac{5}{3}+\kappa-\alpha+\gamma}(\mathcal{N}_{ex}+1)^3\label{control eB' on V_4}
  \end{align}
\end{corollary}

\par The next four lemmas, Lemmas \ref{cal eB' on Q_G_N and E_G_N}-\ref{cal [V_3,B'] lemma}, collect the results about the right hand side of (\ref{define J_N}).
\begin{lemma}\label{cal eB' on Q_G_N and E_G_N}
We have
  \begin{equation}\label{cal e^B' on Q_G_N}
    e^{-B^\prime}Q_{\mathcal{G}_N}e^{B^\prime}=Q_{\mathcal{G}_N}
    +\mathcal{E}_{Q_{\mathcal{G}_N}}
  \end{equation}
  where
  \begin{equation}\label{bound E_Q_G_N}
    \pm\mathcal{E}_{Q_{\mathcal{G}_N}}\lesssim
    N^{-\frac{1}{3}+\frac{\kappa}{2}-\frac{\alpha}{2}}(\mathcal{N}_{ex}+1).
  \end{equation}
  On the other hand,
  \begin{equation}\label{bound eB' on E_G_N}
  \begin{aligned}
    \pm e^{-B^\prime}\mathcal{E}_{\mathcal{G}_N}e^{B^\prime}\lesssim &
    \big(N^{-\frac{1}{6}}+N^{-\alpha-\gamma}
    +N^{-\frac{1}{2}+\alpha+2\gamma}\big)
    \big(\mathcal{K}_s+N^{\frac{\kappa}{2}}\ln N
    (\mathcal{N}_{ex}+1)\big)\\
    &+\big(\mu+N^{-\gamma}\big)\big(\mathcal{V}_4+(\mathcal{N}_{ex}+1)
    +N^{-\frac{5}{3}+\kappa-\alpha+\gamma}(\mathcal{N}_{ex}+1)^3\big)\\
    &+P_{\mathcal{G}_N}(\mathcal{N}_{ex})+N^{\frac{1}{3}-\frac{\alpha}{2}}
   \end{aligned}
  \end{equation}
  with $0<\mu<1$, $0<\gamma<\alpha$ and $P_{\mathcal{G}_N}(\mathcal{N}_{ex})$ given in (\ref{bound P_G_N}).
\end{lemma}
\begin{proof}
  This is a direct consequence of Proposition \ref{p1}, Lemma \ref{control eB' on N_ex lemma} and Corollary \ref{control eB' on K,V_4}.
\end{proof}

\begin{lemma}\label{control eB'on V_21' Omega lemma}
We have
  \begin{equation}\label{control eB'on V_21' Omega}
    e^{-B^\prime}\big(\mathcal{V}_{21}^\prime
  +\Omega\big)e^{B^\prime}=\mathcal{V}_{21}^\prime
  +\Omega+\mathcal{E}_{(\mathcal{V}_{21}^\prime
  +\Omega)}
  \end{equation}
  with
  \begin{equation}\label{bound E_V_21'+Omega}
  \begin{aligned}
    \pm\mathcal{E}_{(\mathcal{V}_{21}^\prime
  +\Omega)}\lesssim& N^{-\frac{1}{3}-2\alpha+\frac{\kappa}{2}}(\mathcal{N}_{ex}+1)
  +N^{-\frac{1}{3}-\frac{\alpha}{2}+\frac{\kappa}{2}}(\mathcal{N}_{ex}+1)^{\frac{3}{2}}\\
  &N^{-\frac{4}{3}+\alpha+\gamma+\kappa}(\mathcal{N}_{ex}+1)^2
  +N^{-\alpha-\gamma}\big(\mathcal{K}_s+N^{\frac{\kappa}{2}}\ln N(\mathcal{N}_{ex}+1)\big)\\
  &+N^{\frac{1}{6}+\alpha+\frac{\kappa}{2}}
  \end{aligned}
  \end{equation}
  and the further assumption that $-\frac{11}{48}<\kappa\leq -2(\alpha+\gamma)$.
\end{lemma}
\begin{proof}
  By Newton-Leibniz formula, we immediately obtain (\ref{control eB'on V_21' Omega}) if we put
  \begin{equation}\label{E_V_21' Omega}
    \mathcal{E}_{(\mathcal{V}_{21}^\prime
  +\Omega)}=\int_{0}^{1}e^{-tB}[\mathcal{V}_{21}^\prime+\Omega,B]e^{tB}dt.
  \end{equation}
  We first bound $[\mathcal{V}_{21}^\prime,B]$, then the bound of $[\Omega,B]$ follows similarly. By Lemma \ref{lem commutator}, we have $[\mathcal{V}_{21}^\prime,B]=\sum_{j=1}^{5}\Pi_j$ where
  \begin{equation}\label{[V_21',B']}
    \begin{aligned}
    &\Pi_1=\sum_{\substack{\sigma,\nu\\k,p,q}}\sum_{l,r,s}2W_k\eta_l
    \delta_{p-k,r-l}\delta_{q+k,s+l}a^*_{p,\sigma}a^*_{q,\nu}a_{s,\nu}a_{r,\sigma}\\
    &\quad\quad\times\chi_{p-k\notin B_F^\sigma}\chi_{p\in B_F^\sigma}\chi_{q+k\notin B_F^\nu}\chi_{q\in B_F^\nu}\chi_{r\in B_F^\sigma}\chi_{s\in A_{F,\kappa}^\nu}+h.c.\\
    &\Pi_2=-\sum_{\substack{\sigma,\nu\\k,p,q}}\sum_{\varpi,l,r,s}2W_k\eta_l
    \delta_{p-k,r-l}a^*_{p,\sigma}a^*_{q,\nu}a^*_{s+l,\varpi}a_{q+k,\nu}a_{s,\varpi}a_{r,\sigma}\\
    &\quad\quad\times\chi_{p-k\notin B_F^\sigma}\chi_{p\in B_F^\sigma}\chi_{q+k\notin B_F^\nu}\chi_{q\in B_F^\nu}\chi_{r\in B_F^\sigma}\chi_{s+l\notin B_F^\varpi}\chi_{s\in A_{F,\kappa}^\varpi}+h.c.\\
    &\Pi_3=-\sum_{\substack{\sigma,\nu\\k,p,q}}\sum_{\tau,l,r,s}2W_k\eta_l
    \delta_{q+k,s+l}a^*_{p,\sigma}a^*_{q,\nu}a^*_{r-l,\tau}
    a_{p-k,\sigma}a_{s,\nu}a_{r,\tau}\\
    &\quad\quad\times\chi_{p-k\notin B_F^\sigma}\chi_{p\in B_F^\sigma}\chi_{q+k\notin B_F^\nu}\chi_{q\in B_F^\nu}\chi_{r-l\notin B_F^\tau}\chi_{r\in B_F^\tau}\chi_{s\in A_{F,\kappa}^\nu}+h.c.\\
    &\Pi_4=\sum_{\substack{\sigma,\nu\\k,p,q}}\sum_{\tau,l,r,s}2W_k\eta_l
    \delta_{q+k,s}a^*_{r-l,\tau}a^*_{s+l,\nu}a^*_{p-k,\sigma}a_{r,\tau}a_{q,\nu}a_{p,\sigma}\\
    &\quad\quad\times
    \chi_{p-k\notin B_F^\sigma}\chi_{p\in B_F^\sigma}\chi_{q+k\in  A_{F,\kappa}^\nu}
    \chi_{q\in B_F^\nu}\chi_{r-l\notin B_F^\tau}\chi_{r\in B_F^\tau}\chi_{s+l\notin B_F^\nu}
    +h.c.\\
    &\Pi_5=\sum_{\substack{\sigma,\nu\\k,p,q}}\sum_{\varpi,l,r,s}2W_k\eta_l
    \delta_{p,r}a^*_{r-l,\sigma}a^*_{s+l,\varpi}a^*_{q,\nu}
    a_{s,\varpi}a_{q+k,\nu}a_{p-k,\sigma}\\
    &\quad\quad\times\chi_{p-k\notin B_F^\sigma}\chi_{p\in B_F^\sigma}\chi_{q+k\notin B_F^\nu}\chi_{q\in B_F^\nu}\chi_{r-l\notin B_F^\sigma}\chi_{s+l\notin B_F^\varpi}\chi_{s\in A_{F,\kappa}^\varpi}+h.c.
    \end{aligned}
  \end{equation}
  We then analyse (\ref{[V_21',B']}) term by term.
  \begin{flushleft}
    \textbf{Analysis of $\Pi_1$:}
  \end{flushleft}
  Switching to the position space, we have
  \begin{equation}\label{Pi_1}
    \begin{aligned}
    \Pi_1&=2\sum_{\sigma,\nu}\int_{\Lambda^4}W(x_1-x_3)\eta(x_2-x_4)
    \Big(\sum_{p_1\notin B_F^\sigma}e^{ip_1(x_1-x_2)}\Big)
    \Big(\sum_{p_2\notin B_F^\nu}e^{ip_2(x_3-x_4)}\Big)\\
    &\quad\quad\times a^*(g_{x_1,\sigma}) a^*(g_{x_3,\nu})
    a(L_{x_4,\nu}[\kappa])a(g_{x_2,\sigma})dx_1dx_2dx_3dx_4+h.c.
    \end{aligned}
  \end{equation}
  Using
  \begin{equation*}
     \sum_{p_1\notin B_F^\sigma}\sum_{p_2\notin B_F^\nu}=
  \sum_{p_1,p_2}-\sum_{p_1}\sum_{p_2\in B_F^\nu}-\sum_{p_2}\sum_{p_1\in B_F^\sigma}
  +\sum_{p_1\in B_F^\sigma}\sum_{p_2\in B_F^\nu},
  \end{equation*}
  we rewrite $\Pi_1=\sum_{j=1}^{4}\Pi_{1j}$. Using the estimates in Section \ref{scattering eqn sec}, we can bound them easily by
  \begin{equation}\label{Pi_1 bound}
    \begin{aligned}
    \pm\Pi_{11}&\lesssim N^{\frac{11}{6}+\frac{\kappa}{2}}\Vert W\Vert_2\Vert\eta\Vert_2\lesssim N^{\frac{1}{6}+\alpha+\frac{\kappa}{2}}\\
    \pm\Pi_{12},\pm\Pi_{13}&\lesssim N^{\frac{7}{3}+\frac{\kappa}{2}}\Vert W\Vert_1\Vert\eta\Vert_2\lesssim N^{\frac{1}{6}-\frac{\alpha}{2}+\frac{\kappa}{2}}\\
    \pm\Pi_{14}&\lesssim N^{\frac{17}{6}+\frac{\kappa}{2}}\Vert W\Vert_1\Vert\eta\Vert_1\lesssim N^{\frac{1}{6}-2\alpha+\frac{\kappa}{2}}
    \end{aligned}
  \end{equation}
  Together, we have
  \begin{equation}\label{Pi_1 bound true}
    \pm\Pi_1\lesssim N^{\frac{1}{6}+\alpha+\frac{\kappa}{2}}.
  \end{equation}
  \begin{flushleft}
    \textbf{Analysis of $\Pi_2$:}
  \end{flushleft}
   Switching to the position space, we have
  \begin{equation}\label{Pi_2}
    \begin{aligned}
    \Pi_2&=-2\sum_{\sigma,\nu,\varpi}\int_{\Lambda^4}W(x_1-x_3)\eta(x_2-x_4)
    \Big(\sum_{p_1\notin B_F^\sigma}e^{ip_1(x_1-x_2)}\Big)
    a^*(g_{x_1,\sigma}) a^*(g_{x_3,\nu})\\
    &\quad\quad\times a^*(h_{x_4,\varpi})
    a(h_{x_3,\nu})a(L_{x_4,\varpi}[\kappa])a(g_{x_2,\sigma})dx_1dx_2dx_3dx_4+h.c.
    \end{aligned}
  \end{equation}
  Using $\sum_{p_1\notin B_F^\sigma}=\sum_{p_1}-\sum_{p_1\in B_F^\sigma}$, we can write $\Pi_2=\Pi_{21}+\Pi_{22}$. Therefore we can bound them easily and hence reach
  \begin{equation}\label{Pi_2 bound}
    \pm\Pi_2\lesssim N^{\frac{7}{3}+\frac{\kappa}{2}}\Vert W\Vert_1\Vert\eta\Vert_1
    (\mathcal{N}_{ex}+1)\lesssim N^{-\frac{1}{3}-2\alpha+\frac{\kappa}{2}}
    (\mathcal{N}_{ex}+1).
  \end{equation}

  \begin{flushleft}
    \textbf{Analysis of $\Pi_3$:}
  \end{flushleft}
   Switching to the position space, we have
  \begin{equation}\label{Pi_3}
    \begin{aligned}
    \Pi_3&=-2\sum_{\sigma,\nu,\tau}\int_{\Lambda^4}W(x_1-x_3)\eta(x_2-x_4)
    \Big(\sum_{p_2\notin B_F^\nu}e^{ip_2(x_3-x_4)}\Big)
    a^*(g_{x_1,\sigma}) a^*(g_{x_3,\nu})\\
    &\quad\quad\times a^*(h_{x_2,\tau})
    a(h_{x_1,\sigma})a(L_{x_4,\nu}[\kappa])a(g_{x_2,\tau})dx_1dx_2dx_3dx_4+h.c.
    \end{aligned}
  \end{equation}
  Similar to the bound of $\Pi_2$, we have
  \begin{equation}\label{Pi_3 bound}
    \pm\Pi_3\lesssim N^{\frac{7}{3}+\frac{\kappa}{2}}\Vert W\Vert_1\Vert\eta\Vert_1
    (\mathcal{N}_{ex}+1)\lesssim N^{-\frac{1}{3}-2\alpha+\frac{\kappa}{2}}
    (\mathcal{N}_{ex}+1).
  \end{equation}

   \begin{flushleft}
    \textbf{Analysis of $\Pi_4$:}
  \end{flushleft}
  Switching to the position space, we have
  \begin{equation}\label{Pi_4}
    \begin{aligned}
    \Pi_4&=2\sum_{\sigma,\nu,\tau}\int_{\Lambda^4}W(x_1-x_3)\eta(x_2-x_4)
    \Big(\sum_{p_2\in A_{F,\kappa}^\nu}e^{ip_2(x_3-x_4)}\Big)
    a^*(h_{x_2,\tau}) a^*(h_{x_4,\nu})\\
    &\quad\quad\times a^*(h_{x_1,\sigma})
    a(g_{x_2,\tau})a(g_{x_3,\nu})a(g_{x_1,\sigma})dx_1dx_2dx_3dx_4+h.c.
    \end{aligned}
  \end{equation}
  Notice that by (\ref{ineqn N_l and N_s}), for $j=3,4$
  \begin{equation*}
    \int_{\Lambda}\Big\vert\sum_{p_2\in A_{F,\kappa}^\nu}e^{ip_2(x_3-x_4)}\Big\vert^2dx_j
    =\#A_{F,\kappa}^\nu\lesssim N^{\frac{2}{3}+\kappa}.
  \end{equation*}
  Hence we can bound $\Pi_4$ by
  \begin{equation}\label{Pi_4 bound}
    \pm\Pi_4\lesssim N^{\frac{11}{6}+\frac{\kappa}{2}}\Vert W\Vert_1\Vert\eta\Vert_2
    (\mathcal{N}_{ex}+1)^{\frac{3}{2}}\lesssim N^{-\frac{1}{3}+\frac{\kappa}{2}-\frac{\alpha}{2}}(\mathcal{N}_{ex}+1)^{\frac{3}{2}}.
  \end{equation}

   \begin{flushleft}
    \textbf{Analysis of $\Pi_5$:}
  \end{flushleft}
  Switching to the position space, we have
  \begin{equation}\label{Pi_5}
    \begin{aligned}
    \Pi_5&=2\sum_{\sigma,\nu,\varpi}\int_{\Lambda^4}W(x_1-x_3)\eta(x_2-x_4)
    \Big(\sum_{q_1\in B_{F}^\sigma}e^{iq_1(x_1-x_2)}\Big)
    a^*(h_{x_2,\sigma}) a^*(h_{x_4,\varpi})\\
    &\quad\quad\times a^*(g_{x_3,\nu})
    a(L_{x_4,\varpi}[\kappa])a(h_{x_3,\nu})a(h_{x_1,\sigma})dx_1dx_2dx_3dx_4+h.c.
    \end{aligned}
  \end{equation}
  For $-\frac{11}{48}<\delta\leq\frac{1}{3}$ to be determined later, we use $a(h_{x_3,\nu})=
  a(H_{x_3,\nu}[\delta])+a(L_{x_3,\nu}[\delta])$ to rewrite $\Pi_5=\Pi_{51}+\Pi_{52}$. For $\Pi_{52}$, we can bound it easily by
  \begin{equation}\label{Pi_52 bound}
    \begin{aligned}
    \pm\Pi_{52}\lesssim N^{\frac{5}{3}+\frac{\kappa}{2}+\frac{\delta}{2}}
    \Vert W\Vert_1\Vert\eta\Vert_2 (\mathcal{N}_{ex}+1)^{\frac{3}{2}}
    \lesssim N^{-\frac{1}{2}-\frac{\alpha}{2}+\frac{\kappa}{2}+\frac{\delta}{2}}
    (\mathcal{N}_{ex}+1)^{\frac{3}{2}}.
    \end{aligned}
  \end{equation}
  For the estimate of $\Pi_{51}$, we rewrite it by
  \begin{equation}\label{Pi_51re}
    \begin{aligned}
    \Pi_{51}&=2\sum_{\sigma,\nu,\varpi,p_1}\chi_{p_1\notin B_F^\sigma}e^{-ip_1x_3}
    \int_{\Lambda^3}\eta(x_2-x_4)
    \Big(\sum_{q_1\in B_{F}^\sigma}W_{q_1-p_1}e^{iq_1(x_3-x_2)}\Big)
    \\
    &\quad\quad\times a^*(h_{x_2,\sigma}) a^*(h_{x_4,\varpi})a^*(g_{x_3,\nu})
    a(L_{x_4,\varpi}[\kappa])a(H_{x_3,\nu}[\delta])a_{p_1,\sigma}dx_2dx_3dx_4+h.c.
    \end{aligned}
  \end{equation}
  Thence we can bound it using (\ref{energy est bog}), by
  \begin{equation}\label{Pi_51bound}
    \begin{aligned}
    \pm\Pi_{51}&\lesssim N^{\frac{5}{6}+\frac{\kappa}{2}}\Vert\eta\Vert_2
    \Big(\sum_{\sigma}\sum_{q_1\in B_F^\sigma,p_1\notin B_F^\sigma}
    \frac{\vert W_{q_1-p_1}\vert^2}{\vert p_1\vert^2-\vert k_F^\sigma\vert^2}\Big)^{\frac{1}{2}}\\
    &\quad\quad\quad\quad\times
    \big(\theta^{-1}\mathcal{K}_s+\theta(\mathcal{N}_{ex}+1)^2(\mathcal{N}_h[\delta]+1)\big)\\
    &\lesssim \big(N^{-\alpha-\gamma}+N^{\frac{1}{3}+\kappa+\alpha+\gamma
    -\delta}\big)\mathcal{K}_s+N^{-\frac{4}{3}+\kappa+\alpha+\gamma}(\mathcal{N}_{ex}+1)^2
    \end{aligned}
  \end{equation}
  where we choose $\theta=N^{-\frac{2}{3}+\frac{\kappa}{2}+\alpha+\gamma}$. With the further requirement that $\kappa\leq-2(\alpha+\gamma)$, we choose $\delta=\frac{1}{3}+\kappa+2\alpha+2\gamma$, such that $-\frac{11}{48}<\delta\leq\frac{1}{3}$. Combining (\ref{Pi_51bound}) and (\ref{Pi_52 bound}), we reach
  \begin{equation}\label{Pi_5 bound}
    \pm\Pi_5\lesssim N^{-\alpha-\gamma}\mathcal{K}_s
    +N^{-\frac{4}{3}+\kappa+\alpha+\gamma}(\mathcal{N}_{ex}+1)^2
    +N^{-\frac{1}{3}+\kappa+\frac{\alpha}{2}+\gamma}
    (\mathcal{N}_{ex}+1)^{\frac{3}{2}}.
  \end{equation}

   \begin{flushleft}
    \textbf{Conclusion of Lemma \ref{control eB'on V_21' Omega lemma}:}
  \end{flushleft}
  Collecting (\ref{Pi_1 bound true}), (\ref{Pi_2 bound}), (\ref{Pi_3 bound}), (\ref{Pi_4 bound}) and (\ref{Pi_5 bound}), with the further assumption that $\kappa\leq-2(\alpha+\gamma)$, we arrive at
  \begin{equation}\label{[V_21' B'] bound}
    \begin{aligned}
    \pm[\mathcal{V}_{21}^\prime,B^\prime]\lesssim& N^{-\frac{1}{3}-2\alpha+\frac{\kappa}{2}}(\mathcal{N}_{ex}+1)
  +N^{-\frac{1}{3}-\frac{\alpha}{2}+\frac{\kappa}{2}}(\mathcal{N}_{ex}+1)^{\frac{3}{2}}\\
  &N^{-\frac{4}{3}+\alpha+\gamma+\kappa}(\mathcal{N}_{ex}+1)^2
  +N^{-\alpha-\gamma}\mathcal{K}_s+N^{\frac{1}{6}+\alpha+\frac{\kappa}{2}}
  \end{aligned}
  \end{equation}
  \par For the estimate of $[\Omega,B^\prime]$, we can similarly write $[\Omega,B^\prime]=\sum_{j=1}^{5}\tilde{\Pi}_j$, with the coefficient $W_k$ being replaced by $\eta_kk(q-p)$. The bound of each term, except $\tilde{\Pi}_{11}$, follows similarly to the above discussion. We only need to, in the corresponding estimates, replace the estimate $\Vert W\Vert_1\lesssim a$ by $N^{\frac{1}{3}}\Vert\nabla\eta\Vert_1\lesssim N^{\frac{1}{3}}\ell a$, and use the second inequality in (\ref{energy est bog}) rather than the first one. The only difference is the estimate of $\tilde{\Pi}_{11}$, which is
  \begin{equation}\label{tiled Pi_11}
    \begin{aligned}
    \tilde{\Pi}_{11}&=C\sum_{\sigma,\nu}\int_{\Lambda^2}\nabla_{x_1}\eta(x_1-x_3)\eta(x_1-x_3)\\
    &\quad\quad\times a^*(\nabla_{x_1} g_{x_1,\sigma}) a^*(g_{x_3,\nu})
    a(L_{x_3,\nu}[\kappa])a(g_{x_1,\sigma})dx_1dx_3+h.c.
    \end{aligned}
  \end{equation}
  where we denote an arbitrary universal constant by $C$. Integrating by parts yields
  \begin{equation}\label{tilde Pi_11}
    \begin{aligned}
    \tilde{\Pi}_{11}&=-C\sum_{\sigma,\nu}\int_{\Lambda^2}\eta(x_1-x_3)\nabla_{x_1}\eta(x_1-x_3)\\
    &\quad\quad\times a^*(\nabla_{x_1} g_{x_1,\sigma}) a^*(g_{x_3,\nu})
    a(L_{x_3,\nu}[\kappa])a(g_{x_1,\sigma})dx_1dx_3+h.c.\\
    &-C\sum_{\sigma,\nu}\int_{\Lambda^2}\vert\eta(x_1-x_3)\vert^2\\
    &\quad\quad\times a^*(\Delta_{x_1} g_{x_1,\sigma}) a^*(g_{x_3,\nu})
    a(L_{x_3,\nu}[\kappa])a(g_{x_1,\sigma})dx_1dx_3+h.c.\\
    &-C\sum_{\sigma,\nu}\int_{\Lambda^2}\vert\eta(x_1-x_3)\vert^2\\
    &\quad\quad\times a^*(\nabla_{x_1} g_{x_1,\sigma}) a^*(g_{x_3,\nu})
    a(L_{x_3,\nu}[\kappa])a(\nabla_{x_1} g_{x_1,\sigma})dx_1dx_3+h.c.
    \end{aligned}
  \end{equation}
  Combining (\ref{tiled Pi_11}) and (\ref{tilde Pi_11}), we have
  \begin{equation*}
  \begin{aligned}
    \tilde{\Pi}_{11}&=
    -C\sum_{\sigma,\nu}\int_{\Lambda^2}\vert\eta(x_1-x_3)\vert^2\\
    &\quad\quad\times a^*(\Delta_{x_1} g_{x_1,\sigma}) a^*(g_{x_3,\nu})
    a(L_{x_3,\nu}[\kappa])a(g_{x_1,\sigma})dx_1dx_3+h.c.\\
    &-C\sum_{\sigma,\nu}\int_{\Lambda^2}\vert\eta(x_1-x_3)\vert^2\\
    &\quad\quad\times a^*(\nabla_{x_1} g_{x_1,\sigma}) a^*(g_{x_3,\nu})
    a(L_{x_3,\nu}[\kappa])a(\nabla_{x_1} g_{x_1,\sigma})dx_1dx_3+h.c.
  \end{aligned}
  \end{equation*}
  and we bound $\tilde{\Pi}_{11}$ by
  \begin{equation}\label{tilde Pi_11 bound}
    \tilde{\Pi}_{11}\lesssim N^{\frac{5}{2}+\frac{\kappa}{2}}\Vert\eta\Vert_1\lesssim
    N^{-\frac{1}{6}-2\alpha+\frac{\kappa}{2}}.
  \end{equation}
  \par Therefore, applying Lemma \ref{control eB' on N_ex lemma} and Corollary \ref{control eB' on K,V_4} to (\ref{[V_21' B'] bound}), we reach (\ref{bound E_V_21'+Omega}).
\end{proof}

\begin{lemma}\label{bound eB' on Gamma' lemma}
Recall that $0<\gamma<\alpha<\frac{1}{6}$, we have
  \begin{equation}\label{bound eB' on Gamma'}
    \begin{aligned}
    \pm\int_{0}^{1}e^{-tB\prime}\Gamma^\prime e^{tB^\prime}dt
    \lesssim& \big(N^{-\frac{1}{6}+\frac{\alpha}{2}+\frac{\kappa}{2}}+
    N^{-\frac{1}{3}-\kappa+\gamma}\big)
    \big(\mathcal{K}_s+N^{\frac{\kappa}{2}}\ln N
    (\mathcal{N}_{ex}+1)\big)\\
    &+N^{-\gamma}\big(\mathcal{V}_4+(\mathcal{N}_{ex}+1)\big)
    +N^{-\frac{5}{3}+\kappa-\alpha+\gamma}(\mathcal{N}_{ex}+1)^3.
    \end{aligned}
  \end{equation}
\end{lemma}
\begin{proof}
  \par By Lemma \ref{cal [K,V_4,B'] lemma} and (\ref{eqn of eta_p rewrt}), we have
  \begin{equation}\label{Gamma' re}
    \Gamma^\prime=\mathcal{E}_{[\mathcal{K},B^\prime]}+\mathcal{E}_{[\mathcal{V}_4,B^\prime]}
    +\Gamma^\prime_1+\Gamma^\prime_2
  \end{equation}
  with
  \begin{equation}\label{Gamma'_j}
    \begin{aligned}
    &\Gamma^\prime_1=\sum_{k,p,q,\sigma,\nu}
  2W_k(a^*_{p-k,\sigma}a^*_{q+k,\nu}a_{q,\nu}a_{p,\sigma}+h.c.)\chi_{p-k\notin B^{\sigma}_F}\chi_{q+k\notin B^{\nu}_F}\chi_{q\in A^\nu_{F,\kappa}}
  \chi_{p\in B^{\sigma}_F}\\
    &\Gamma^\prime_2=\sum_{k,p,q,\sigma,\nu}
  \hat{v}_k(a^*_{p-k,\sigma}a^*_{q+k,\nu}a_{q,\nu}a_{p,\sigma}+h.c.)\chi_{p-k\notin B^{\sigma}_F}\chi_{q+k\notin B^{\nu}_F}\chi_{q\in P^\nu_{F,\kappa}}
  \chi_{p\in B^{\sigma}_F}
    \end{aligned}
  \end{equation}
  The bound of $\Gamma_1^\prime$ is analogous to the bound of $\Omega^\prime$ in Lemma \ref{cal [K,V_4,B'] lemma} (see (\ref{Omega_11' bound}) and (\ref{Omega_2' bound})). We only need to, in the corresponding estimates, replace the estimate $N^{\frac{1}{3}}\Vert\nabla\eta\Vert_1\lesssim N^{\frac{1}{3}}\ell a$ by $\Vert W\Vert_1\lesssim a$, and use the first inequality in (\ref{energy est bog}) rather than the second one. We state the result omitting further details:
  \begin{equation}\label{Gamma'_1 bound}
    \pm\Gamma^\prime_1\lesssim N^{-\frac{1}{6}+\frac{\alpha}{2}+\frac{\kappa}{2}}\mathcal{K}_s
    +N^{-\frac{1}{6}}(\mathcal{N}_{ex}+1).
  \end{equation}
  For the bound of $\Gamma_2^\prime$, we write it in the position space:
  \begin{equation}\label{Gamma'2}
    \begin{aligned}
    \Gamma_2^\prime=\int_{\Lambda^2}v_a(x-y)a^*(h_{x,\sigma})a^*(h_{y,\nu})
    a(H_{y,\nu}[\kappa])a(g_{x,\sigma})dxdy+h.c.
    \end{aligned}
  \end{equation}
  and bound it by
  \begin{equation}\label{Gamma'2 bound}
    \pm\Gamma_2^\prime\lesssim N^{\frac{1}{2}}\Vert v_a\Vert_1^\frac{1}{2}
    \big(\theta\mathcal{V}_{4}+\theta^{-1}\mathcal{N}_h[\kappa]\big)
    \lesssim N^{-\gamma}\mathcal{V}_4+N^{-\frac{1}{3}-\kappa+\gamma}\mathcal{K}_s
  \end{equation}
  where we choose $\theta=N^{-\gamma}$. Combining (\ref{E_[K,B']}), (\ref{E_[V_4,B']}), (\ref{Gamma'_1 bound}) and (\ref{Gamma'2 bound}), we have
  \begin{equation}\label{Gamma' bound}
    \begin{aligned}
    \pm\Gamma^\prime
    \lesssim& \big(N^{-\frac{1}{6}+\frac{\alpha}{2}+\frac{\kappa}{2}}+
    N^{-\frac{1}{3}-\kappa+\gamma}\big)
    \mathcal{K}_s+N^{-\frac{1}{6}}
    (\mathcal{N}_{ex}+1)\\
    &+N^{-\gamma}\mathcal{V}_4
    +N^{-\frac{5}{3}+\kappa-\alpha+\gamma}(\mathcal{N}_{ex}+1)^3.
    \end{aligned}
  \end{equation}
  Applying Lemma \ref{control eB' on N_ex lemma} and Corollary \ref{control eB' on K,V_4} to (\ref{Gamma' bound}), we reach (\ref{bound eB' on Gamma'}).
\end{proof}

\begin{lemma}\label{cal [V_3,B'] lemma}
We have
  \begin{equation}\label{cal [V_3,B']}
    \begin{aligned}
    \int_{0}^{1}\int_{t}^{1}e^{-sB^\prime}[\mathcal{V}_3,B^\prime]e^{sB^\prime}dsdt
    =\sum_{k}\hat{v}_k\eta_k\sum_{\sigma\neq\nu}N_\sigma \mathcal{N}_{ex,\nu}
    +\mathcal{E}_{[\mathcal{V}_3,B^\prime]}
    \end{aligned}
  \end{equation}
  where
  \begin{equation}\label{bound E_[V_3,B']}
    \begin{aligned}
    \pm\mathcal{E}_{[\mathcal{V}_3,B^\prime]}
    \lesssim& \big(N^{-\frac{1}{3}-\kappa}+
    N^{\kappa+\gamma-\alpha}\big)
    \big(\mathcal{K}_s+N^{\frac{\kappa}{2}}\ln N
    (\mathcal{N}_{ex}+1)\big)\\
    &+N^{-\gamma}\big(\mathcal{V}_4+(\mathcal{N}_{ex}+1)\big)
    +N^{-\frac{3}{2}+\frac{\kappa}{2}-\frac{\alpha}{2}}(\mathcal{N}_{ex}+1)^3\\
    &+\big(N^{-1}+N^{-\frac{5}{6}-\frac{\alpha}{2}+\frac{\kappa}{2}}
    +N^{-\frac{2}{3}-\alpha+\kappa+\gamma}\big)
    (\mathcal{N}_{ex}+1)^2.
    \end{aligned}
  \end{equation}
\end{lemma}
\begin{proof}
  \par By Lemma \ref{lem commutator}, we have
  \begin{equation}\label{Upsilon}
    [\mathcal{V}_3,B^\prime]=\sum_{j=1}^{12}\Upsilon_j
  \end{equation}
  with
  \begin{align*}
    &\Upsilon_1=\sum_{\substack{\sigma,\nu\\k,p,q}}\sum_{l,r,s}
    \hat{v}_k\eta_l\delta_{p-k,r-l}\delta_{q+k,s+l}
    a^*_{p,\sigma}a^*_{q,\nu}a_{s,\nu}a_{r,\sigma}\\
    &\quad\quad\times
    \chi_{p-k\notin B_F^\sigma}\chi_{p\in B_F^\sigma}\chi_{q+k,q\notin B_F^\nu}
  \chi_{r\in B_F^\sigma}\chi_{s\in A_{F,\kappa}^\nu}+h.c.\\
  &\Upsilon_2=-\sum_{\substack{\sigma,\nu\\k,p,q}}\sum_{l,r,s}
  \hat{v}_k\eta_l\delta_{p-k,s+l}\delta_{q+k,r-l}
    a^*_{p,\sigma}a^*_{q,\nu}a_{s,\sigma}a_{r,\nu}\\
    &\quad\quad\times
    \chi_{p-k\notin B_F^\sigma}\chi_{p\in B_F^\sigma}\chi_{q+k,q\notin B_F^\nu}
  \chi_{r\in B_F^\nu}\chi_{s\in A_{F,\kappa}^\sigma}+h.c.\\
  &\Upsilon_3=\sum_{\substack{\sigma,\nu\\k,p,q}}\sum_{\varpi,l,r,s}
  \hat{v}_k\eta_l\delta_{q,r-l}
  a^*_{p-k,\sigma}a^*_{q+k,\nu}a^*_{s+l,\varpi}a_{p,\sigma}a_{s,\varpi}a_{r,\nu}\\
  &\quad\quad\times
  \chi_{p-k\notin B_F^\sigma}\chi_{p\in B_F^\sigma}\chi_{q+k,q\notin B_F^\nu}
  \chi_{r\in B_F^\nu}\chi_{s+l\notin B^\varpi_F}\chi_{s\in A_{F,\kappa}^\varpi}+h.c.\\
  &\Upsilon_4=-\sum_{\substack{\sigma,\nu\\k,p,q}}\sum_{\tau,l,r,s}
  \hat{v}_k\eta_l\delta_{q,s+l}
  a^*_{p-k,\sigma}a^*_{q+k,\nu}a^*_{r-l,\tau}a_{p,\sigma}a_{s,\nu}a_{r,\tau}\\
  &\quad\quad\times
  \chi_{p-k\notin B_F^\sigma}\chi_{p\in B_F^\sigma}\chi_{q+k,q\notin B_F^\nu}
  \chi_{r-l\notin B_F^\tau}\chi_{r\in B_F^\tau}\chi_{s\in A_{F,\kappa}^\nu}+h.c.\\
  &\Upsilon_5=\sum_{\substack{\sigma,\nu\\k,p,q}}\sum_{\tau,l,r,s}
  \hat{v}_k\eta_l\delta_{q+k,s}
  a^*_{r-l,\tau}a^*_{s+l,\nu}a^*_{p-k,\sigma}a_{r,\tau}a_{q,\nu}a_{p,\sigma}\\
  &\quad\quad\times
  \chi_{p-k\notin B_F^\sigma}\chi_{p\in B_F^\sigma}\chi_{q+k\in A_{F,\kappa}^\nu}
  \chi_{q\notin B_F^\nu}\chi_{r-l\notin B_F^\tau}\chi_{r\in B_F^\tau}
  \chi_{s+l\notin B^\nu_F}+h.c.\\
  &\Upsilon_6=-\sum_{\substack{\sigma,\nu\\k,p,q}}\sum_{\tau,l,r,s}
  \hat{v}_k\eta_l\delta_{p-k,s}
  a^*_{r-l,\tau}a^*_{s+l,\sigma}a^*_{q+k,\nu}a_{r,\tau}a_{q,\nu}a_{p,\sigma}\\
  &\quad\quad\times
  \chi_{p-k\in A_{F,\kappa}^\sigma}\chi_{p\in B_F^\sigma}\chi_{q+k,q\notin B_F^\nu}
  \chi_{r-l\notin B_F^\tau}\chi_{r\in B_F^\tau}\chi_{s+l\notin B^\sigma_F}+h.c.\\
  &\Upsilon_7=\sum_{\substack{\sigma,\nu\\k,p,q}}\sum_{\varpi,l,r,s}
  \hat{v}_k\eta_l\delta_{q+k,r-l}
  a^*_{p,\sigma}a^*_{q,\nu}a^*_{s+l,\varpi}a_{p-k,\sigma}a_{s,\varpi}a_{r,\nu}\\
  &\quad\quad\times
  \chi_{p-k\notin B_F^\sigma}\chi_{p\in B_F^\sigma}\chi_{q+k,q\notin B_F^\nu}
  \chi_{r\in B_F^\nu}\chi_{s+l\notin B^\varpi_F}\chi_{s\in A_{F,\kappa}^\varpi}+h.c.\\
  &\Upsilon_8=-\sum_{\substack{\sigma,\nu\\k,p,q}}\sum_{\varpi,l,r,s}
  \hat{v}_k\eta_l\delta_{p-k,r-l}
  a^*_{p,\sigma}a^*_{q,\nu}a^*_{s+l,\varpi}a_{q+k,\nu}a_{s,\varpi}a_{r,\sigma}\\
  &\quad\quad\times
  \chi_{p-k\notin B_F^\sigma}\chi_{p\in B_F^\sigma}\chi_{q+k,q\notin B_F^\nu}
  \chi_{r\in B_F^\sigma}\chi_{s+l\notin B^\varpi_F}\chi_{s\in A_{F,\kappa}^\varpi}+h.c.\\
  &\Upsilon_9=\sum_{\substack{\sigma,\nu\\k,p,q}}\sum_{\tau,l,r,s}
  \hat{v}_k\eta_l\delta_{p-k,s+l}
  a^*_{p,\sigma}a^*_{q,\nu}a^*_{r-l,\tau}a_{q+k,\nu}a_{s,\sigma}a_{r,\tau}\\
  &\quad\quad\times
  \chi_{p-k\notin B_F^\sigma}\chi_{p\in B_F^\sigma}\chi_{q+k,q\notin B_F^\nu}
  \chi_{r-l\notin B_F^\tau}\chi_{r\in B_F^\tau}\chi_{s\in A_{F,\kappa}^\sigma}+h.c.\\
  &\Upsilon_{10}=-\sum_{\substack{\sigma,\nu\\k,p,q}}\sum_{\tau,l,r,s}
  \hat{v}_k\eta_l\delta_{q+k,s+l}
  a^*_{p,\sigma}a^*_{q,\nu}a^*_{r-l,\tau}a_{p-k,\sigma}a_{s,\nu}a_{r,\tau}\\
  &\quad\quad\times
  \chi_{p-k\notin B_F^\sigma}\chi_{p\in B_F^\sigma}\chi_{q+k,q\notin B_F^\nu}
  \chi_{r-l\notin B_F^\tau}\chi_{r\in B_F^\tau}\chi_{s\in A_{F,\kappa}^\nu}+h.c.\\
  &\Upsilon_{11}=\sum_{\substack{\sigma,\nu\\k,p,q}}\sum_{\varpi,l,r,s}
  \hat{v}_k\eta_l\delta_{p,r}
  a^*_{r-l,\sigma}a^*_{s+l,\varpi}a^*_{q,\nu}a_{s,\varpi}a_{q+k,\nu}a_{p-k,\sigma}\\
  &\quad\quad\times
  \chi_{p-k\notin B_F^\sigma}\chi_{p\in B_F^\sigma}\chi_{q+k,q\notin B_F^\nu}
  \chi_{r-l\notin B_F^\sigma}\chi_{s+l\notin B^\varpi_F}\chi_{s\in A_{F,\kappa}^\varpi}+h.c.\\
  &\Upsilon_{12}=-\sum_{\substack{\sigma,\nu\\k,p,q}}\sum_{\tau,l,r,s}
  \hat{v}_k\eta_l\delta_{q,s}
  a^*_{r-l,\tau}a^*_{s+l,\nu}a_{r,\tau}a^*_{p,\sigma}a_{q+k,\nu}a_{p-k,\sigma}\\
  &\quad\quad\times
  \chi_{p-k\notin B_F^\sigma}\chi_{p\in B_F^\sigma}\chi_{q+k\notin B_F^\nu}\chi_{q\in A_{F,\kappa}^\nu}
  \chi_{r-l\notin B_F^\tau}\chi_{r\in B_F^\tau}\chi_{s+l\notin B^\nu_F}+h.c.
  \end{align*}
We then analyse (\ref{Upsilon}) term by term.
\begin{flushleft}
  \textbf{Analysis of $\Upsilon_1\&\Upsilon_2$:}
\end{flushleft}
Using $1=\chi_{k=l}+\chi_{k\neq l}$, we rewrite $\Upsilon_1=\Upsilon_{1,1}+\Upsilon_{1,2}$. For $\Upsilon_{1,1}$, we further decompose it by $\Upsilon_{1,1}=\sum_{j=1}^{3}\Upsilon_{1,1j}$ with
\begin{equation}\label{Upsilon_11j}
  \begin{aligned}
  &\Upsilon_{1,11}=2\sum_{\substack{\sigma,\nu\\k,p,q}}\hat{v}_k\eta_k
  a^*_{p,\sigma}a_{p,\sigma}a^*_{q,\nu}a_{q,\nu}
  \chi_{p\in B_F^\sigma}
  \chi_{q\notin B_{F}^\nu}\\
  &\Upsilon_{1,12}=-2\sum_{\substack{\sigma,\nu\\k,p,q}}\hat{v}_k\eta_k
  a^*_{p,\sigma}a_{p,\sigma}a^*_{q,\nu}a_{q,\nu}
  \chi_{p\in B_F^\sigma}
  \chi_{q\in P_{F,\kappa}^\nu}\\
  &\Upsilon_{1,13}=2\sum_{\substack{\sigma,\nu\\k,p,q}}\hat{v}_k\eta_k
  a^*_{p,\sigma}a_{p,\sigma}a^*_{q,\nu}a_{q,\nu}
  (\chi_{p-k\notin B^\sigma_F}\chi_{q+k\notin B_F^\nu}-1)\chi_{p\in B_F^\sigma}
  \chi_{q\in A_{F,\kappa}^\nu}
  \end{aligned}
\end{equation}
Using (\ref{cal N_sigma}) and (\ref{number of exicited particles ops}), we have
\begin{equation}\label{Upsilon111}
\begin{aligned}
  \Upsilon_{1,11}&=2\sum_{k}\hat{v}_k\eta_k\sum_{\sigma,\nu}(N_\sigma
  -\mathcal{N}_{ex,\sigma}) \mathcal{N}_{ex,\nu}\\
  &=2\sum_{k}\hat{v}_k\eta_k\sum_{\sigma,\nu}N_\sigma\mathcal{N}_{ex,\nu}
  +\mathcal{E}_{\Upsilon_{1,11}},
\end{aligned}
\end{equation}
where
\begin{equation}\label{Upsilon 111 error bound}
  \pm\mathcal{E}_{\Upsilon_{1,11}}\lesssim a\mathcal{N}_{ex}^2,
\end{equation}
since $\vert\sum_{k}\hat{v}_k\eta_k\vert\lesssim a$.
\par Recall the definition (\ref{tools2}) and inequality (\ref{ineqn N_h and N_i}) in Lemma \ref{lemma ineqn K_s}, we also have
\begin{equation}\label{Upsilon 112}
  \pm\Upsilon_{1,12}=\mp2\sum_{k}\hat{v}_k\eta_k\sum_{\sigma,\nu}(N_\sigma
  -\mathcal{N}_{ex,\sigma}) \mathcal{N}_{h,\nu}[\kappa]\lesssim N^{-\frac{1}{3}-\kappa}\mathcal{K}_s,
\end{equation}
since obviously $\pm(N_\sigma-\mathcal{N}_{ex,\sigma})\lesssim N$.
\par For $\Upsilon_{1,13}$, recalling that for any $p\in B_F^\sigma$ and $q\in A_{F,\kappa}^\nu$, we have $\vert p\vert, \vert q\vert\lesssim N^{\frac{1}{3}}$. Therefore $(\chi_{p-k\notin B^{\sigma}_F}\chi_{q+k\notin B^{\nu}_F}-1)\neq0$ implies $\vert k\vert\lesssim N^{\frac{1}{3}}$. Since $\vert \hat{v}_k\vert\lesssim a$ and $\vert\eta_k\vert\lesssim a\ell^2$, we can bound it by
\begin{equation}\label{Upsilon 113}
  \pm\Upsilon_{1,13}\lesssim N^2a^2\ell^2\mathcal{N}_{l}[\kappa]\lesssim N^{-\frac{2}{3}-2\alpha}\mathcal{N}_{ex}.
\end{equation}
\par For $\Upsilon_{1,2}$, we further decompose it by $\Upsilon_{1,2}=\Upsilon_{1,21}+\Upsilon_{1,22}$ with
\begin{equation}\label{Upsilon_12j}
  \begin{aligned}
  &\Upsilon_{1,21}=-\sum_{\substack{\sigma,\nu\\k,p,q}}\sum_{l,r,s}
    \hat{v}_k\eta_l\delta_{p-k,r-l}\delta_{q+k,s+l}
    a_{r,\sigma}a^*_{q,\nu}a_{s,\nu}a^*_{p,\sigma}\\
    &\quad\quad\times
    \chi_{p-k\notin B_F^\sigma}\chi_{p\in B_F^\sigma}\chi_{q+k,q\notin B_F^\nu}
  \chi_{r\in B_F^\sigma}\chi_{s\in A_{F,\kappa}^\nu}+h.c.\\
   &\Upsilon_{1,22}=\sum_{\substack{\sigma,\nu\\k,p,q}}\sum_{l,r,s}
    \hat{v}_k\eta_l\delta_{p-k,r-l}\delta_{q+k,s+l}\chi_{k=l}
    a_{r,\sigma}a^*_{q,\nu}a_{s,\nu}a^*_{p,\sigma}\\
    &\quad\quad\times
    \chi_{p-k\notin B_F^\sigma}\chi_{p\in B_F^\sigma}\chi_{q+k,q\notin B_F^\nu}
  \chi_{r\in B_F^\sigma}\chi_{s\in A_{F,\kappa}^\nu}+h.c.
  \end{aligned}
\end{equation}
Obviously,
\begin{equation}\label{Upsilon_122}
  \pm\Upsilon_{1,22}\lesssim a\mathcal{N}_{ex}\mathcal{N}_{l}[\kappa]
  \lesssim a\mathcal{N}_{ex}^{2}.
\end{equation}
Switching to the position space, we have
\begin{equation}\label{Upsilon_121}
  \begin{aligned}
  \Upsilon_{1,21}&=-\sum_{\sigma,\nu}\int_{\Lambda^4}v_a(x_1-x_3)\eta(x_2-x_4)
  \Big(\sum_{p_1\notin B_F^\sigma}e^{ip_1(x_2-x_1)}\Big)
  \Big(\sum_{p_2\notin B_F^\nu}e^{ip_2(x_4-x_3)}\Big)\\
  &\quad\quad\times a(g_{x_2,\sigma})a^*(h_{x_3,\nu})a(L_{x_4,\nu}[\kappa])
  a^*(g_{x_1,\sigma})dx_1dx_2dx_3dx_4+h.c.
  \end{aligned}
\end{equation}
Using
  \begin{equation*}
     \sum_{p_1\notin B_F^\sigma}\sum_{p_2\notin B_F^\nu}=
  \sum_{p_1,p_2}-\sum_{p_1}\sum_{p_2\in B_F^\nu}-\sum_{p_2}\sum_{p_1\in B_F^\sigma}
  +\sum_{p_1\in B_F^\sigma}\sum_{p_2\in B_F^\nu},
  \end{equation*}
  we rewrite $\Upsilon_{1,21}=\sum_{j=1}^{4}\Upsilon_{1,21j}$. Recall (\ref{useful eqn1}), we can bound them by
  \begin{equation}\label{Upsilon121 bound}
    \begin{aligned}
    \pm\Upsilon_{1,211}&\lesssim N^{\frac{5}{6}+\frac{\kappa}{2}}\Vert v_a\Vert_1\Vert\eta\Vert_{\infty}(\mathcal{N}_{ex}+1)
    \lesssim N^{-\frac{1}{6}+\frac{\kappa}{2}}(\mathcal{N}_{ex}+1)\\
    \pm\Upsilon_{1,212},\pm\Upsilon_{1213}&\lesssim N^{\frac{4}{3}+\frac{\kappa}{2}}\Vert v_a\Vert_1\Vert\eta\Vert_{2}(\mathcal{N}_{ex}+1)
    \lesssim N^{-\frac{5}{6}+\frac{\kappa}{2}-\frac{\alpha}{2}}(\mathcal{N}_{ex}+1)\\
    \pm\Upsilon_{1,214}&\lesssim N^{\frac{11}{6}+\frac{\kappa}{2}}\Vert v_a\Vert_1\Vert\eta\Vert_{1}(\mathcal{N}_{ex}+1)
    \lesssim N^{-\frac{5}{6}+\frac{\kappa}{2}-2\alpha}(\mathcal{N}_{ex}+1)
    \end{aligned}
  \end{equation}
  Therefore
  \begin{equation}\label{Upsilon121 bound true}
    \pm\Upsilon_{1,21}\lesssim N^{-\frac{1}{6}+\frac{\kappa}{2}}(\mathcal{N}_{ex}+1).
  \end{equation}
  \par The evaluation of $\Upsilon_2$ goes similarly. We use $1=\delta_{q,s}\delta_{\sigma,\nu}+(1-\delta_{q,s}\delta_{\sigma,\nu})$ to rewrite $\Upsilon_2=\Upsilon_{21}+\Upsilon_{22}$. Similar to the analysis of $\Upsilon_{1,13}$ defined in (\ref{Upsilon_11j}), we calculate
  \begin{equation}\label{Upsilon 21}
    \Upsilon_{21}=-2\sum_{\sigma,k,p,q}\hat{v}_{k+p-q}\eta_{k}
    a^*_{p,\sigma}a^*_{q,\sigma}a_{q,\sigma}a_{p,\sigma}\chi_{p\in B_F^\sigma}
    \chi_{q\in A_{F,\kappa}^\sigma}+\mathcal{E}_{\Upsilon_{21}},
  \end{equation}
  where
  \begin{equation}\label{Upsilon 21 error}
    \pm\mathcal{E}_{\Upsilon_{21}}\lesssim N^{-\frac{2}{3}-2\alpha}\mathcal{N}_{ex}.
  \end{equation}
  Notice that for any $p\in B_F^\sigma$ and $q\in A_{F,\kappa}^\nu$, we have $\vert p\vert, \vert q\vert\lesssim N^{\frac{1}{3}}$, and therefore
  \begin{equation}\label{faklvsndfjd}
    \begin{aligned}
    \vert\hat{v}_{k+p-q}-\hat{v}_k\vert\lesssim\int_{B_a}\vert p-q\vert \vert x\vert v_a(x)
    dx\lesssim N^{\frac{1}{3}}a^2.
    \end{aligned}
  \end{equation}
  Combining (\ref{faklvsndfjd}) with Lemma \ref{l1 lemma}, we know that for any $p\in B_F^\sigma$ and $q\in A_{F,\kappa}^\nu$ fixed,
  \begin{equation}\label{fudsfsdkf}
    \Big\vert\sum_{k}(\hat{v}_{k+p-q}-\hat{v}_k)\eta_k\Big\vert\lesssim
    a^2N^{\frac{1}{3}}.
  \end{equation}
  Therefore we can rewrite
  \begin{equation}\label{Upsilon_21 bound}
    \Upsilon_{21}=-2\sum_{k}\hat{v}_k\eta_k\sum_{\sigma}N_\sigma\mathcal{N}_{ex,\sigma}
    +\mathcal{E}_{\Upsilon_{21}},
  \end{equation}
  where
  \begin{equation}\label{Upsilon 21 error true}
    \pm\mathcal{E}_{\Upsilon_{21}}\lesssim N^{-\frac{2}{3}}\mathcal{N}_{ex}
    +N^{-1}\mathcal{N}_{ex}^2+N^{-\frac{1}{3}-\kappa}\mathcal{K}_s.
  \end{equation}
  For the estimate of $\Upsilon_{22}$, it is analogous to the estimates of (\ref{Upsilon_12j}), hence we state the result without further details:
  \begin{equation}\label{Upsilon_22}
    \pm\Upsilon_{22}\lesssim N^{-\frac{1}{6}+\frac{\kappa}{2}}(\mathcal{N}_{ex}+1).
  \end{equation}
  \par Collecting the evaluations above, together with Lemma \ref{control eB' on N_ex lemma} and Corollary \ref{control eB' on K,V_4}, we obtain
  \begin{equation}\label{Upsilon 1and 2}
     \begin{aligned}
    \int_{0}^{1}\int_{t}^{1}e^{-sB^\prime}(\Upsilon_1+\Upsilon_2)e^{sB^\prime}dsdt
    =\sum_{k}\hat{v}_k\eta_k\sum_{\sigma\neq\nu}N_\sigma \mathcal{N}_{ex,\nu}
    +\mathcal{E}_{\Upsilon_{1\&2}}
    \end{aligned}
  \end{equation}
  with
  \begin{equation}\label{Upsilon 1and 2 error}
 \begin{aligned}
    \pm\mathcal{E}_{\Upsilon_{1\&2}}\lesssim& N^{-\frac{1}{6}+\frac{\kappa}{2}}
    (\mathcal{N}_{ex}+1)
    +N^{-1}(\mathcal{N}_{ex}+1)^2\\
    &+N^{-\frac{1}{3}-\kappa}\big(\mathcal{K}_s+N^{\frac{\kappa}{2}}\ln N(\mathcal{N}_{ex}+1)\big).
 \end{aligned}
  \end{equation}
 \begin{flushleft}
  \textbf{Analysis of $\Upsilon_3\&\Upsilon_4$:}
\end{flushleft}
Switching to the position space, we have
\begin{equation}\label{Upsilon 3and4}
  \begin{aligned}
  &\Upsilon_3=\sum_{\sigma,\nu,\varpi}\int_{\Lambda^4}v_a(x_1-x_3)\eta(x_2-x_4)
  \Big(\sum_{p_5\notin B^\nu_F}e^{-ip_5(x_3-x_2)}\Big)a^*(h_{x_1,\sigma})a^*(h_{x_3,\nu})\\
  &\quad\quad\times a^*(h_{x_4,\varpi})
  a(g_{x_1,\sigma})a(L_{x_4,\varpi}[\kappa])a(g_{x_2,\nu})
  dx_1dx_2dx_3dx_4+h.c.\\
  &\Upsilon_4=-\sum_{\sigma,\nu,\tau}\int_{\Lambda^4}v_a(x_1-x_3)\eta(x_2-x_4)
  \Big(\sum_{p_5\notin B^\nu_F}e^{-ip_5(x_3-x_4)}\Big)a^*(h_{x_1,\sigma})a^*(h_{x_3,\nu})\\
  &\quad\quad\times a^*(h_{x_2,\tau})
  a(g_{x_1,\sigma})a(L_{x_4,\nu}[\kappa])a(g_{x_2,\tau})
  dx_1dx_2dx_3dx_4+h.c.
  \end{aligned}
\end{equation}
Using $\sum_{p_5\notin B_F^\nu}=\sum_{p_5}-\sum_{p_5\in B_F^\nu}$, we bound them by
\begin{equation}\label{Upsilon3and 4bound}
\begin{aligned}
  \pm\Upsilon_3,\pm\Upsilon_4&\lesssim \big(N^{\frac{4}{3}+\frac{\kappa}{2}}
  \Vert v_a\Vert_1^\frac{1}{2}\Vert\eta\Vert_2+N^{\frac{11}{6}+\frac{\kappa}{2}}
  \Vert v_a\Vert_1^\frac{1}{2}\Vert\eta\Vert_1\big)(\mathcal{V}_4+\mathcal{N}_{ex}+1)\\
  &\lesssim N^{-\frac{1}{3}-\frac{\alpha}{2}+\frac{\kappa}{2}}
  (\mathcal{V}_4+\mathcal{N}_{ex}+1).
\end{aligned}
\end{equation}

\begin{flushleft}
  \textbf{Analysis of $\Upsilon_5\&\Upsilon_6$:}
\end{flushleft}
Switching to the position space, we have
\begin{equation}\label{Upsilon5and6}
  \begin{aligned}
  &\Upsilon_5=\sum_{\sigma,\nu,\tau}\int_{\Lambda^4}v_a(x_1-x_3)\eta(x_2-x_4)
  \Big(\sum_{p_2\in A_{F,\kappa}^\nu}e^{-ip_2(x_4-x_3)}\Big)
  a^*(h_{x_2,\tau})a^*(h_{x_4,\nu})\\
  &\quad\quad\times
  a^*(h_{x_1,\sigma})a(g_{x_2,\tau})a(h_{x_3,\nu})a(g_{x_1,\sigma})
  dx_1dx_2dx_3dx_4+h.c.\\
  &\Upsilon_6=-\sum_{\sigma,\nu,\tau}\int_{\Lambda^4}v_a(x_1-x_3)\eta(x_2-x_4)
  \Big(\sum_{p_1\in A_{F,\kappa}^\sigma}e^{-ip_2(x_4-x_1)}\Big)
  a^*(h_{x_2,\tau})a^*(h_{x_4,\sigma})\\
  &\quad\quad\times
  a^*(h_{x_3,\nu})a(g_{x_2,\tau})a(h_{x_3,\nu})a(g_{x_1,\sigma})
  dx_1dx_2dx_3dx_4+h.c.
  \end{aligned}
\end{equation}
We bound them by
\begin{equation}\label{Upsilon5and6bound}
  \pm\Upsilon_5,\pm\Upsilon_6\lesssim
  N^{\frac{4}{3}+\frac{\kappa}{2}}
  \Vert v_a\Vert_1\Vert\eta\Vert_2(\mathcal{N}_{ex}+1)^2\lesssim
  N^{-\frac{5}{6}-\frac{\alpha}{2}+\frac{\kappa}{2}}(\mathcal{N}_{ex}+1)^2.
\end{equation}

\begin{flushleft}
  \textbf{Analysis of $\Upsilon_7$\textendash$\Upsilon_{10}$:}
\end{flushleft}
Switching to the position space, we have
\begin{equation}\label{Upsilon7-10}
  \begin{aligned}
  &\Upsilon_7=\sum_{\sigma,\nu,\varpi}\int_{\Lambda^4}v_a(x_1-x_3)\eta(x_2-x_4)
  \Big(\sum_{p_2\notin B_F^\nu}e^{ip_2(x_2-x_3)}\Big)
  a^*(g_{x_1,\sigma})a^*(h_{x_3,\nu})\\
  &\quad\quad\times
  a^*(h_{x_4,\varpi})a(h_{x_1,\sigma})a(L_{x_4,\varpi}[\kappa])a(g_{x_2,\nu})
  dx_1dx_2dx_3dx_4+h.c.\\
  &\Upsilon_8=-\sum_{\sigma,\nu,\varpi}\int_{\Lambda^4}v_a(x_1-x_3)\eta(x_2-x_4)
  \Big(\sum_{p_1\notin B_F^\sigma}e^{ip_1(x_2-x_1)}\Big)
  a^*(g_{x_1,\sigma})a^*(h_{x_3,\nu})\\
  &\quad\quad\times
  a^*(h_{x_4,\varpi})a(h_{x_3,\nu})a(L_{x_4,\varpi}[\kappa])a(g_{x_2,\sigma})
  dx_1dx_2dx_3dx_4+h.c.\\
  &\Upsilon_9=\sum_{\sigma,\nu,\tau}\int_{\Lambda^4}v_a(x_1-x_3)\eta(x_2-x_4)
  \Big(\sum_{p_1\notin B_F^\sigma}e^{ip_1(x_4-x_1)}\Big)
  a^*(g_{x_1,\sigma})a^*(h_{x_3,\nu})\\
  &\quad\quad\times
  a^*(h_{x_2,\tau})a(h_{x_3,\nu})a(L_{x_4,\sigma}[\kappa])a(g_{x_2,\tau})
  dx_1dx_2dx_3dx_4+h.c.\\
  &\Upsilon_{10}=-\sum_{\sigma,\nu,\tau}\int_{\Lambda^4}v_a(x_1-x_3)\eta(x_2-x_4)
  \Big(\sum_{p_2\notin B_F^\nu}e^{ip_2(x_4-x_3)}\Big)
  a^*(g_{x_1,\sigma})a^*(h_{x_3,\nu})\\
  &\quad\quad\times
  a^*(h_{x_2,\tau})a(h_{x_1,\sigma})a(L_{x_4,\nu}[\kappa])a(g_{x_2,\tau})
  dx_1dx_2dx_3dx_4+h.c.
  \end{aligned}
\end{equation}
Using $\sum_{p_i\notin B_F^\sharp}=\sum_{p_i}-\sum_{p_i\in B_F^\sharp}$ for the suitable subscript $i$ and spin $\sharp$, we can bound for $j=7,8,9,10$:
\begin{equation}\label{Upsilon7-10bound}
  \begin{aligned}
  \Upsilon_j&\lesssim\big(N^{\frac{4}{3}+\frac{\kappa}{2}}
  \Vert v_a\Vert_1\Vert\eta\Vert_2+N^{\frac{11}{6}+\frac{\kappa}{2}}
  \Vert v_a\Vert_1\Vert\eta\Vert_1\big)(\mathcal{N}_{ex}+1)^{\frac{3}{2}}\\
  &\lesssim N^{-\frac{5}{6}-\frac{\alpha}{2}+\frac{\kappa}{2}}
  (\mathcal{N}_{ex}+1)^{\frac{3}{2}}.
  \end{aligned}
\end{equation}

\begin{flushleft}
  \textbf{Analysis of $\Upsilon_{11}$:}
\end{flushleft}
Switching to the position space, we have
\begin{equation}\label{Upsilon11}
  \begin{aligned}
  &\Upsilon_{11}=\sum_{\sigma,\nu,\varpi}\int_{\Lambda^4}v_a(x_1-x_3)\eta(x_2-x_4)
  \Big(\sum_{q_1\in B_F^\sigma}e^{iq_1(x_1-x_2)}\Big)
  a^*(h_{x_2,\sigma})a^*(h_{x_4,\varpi})\\
  &\quad\quad\times a^*(h_{x_3,\nu})
  a(L_{x_4,\varpi}[\kappa])a(h_{x_3,\nu})a(h_{x_1,\sigma})
  dx_1dx_2dx_3dx_4+h.c.
  \end{aligned}
\end{equation}
We bound it by
\begin{equation}\label{Upsilon11bound}
  \begin{aligned}
  \Upsilon_{11}&\lesssim N^{\frac{5}{6}+\frac{\kappa}{2}}\Vert v_a\Vert_1^\frac{1}{2}
  \Vert\eta\Vert_2\big(\theta\mathcal{V}_4+\theta^{-1}
  (\mathcal{N}_{ex}+1)^3\big)\\
  &\lesssim N^{-\frac{1}{6}-\frac{\alpha}{2}+\frac{\kappa}{2}}\mathcal{V}_4
  +N^{-\frac{9}{6}-\frac{\alpha}{2}+\frac{\kappa}{2}}(\mathcal{N}_{ex}+1)^3,
  \end{aligned}
\end{equation}
where we choose $\theta=N^{\frac{2}{3}}$.

\begin{flushleft}
  \textbf{Analysis of $\Upsilon_{12}$:}
\end{flushleft}
Switching to the position space, we have
\begin{equation}\label{Upsilon12}
  \begin{aligned}
  &\Upsilon_{12}=-\sum_{\sigma,\nu,\tau}\int_{\Lambda^4}v_a(x_1-x_3)\eta(x_2-x_4)
  \Big(\sum_{p_5\in A_{F,\kappa}^\nu}e^{ip_5(x_3-x_4)}\Big)
  a^*(h_{x_2,\tau})a^*(h_{x_4,\nu})\\
  &\quad\quad\times
  a(g_{x_2,\tau})a^*(g_{x_1,\sigma})a(h_{x_3,\nu})a(h_{x_1,\sigma})
  dx_1dx_2dx_3dx_4+h.c.
  \end{aligned}
\end{equation}
We bound it by
\begin{equation}\label{Upsilon12bound}
\begin{aligned}
  \pm\Upsilon_{12}&\lesssim N^{\frac{4}{3}+\frac{\kappa}{2}}
  \Vert v_a\Vert_1^{\frac{1}{2}}\Vert\eta\Vert_2
  \big(\theta\mathcal{V}_4+\theta^{-1}
  (\mathcal{N}_{ex}+1)^2\big)\\
  &\lesssim N^{-\gamma}\mathcal{V}_4+N^{-\frac{2}{3}-\alpha+\kappa+\gamma}
  (\mathcal{N}_{ex}+1)^2,
\end{aligned}
\end{equation}
where we choose $\theta=N^{\frac{1}{3}+\frac{\alpha}{2}-\frac{\kappa}{2}-\gamma}$.

\begin{flushleft}
  \textbf{Conclusion of Lemma \ref{cal [V_3,B'] lemma}:}
\end{flushleft}
We close the proof of Lemma \ref{cal [V_3,B'] lemma} by collecting all the above results regarding $\Upsilon_{1}$\textendash$\Upsilon_{12}$.
\end{proof}

\begin{proof}[Proof of Proposition \ref{p2}]
  \par We set $\mu=N^{-\gamma}$ and $\kappa=-2(\alpha+\gamma)$, with the further requirement that $0<\gamma<\alpha<\frac{11}{192}$. Then Proposition \ref{p2} is deduced by combining Lemmas \ref{cal eB' on Q_G_N and E_G_N}-\ref{cal [V_3,B'] lemma} and equation (\ref{define J_N}).
\end{proof}

\section{Bogoliubov Transformation}\label{bog}
\par In this section, we do the Bogoliubov transformation and analyse the excitation Hamiltonian $\mathcal{Z}_N$ defined in (\ref{define Z_N 0}) and prove Proposition \ref{p3}. Recall (\ref{define B tilde 0}-\ref{define xi_k,q,p,nu,sigma 0})
\begin{equation}\label{define B tilde}
  \tilde{B}=\frac{1}{2}(\tilde{A}-\tilde{A}^*)
\end{equation}
with
\begin{equation}\label{define A tilde}
  \tilde{A}=\sum_{k,p,q,\sigma,\nu}
  \xi_{k,q,p}^{\nu,\sigma}a^*_{p-k,\sigma}a^*_{q+k,\nu}a_{q,\nu}a_{p,\sigma}\chi_{p-k\notin B^{\sigma}_F}\chi_{q+k\notin B^{\nu}_F}\chi_{p\in B^{\sigma}_F}\chi_{q\in B^{\nu}_F}.
\end{equation}
and the coefficients as in Section \ref{coeffsub}, we choose
\begin{equation}\label{define xi_k,q,p,nu,sigma}
  \xi_{k,q,p}^{\nu,\sigma}=\frac{-\big(W_k+\eta_kk(q-p)\big)}
  {\frac{1}{2}\big(\vert q+k\vert^2+\vert p-k\vert^2-\vert q\vert^2-\vert p\vert^2\big)}
  \chi_{p-k\notin B^{\sigma}_F}\chi_{q+k\notin B^{\nu}_F}\chi_{p\in B^{\sigma}_F}\chi_{q\in B^{\nu}_F}.
\end{equation}
For $p\in B_F^\sigma$ and $q\in B_F^\nu$, we let
\begin{equation}\label{define xi}
  \xi_{q,p}^{\nu,\sigma}(x)=\sum_{k} \xi_{k,q,p}^{\nu,\sigma}e^{ikx}.
\end{equation}
Using (\ref{cal J_N p2}) and Newton-Leibniz formula,  we rewrite $\mathcal{Z}_N$ by
\begin{equation}\label{define Z_N}
  \begin{aligned}
  \mathcal{Z}_N\coloneqq e^{-\tilde{B}}\mathcal{J}_Ne^{\tilde{B}}=
  &C_{\mathcal{G}_N}
  +\mathcal{K}+e^{-\tilde{B}}\mathcal{V}_{4}e^{\tilde{B}}
  +e^{-\tilde{B}}\mathcal{E}_{\mathcal{J}_N}e^{\tilde{B}}
  +\int_{0}^{1}e^{-t\tilde{B}}\tilde{\Gamma}e^{t\tilde{B}}dt\\
  &+\int_{0}^{1}\int_{t}^{1}e^{-s\tilde{B}}[\mathcal{V}_{21}^\prime
  +\Omega,\tilde{B}]e^{s\tilde{B}}dsdt
  \end{aligned}
\end{equation}
with
\begin{equation}\label{define Gamme tilde}
  \tilde{\Gamma}=[\mathcal{K},\tilde{B}]+\mathcal{V}_{21}^\prime
  +\Omega.
\end{equation}
We point out that we choose $\xi_{k,q,p}^{\nu,\sigma}$ so that $\tilde{\Gamma}=0$, as we will see that in Lemma \ref{control eBtiled on K_s lemma}. We then prove Proposition \ref{p3} by analysing each terms on the right hand side of (\ref{define Z_N}) up to some error terms that may be discarded in the large $N$ limit. Corresponding results and detailed proofs are collected in Lemmas \ref{control eBtiled on K_s lemma}, \ref{control eBtilde on E_J_N lemma} and \ref{cal com bog lemma}. To control some of these error terms that may not contribute to the ground state energy up to the second order, we first prove the following lemmas which control the action of $e^{\tilde{B}}$ on $\mathcal{N}_{ex}$, $\mathcal{K}_s$ and $\mathcal{V}_4$, and are fairly useful in the up-coming discussion. Different from previous sections, in which we first control the particle number operator, we first establish the estimate controlling the action of $e^{\tilde{B}}$ on the excited momentum energy operator $\mathcal{K}_s$ in Lemma \ref{control eBtiled on K_s lemma}.
\begin{lemma}\label{control eBtiled on K_s lemma}
We have
  \begin{equation}\label{cal Gammatilde}
    \tilde{\Gamma}=0.
  \end{equation}
  Hence for $\vert t\vert\leq 1$,
  \begin{equation}\label{control eBtiled on K_s}
    e^{-t\tilde{B}}\mathcal{K}_se^{t\tilde{B}}\lesssim \mathcal{K}_s+N^{\frac{1}{3}+\alpha}.
  \end{equation}
\end{lemma}
\begin{proof}
  \par By Lemma \ref{lem commutator} and the definition of $\xi_{k,q,p}^{\nu,\sigma}$ in (\ref{define xi_k,q,p,nu,sigma}), we immediately know that
  \begin{equation}\label{trump}
    [\mathcal{K},\tilde{B}]=\frac{1}{2}\big([\mathcal{K},\tilde{A}]
    +[\mathcal{K},\tilde{A}]^*\big)=-\mathcal{V}_{21}^\prime-\Omega.
  \end{equation}
  Therefore we have $\tilde{\Gamma}=0$.
  \par Moreover, by (\ref{K_s}), we have $[\mathcal{K}_s,\tilde{B}]=[\mathcal{K},\tilde{B}]$. Therefore, using (\ref{est V_21'}) and (\ref{est Omega}) in Lemma \ref{contol V_21'Omega [K,B]} with $\beta=0$ (recall that we have assumed $0<\alpha<\frac{11}{192}$ in Propositon \ref{p3}), together with Gronwall's inequality, we reach (\ref{control eBtiled on K_s}).

\end{proof}

\par Lemma \ref{control eBtilde on N_ex lemma} controls the action of $e^{\tilde{B}}$ on $\mathcal{N}_{ex}$:
\begin{lemma}\label{control eBtilde on N_ex lemma}
For $\vert t\vert\leq 1$, we have
  \begin{equation}\label{control eBtilde on N_ex}
    e^{-t\tilde{B}}(\mathcal{N}_{ex}+1)e^{t\tilde{B}}
    \lesssim (\mathcal{N}_{ex}+1)+N^{-\frac{2}{33}+\varepsilon}
    \big(\mathcal{K}_s+N^{\frac{1}{3}+\alpha}\big)
  \end{equation}
  for some $\varepsilon>0$ small enough but fixed.
\end{lemma}
\begin{proof}
  \par By Lemma \ref{lem commutator}, we know that $[\mathcal{N}_{ex},\tilde{B}]=\tilde{A}+\tilde{A}^*$. In order to prove (\ref{control eBtilde on N_ex}), we would like to prove
  \begin{equation}\label{aim}
    \pm(\tilde{A}+\tilde{A}^*)\lesssim (\mathcal{N}_{ex}+1)+N^{-\frac{2}{33}
    +\varepsilon}\mathcal{K}_s
  \end{equation}
  for some $\varepsilon>0$ small enough but fixed. Then (\ref{control eBtilde on N_ex}) follows by Gronwall's inequality and Lemma \ref{control eBtiled on K_s lemma}.
  \par Let $0< d<\frac{1}{3}$ to be determined later. We decompose $[\mathcal{N}_{ex},\tilde{B}]$ by
  \begin{equation}\label{decompose Atilde}
    \begin{aligned}
    \Psi_1&=\sum_{k,p,q,\sigma,\nu}
  \xi_{k,q,p}^{\nu,\sigma}a^*_{p-k,\sigma}a^*_{q+k,\nu}a_{q,\nu}a_{p,\sigma}\chi_{p-k\notin B^{\sigma}_F}\chi_{q+k\notin B^{\nu}_F}\chi_{p\in \underline{A}^{\sigma}_{F,d}}\chi_{q\in \underline{A}^{\nu}_{F,d}}+h.c.\\
  \Psi_2&=\sum_{k,p,q,\sigma,\nu}
  \xi_{k,q,p}^{\nu,\sigma}a^*_{p-k,\sigma}a^*_{q+k,\nu}a_{q,\nu}a_{p,\sigma}\chi_{p-k\notin B^{\sigma}_F}\chi_{q+k\notin B^{\nu}_F}\chi_{p\in \underline{A}^{\sigma}_{F,d}}\chi_{q\in \underline{B}^{\nu}_{F,d}}+h.c.\\
  \Psi_3&=\sum_{k,p,q,\sigma,\nu}
  \xi_{k,q,p}^{\nu,\sigma}a^*_{p-k,\sigma}a^*_{q+k,\nu}a_{q,\nu}a_{p,\sigma}\chi_{p-k\notin B^{\sigma}_F}\chi_{q+k\notin B^{\nu}_F}\chi_{p\in \underline{B}^{\sigma}_{F,d}}\chi_{q\in \underline{A}^{\nu}_{F,d}}+h.c.\\
  \Psi_4&=\sum_{k,p,q,\sigma,\nu}
  \xi_{k,q,p}^{\nu,\sigma}a^*_{p-k,\sigma}a^*_{q+k,\nu}a_{q,\nu}a_{p,\sigma}\chi_{p-k\notin B^{\sigma}_F}\chi_{q+k\notin B^{\nu}_F}\chi_{p\in \underline{B}^{\sigma}_{F,d}}\chi_{q\in \underline{B}^{\nu}_{F,d}}+h.c.
    \end{aligned}
  \end{equation}
  To bound $\Psi_1$, we adopt the notation
  \begin{equation*}
    k_{x_1,x_2,\nu}^{(q_1,\sigma)}(z)\coloneqq
    \sum_{q_2\in\underline{A}_{F,d}^\nu}\overline{\xi_{q_2,q_1}^{\nu,\sigma}}(x_1-x_2)
    e^{iq_2x_2}f_{q_2,\nu}(z).
  \end{equation*}
  It is easy to verify that, for $j=1,2$
  \begin{equation*}
    \int_{\Lambda}\Vert a(k_{x_1,x_2,\nu}^{(q_1,\sigma)})\Vert^2dx_j\leq
    \int_{\Lambda}\Vert k_{x_1,x_2,\nu}^{(q_1,\sigma)}\Vert^2dx_j
    \leq \sum_{q_2\in\underline{A}_{F,d}^\nu}\Vert\xi_{q_2,q_1}^{\nu,\sigma}\Vert^2
    =\sum_{q_2\in\underline{A}_{F,d}^\nu}\sum_{k}\vert\xi_{k,q_2,q_1}^{\nu,\sigma}\vert^2.
  \end{equation*}
  Thus we can write $\Psi_1$ by
  \begin{equation}\label{Psi_1}
    \begin{aligned}
    \Psi_1&=\sum_{\sigma,\nu,q_1}\chi_{q_1\in\underline{A}_{F,d}^\sigma}\int_{\Lambda^2}
    e^{-iq_1x}a^*(h_{x,\sigma})a^*(h_{y,\nu})a(k_{x,y,\nu}^{(q_1,\sigma)})a_{q_1,\sigma}dxdy+h.c.
    \end{aligned}
  \end{equation}
  Let $-\frac{11}{48}<\delta\leq\frac{1}{3}$ to be determined, we continue by decomposing $\Psi_1=\Psi_{11}+\Psi_{12}$ using the relation $a^*(h_{y,\nu})=a^*(H_{y,\nu}[\delta])+a^*(L_{y,\nu}[\delta])$. For $\psi\in\mathcal{H}^{\wedge N}(\{N_{\varsigma_i}\})$, we bound $\Psi_{11}$ with the help of (\ref{ineqn N_l and N_s}) in Lemma \ref{lemma ineqn K_s} and (\ref{xi l2}) in Lemma \ref{xi l2 lemma}:
  \begin{equation}\label{bound Psi_11}
    \begin{aligned}
      \vert\langle\Psi_{11}\psi,\psi\rangle\vert&\lesssim \sum_{\sigma,\nu,q_1}\chi_{q_1\in\underline{A}_{F,d}^\sigma}\int_{\Lambda^2}
    \Vert (\mathcal{N}_{ex}+2)^{-\frac{1}{2}}a(h_{x,\sigma})a(H_{y,\nu}[\delta])\psi\Vert\\
    &\quad\quad\quad\quad\quad\quad\quad\quad\times
    \Vert (\mathcal{N}_{ex}+2)^{\frac{1}{2}}a(k_{x,y,\nu}^{(q_1,\sigma)})a_{q_1,\sigma}\psi
    \Vert dxdy\\
    &\lesssim\Big( \sum_{\sigma,q_1}\chi_{q_1\in\underline{A}_{F,d}^\sigma}
    \langle\mathcal{N}_h[\delta]\psi,\psi\rangle\Big)^{\frac{1}{2}}\\
    &\quad\quad\times
    \Big(\sum_{\substack{\sigma,\nu\\k,q_1,q_2}}\chi_{q_1\in\underline{A}_{F,d}^\sigma}
    \chi_{q_2\in\underline{A}_{F,d}^\nu}\vert\xi_{k,q_2,q_1}^{\nu,\sigma}\vert^2
    \langle(\mathcal{N}_{ex}+1)\psi,\psi\rangle\Big)^{\frac{1}{2}}\\
    &\lesssim \langle(\mathcal{N}_{ex}+1)\psi,\psi\rangle
    +N^{-\frac{1}{3}+2d-\delta+\varepsilon}
    \langle\mathcal{K}_s\psi,\psi\rangle.
    \end{aligned}
  \end{equation}
Using (\ref{ineqn N_l and N_s}) in Lemma \ref{lemma ineqn K_s} and (\ref{xi l2}) in Lemma \ref{xi l2 lemma}, we also bound $\Psi_{12}$ by
\begin{equation}\label{Psi_12 bound}
  \begin{aligned}
   \vert\langle\Psi_{12}\psi,\psi\rangle\vert&\lesssim \sum_{\sigma,\nu,q_1}\chi_{q_1\in\underline{A}_{F,d}^\sigma}\int_{\Lambda^2}
    \Vert (\mathcal{N}_{ex}+2)^{-\frac{1}{2}}a(h_{x,\sigma})a_{q_1,\sigma}^*\psi\Vert\\
    &\quad\quad\quad\quad\quad\quad\quad\quad\times
    \Vert (\mathcal{N}_{ex}+2)^{\frac{1}{2}}a^*(L_{y,\nu}[\delta])
    a(k_{x,y,\nu}^{(q_1,\sigma)})\psi
    \Vert dxdy\\
    &\lesssim\Big( \sum_{\sigma,q_1}
    \langle a_{q_1,\sigma}a_{q_1,\sigma}^*\psi,\psi\rangle\Big)^{\frac{1}{2}}\\
    &\quad\quad\times
    \Big(N^{\frac{2}{3}+\delta}
    \sum_{\substack{\sigma,\nu\\k,q_1,q_2}}\chi_{q_1\in\underline{A}_{F,d}^\sigma}
    \chi_{q_2\in\underline{A}_{F,d}^\nu}\vert\xi_{k,q_2,q_1}^{\nu,\sigma}\vert^2
    \langle(\mathcal{N}_{ex}+1)\psi,\psi\rangle\Big)^{\frac{1}{2}}\\
    &\lesssim N^{\frac{\delta}{2}+\frac{d}{2}+\frac{\varepsilon}{2}}
    \langle(\mathcal{N}_{ex}+1)\psi,\psi\rangle.
  \end{aligned}
\end{equation}
With (\ref{bound Psi_11}) and (\ref{Psi_12 bound}), we choose $\delta=-d-\varepsilon$, and we conclude that
\begin{equation}\label{Psi_1 bound}
  \pm\Psi_1\lesssim (\mathcal{N}_{ex}+1)+N^{-\frac{1}{3}+3d+2\varepsilon}\mathcal{K}_s.
\end{equation}
\par The estimates of $\Psi_j$ for $j=2,3,4$ are similar, we take $\Psi_4$ for instance. We adopt the notation
\begin{equation*}
    \tilde{k}_{x_1,x_2,\nu}^{(q_1,\sigma)}(z)\coloneqq
    \sum_{q_2\in\underline{B}_{F,d}^\nu}\overline{\xi_{q_2,q_1}^{\nu,\sigma}}(x_1-x_2)
    e^{iq_2x_2}f_{q_2,\nu}(z).
  \end{equation*}
  It is easy to verify for $j=1,2$
  \begin{equation*}
    \int_{\Lambda}\Vert a(\tilde{k}_{x_1,x_2,\nu}^{(q_1,\sigma)})\Vert^2dx_j\leq
    \int_{\Lambda}\Vert \tilde{k}_{x_1,x_2,\nu}^{(q_1,\sigma)}\Vert^2dx_j
    \leq \sum_{q_2\in\underline{B}_{F,d}^\nu}\Vert\xi_{q_2,q_1}^{\nu,\sigma}\Vert^2
    =\sum_{q_2\in\underline{B}_{F,d}^\nu}\sum_{k}\vert\xi_{k,q_2,q_1}^{\nu,\sigma}\vert^2.
  \end{equation*}
  Switching to the position space, we have
  \begin{equation}\label{Psi_4}
    \begin{aligned}
    \Psi_4&=\sum_{\sigma,\nu,q_1}\chi_{q_1\in\underline{B}_{F,d}^\sigma}\int_{\Lambda^2}
    e^{-iq_1x}a^*(h_{x,\sigma})a^*(h_{y,\nu})
    a(\tilde{k}_{x,y,\nu}^{(q_1,\sigma)})a_{q_1,\sigma}dxdy+h.c.
    \end{aligned}
  \end{equation}
Let $-\frac{11}{48}<\delta^\prime\leq\frac{1}{3}$ to be determined, we decompose
$\Psi_4=\Psi_{41}+\Psi_{42}$ using $a^*(h_{y,\nu})=a^*(H_{y,\nu}[\delta^\prime])+a^*(L_{y,\nu}[\delta^\prime])$. We bound $\Psi_{41}$ via (\ref{xi l2 cut off}) in Lemma \ref{xi l2 lemma}:
 \begin{equation}\label{bound Psi_41}
    \begin{aligned}
      \vert\langle\Psi_{41}\psi,\psi\rangle\vert&\lesssim \sum_{\sigma,\nu,q_1}\chi_{q_1\in\underline{B}_{F,d}^\sigma}\int_{\Lambda^2}
      \Vert (\mathcal{N}_{ex}+3)^{\frac{1}{2}}(\mathcal{N}_{h}[\delta^\prime]+3)^{\frac{1}{2}}
    a(k_{x,y,\nu}^{(q_1,\sigma)})\psi\Vert\\
    &\quad\quad\times
     \Vert
     (\mathcal{N}_{ex}+3)^{-\frac{1}{2}}(\mathcal{N}_{h}[\delta^\prime]+3)^{-\frac{1}{2}}
    a(H_{y,\nu}[\delta^\prime])a(h_{x,\sigma})a_{q_1,\sigma}^*\psi\Vert
     dxdy\\
      &\lesssim\Big(\sum_{\substack{\sigma,\nu\\k,q_1,q_2}}\frac{
      \chi_{q_1\in\underline{B}_{F,d}^\sigma}
    \chi_{q_2\in\underline{B}_{F,d}^\nu}\vert\xi_{k,q_2,q_1}^{\nu,\sigma}\vert^2}
    {\big(\vert k_F^\sigma\vert^2-\vert q_1\vert^2\big)}
    \langle(\mathcal{N}_{ex}+1)
    (\mathcal{N}_{h}[\delta^\prime]+1)\psi,\psi\rangle\Big)^{\frac{1}{2}}\\
    &\quad\quad\times
    \Big( \sum_{\sigma,q_1}\big(\vert k_F^\sigma\vert^2-\vert q_1\vert^2\big)
    \langle a_{q_1,\sigma}a^*_{q_1,\sigma}\psi,\psi\rangle\Big)^{\frac{1}{2}}
    \\
    &\lesssim N^{-\frac{2}{3}-\frac{d}{2}+\frac{\delta^\prime}{2}}
    \langle(\mathcal{N}_{ex}+1)\psi,\psi\rangle
    +N^{-\frac{\delta^\prime}{2}-\frac{d}{2}}
    \langle\mathcal{K}_s\psi,\psi\rangle.
    \end{aligned}
  \end{equation}
 We also use (\ref{xi l2 cut off}) in Lemma \ref{xi l2 lemma} to bound $\Psi_{42}$ by
 \begin{equation}\label{Psi_42 bound}
   \begin{aligned}
   \vert\langle\Psi_{42}\psi,\psi\rangle\vert&\lesssim \sum_{\sigma,\nu,q_1}\chi_{q_1\in\underline{B}_{F,d}^\sigma}\int_{\Lambda^2}
      \Vert (\mathcal{N}_{ex}+2)^{\frac{1}{2}}a^*(L_{y,\nu}[\delta^\prime])
    a(k_{x,y,\nu}^{(q_1,\sigma)})\psi\Vert\\
    &\quad\quad\times
     \Vert
     (\mathcal{N}_{ex}+2)^{-\frac{1}{2}}
    a(h_{x,\sigma})a_{q_1,\sigma}^*\psi\Vert
     dxdy\\
      &\lesssim\Big(N^{\frac{2}{3}+\delta^\prime}
      \sum_{\substack{\sigma,\nu\\k,q_1,q_2}}\frac{
      \chi_{q_1\in\underline{B}_{F,d}^\sigma}
    \chi_{q_2\in\underline{B}_{F,d}^\nu}\vert\xi_{k,q_2,q_1}^{\nu,\sigma}\vert^2}
    {\big(\vert k_F^\sigma\vert^2-\vert q_1\vert^2\big)}
    \langle(\mathcal{N}_{ex}+1)
    \psi,\psi\rangle\Big)^{\frac{1}{2}}\\
    &\quad\quad\times
    \Big( \sum_{\sigma,q_1}\big(\vert k_F^\sigma\vert^2-\vert q_1\vert^2\big)
    \langle a_{q_1,\sigma}a^*_{q_1,\sigma}\psi,\psi\rangle\Big)^{\frac{1}{2}}
    \\
    &\lesssim
    \langle(\mathcal{N}_{ex}+1)\psi,\psi\rangle
    +N^{{\delta^\prime}-d}
    \langle\mathcal{K}_s\psi,\psi\rangle.
   \end{aligned}
 \end{equation}
 We optimize (\ref{bound Psi_41}) and  (\ref{Psi_42 bound}) by choosing $\delta^\prime=\frac{d}{3}$, and we conclude that
 \begin{equation}\label{Psi_4 bound}
   \pm\Psi_4\lesssim (\mathcal{N}_{ex}+1)+N^{-\frac{2}{3}d}\mathcal{K}_{s}.
 \end{equation}
 We optimize (\ref{Psi_1 bound}) and (\ref{Psi_4 bound}) by choosing $d=\frac{1}{11}$ and set $\varepsilon^\prime=2\varepsilon>0$ small enough but fixed. It is easy to verify $-\frac{11}{48}<\delta,\delta^\prime<\frac{1}{3}$. Then we reach (\ref{aim}) and hence (\ref{control eBtilde on N_ex}).

 %
\end{proof}

\par Though Lemma \ref{control eBtilde on V_4 lemma} is a rough estimate of the action of $e^{\tilde{B}}$ on $\mathcal{V}_4$, It is useful in the proof of the energy upper bound in Section \ref{proof core}.
\begin{lemma}\label{control eBtilde on V_4 lemma}
We have
\begin{equation}\label{cal [V_4,Btilde]}
\begin{aligned}
  \pm[\mathcal{V}_4,\tilde{B}]\lesssim& \mathcal{V}_4+N^{-\frac{4}{3}+\varepsilon}
  (\mathcal{N}_{ex}+1)^3\\
  &+\big(N^{-\frac{1}{3}+\varepsilon}+N^{-\frac{1}{3}+2\alpha}\big)(\mathcal{N}_{ex}+1)
\end{aligned}
\end{equation}
for any $\varepsilon>0$ small enough but fixed. Hence for $\vert t\vert\leq1$,
  \begin{equation}\label{control eBtilde on V_4}
  \begin{aligned}
    e^{-t\tilde{B}}\mathcal{V}_4e^{t\tilde{B}}
    \lesssim& \mathcal{V}_4+N^{-\frac{4}{3}+\varepsilon}\int_{0}^{t}
  e^{-s\tilde{B}}(\mathcal{N}_{ex}+1)^3e^{s\tilde{B}}ds\\
    &+\big(N^{-\frac{1}{3}+\varepsilon}+N^{-\frac{1}{3}+2\alpha}\big)(\mathcal{N}_{ex}+1)\\
    &+N^{-\frac{2}{33}+\varepsilon}
    \big(N^{-\frac{1}{3}+\varepsilon}+N^{-\frac{1}{3}+2\alpha}\big)
    \big(\mathcal{K}_s+N^{\frac{1}{3}+\alpha}\big).
  \end{aligned}
  \end{equation}
\end{lemma}
\begin{proof}
  \par By Lemma \ref{lem commutator}, we have
  \begin{equation}\label{com V_4 eBtilde}
    [\mathcal{V}_4,\tilde{B}]=\tilde{\Theta}_1+\tilde{\Theta}_2
  \end{equation}
  with
  \begin{equation}\label{Thetatilde}
   \begin{aligned}
    \tilde{\Theta}_1&=\frac{1}{2}\sum_{l,k,p,q,\sigma,\nu}\hat{v}_{k}\xi_{l,s,r}^{\nu,\sigma}
    (a^*_{p-k,\sigma}a^*_{q+k,\nu}a_{s,\nu}a_{r,\sigma}+h.c.)\\
    &\quad\quad\times
    \chi_{p-k,p\notin B^{\sigma}_F}\chi_{q+k,q\notin B^{\nu}_F}
    \chi_{r\in B_F^\sigma}\chi_{s\in B_F^\nu}
    \delta_{p,r-l}\delta_{q,s+l}
    \\
    \tilde{\Theta}_2&=-\sum_{k,p,q,l,r,s,\sigma,\nu,\varpi}
    \hat{v}_{k}\xi_{l,s,r}^{\varpi,\sigma}
    (a^*_{p-k,\sigma}a^*_{q+k,\nu}a^*_{s+l,\varpi}a_{q,\nu}
    a_{s,\varpi}a_{r,\sigma}+h.c.)\\
    &\quad\quad\times
    \chi_{p-k,p\notin B^{\sigma}_F}\chi_{q+k,q\notin B^{\nu}_F}
    \chi_{r\in B_F^\sigma}
    \chi_{s+l\notin B_F^\varpi}\chi_{s\in B_F^\varpi}\delta_{p,r-l}
    \end{aligned}
  \end{equation}
  Switching to the position space, we have
  \begin{equation}\label{Thetatilde1}
    \begin{aligned}
    \tilde{\Theta}_1&=\frac{1}{2}\sum_{\substack{\sigma,\nu\\p_3,p_4,q_3,q_4}}
    \chi_{p_3\notin B_F^\sigma}\chi_{p_4\notin B_F^\nu}\chi_{q_3\in B_F^\sigma}
    \chi_{q_4\in B_F^\nu}
    \int_{\Lambda^3}v_a(x_1-x_2) \xi_{q_3-p_3,q_4,q_3}^{\nu,\sigma}\\
    &\quad\quad\times
   e^{ip_3(x_3-x_1)}e^{ip_4(x_3-x_2)}
    e^{-i(q_3+q_4)x_3}
    a^*(h_{x_1,\sigma})a^*(h_{x_2,\nu})a_{q_4,\nu}a_{q_3,\sigma}+h.c.
    \end{aligned}
  \end{equation}
  We adopt the notation
  \begin{equation}\label{notation}
    k_{x_1,x_3,\sigma}^{(q_4,\nu)}(z)\coloneqq
    \sum_{q_3\in B_F^\sigma}\Big(\sum_{p_3\notin B_F^\sigma}
    \xi_{q_3-p_3,q_4,q_3}^{\nu,\sigma}e^{-ip_3(x_3-x_1)}\Big)e^{iq_3x_3}f_{q_3,\sigma}(z).
  \end{equation}
  By (\ref{xi l2}) in Lemma \ref{xi l2 lemma}, we can check that for $j=1,3$
  \begin{equation}\label{12thetatilde}
    \begin{aligned}
    \int_{\Lambda}\Vert a(k_{x_1,x_3,\sigma}^{(q_4,\nu)})\Vert^2dx_j
    &\leq \sum_{q_3\in B_F^\sigma}\int_{\Lambda}\Big\vert
    \sum_{p_3\notin B_F^\sigma}
    \xi_{q_3-p_3,q_4,q_3}^{\nu,\sigma}e^{-ip_3(x_3-x_1)}\Big\vert^2dx_j\\
    &\leq \sum_{q_3\in B_F^\sigma}\sum_k\vert \xi_{k,q_4,q_3}^{\nu,\sigma}\vert^2
    \lesssim N^{-\frac{4}{3}+\varepsilon},
    \end{aligned}
  \end{equation}
  for $\varepsilon>0$ small but fixed. On the other hand, by (\ref{L1 Linfty est}) and (\ref{too naive}) in Lemma \ref{L1 Linfty est lemma} and the definition (\ref{define xi_k,q,p,nu,sigma}), we have
  \begin{equation}\label{11thetatilde}
    \begin{aligned}
    \Vert a(k_{x_1,x_3,\sigma}^{(q_4,\nu)})\Vert^2
    &\leq \sum_{q_3\in B_F^\sigma}\Big\vert
    \sum_{p_3\notin B_F^\sigma}
    \xi_{q_3-p_3,q_4,q_3}^{\nu,\sigma}e^{-ip_3(x_3-x_1)}\Big\vert^2\\
    &\leq  \sum_{q_3\in B_F^\sigma}\Big(
    \sum_{k}\vert
    \xi_{k,q_4,q_3}^{\nu,\sigma}\vert\Big)^2
    \lesssim Na^2\ell^{-2}=N^{-\frac{1}{3}+2\alpha}.
    \end{aligned}
  \end{equation}
  With the notation above, we rewrite $\tilde{\Theta}_1$ as
  \begin{equation}\label{Thetatilde1re}
     \begin{aligned}
    \tilde{\Theta}_1&=\frac{1}{2}\sum_{\substack{\sigma,\nu,q_4}}
    \chi_{q_4\in B_F^\nu}
    \int_{\Lambda^3}v_a(x_1-x_2) \Big(\sum_{p_4\notin B_F^\nu}e^{ip_4(x_3-x_2)}\Big) \\
    &\quad\quad\times
    e^{-iq_4x_3}
    a^*(h_{x_1,\sigma})a^*(h_{x_2,\nu})a_{q_4,\nu}a(k_{x_1,x_3,\sigma}^{(q_4,\nu)})+h.c.
    \end{aligned}
  \end{equation}
  Using $\sum_{p_4\notin B_F^\nu}=\sum_{p_4}-\sum_{p_4\in B_F^\nu}$, we decompose $\tilde{\Theta}_1=\tilde{\Theta}_{11}+\tilde{\Theta}_{12}$. For $\tilde{\Theta}_{11}$, we bound for any $\psi\in\mathcal{H}^{\wedge N}(\{N_{\varsigma_i}\})$:
  \begin{equation}\label{Thetatilde11 bound}
    \begin{aligned}
    \vert\langle\tilde{\Theta}_{11}\psi,\psi\rangle\vert
    &\lesssim\sum_{\substack{\sigma,\nu,q_4}}
    \chi_{q_4\in B_F^\nu}
    \int_{\Lambda^2}v_a(x_1-x_2)
    \Vert (\mathcal{N}_{ex}+3)^{\frac{1}{2}}a(k_{x_1,x_2,\sigma}^{(\nu)})
    \psi\Vert\\
    &\quad\quad\times
    \Vert(\mathcal{N}_{ex}+3)^{-\frac{1}{2}}a^*_{q_4,\nu}
    a(h_{x_2,\nu})a(h_{x_1,\sigma})\psi\Vert dx_1dx_2
    \end{aligned}
  \end{equation}
  Notice that since $a(h_{x_2,\nu})a(h_{x_1,\sigma})\psi\in\mathcal{H}^{\wedge (N-2)}$, and on the subspace of $\mathcal{H}^{\wedge (N-2)}$ generated by $a(h_{x_2,\nu})a(h_{x_1,\sigma})\psi$, it is clear that
  \begin{equation*}
    N_{\nu}-2\leq\mathcal{N}_{\nu}\leq N_{\nu}-1,
  \end{equation*}
  where $\mathcal{N}_{\nu}$ has been defined in (\ref{define op N_sigma}). Hence on the subspace of $\mathcal{H}^{\wedge (N-2)}$ generated by $a(h_{x_2,\nu})a(h_{x_1,\sigma})\psi$, we have
  \begin{equation*}
    \sum_{q_4\in B_F^\nu}a_{q_4,\nu}a^*_{q_4,\nu}
    =N_{\nu}-\sum_{q_4\in B_F^\nu}a^*_{q_4,\nu}a_{q_4,\nu}
    \leq \mathcal{N}_{ex,\nu}+2\leq\mathcal{N}_{ex}+2.
  \end{equation*}
  With the above fact, and recall (\ref{V_4position}) and (\ref{11thetatilde}), we can bound
  \begin{equation}\label{boundtildeTheta11}
    \begin{aligned}
    \vert\langle\tilde{\Theta}_{11}\psi,\psi\rangle\vert
    &\lesssim\Big(N^{-\frac{1}{3}+2\alpha}\sum_{\substack{\sigma,\nu,q_4}}
    \chi_{q_4\in B_F^\nu}
    \int_{\Lambda^2}v_a(x_1-x_2)
    \Vert (\mathcal{N}_{ex}+1)^{\frac{1}{2}}\psi\Vert^2\Big)^{\frac{1}{2}}\\
    &\quad\quad\times\Big(\sum_{\sigma,\nu}\int_{\Lambda^2}v_a(x_1-x_2)
    \Vert a(h_{x_2,\nu})a(h_{x_1,\sigma})\psi\Vert^2 \Big)^{\frac{1}{2}}\\
    &\lesssim
    \langle\mathcal{V}_{4}\psi,\psi\rangle
    +N^{\frac{2}{3}+2\alpha}\Vert v_a\Vert_1\langle(\mathcal{N}_{ex}+1)\psi,\psi\rangle\\
    &\lesssim \langle\mathcal{V}_{4}\psi,\psi\rangle
    +N^{-\frac{1}{3}+2\alpha}\langle(\mathcal{N}_{ex}+1)\psi,\psi\rangle.
    \end{aligned}
  \end{equation}
  Similarly, we use (\ref{12thetatilde}) to bound $\tilde{\Theta}_{12}$ by
  \begin{equation}\label{tildeTheta12}
    \begin{aligned}
    \vert\langle\tilde{\Theta}_{12}\psi,\psi\rangle\vert
    &\lesssim\sum_{\substack{\sigma,\nu,q_4}}
    \chi_{q_4\in B_F^\nu}
    \int_{\Lambda^3}v_a(x_1-x_2)
    \Big\vert\sum_{p_4\in B_F^\nu}e^{ip_4(x_3-x_2)}\Big\vert\\
     &\quad\quad\times
     \Vert(\mathcal{N}_{ex}+3)^{-\frac{1}{2}}a^*_{q_4,\nu}
    a(h_{x_2,\nu})a(h_{x_1,\sigma})\psi\Vert
     \\
    &\quad\quad\times
    \Vert (\mathcal{N}_{ex}+3)^{\frac{1}{2}}a(k_{x_1,x_3,\sigma}^{(q_4,\nu)})
    \psi\Vert
     dx_1dx_2\\
     &\lesssim\Big(\sum_{\substack{\sigma,\nu,q_4}}
    \chi_{q_4\in B_F^\nu}
    \int_{\Lambda^3}v_a(x_1-x_2)\Vert a(k_{x_1,x_3,\sigma}^{(q_4,\nu)})\Vert^2
    \Vert (\mathcal{N}_{ex}+1)^{\frac{1}{2}}\psi\Vert^2\Big)^{\frac{1}{2}}\\
    &\quad\quad\times\Big(\sum_{\sigma,\nu}\int_{\Lambda^3}
    \Big\vert\sum_{p_4\in B_F^\nu}e^{ip_4(x_3-x_2)}\Big\vert^2
    v_a(x_1-x_2)
    \Vert a(h_{x_2,\nu})a(h_{x_1,\sigma})\psi\Vert^2 \Big)^{\frac{1}{2}}\\
    &\lesssim \langle\mathcal{V}_{4}\psi,\psi\rangle
    +N^{-\frac{1}{3}+\varepsilon}\langle(\mathcal{N}_{ex}+1)\psi,\psi\rangle.
    \end{aligned}
  \end{equation}
  For $\tilde{\Theta}_2$, we use the above notation to write
  \begin{equation}\label{Thetatilde2}
     \begin{aligned}
    \tilde{\Theta}_2&=-\sum_{\substack{\sigma,\nu,\varpi,q_4}}
    \chi_{q_4\in B_F^\varpi}
    \int_{\Lambda^3}v_a(x_1-x_2) e^{-iq_4x_3}
    a^*(h_{x_1,\sigma})a^*(h_{x_2,\nu})a^*(h_{x_3,\varpi}) \\
    &\quad\quad\times
    a(h_{x_2,\nu})a_{q_4,\varpi}a(k_{x_1,x_3,\sigma}^{(q_4,\varpi)})+h.c.
    \end{aligned}
  \end{equation}
  Similar to (\ref{tildeTheta12}), we can bound $\tilde{\Theta}_2$ by
  \begin{equation}\label{tildeTheta2}
    \begin{aligned}
    &\vert\langle\tilde{\Theta}_{2}\psi,\psi\rangle\vert
    \lesssim\sum_{\substack{\sigma,\nu,\varpi,q_4}}
    \chi_{q_4\in B_F^\varpi}
    \int_{\Lambda^3}v_a(x_1-x_2)
    \Vert (\mathcal{N}_{ex}+4)a(k_{x_1,x_3,\sigma}^{(q_4,\varpi)})a(h_{x_2,\nu})
    \psi\Vert\\
     &\quad\quad\times
     \Vert(\mathcal{N}_{ex}+4)^{-1}a(h_{x_3,\varpi})a^*_{q_4,\varpi}
    a(h_{x_2,\nu})a(h_{x_1,\sigma})\psi\Vert
     dx_1dx_2dx_3\\
     &\lesssim\Big(\sum_{\sigma,\nu}\int_{\Lambda^2}
    v_a(x_1-x_2)
    \Vert a(h_{x_2,\nu})a(h_{x_1,\sigma})\psi\Vert^2 \Big)^{\frac{1}{2}}\\
    &\quad\quad\times
    \Big(\sum_{\substack{\sigma,\nu,\varpi,q_4}}
    \chi_{q_4\in B_F^\varpi}
    \int_{\Lambda^3}v_a(x_1-x_2)\Vert a(k_{x_1,x_3,\sigma}^{(q_4,\varpi)})\Vert^2
    \Vert a(h_{x_2,\nu})(\mathcal{N}_{ex}+3)\psi\Vert^2\Big)^{\frac{1}{2}}
    \\
    &\lesssim \langle\mathcal{V}_{4}\psi,\psi\rangle
    +N^{-\frac{4}{3}+\varepsilon}\langle(\mathcal{N}_{ex}+1)^3\psi,\psi\rangle.
    \end{aligned}
  \end{equation}
  Combining (\ref{boundtildeTheta11}), (\ref{tildeTheta12}) and (\ref{tildeTheta2}) we reach (\ref{cal [V_4,Btilde]}), and (\ref{control eBtilde on V_4}) follows via Lemma \ref{control eBtilde on N_ex lemma} and Gronwall's inequality.
\end{proof}

\par As a direct consequence of (\ref{bound E_J_N p2}) together with Lemmas \ref{control eBtiled on K_s lemma} and \ref{control eBtilde on N_ex lemma}, we have
\begin{lemma}\label{control eBtilde on E_J_N lemma}
Under the requirement that $0<\gamma<\alpha<\frac{11}{192}$, and $\varepsilon>0$ small enough
  \begin{equation}\label{control eBtilde on E_J_N}
    \begin{aligned}
    e^{-\tilde{B}}\mathcal{E}_{\mathcal{J}_N}e^{\tilde{B}}
    \lesssim& e^{-\tilde{B}}P_{\mathcal{J}_N}(\mathcal{N}_{ex})e^{\tilde{B}}
    +N^{-\gamma}\big(e^{-\tilde{B}}\mathcal{V}_{4}e^{\tilde{B}}
    +(\mathcal{N}_{ex}+1)\big)\\
    &+\big(N^{-\alpha-\gamma}+N^{-\frac{1}{3}+2\alpha+3\gamma}\big)\mathcal{K}_s
    +N^{\frac{1}{3}-\frac{\alpha}{2}}+N^{\frac{1}{3}-\gamma}+N^{3\alpha+3\gamma}.
    \end{aligned}
  \end{equation}
\end{lemma}

\par Lemma \ref{cal com bog lemma} establishes the evaluation of the last term of (\ref{define Z_N}).
\begin{lemma}\label{cal com bog lemma}
Under the assumption that $0<\gamma<\alpha<\frac{1}{168}$,
  \begin{equation}\label{cal com bog}
  \begin{aligned}
    &\int_{0}^{1}\int_{t}^{1}e^{-s\tilde{B}}[\mathcal{V}_{21}^\prime
  +\Omega,\tilde{B}]e^{s\tilde{B}}dsdt\\
  =&\sum_{k,p,q,\sigma,\nu}\big(W_k+\eta_kk(q-p)\big)\xi_{k,q,p}^{\nu,\sigma}
  \chi_{p-k\notin B^{\sigma}_F}\chi_{q+k\notin B^{\nu}_F}
  \chi_{p\in B^{\sigma}_F}\chi_{q\in B^{\nu}_F}\\
  &-\sum_{k,p,q,\sigma}\big(W_k+\eta_kk(q-p)\big)\xi_{(k+q-p),p,q}^{\sigma,\sigma}
  \chi_{p-k,q+k\notin B^{\sigma}_F}\chi_{p,q\in B^{\sigma}_F}\chi_{p\neq q}\\
  &+\mathcal{E}_{[\mathcal{V}_{21}^\prime
  +\Omega,\tilde{B}]}
  \end{aligned}
  \end{equation}
  where
  \begin{equation}\label{bound E com bog}
    \begin{aligned}
    \pm\mathcal{E}_{[\mathcal{V}_{21}^\prime
  +\Omega,\tilde{B}]}\lesssim
   &\big(N^{-\alpha-\gamma}+
  N^{-\frac{1}{21}-\frac{2}{33}+\frac{33}{7}\alpha+\frac{23}{7}\gamma+\varepsilon}
  \big)
  \big(\mathcal{K}_s+N^{\frac{1}{3}+\alpha}\big)\\
  &N^{-\frac{1}{21}+\frac{33}{7}\alpha+\frac{23}{7}\gamma+\varepsilon}(\mathcal{N}_{ex}+1).
   \end{aligned}
  \end{equation}
  for any $\varepsilon>0$ small enough but fixed.
\end{lemma}
\begin{proof}
  \par We first evaluate $[\mathcal{V}_{21}^\prime,\tilde{B}]$, then $[\Omega,\tilde{B}]$ follows analogously. By Lemma \ref{lem commutator}, We have
  \begin{equation}\label{Phi}
   [\mathcal{V}_{21}^\prime,\tilde{B}]=\sum_{j=1}^{4}\Phi_j
  \end{equation}
  with
  \begin{equation}\label{Phi_j}
    \begin{aligned}
     &\Phi_1=\sum_{\substack{\sigma,\nu\\k,p,q}}\sum_{l,r,s}W_k\xi_{l,s,r}^{\nu,\tau}
    \delta_{p-k,r-l}\delta_{q+k,s+l}a^*_{p,\sigma}a^*_{q,\nu}a_{s,\nu}a_{r,\sigma}\\
    &\quad\quad\quad\quad\times
    \chi_{p-k\notin B_F^\sigma}\chi_{p,r\in B_F^\sigma}\chi_{q+k\notin B_F^\nu}
    \chi_{q,s\in B_F^\nu}+h.c.\\
    &\Phi_2=\sum_{\substack{\sigma,\nu\\k,p,q}}\sum_{\tau,l,r,s}2W_k\xi_{l,s,r}^{\sigma,\tau}
    \delta_{p-k,s+l}a^*_{p,\sigma}a^*_{q,\nu}a^*_{r-l,\tau}
    a_{q+k,\nu}a_{s,\sigma}a_{r,\tau}\\
    &\quad\quad\quad\quad\times
    \chi_{p-k\notin B_F^\sigma}\chi_{p,s\in B_F^\sigma}\chi_{q+k\notin B_F^\nu}
    \chi_{q\in B_F^\nu}\chi_{r-l\notin B^\tau_F}\chi_{r\in B_F^\tau}+h.c.\\
    &\Phi_3=\sum_{\substack{\sigma,\nu\\k,p,q}}\sum_{l,r,s}W_k\xi_{l,s,r}^{\nu,\sigma}
    \delta_{p,r}\delta_{q,s}a^*_{r-l,\sigma}a^*_{s+l,\nu}
    a_{q+k,\nu}a_{p-k,\sigma}\\
    &\quad\quad\quad\quad\times
    \chi_{p-k,r-l\notin B_F^\sigma}\chi_{p\in B_F^\sigma}\chi_{q+k,s+l\notin B_F^\nu}
    \chi_{q\in B_F^\nu}+h.c.\\
    &\Phi_4=\sum_{\substack{\sigma,\nu\\k,p,q}}\sum_{\varpi,l,r,s}-2W_k\xi_{l,s,r}^{\varpi,\sigma}
    \delta_{p,r}a^*_{r-l,\sigma}a^*_{s+l,\varpi}a_{s,\varpi}a^*_{q,\nu}
    a_{q+k,\nu}a_{p-k,\sigma}\\
    &\quad\quad\quad\quad\times
    \chi_{p-k,r-l\notin B_F^\sigma}\chi_{p\in B_F^\sigma}\chi_{q+k\notin B_F^\nu}
    \chi_{q\in B_F^\nu}\chi_{s+l\notin B^\varpi_F}\chi_{s\in B^\varpi_F}+h.c.
    \end{aligned}
  \end{equation}
  Similarly, we have $[\Omega,\tilde{B}]=\sum_{j=1}^{4}\tilde{\Phi}_j$, with $W_k$ in $\Phi_j$ being replaced by $\eta_kk(q-p)$ (they are indeed similar). We then evaluate (\ref{Phi_j}) term by term.
  \begin{flushleft}
    \textbf{Analysis of $\Phi_1$:}
  \end{flushleft}
  Using $1=\chi_{k=l}+\chi_{k\neq l}$, we decompose $\Psi_1=\Psi_{11}+\Psi_{12}$. For $\Psi_{11}$, we furthermore split it using
  \begin{equation*}
  \begin{aligned}
    &a^*_{p,\sigma}a^*_{q,\nu}a_{q,\nu}a_{p,\sigma}\\
    =&1-(a_{p,\sigma}a_{p,\sigma}^*+a_{q,\nu}a^*_{q,\nu})+a_{p,\sigma}a_{q,\nu}
    a^*_{q,\nu}a^*_{p,\sigma}+\delta_{p,q}\delta_{\sigma,\nu}(a_{q,\nu}a^*_{p,\sigma}
    +a_{p,\sigma}a^*_{q,\nu}-1)
  \end{aligned}
  \end{equation*}
  That is, we write $\Phi_{11}=\sum_{j=1}^{4}\Phi_{11j}$ with (notice that $\xi_{k,q,p}^{\nu,\sigma}=\xi_{-k,p,q}^{\sigma,\nu}$):
  \begin{equation}\label{Psi_11j}
    \begin{aligned}
    &\Phi_{111}=\sum_{k,p,q,\sigma,\nu}2W_k\xi_{k,q,p}^{\nu,\sigma}
    \chi_{p-k\notin B^{\sigma}_F}\chi_{q+k\notin B^{\nu}_F}\chi_{p\in B^{\sigma}_F}\chi_{q\in B^{\nu}_F}\\
    &\Phi_{112}=-\sum_{k,p,q,\sigma,\nu}4W_k\xi_{k,q,p}^{\nu,\sigma}
    a_{p,\sigma}a_{p,\sigma}^*
    \chi_{p-k\notin B^{\sigma}_F}\chi_{q+k\notin B^{\nu}_F}\chi_{p\in B^{\sigma}_F}
    \chi_{q\in B^{\nu}_F}\\
    &\Phi_{113}=\sum_{k,p,q,\sigma,\nu}2W_k\xi_{k,q,p}^{\nu,\sigma}a_{p,\sigma}a_{q,\nu}
    a^*_{q,\nu}a^*_{p,\sigma}
    \chi_{p-k\notin B^{\sigma}_F}\chi_{q+k\notin B^{\nu}_F}\chi_{p\in B^{\sigma}_F}
    \chi_{q\in B^{\nu}_F}\\
    &\Phi_{114}=\sum_{k,p,\sigma}2W_k\xi_{k,p,p}^{\sigma,\sigma}(2a_{p,\sigma}a^*_{p,\sigma}-1)
    \chi_{p-k\notin B^{\sigma}_F}\chi_{p\in B^{\sigma}_F}
    \end{aligned}
      \end{equation}
    Using (\ref{energy est bog}) in Lemma \ref{energy est bog lem}, we easily bound for $j=2,3,4$:
    \begin{equation}\label{Phi_11j bound234}
      \pm\Phi_{11j}\lesssim Na^2\ell^{-1}(\mathcal{N}_{ex}+1)\leq N^{-\frac{2}{3}+\alpha}
      (\mathcal{N}_{ex}+1).
    \end{equation}
    \par For $\Phi_{12}$, we similarly decompose it by $\Phi_{12}=\sum_{j=1}^{3}\Phi_{12j}$ with
    \begin{equation}\label{Phi_12j}
      \begin{aligned}
      &\Phi_{121}=-\sum_{k,p,q,\sigma}2W_k\xi_{(k+q-p),p,q}^{\sigma,\sigma}\chi_{p-k,q+k\notin B^{\sigma}_F}\chi_{p,q\in B^{\sigma}_F}\chi_{p\neq q}\\
      &\Phi_{122}=\sum_{k,p,q,\sigma}4W_k\xi_{(k+q-p),p,q}^{\sigma,\sigma}
      a_{p,\sigma}a_{p,\sigma}^*
      \chi_{p-k,q+k\notin B^{\sigma}_F}\chi_{p,q\in B^{\sigma}_F}\chi_{p\neq q}\\
      &\Phi_{123}=\sum_{k,p,q,\sigma,\nu}\sum_{l,r,s}W_k\xi_{l,s,r}^{\nu,\sigma}
      a_{r,\sigma}a_{s,\nu}a^*_{q,\nu}a^*_{p,\sigma}
      \delta_{p-k,r-l}\delta_{q+k,s+l}\chi_{k\neq l}\\
    &\quad\quad\quad\quad\times
    \chi_{p-k\notin B_F^\sigma}\chi_{p,r\in B_F^\sigma}\chi_{q+k\notin B_F^\nu}
    \chi_{q,s\in B_F^\nu}+h.c.
      \end{aligned}
    \end{equation}
    Using (\ref{energy est bog}) in Lemma \ref{energy est bog lem}, we easily bound
    \begin{equation}\label{Phi_122 bound}
      \pm\Phi_{122}\lesssim Na^2\ell^{-1}(\mathcal{N}_{ex}+1)\leq N^{-\frac{2}{3}+\alpha}
      (\mathcal{N}_{ex}+1).
    \end{equation}
    For $\Phi_{123}$, we use$\chi_{k\neq l}=1-\chi_{k=l}$ to further decompose it by $\Phi_{123}=\Phi_{1231}+\Phi_{1232}$. It is easy to find $\Phi_{1232}=-\Phi_{113}$ and hence
    \begin{equation}\label{Phi1232bound}
      \pm\Phi_{1232}\lesssim N^{-\frac{2}{3}+\alpha}(\mathcal{N}_{ex}+1).
    \end{equation}
    For $\Phi_{1231}$, we write it out as
    \begin{equation}\label{Phi_1231}
      \begin{aligned}
      \Phi_{1231}&=\sum_{\substack{\sigma,\nu\\q_3,q_4}}
      \chi_{q_3\in B_F^\sigma}\chi_{q_4\in B_F^\nu}\int_{\Lambda^4}W(x_1-x_3)
      \xi_{q_4,q_3}^{\nu,\sigma}(x_2-x_4)\\
      &\quad\quad\times
      \Big(\sum_{p_1\notin B_F^\sigma}e^{ip_1(x_2-x_1)}\Big)
      \Big(\sum_{p_2\notin B_F^\nu}e^{ip_2(x_4-x_3)}\Big)
      e^{-iq_4x_4}e^{-iq_3x_2}\\
      &\quad\quad\times
      a_{q_3,\sigma}a_{q_4,\nu}a^*(g_{x_3,\nu})a^*(g_{x_1,\sigma})dx_1dx_2dx_3dx_4+h.c.\\
      &=\sum_{\substack{\sigma,\nu\\q_3,q_4}}
      \chi_{q_3\in B_F^\sigma}\chi_{q_4\in B_F^\nu}\int_{\Lambda^2}W(x_1-x_3)\\
      &\quad\quad\times
      \Big(\sum_{p_2\notin B_F^\nu}\xi_{p_2-q_4,q_4,q_3}^{\nu,\sigma}
      e^{ip_2(x_1-x_3)}\Big)e^{-i(q_3+q_4)x_1}\\
      &\quad\quad\times
      a_{q_3,\sigma}a_{q_4,\nu}a^*(g_{x_3,\nu})a^*(g_{x_1,\sigma})dx_1dx_3+h.c.
      \end{aligned}
    \end{equation}
    Using the notation defined in (\ref{notation}) (notice that $\xi_{p_2-q_4,q_4,q_3}^{\nu,\sigma}=\xi_{q_4-p_2,q_3,q_4}^{\sigma,\nu}$), we can write it as
    \begin{equation}\label{Phi_1231re}
      \begin{aligned}
      \Phi_{1231}&=\sum_{\substack{\sigma,\nu,q_3}}
      \chi_{q_3\in B_F^\sigma}\int_{\Lambda^2}W(x_1-x_3)e^{-iq_3x_1}\\
      &\quad\quad\times
      a_{q_3,\sigma}a(k_{x_3,x_1,\nu}^{(q_3,\sigma)})
      a^*(g_{x_3,\nu})a^*(g_{x_1,\sigma})dx_1dx_3+h.c.
      \end{aligned}
    \end{equation}
    Using (\ref{number of exicited particles ops2}) and (\ref{11thetatilde}), we bound it, for $\psi\in\mathcal{H}^{\wedge N}
    (\{N_{\varsigma_i}\})$ that
    \begin{equation}\label{Psi1231bound}
      \begin{aligned}
      \vert\langle\Phi_{1231}\psi,\psi\rangle\vert
      &\lesssim \sum_{\substack{\sigma,\nu,q_3}}
      \chi_{q_3\in B_F^\sigma}\int_{\Lambda^2}W(x_1-x_3)
      \Vert a^*(k_{x_3,x_1,\nu}^{(q_3,\sigma)})a_{q_3,\sigma}^*\psi\Vert\\
      &\quad\quad\times
      \Vert a^*(g_{x_3,\nu})a^*(g_{x_1,\sigma})\psi\Vert dx_1dx_3\\
      &\lesssim N^{\frac{5}{6}+\alpha}\Vert W\Vert_1
      \langle(\mathcal{N}_{ex}+1)\psi,\psi\rangle
      \lesssim N^{-\frac{1}{6}+\alpha}
      \langle(\mathcal{N}_{ex}+1)\psi,\psi\rangle.
      \end{aligned}
    \end{equation}
    \vspace{1em}
    \par For $[\Omega,\tilde{B}]$, $\tilde{\Phi}_1$ can be evaluated analogously. We only need to replace the estimate $\Vert W\Vert_1\lesssim a$ by $N^{\frac{1}{3}}\Vert\nabla\eta\Vert_1\lesssim aN^{-\alpha}$. Collecting all the results above, we conclude that
    \begin{equation}\label{Phi_1 conclusion}
      \begin{aligned}
      &\int_{0}^{1}\int_{t}^{1}e^{-s\tilde{B}}(\Phi_1+\tilde{\Phi}_1)e^{s\tilde{B}}dsdt\\
  =&\sum_{k,p,q,\sigma,\nu}\big(W_k+\eta_kk(q-p)\big)\xi_{k,q,p}^{\nu,\sigma}
  \chi_{p-k\notin B^{\sigma}_F}\chi_{q+k\notin B^{\nu}_F}
  \chi_{p\in B^{\sigma}_F}\chi_{q\in B^{\nu}_F}\\
  &-\sum_{k,p,q,\sigma}\big(W_k+\eta_kk(q-p)\big)\xi_{(k+q-p),p,q}^{\sigma,\sigma}
  \chi_{p-k,q+k\notin B^{\sigma}_F}\chi_{p,q\in B^{\sigma}_F}\chi_{p\neq q}\\
  &+\mathcal{E}_{\Phi_1+\tilde{\Phi}_1},
      \end{aligned}
    \end{equation}
    and by Lemma \ref{control eBtilde on N_ex lemma} we have
    \begin{equation}\label{E_phi_1}
      \pm\mathcal{E}_{\Phi_1+\tilde{\Phi}_1}\lesssim
      N^{-\frac{1}{6}+\alpha}\big((\mathcal{N}_{ex}+1)+N^{-\frac{2}{33}+\varepsilon}
      (\mathcal{K}_s+N^{\frac{1}{3}+\alpha})\big)
    \end{equation}
    for some $\varepsilon>0$ small enough but fixed.

\begin{flushleft}
    \textbf{Analysis of $\Phi_2$:}
\end{flushleft}
Using the notation defined in (\ref{notation}), we write $\Phi_2$ as
\begin{equation}\label{Phi_2}
  \begin{aligned}
  \Phi_2&=\sum_{\substack{\sigma,\nu,\tau\\
  q_3,q_4}}\chi_{q_3\in B_F^\tau}\chi_{q_4\in B_F^\sigma}\int_{\Lambda^3}
  2W(x_1-x_3)e^{-i(q_3+q_4)x_2}
  \Big(\sum_{p_1\notin B_F^\sigma}\xi_{p_1-q_4,q_4,q_3}^{\sigma,\tau}e^{ip_1(x_2-x_1)}\Big)\\
  &\quad\quad\times
  a^*(g_{x_1,\sigma})a^*(g_{x_3,\nu})a^*(h_{x_2,\tau})a(h_{x_3,\nu})
  a_{q_4,\sigma}a_{q_3,\tau}dx_1dx_2dx_3+h.c.\\
  &=\sum_{\substack{\sigma,\nu,\tau,q_3}}\chi_{q_3\in B_F^\tau}\int_{\Lambda^3}
  2W(x_1-x_3)e^{-iq_3x_2}\\
  &\quad\quad\times
  a^*(g_{x_1,\sigma})a^*(g_{x_3,\nu})a^*(h_{x_2,\tau})a(h_{x_3,\nu})
  a(k_{x_1,x_2,\sigma}^{(q_3,\tau)})a_{q_3,\tau}dx_1dx_2dx_3+h.c.
  \end{aligned}
\end{equation}
For $0<d<\frac{1}{3}$ to be determined later, we use the fact that
\begin{equation*}
\begin{aligned}
  \chi_{q_3\in B_F^\tau}=
  \chi_{q_3\in \underline{A}_{F,d}^\tau}+
  \chi_{q_3\in \underline{B}_{F,d}^\tau}
\end{aligned}
\end{equation*}
to decompose $\Phi_2=\Phi_{21}+\Phi_{22}$.
\par For $\Phi_{21}$, we further decompose it into $\Phi_{21}=\sum_{j=1}^{4}\Phi_{21j}$ using the fact that
\begin{equation*}
\begin{aligned}
  a^*(h_{x_2,\tau})a(h_{x_3,\nu})
  =&a^*(H_{x_2,\tau}[\delta])a(H_{x_3,\nu}[\delta^\prime])
  +a^*(H_{x_2,\tau}[\delta])a(L_{x_3,\nu}[\delta^\prime])\\
  &+a^*(L_{x_2,\tau}[\delta])a(H_{x_3,\nu}[\delta^\prime])
  +a^*(L_{x_2,\tau}[\delta])a(L_{x_3,\nu}[\delta^\prime]),
\end{aligned}
\end{equation*}
where $-\frac{11}{48}<\delta,\delta^\prime\leq\frac{1}{3}$ to be determined later. For $\Phi_{211}$ we can use (\ref{ineqn N_h and N_i}), (\ref{ineqn N_l and N_s}) and (\ref{12thetatilde}) to bound for any $\psi\in\mathcal{H}^{\wedge N}(\{N_{\varsigma_i}\})$:
\begin{equation}\label{Phi211 bound}
  \begin{aligned}
  \vert\langle\Phi_{211}\psi,\psi\rangle\vert&\lesssim
  \sum_{\substack{\sigma,\nu,\tau,q_3}}\chi_{q_3\in\underline{A}_{F,d}^\tau}\int_{\Lambda^3}
  W(x_1-x_3)\Vert a(g_{x_3,\nu})a(g_{x_1,\sigma})a(H_{x_2,\tau}[\delta])\psi\Vert\\
  &\quad\quad\times
  \Vert a(k_{x_1,x_2,\sigma}^{(q_3,\tau)})a_{q_3,\tau}a(H_{x_3,\nu}[\delta^\prime])\psi\Vert
  dx_1dx_2dx_3\\
  &\lesssim\Big(\sum_{\substack{\sigma,\nu,\tau,q_3}}
  \chi_{q_3\in\underline{A}_{F,d}^\tau}\int_{\Lambda^3}
  W(x_1-x_3)N^2\Vert a(H_{x_2,\tau}[\delta])\psi\Vert^2\Big)^{\frac{1}{2}}\\
   &\quad\quad\times
  \Big(\sum_{\substack{\sigma,\nu,\tau,q_3}}
  \chi_{q_3\in\underline{A}_{F,d}^\tau}\int_{\Lambda^2}
  W(x_1-x_3)N^{-\frac{4}{3}+\varepsilon}
  \Vert a(H_{x_3,\nu}[\delta^\prime])\psi\Vert^2\Big)^{\frac{1}{2}}\\
  &\lesssim N^{\frac{2}{3}-\frac{\delta}{2}-\frac{\delta^\prime}{2}+\frac{\varepsilon}{2}+d}
  \Vert W\Vert_1\langle\mathcal{K}_s\psi,\psi\rangle
  \lesssim N^{-\frac{1}{3}-\frac{\delta}{2}-\frac{\delta^\prime}{2}+\frac{\varepsilon}{2}+d}
  \langle\mathcal{K}_s\psi,\psi\rangle.
  \end{aligned}
\end{equation}
For the estimate of $\Phi_{213}$, since $\Vert a(L_{x_2,\tau}[\delta])\Vert^2\lesssim N^{\frac{2}{3}+\delta}$, we have
\begin{equation}\label{Phi213bound1st}
  \begin{aligned}
  \vert\langle\Phi_{213}\psi,\psi\rangle\vert&\lesssim
  \sum_{\substack{\sigma,\nu,\tau,q_3}}\chi_{q_3\in\underline{A}_{F,d}^\tau}\int_{\Lambda^3}
  W(x_1-x_3)\Vert a^*_{q_3,\tau}a(g_{x_3,\nu})a(g_{x_1,\sigma})\psi\Vert\\
  &\quad\quad\times\Vert a^*(L_{x_2,\tau}[\delta])a(k_{x_1,x_2,\sigma}^{(q_3,\tau)})
  a(H_{x_3,\nu}[\delta^\prime])\psi\Vert dx_1dx_2dx_3\\
  &\lesssim\Big(\sum_{\sigma,\nu,\tau,q_3}\chi_{q_3\in B_F^\tau}
  \int_{\Lambda^3}W(x_1-x_3)\Vert a^*_{q_3,\tau}a(g_{x_3,\nu})a(g_{x_1,\sigma})\psi\Vert^2
  \Big)^{\frac{1}{2}}\\
  &\quad\quad\times
  \Big(\sum_{\substack{\sigma,\nu,\tau,q_3}}\chi_{q_3\in\underline{A}_{F,d}^\tau}\int_{\Lambda^2}
  W(x_1-x_3)N^{-\frac{2}{3}+\delta+\varepsilon}
  \Vert a(H_{x_3,\nu}[\delta^\prime])\psi\Vert^2\Big)^{\frac{1}{2}}
  \end{aligned}
\end{equation}
Notice that since $a(g_{x_3,\nu})a(g_{x_1,\sigma})\psi\in\mathcal{H}^{\wedge (N-2)}$, and on the subspace of $\mathcal{H}^{\wedge (N-2)}$ generated by $a(g_{x_3,\nu})a(g_{x_1,\sigma})\psi$, since by Lemma \ref{lem commutator}, we have
\begin{equation*}
  \mathcal{N}a(g_{x_3,\nu})a(g_{x_1,\sigma})\psi=a(g_{x_3,\nu})a(g_{x_1,\sigma})(\mathcal{N}-2)
  \psi=a(g_{x_3,\nu})a(g_{x_1,\sigma})(N-2)
  \psi.
\end{equation*}
Therefore, on the subspace of $\mathcal{H}^{\wedge (N-2)}$ generated by $a(g_{x_3,\nu})a(g_{x_1,\sigma})\psi$, we have
\begin{equation}\label{fun fact}
  \sum_{\tau,q_3\in B_F^\tau}a_{q_3,\tau}a^*_{q_3,\tau}
  =N-\sum_{\tau,q_3\in B_F^\tau}a^*_{q_3,\tau}a_{q_3,\tau}
  =\mathcal{N}_{ex}+2.
\end{equation}
Plugging (\ref{fun fact}) into (\ref{Phi213bound1st}), we have
\begin{equation}\label{Phi213bound}
  \begin{aligned}
  \vert\langle\Phi_{213}\psi,\psi\rangle\vert&\lesssim
  N^{1+\frac{\delta}{2}+\frac{d}{2}+\frac{\varepsilon}{2}}\Vert W\Vert_1
  \langle(\mathcal{N}_{ex}+1)\psi,\psi\rangle^{\frac{1}{2}}
  \langle\mathcal{N}_h[\delta^\prime]\psi,\psi\rangle^{\frac{1}{2}}\\
  &\lesssim N^{-\frac{1}{6}+\frac{\delta}{2}+\frac{\delta^\prime}{2}
  +\frac{d}{2}+\frac{\varepsilon}{2}}\langle(\mathcal{N}_{ex}+1)\psi,\psi\rangle
  +N^{-\frac{1}{6}+\frac{\delta}{2}-\frac{3}{2}\delta^\prime
  +\frac{d}{2}+\frac{\varepsilon}{2}}
  \langle\mathcal{K}_s\psi,\psi\rangle.
  \end{aligned}
\end{equation}
For $\Phi_{212}$ and $\Phi_{214}$, we first switch back to the momentum space and use (\ref{anticommutator}) to write
\begin{equation}\label{Phi_212Phi214re}
  \Phi_{212}+\Phi_{214}=\check{\Phi}_{212}+\check{\Phi}_{214}+
  \check{\Phi}_{21,r1}+\check{\Phi}_{21,r2}
\end{equation}
with
\begin{equation}\label{checkphi}
  \begin{aligned}
  \check{\Phi}_{212}&=\sum_{\substack{\sigma,\nu\\k,p,q}}\sum_{\tau,l,r,s}
  2W_k\xi_{l,s,r}^{\sigma,\tau}
    \delta_{p-k,s+l}a^*_{p,\sigma}a^*_{r-l,\tau}
    a_{q+k,\nu}a_{s,\sigma}a_{r,\tau}a^*_{q,\nu}\\
    &\quad\quad\quad\quad\times
    \chi_{p-k\notin B_F^\sigma}\chi_{p,s\in B_F^\sigma}\chi_{q+k\in A_{F,\delta^\prime}^\nu}
    \chi_{q\in B_F^\nu}\chi_{r-l\in P^\tau_{F,\delta}}\chi_{r\in \underline{A}_{F,d}^\tau}+h.c.\\
 \check{\Phi}_{214}&=\sum_{\substack{\sigma,\nu\\k,p,q}}\sum_{\tau,l,r,s}
 2W_k\xi_{l,s,r}^{\sigma,\tau}
    \delta_{p-k,s+l}a^*_{p,\sigma}a^*_{r-l,\tau}
    a_{q+k,\nu}a_{s,\sigma}a_{r,\tau}a^*_{q,\nu}\\
    &\quad\quad\quad\quad\times
    \chi_{p-k\notin B_F^\sigma}\chi_{p,s\in B_F^\sigma}\chi_{q+k\in A_{F,\delta^\prime}^\nu}
    \chi_{q\in B_F^\nu}\chi_{r-l\in A^\tau_{F,\delta}}\chi_{r\in \underline{A}_{F,d}^\tau}+h.c.\\
 \check{\Phi}_{21,r1}&=\sum_{\substack{\sigma,k,p,q}}\sum_{\tau,l,r,s}
 2W_k\xi_{l,s,r}^{\sigma,\tau}
    \delta_{p-k,s+l}\delta_{q,s}a^*_{p,\sigma}a^*_{r-l,\tau}
    a_{q+k,\sigma}a_{r,\tau}\\
    &\quad\quad\quad\quad\times
    \chi_{p-k\notin B_F^\sigma}\chi_{p,q\in B_F^\sigma}\chi_{q+k\in A_{F,\delta^\prime}^\sigma}
    \chi_{r-l\in B^\tau_{F}}\chi_{r\in \underline{A}_{F,d}^\tau}+h.c.\\
 \check{\Phi}_{21,r2}&=-\sum_{\substack{\sigma,k,p,q}}\sum_{\tau,l,r,s}
 2W_k\xi_{l,s,r}^{\sigma,\tau}
    \delta_{p-k,s+l}\delta_{q,r}a^*_{p,\sigma}a^*_{r-l,\tau}
    a_{q+k,\tau}a_{s,\sigma}\\
    &\quad\quad\quad\quad\times
    \chi_{p-k\notin B_F^\sigma}\chi_{p,s\in B_F^\sigma}\chi_{q+k\in A_{F,\delta^\prime}^\tau}
    \chi_{q\in \underline{A}_{F,d}^\tau}\chi_{r-l\in B^\tau_{F}}+h.c.
  \end{aligned}
\end{equation}
We bound $\check{\Phi}_{212}$ by
\begin{equation}\label{Phi212bound}
  \begin{aligned}
  \vert\langle\check{\Phi}_{212}\psi,\psi\rangle\vert
  &\lesssim \sum_{\sigma,\nu,\tau,q_3}\chi_{q_3\in\underline{A}_{F,d}^\tau}\int_{\Lambda^3}
  W(x_1-x_3)\Vert a^*(L_{x_3,\nu}[\delta^\prime])a(g_{x_1,\sigma})
  a(H_{x_2,\tau}[\delta])\psi\Vert\\
  &\quad\quad\times
  \Vert a(k_{x_1,x_2,\sigma}^{q_3,\tau})a_{q_3,\tau}a^*(g_{x_3,\nu})\psi\Vert
  dx_1dx_2dx_3\\
  &\lesssim\Big(\sum_{\sigma,\nu,\tau,q_3}\chi_{q_3\in\underline{A}_{F,d}^\tau}\int_{\Lambda^3}
  W(x_1-x_3)N^{\frac{5}{3}+\delta^\prime}\Vert
  a(H_{x_2,\tau}[\delta])\psi\Vert^2\Big)^{\frac{1}{2}}\\
  &\quad\quad\times
  \Big(\sum_{\sigma,\nu,\tau,q_3}\chi_{q_3\in\underline{A}_{F,d}^\tau}\int_{\Lambda^2}
  W(x_1-x_3)N^{-\frac{4}{3}+\varepsilon}
  \Vert a^*(g_{x_3,\nu})\psi\Vert^2 \Big)^{\frac{1}{2}}\\
  &\lesssim N^{\frac{5}{6}+d+\frac{\delta^\prime}{2}+\frac{\varepsilon}{2}}\Vert W\Vert_1
  \langle(\mathcal{N}_{ex}+1)\psi,\psi\rangle^{\frac{1}{2}}
  \langle\mathcal{N}_h[\delta]\psi,\psi\rangle^{\frac{1}{2}}\\
  &\lesssim N^{-\frac{1}{3}-\frac{\delta}{2}+\frac{3}{2}\delta^\prime
  +\frac{\varepsilon}{2}+d}\langle(\mathcal{N}_{ex}+1)\psi,\psi\rangle
  +N^{-\frac{1}{3}-\frac{\delta}{2}-\frac{\delta^\prime}{2}
  +\frac{\varepsilon}{2}+d}
  \langle\mathcal{K}_s\psi,\psi\rangle.
  \end{aligned}
\end{equation}
With an observation similar to (\ref{fun fact}), we bound $\check{\Phi}_{214}$ by
\begin{equation}\label{Phi214bound}
  \begin{aligned}
  \vert\langle\check{\Phi}_{214}\psi,\psi\rangle\vert
  &\lesssim \sum_{\sigma,\nu,\tau,q_3}\chi_{q_3\in\underline{A}_{F,d}^\tau}\int_{\Lambda^3}
  W(x_1-x_3)\Vert a(L_{x_2,\tau}[\delta])a_{q_3,\tau}^*a(g_{x_1,\sigma})
  \psi\Vert\\
  &\quad\quad\times
  \Vert a(L_{x_3,\nu}[\delta^\prime])a(k_{x_1,x_2,\sigma}^{q_3,\tau})a^*(g_{x_3,\nu})\psi\Vert
  dx_1dx_2dx_3\\
  &\lesssim\Big(\sum_{\sigma,\nu,\tau,q_3}\chi_{q_3\in B_{F}^\tau}\int_{\Lambda^3}
  W(x_1-x_3)N^{\frac{2}{3}+\delta}\Vert a_{q_3,\tau}^*a(g_{x_1,\sigma})
  \psi\Vert^2\Big)^\frac{1}{2}\\
  &\quad\quad\times
  \Big(\sum_{\sigma,\nu,\tau,q_3}\chi_{q_3\in\underline{A}_{F,d}^\tau}\int_{\Lambda^2}
  W(x_1-x_3)N^{-\frac{2}{3}+\delta^\prime+\varepsilon}
  \Vert a^*(g_{x_3,\nu})\psi\Vert
  \Big)^\frac{1}{2}\\
  &\lesssim N^{\frac{5}{6}+\frac{\delta}{2}+\frac{\delta^\prime}{2}
  +\frac{d}{2}+\frac{\varepsilon}{2}}\Vert W\Vert_1\langle(\mathcal{N}_{ex}+1)\psi,\psi\rangle\\
  &\lesssim N^{-\frac{1}{6}+\frac{\delta}{2}+\frac{\delta^\prime}{2}
  +\frac{d}{2}+\frac{\varepsilon}{2}}\langle(\mathcal{N}_{ex}+1)\psi,\psi\rangle.
  \end{aligned}
\end{equation}
For the third term of (\ref{checkphi}), we first introduce the notation
\begin{equation*}
\begin{aligned}
  &k_{x_1,x_2,x_3,\tau}^{(\sigma)}(z)
  \\
  &\coloneqq\sum_{q_3\in\underline{A}_{F,d}^{\tau}}
  \Big(\sum_{p_1\notin B_F^\sigma}\sum_{q_2\in B_F^\sigma}
  \xi_{p_1-q_2,q_2,q_3}^{\sigma,\tau}e^{-iq_2(x_3-x_2)}e^{-ip_1(x_2-x_1)}\Big)
  e^{iq_3x_2}f_{q_3,\tau}(z)
\end{aligned}
\end{equation*}
By (\ref{xi l2}) in Lemma \ref{xi l2 lemma}, we have
\begin{equation}\label{hehe}
  \begin{aligned}
  &\int_{\Lambda}\Vert a(k_{x_1,x_2,x_3,\tau}^{(\sigma)})\Vert^2dx_2\\
  &\leq\int_{\Lambda}\sum_{q_3\in\underline{A}_{F,d}^{\tau}}
  \Big\vert\sum_{p_1\notin B_F^\sigma}\sum_{q_2\in B_F^\sigma}
  \xi_{p_1-q_2,q_2,q_3}^{\sigma,\tau}e^{iq_2(x_3-x_2)}e^{ip_1(x_2-x_1)}\Big\vert^2dx_2\\
  &\leq\int_{\Lambda}\sum_{q_3\in\underline{A}_{F,d}^{\tau}}
  \Big\vert\sum_{k}\sum_{q_2\in B_F^\sigma}\Big(\chi_{q_2+k\notin B_F^\sigma}
  \xi_{k,q_2,q_3}^{\sigma,\tau}e^{iq_2(x_3-x_1)}\Big)e^{ik(x_2-x_1)}\Big\vert^2dx_2\\
  &=\sum_{q_3\in\underline{A}_{F,d}^{\tau}}
  \sum_{k}\Big\vert\sum_{q_2\in B_F^\sigma}\Big(\chi_{q_2+k\notin B_F^\sigma}
  \xi_{k,q_2,q_3}^{\sigma,\tau}e^{iq_2(x_3-x_1)}\Big)\Big\vert^2\\
  &\leq \sum_{q_3\in\underline{A}_{F,d}^{\tau}}
  \sum_{q_2,q_2^\prime\in B_F^\sigma}\sum_{k}\big\vert\xi_{k,q_2,q_3}^{\sigma,\tau}\big\vert
  \big\vert\xi_{k,q_2^\prime,q_3}^{\sigma,\tau}\big\vert
  \lesssim N^{\frac{1}{3}+d+\varepsilon}.
  \end{aligned}
\end{equation}
With the above notation, we write $\check{\Phi}_{21,r1}$ by
\begin{equation}\label{checkPhi_21r1}
  \begin{aligned}
  \check{\Phi}_{21,r1}=&2\sum_{\sigma,\tau}\int_{\Lambda^3}W(x_1-x_3)a^*(g_{x_1,\sigma})
  a^*(h_{x_2,\tau})\\
  &\quad\quad\quad\quad\times
  a(L_{x_3,\sigma}[\delta^\prime])a(k_{x_1,x_2,x_3,\tau}^{(\sigma)})
  dx_1dx_2dx_3+h.c.
  \end{aligned}
\end{equation}
By (\ref{hehe}), we have
\begin{equation}\label{checkPhi_21r1bound}
  \pm\check{\Phi}_{21,r1}\lesssim
  N^{\frac{2}{3}+\frac{d}{2}+\frac{\varepsilon}{2}}\Vert W\Vert_1
  (\mathcal{N}_{ex}+\mathcal{N}_l[\delta^\prime])
  \lesssim N^{-\frac{1}{3}+\frac{d}{2}+\frac{\varepsilon}{2}}\mathcal{N}_{ex}.
\end{equation}
For the last term of (\ref{checkphi}), we introduce the notation
\begin{equation*}
\begin{aligned}
  &\tilde{k}_{x_1,x_2,x_3,\sigma}^{(\tau)}(z)
  \\
  &\coloneqq\sum_{q_4\in B_F^\sigma}
  \Big(\sum_{p_1\notin B_F^\sigma}\sum_{q_2\in \underline{A}_{F,d}^{\tau}}
  \xi_{p_1-q_4,q_4,q_2}^{\sigma,\tau}e^{-iq_2(x_3-x_2)}e^{-ip_1(x_2-x_1)}\Big)
  e^{iq_4x_2}f_{q_4,\sigma}(z)
\end{aligned}
\end{equation*}
Similar to (\ref{hehe}), we can bound
\begin{equation}\label{hehe2}
  \begin{aligned}
  \int_{\Lambda}\Vert a(\tilde{k}_{x_1,x_2,x_3,\sigma}^{(\tau)})\Vert^2dx_2
  \leq \sum_{q_4\in B_F^\sigma}
  \sum_{q_2,q_2^\prime\in \underline{A}_{F,d}^{\tau}}
  \sum_{k}\big\vert\xi_{k,q_4,q_2}^{\sigma,\tau}\big\vert^2
  \lesssim N^{2d+\varepsilon}.
  \end{aligned}
\end{equation}
With the above notation, we write $\check{\Phi}_{21,r2}$ by
\begin{equation}\label{checkPhi_21r2}
  \begin{aligned}
  \check{\Phi}_{21,r2}=&2\sum_{\sigma,\tau}\int_{\Lambda^3}W(x_1-x_3)a^*(g_{x_1,\sigma})
  a^*(h_{x_2,\tau})\\
  &\quad\quad\quad\quad\times
  a(L_{x_3,\tau}[\delta^\prime])a(k_{x_1,x_2,x_3,\sigma}^{(\tau)})
  dx_1dx_2dx_3+h.c.
  \end{aligned}
\end{equation}
By (\ref{hehe2}), we have
\begin{equation}\label{checkPhi_21r2bound}
  \pm\check{\Phi}_{21,r2}\lesssim
  N^{\frac{1}{2}+d+\frac{\varepsilon}{2}}\Vert W\Vert_1
  (\mathcal{N}_{ex}+\mathcal{N}_l[\delta^\prime])
  \lesssim N^{-\frac{1}{2}+d+\frac{\varepsilon}{2}}\mathcal{N}_{ex}.
\end{equation}
Collecting (\ref{Phi211 bound}-\ref{checkPhi_21r2bound}), and setting
\begin{equation*}
  \delta=-\frac{1}{6}+\frac{1}{33}+\frac{d}{2},\quad
  \delta^\prime=\frac{1}{33},
\end{equation*}
we conclude that
\begin{equation}\label{Phi_21bound}
  \pm\Phi_{21}\lesssim N^{-\frac{1}{4}-\frac{1}{33}+\frac{3}{4}d+\frac{\varepsilon}{2}}
  \mathcal{K}_s
  +N^{-\frac{1}{4}+\frac{1}{33}+\frac{3}{4}d+\frac{\varepsilon}{2}}
  (\mathcal{N}_{ex}+1).
\end{equation}
Notice that since we demand $0<d<\frac{1}{3}$, we can check that $-\frac{11}{48}<\delta,\delta^\prime\leq\frac{1}{3}$. We will optimize (\ref{Phi_21bound}) in $d$ when we finish (\ref{Phi_22bound}).
\par For the estimate of $\Phi_{22}$, we notice that by (\ref{anticommutator}), we have
\begin{equation*}
  a_{q_3,\tau}a^*(g_{x,\sigma})+a^*(g_{x,\sigma})a_{q_3,\tau}
  =\delta_{\sigma,\tau}e^{iq_3x}.
\end{equation*}
Therefore, we decompose $\Phi_{22}=\sum_{j=1}^{3}\Phi_{22j}$ with
\begin{equation}\label{Phi_22j}
  \begin{aligned}
  \Phi_{221}&=-\sum_{\substack{\sigma,\nu,\tau,q_3}}\chi_{q_3\in \underline{B}_{F,d}^\tau}\int_{\Lambda^3}
  2W(x_1-x_3)e^{-iq_3x_2}\\
  &\quad\quad\times
 a_{q_3,\tau}a^*(g_{x_1,\sigma})a^*(g_{x_3,\nu})a^*(h_{x_2,\tau})a(h_{x_3,\nu})
  a(k_{x_1,x_2,\sigma}^{(q_3,\tau)})dx_1dx_2dx_3+h.c.\\
  \Phi_{222}&=-\sum_{\substack{\sigma,\tau,q_3}}\chi_{q_3\in \underline{B}_{F,d}^\tau}\int_{\Lambda^3}
  2W(x_1-x_3)e^{-iq_3(x_2-x_3)}\\
  &\quad\quad\times
 a^*(g_{x_1,\sigma})a^*(h_{x_2,\tau})a(h_{x_3,\tau})
  a(k_{x_1,x_2,\sigma}^{(q_3,\tau)})dx_1dx_2dx_3+h.c.\\
  \Phi_{223}&=\sum_{\substack{\sigma,\nu,q_3}}\chi_{q_3\in \underline{B}_{F,d}^\sigma}\int_{\Lambda^3}
  2W(x_1-x_3)e^{-iq_3(x_2-x_1)}\\
  &\quad\quad\times
 a^*(g_{x_3,\nu})a^*(h_{x_2,\sigma})a(h_{x_3,\nu})
  a(k_{x_1,x_2,\sigma}^{(q_3,\sigma)})dx_1dx_2dx_3+h.c.
  \end{aligned}
\end{equation}
Since $\#\underline{B}_{F,d}^\tau\leq N_{\tau}\leq N$, we then bound the last two terms of (\ref{Phi_22j}) using (\ref{12thetatilde}):
\begin{equation}\label{Phi_222223bound}
  \pm\Phi_{222},\pm\Phi_{223}\lesssim N^{\frac{5}{6}+\frac{\varepsilon}{2}}
  \Vert W\Vert_1\mathcal{N}_{ex}\lesssim N^{-\frac{1}{6}+\frac{\varepsilon}{2}}
  \mathcal{N}_{ex}.
\end{equation}
For the estimate of $\Phi_{221}$, we use the fact that $a(h_{x_3,\nu})=a(H_{x_3,\nu}[\delta^\prime])+a(L_{x_3,\nu}[\delta^\prime])$ (with $\delta^\prime=\frac{1}{33}$ as chosen above) to decompose $\Phi_{221}=\Phi_{2211}+\Phi_{2212}$. Using (\ref{12thetatilde}), and (\ref{xi l2 cut off}) in Lemma \ref{xi l2 lemma}, we bound them respectively by
\begin{equation}\label{Phi_2211bound}
  \begin{aligned}
  &\vert\langle\Phi_{2211}\psi,\psi\rangle\vert\lesssim
  \sum_{\substack{\sigma,\nu,\tau,q_3}}\chi_{q_3\in \underline{B}_{F,d}^\tau}\int_{\Lambda^3}
  W(x_1-x_3)\Vert (\mathcal{N}_{ex}+2)^{-\frac{1}{2}}a(h_{x_2,\tau})a_{q_3,\tau}^*\psi\Vert\\
  &\quad\quad\quad\quad\quad\quad\times
  \Vert a^*(g_{x_1,\sigma})a^*(g_{x_3,\nu})a(k_{x_1,x_2,\sigma}^{(q_3,\tau)})
  (\mathcal{N}_{ex}+2)^{\frac{1}{2}}
  a(H_{x_3,\nu}[\delta^\prime])\psi\Vert\\
  &\lesssim\Big(\sum_{\substack{\sigma,\nu,\tau,q_3}}\int_{\Lambda^2}
  W(x_1-x_3)\big(\vert k_F^\tau\vert^2-\vert q_3\vert^2\big)
  \Vert a_{q_3,\tau}^*\psi\Vert^2\Big)^{\frac{1}{2}}\\
  &\quad\quad\times
  \Big(\sum_{\substack{\sigma,\nu,\tau\\q_4\in B_F^\sigma}}\sum_{\substack{q_3\in \underline{B}_{F,d}^\tau\\k}}\frac{N^{2}\vert\xi_{k,q_4,q_3}^{\sigma,\tau}\vert^2}{\vert k_F^\tau\vert^2-\vert q_3\vert^2}
  \int_{\Lambda^2}
  W(x_1-x_3)\Vert
  (\mathcal{N}_{ex}+1)^{\frac{1}{2}}
  a(H_{x_3,\nu}[\delta^\prime])\psi\Vert^2\Big)^{\frac{1}{2}}\\
  &\lesssim N^{\frac{7}{6}-\frac{d}{2}}\Vert W\Vert_1
  \langle\mathcal{K}_s\psi,\psi\rangle^{\frac{1}{2}}
  \langle\mathcal{N}_h[\delta^\prime]\psi,\psi\rangle^{\frac{1}{2}}
  \lesssim N^{-\frac{d}{2}-\frac{1}{66}}\langle\mathcal{K}_s\psi,\psi\rangle
  \end{aligned}
\end{equation}
and
\begin{equation}\label{Phi_2212bound}
  \begin{aligned}
  &\vert\langle\Phi_{2212}\psi,\psi\rangle\vert\lesssim
  \sum_{\substack{\sigma,\nu,\tau,q_3}}\chi_{q_3\in \underline{B}_{F,d}^\tau}\int_{\Lambda^3}
  W(x_1-x_3)\Vert (\mathcal{N}_{ex}+2)^{-\frac{1}{2}}a(h_{x_2,\tau})a_{q_3,\tau}^*\psi\Vert\\
  &\quad\quad\quad\quad\quad\quad\times
  \Vert a^*(g_{x_1,\sigma})a^*(g_{x_3,\nu})a(k_{x_1,x_2,\sigma}^{(q_3,\tau)})
  a(L_{x_3,\nu}[\delta^\prime])(\mathcal{N}_{ex}+1)^{\frac{1}{2}}
  \psi\Vert\\
  &\lesssim\Big(\sum_{\substack{\sigma,\nu,\tau,q_3}}\int_{\Lambda^2}
  W(x_1-x_3)\big(\vert k_F^\tau\vert^2-\vert q_3\vert^2\big)
  \Vert a_{q_3,\tau}^*\psi\Vert^2\Big)^{\frac{1}{2}}\\
  &\quad\quad\times
  \Big(\sum_{\substack{\sigma,\nu,\tau\\q_4\in B_F^\sigma}}\sum_{\substack{q_3\in \underline{B}_{F,d}^\tau\\k}}\frac{N^{\frac{8}{3}+\delta^\prime}
  \vert\xi_{k,q_4,q_3}^{\sigma,\tau}\vert^2}{\vert k_F^\tau\vert^2-\vert q_3\vert^2}
  \int_{\Lambda^2}
  W(x_1-x_3)\Vert
  (\mathcal{N}_{ex}+1)^{\frac{1}{2}}
  \psi\Vert^2\Big)^{\frac{1}{2}}\\
  &\lesssim N^{1-\frac{d}{2}+\frac{\delta^\prime}{2}}\Vert W\Vert_1
  \langle\mathcal{K}_s\psi,\psi\rangle^{\frac{1}{2}}
  \langle(\mathcal{N}_{ex}+1)\psi,\psi\rangle^{\frac{1}{2}}\\
  &
  \lesssim N^{-\frac{d}{2}-\frac{1}{66}}\langle\mathcal{K}_s\psi,\psi\rangle
  +N^{-\frac{d}{2}+\frac{1}{22}}\langle(\mathcal{N}_{ex}+1)\psi,\psi\rangle
  \end{aligned}
\end{equation}
Collecting (\ref{Phi_22j}-\ref{Phi_2212bound}), we conclude that
\begin{equation}\label{Phi_22bound}
  \pm\Phi_{22}\lesssim N^{-\frac{d}{2}-\frac{1}{66}}\mathcal{K}_s
  +\big(N^{-\frac{d}{2}+\frac{1}{22}}+N^{-\frac{1}{6}+\frac{\varepsilon}{2}}\big)
  (\mathcal{N}_{ex}+1).
\end{equation}
\vspace{1em}
\par Putting (\ref{Phi_21bound}) and (\ref{Phi_22bound}) together, we choose $d=\frac{1}{5}+\frac{4}{5}\times\frac{1}{66}$ to optimize, and we conclude that
\begin{equation}\label{Phi_2bound}
  \pm\Phi_2\lesssim N^{-\frac{1}{10}-\frac{7}{5}\times\frac{1}{66}+\frac{\varepsilon}{2}}
  \mathcal{K}_s+N^{-\frac{1}{10}+\frac{13}{5}\times\frac{1}{66}+\frac{\varepsilon}{2}}
  (\mathcal{N}_{ex}+1).
\end{equation}
\par For $[\Omega,\tilde{B}]$, $\tilde{\Phi}_2$ can be bounded similarly. We only need to replace the estimate $\Vert W\Vert_1\lesssim a$ by $N^{\frac{1}{3}}\Vert\nabla\eta\Vert_1\lesssim aN^{-\alpha}$. By Lemmas \ref{control eBtiled on K_s lemma} and \ref{control eBtilde on N_ex lemma}, we obtain
\begin{equation}\label{Phi_2boundfinale}
\begin{aligned}
  \int_{0}^{1}\int_{t}^{1}e^{-s\tilde{B}}(\Phi_2+\tilde{\Phi}_2)e^{s\tilde{B}}dsdt
  \lesssim& N^{-\frac{1}{10}-\frac{7}{5}\times\frac{1}{66}+\varepsilon}
  \big(\mathcal{K}_s+N^{\frac{1}{3}+\alpha}\big)\\
  &+N^{-\frac{1}{10}+\frac{13}{5}\times\frac{1}{66}+\varepsilon}
  (\mathcal{N}_{ex}+1)
\end{aligned}
\end{equation}
for some $\varepsilon>0$ small enough but fixed.

\begin{flushleft}
    \textbf{Analysis of $\Phi_3$:}
\end{flushleft}
We first introduce the notation $\zeta_{k,q,p}^{\nu,\sigma}$:
\begin{equation}\label{define zeta_k,q,p,nu,sigma}
  \zeta_{k,q,p}^{\nu,\sigma}=\frac{-\big(W_k+\eta_kk\big((q-k)-(p+k)\big)\big)}
  {\frac{1}{2}\big(\vert q\vert^2+\vert p\vert^2-\vert q-k\vert^2-\vert p+k\vert^2\big)}
  \chi_{p\notin B^{\sigma}_F}\chi_{q\notin B^{\nu}_F}\chi_{p+k\in B^{\sigma}_F}\chi_{q-k\in B^{\nu}_F}.
\end{equation}
It is easy to verify that
\begin{equation}\label{relation}
  \xi_{k,q,p}^{\nu,\sigma}=\zeta_{k,q+k,p-k}^{\nu,\sigma},\quad
  \zeta_{k,q,p}^{\nu,\sigma}=\xi_{k,q-k,p+k}^{\nu,\sigma}.
\end{equation}
Using the fact that $\chi_{p-k\notin B_F^\sigma}=\chi_{p-k\in A_{F,1/3}^\sigma}+\chi_{p-k\in P_{F,1/3}^\sigma}$, we decompose $\Phi_{3}=\Phi_{31}+\Phi_{32}$. Using the notation (\ref{define zeta_k,q,p,nu,sigma}), we can write
\begin{equation}\label{Phi_3j}
  \begin{aligned}
  &\Phi_{31}=\sum_{\substack{\sigma,\nu\\k,p,q}}\sum_{l,r,s}W_k
  \zeta_{l,s+l,r-l}^{\nu,\sigma}
    \delta_{p,r}\delta_{q,s}a^*_{r-l,\sigma}a^*_{s+l,\nu}
    a_{q+k,\nu}a_{p-k,\sigma}\\
    &\quad\quad\quad\quad\times
    \chi_{p-k\in A_{F,1/3}^\sigma}
    \chi_{r-l\notin B_F^\sigma}\chi_{p\in B_F^\sigma}\chi_{q+k,s+l\notin B_F^\nu}
    \chi_{q\in B_F^\nu}+h.c.\\
    &\Phi_{32}=\sum_{\substack{\sigma,\nu\\k,p,q}}\sum_{l,r,s}W_k
    \zeta_{l,s+l,r-l}^{\nu,\sigma}
    \delta_{p,r}\delta_{q,s}a^*_{r-l,\sigma}a^*_{s+l,\nu}
    a_{q+k,\nu}a_{p-k,\sigma}\\
    &\quad\quad\quad\quad\times
    \chi_{p-k\in P_{F,1/3}^\sigma}
    \chi_{r-l\notin B_F^\sigma}\chi_{p\in B_F^\sigma}\chi_{q+k,s+l\notin B_F^\nu}
    \chi_{q\in B_F^\nu}+h.c.
  \end{aligned}
\end{equation}
For $\Phi_{31}$, we rewrite it by
\begin{equation*}
  \begin{aligned}
  \Phi_{31}&=\sum_{\substack{\sigma,\nu\\p_3,p_4,q_1,q_2}}
  \chi_{q_1\in B_F^\sigma}\chi_{q_2\in B_F^\nu}
  \chi_{p_3\notin B_F^\sigma}\chi_{p_4\notin B_F^\nu}
  \int_{\Lambda^3}e^{-i(q_1-p_3-p_4+q_2)x_2}\zeta_{p_4-q_2,p_4,p_3}^{\nu,\sigma}\\
  &\quad\quad\times
  e^{iq_1x_1}e^{iq_2x_3}W(x_1-x_3)
  a_{p_3,\sigma}^*a^*_{p_4,\nu}a(h_{x_3,\nu})a(L_{x_1,\sigma}[1/3])dx_1dx_2dx_3+h.c.\\
  &=\sum_{\substack{\sigma,\nu\\p_3,p_4,q_2}}
  \chi_{q_2\in B_F^\nu}
  \chi_{p_3\notin B_F^\sigma}\chi_{p_4\notin B_F^\nu}
  \int_{\Lambda^2}W(x_1-x_3)\zeta_{p_4-q_2,p_4,p_3}^{\nu,\sigma}\\
  &\quad\quad\times
  e^{i(p_3+p_4)x_1}e^{iq_2(x_3-x_1)}
  a_{p_3,\sigma}^*a^*_{p_4,\nu}a(h_{x_3,\nu})a(L_{x_1,\sigma}[1/3])dx_1dx_3+h.c.
  \end{aligned}
\end{equation*}
We then introduce the notation
\begin{equation}\label{notation j}
  \begin{aligned}
  j_{x_1,x_3,\nu}^{(p_3,\sigma)}(z)\coloneqq
  \sum_{p_4\notin B_F^\nu}\Big(\sum_{q_2\in B_F^\nu}\zeta_{p_4-q_2,p_4,p_3}^{\nu,\sigma}
  e^{iq_2(x_3-x_1)}\Big)e^{ip_4x_1}f_{p_4,\nu}(z).
  \end{aligned}
\end{equation}
By (\ref{xi l2}) in Lemma \ref{xi l2 lemma}, we know that for $j=1,3$
\begin{equation}\label{jl2}
  \begin{aligned}
  \sum_{p_3\notin B_F^\sigma}\int_{\Lambda}\Vert a(j_{x_1,x_3,\nu}^{(p_3,\sigma)})\Vert^2dx_j&\leq
  \sum_{p_3\notin B_F^\sigma}\sum_{p_4\notin B_F^\nu}\int_{\Lambda}
  \Big\vert\sum_{q_2\in B_F^\nu}\zeta_{p_4-q_2,p_4,p_3}^{\nu,\sigma}
  e^{iq_2(x_3-x_1)}\Big\vert^2dx_j\\
  &\leq\sum_{p_3\notin B_F^\sigma}\sum_{p_4\notin B_F^\nu}
  \sum_{k}\vert\zeta_{k,p_4,p_3}^{\nu,\sigma}\vert^2\\
  &=\sum_{q_3\in B_F^\sigma}\sum_{q_4\in B_F^\nu}
  \sum_{k}\vert\xi_{k,q_4,q_3}^{\nu,\sigma}\vert^2\lesssim N^{-\frac{1}{3}+\varepsilon}.
  \end{aligned}
\end{equation}
By the definition of $\zeta_{k,q,p}^{\nu,\sigma}$ in (\ref{define zeta_k,q,p,nu,sigma}) and the proof of Lemma \ref{L1 Linfty est lemma}, we can check that
\begin{equation}\label{jlinfty}
  \begin{aligned}
  &\sum_{p_3\notin B_F^\sigma}\Vert a(j_{x_1,x_3,\nu}^{(p_3,\sigma)})\Vert^2\leq
  \sum_{p_3\notin B_F^\sigma}\sum_{p_4\notin B_F^\nu}
  \Big\vert\sum_{q_2\in B_F^\nu}\zeta_{p_4-q_2,p_4,p_3}^{\nu,\sigma}
  e^{iq_2(x_3-x_1)}\Big\vert^2\\
  &\leq\sum_{p_3\notin B_F^\sigma}\sum_{p_4\notin B_F^\nu}
  \Big(\sum_{k}\vert\zeta_{k,p_4,p_3}^{\nu,\sigma}\vert\Big)^2
  \leq\sum_{p_3\notin B_F^\sigma}\sum_{p_4\notin B_F^\nu}
  \sum_{k,k^\prime}\vert\zeta_{k,p_4,p_3}^{\nu,\sigma}\vert
  \vert\zeta_{k^\prime,p_4,p_3}^{\nu,\sigma}\vert\\
  &\lesssim\sum_{p_3\notin B_F^\sigma}\sum_{p_4\notin B_F^\nu}
  \sum_{k,k^\prime}\frac{\big\vert W_{k}+N^{\frac{1}{3}}\eta_{k}\vert k\vert\big\vert
  \chi_{p_4-k\in B_F^\nu}}
  {\vert p_4\vert^2-\vert p_4-k\vert^2}
  \frac{\big\vert W_{k^\prime}+N^{\frac{1}{3}}\eta_{k^\prime}\vert k^\prime\vert\big\vert
  \chi_{p_3+k\in B_F^\sigma}}
  {\vert p_3\vert^2-\vert p_3+k\vert^2}\\
  &\lesssim\Big(\sum_{q_4\in B_F^\nu}\sum_{k}\frac{\big\vert W_{k}+N^{\frac{1}{3}}\eta_{k}\vert k\vert\big\vert
  \chi_{q_4+k\notin B_F^\nu}}
  {\vert q_4+k\vert^2-\vert q_4\vert^2}\Big)
  \Big(\sum_{q_3\in B_F^\sigma}\sum_{k}\frac{\big\vert W_{k}+N^{\frac{1}{3}}\eta_{k}\vert k\vert\big\vert
  \chi_{q_3-k\notin B_F^\sigma}}
  {\vert q_3-k\vert^2-\vert q_3\vert^2}\Big)\\
  &\lesssim (Na\ell^{-1})^2\lesssim N^{\frac{2}{3}+2\alpha}.
  \end{aligned}
\end{equation}
With the notation given in (\ref{notation j}), we have
\begin{equation}\label{Phi_31re}
  \begin{aligned}
  \Phi_{31}&=\sum_{\substack{\sigma,\nu,p_3}}
  \chi_{p_3\notin B_F^\sigma}
  \int_{\Lambda^2}W(x_1-x_3)e^{ip_3x_1}\\
  &\quad\quad\times
  a_{p_3,\sigma}^*a^*(j_{x_1,x_3,\nu}^{(p_3,\sigma)})
  a(h_{x_3,\nu})a(L_{x_1,\sigma}[1/3])dx_1dx_3+h.c.
  \end{aligned}
\end{equation}
For $\psi\in\mathcal{H}^{\wedge N}(\{N_{\varsigma}\})$, we bound it via (\ref{jlinfty}):
\begin{equation}\label{Phi_31bound}
  \begin{aligned}
  \vert\langle\Phi_{31}\psi,\psi\rangle\vert&\lesssim\sum_{\substack{\sigma,\nu,p_3}}
  \chi_{p_3\notin B_F^\sigma}
  \int_{\Lambda^2}W(x_1-x_3)\Vert a_{p_3,\sigma}\psi\Vert\\
  &\quad\quad\times
  \Vert a^*(j_{x_1,x_3,\nu}^{(p_3,\sigma)})a(L_{x_1,\sigma}[1/3])a(h_{x_3,\nu})\psi\Vert dx_1dx_3\\
  &\lesssim\Big(\sum_{\substack{\sigma,\nu,p_3}}
  \chi_{p_3\notin B_F^\sigma}\int_{\Lambda^2}
  W(x_1-x_3)\Vert a_{p_3,\sigma}\psi\Vert^2\Big)^{\frac{1}{2}}\\
  &\quad\quad\times
  \Big(\sum_{\substack{\sigma,\nu}}\int_{\Lambda^2}
  W(x_1-x_3)N^{\frac{5}{3}+2\alpha}\Vert a(h_{x_3,\nu})\psi\Vert^2\Big)^{\frac{1}{2}}\\
  &\lesssim N^{\frac{5}{6}+\alpha}\Vert W\Vert_1
  \langle\mathcal{N}_{ex}\psi,\psi\rangle
  \lesssim N^{-\frac{1}{6}+\alpha}
  \langle\mathcal{N}_{ex}\psi,\psi\rangle.
  \end{aligned}
\end{equation}
As for $\Phi_{32}$, using (\ref{notation j}), and notice that $\zeta_{k,q,p}^{\nu,\sigma}=\zeta_{-k,p,q}^{\sigma,\nu}$, we have
\begin{equation}\label{Phi_32}
  \begin{aligned}
  \Phi_{32}&=\sum_{\substack{\sigma,\nu\\p_3,p_4,q_1,q_2}}
  \chi_{q_1\in B_F^\sigma}\chi_{q_2\in B_F^\nu}
  \chi_{p_3\notin B_F^\sigma}\chi_{p_4\notin B_F^\nu}
  \int_{\Lambda^3}e^{-i(q_1-p_3-p_4+q_2)x_2}\zeta_{p_4-q_2,p_4,p_3}^{\nu,\sigma}\\
  &\quad\quad\times
  e^{iq_1x_1}e^{iq_2x_3}W(x_1-x_3)
  a_{p_3,\sigma}^*a^*_{p_4,\nu}a(h_{x_3,\nu})a(H_{x_1,\sigma}[1/3])dx_1dx_2dx_3+h.c.\\
  &=\sum_{\substack{\sigma,\nu\\p_2,p_3,p_4,q_1,q_2}}
  \chi_{q_1\in B_F^\sigma}\chi_{q_2\in B_F^\nu}
  \chi_{p_3\notin B_F^\sigma}\chi_{p_2,p_4\notin B_F^\nu}
  \int_{\Lambda^2}W_{q_2-p_2}e^{iq_2(x_1-x_2)}\zeta_{q_1-p_3,p_4,p_3}^{\nu,\sigma}\\
  &\quad\quad\times
  e^{iq_1(x_1-x_2)}e^{ip_3x_2}e^{-ip_2x_1}e^{ip_4x_2}
  a_{p_3,\sigma}^*a^*_{p_4,\nu}a_{p_2,\nu}a(H_{x_1,\sigma}[1/3])dx_1dx_2+h.c.\\
  &=\sum_{\substack{\sigma,\nu\\p_2,p_4}}
  \chi_{p_2,p_4\notin B_F^\nu}\int_{\Lambda^2}
  \Big(\sum_{q_2\in B_F^\nu}W_{q_2-p_2}e^{iq_2(x_1-x_2)}\Big)\\
  &\quad\quad\times e^{-ip_2x_1}e^{ip_4x_2}
  a^*(j_{x_2,x_1,\sigma}^{(p_4,\nu)})a_{p_4,\nu}^*a_{p_2,\nu}a(H_{x_1,\sigma}[1/3])
  dx_1dx_2+h.c.
  \end{aligned}
\end{equation}
For technical reasons, we decompose $j_{x_1,x_3,\nu}^{(p_3,\sigma)}$ by
\begin{equation*}
  \begin{aligned}
  j_{x_1,x_3,\nu}^{(p_3,\sigma),h}(z)&\coloneqq
  \sum_{p_4\in P_{F,1/3}^\nu}\Big(\sum_{q_2\in B_F^\nu}\zeta_{p_4-q_2,p_4,p_3}^{\nu,\sigma}
  e^{iq_2(x_3-x_1)}\Big)e^{ip_4x_1}f_{p_4,\nu}(z)\\
  j_{x_1,x_3,\nu}^{(p_3,\sigma),l}(z)&\coloneqq
  \sum_{p_4\in A_{F,1/3}^\nu}\Big(\sum_{q_2\in B_F^\nu}\zeta_{p_4-q_2,p_4,p_3}^{\nu,\sigma}
  e^{iq_2(x_3-x_1)}\Big)e^{ip_4x_1}f_{p_4,\nu}(z)
  \end{aligned}
\end{equation*}
It is easy to verify that both $a(j_{x_1,x_3,\nu}^{(p_3,\sigma),h})$ and $a(j_{x_1,x_3,\nu}^{(p_3,\sigma),l})$ also satisfy (\ref{jl2}) and (\ref{jlinfty}). Notice that by Lemma \ref{lem commutator} (or equivalently by the Fock space formalism (\ref{rewrite psi for fock space}) and formulae (\ref{formulae of crea and annihi ops})), we obtain
\begin{equation}\label{OBS}
  \begin{aligned}
  &(\mathcal{N}_{h}[1/3]+2)^{\frac{1}{2}}a_{p_4,\nu}\chi_{p_4\in P_{F,1/3}^\nu}
  =a_{p_4,\nu}\chi_{p_4\in P_{F,1/3}^\nu}(\mathcal{N}_{h}[1/3]+1)^{\frac{1}{2}},\\
  &(\mathcal{N}_{h}[1/3]+2)^{\frac{1}{2}}a_{p_4,\nu}\chi_{p_4\in A_{F,1/3}^\nu}
  =a_{p_4,\nu}\chi_{p_4\in A_{F,1/3}^\nu}(\mathcal{N}_{h}[1/3]+2)^{\frac{1}{2}}.\\
  &(\mathcal{N}_{h}[1/3]+2)^{-\frac{1}{2}}a^*(j_{x_2,x_1,\sigma}^{(p_4,\nu),h})
  =a^*(j_{x_2,x_1,\sigma}^{(p_4,\nu),h})(\mathcal{N}_{h}[1/3]+3)^{-\frac{1}{2}},\\
  &(\mathcal{N}_{h}[1/3]+2)^{-\frac{1}{2}}a^*(j_{x_2,x_1,\sigma}^{(p_4,\nu),l})
  =a^*(j_{x_2,x_1,\sigma}^{(p_4,\nu),l})(\mathcal{N}_{h}[1/3]+2)^{-\frac{1}{2}}.
  \end{aligned}
\end{equation}
Since
\begin{equation*}
  \chi_{p_4\notin B_F^\nu}a^*(j_{x_2,x_1,\sigma}^{(p_4,\nu)})
  =\big(\chi_{p_4\in P_{F,1/3}^\nu}+\chi_{p_4\in A_{F,1/3}^\nu}\big)
  \big(a^*(j_{x_2,x_1,\sigma}^{(p_4,\nu),h})+a^*(j_{x_2,x_1,\sigma}^{(p_4,\nu),l})\big),
\end{equation*}
we can decompose $\Phi_{32}=\sum_{j=1}^{4}\Phi_{32j}$. Similar observations will also be used during the analysis of $\Phi_4$. We omit further redundant details but write directly the result:
\begin{equation}\label{Phi_32bound}
  \begin{aligned}
  &\vert\langle\Phi_{32}\psi,\psi\rangle\vert\\
  &\lesssim\sum_{\substack{\sigma,\nu\\p_2,p_4}}
  \chi_{p_2,p_4\notin B_F^\nu}\int_{\Lambda^2}
  \Big(\sum_{q_2\in B_F^\nu}W_{q_2-p_2}e^{iq_2(x_1-x_2)}\Big)
  \Vert(\mathcal{N}_{h}[1/3]+2)^{\frac{1}{2}}a_{p_4,\nu}\psi\Vert\\
  &\quad\quad\times
  \Vert (\mathcal{N}_h[1/3]+2)^{-\frac{1}{2}}a^*(j_{x_2,x_1,\sigma}^{(p_4,\nu)})
  a(H_{x_1,\sigma}[1/3])a_{p_2,\nu}\psi\Vert dx_1dx_2\\
  &\lesssim\Big(\sum_{\substack{\sigma,\nu\\p_2,p_4}}
  \chi_{p_2,p_4\notin B_F^\nu}\int_{\Lambda^2}
  \Big\vert\sum_{q_2\in B_F^\nu}W_{q_2-p_2}e^{iq_2(x_1-x_2)}\Big\vert^2
  \frac{\Vert(\mathcal{N}_{h}[1/3]+2)^{\frac{1}{2}}a_{p_4,\nu}\psi\Vert^2}
  {\vert p_2\vert^2-\vert k_F^\nu\vert^2}\Big)^{\frac{1}{2}}\\
  &\quad\quad\times
  \Big(\sum_{\substack{\sigma,\nu\\p_2,p_4}}
  \chi_{p_2,p_4\notin B_F^\nu}\int_{\Lambda^2}
  \big(\vert p_2\vert^2-\vert k_F^\nu\vert^2\big)
  \Vert a^*(j_{x_2,x_1,\sigma}^{(p_4,\nu)})\Vert^2\\
  &\quad\quad\quad\quad\quad\quad\times
  \Vert (\mathcal{N}_h[1/3]+2)^{-\frac{1}{2}}
  a(H_{x_1,\sigma}[1/3])a_{p_2,\nu}\psi\Vert^2\Big)^{\frac{1}{2}}\\
  &\lesssim N^{-\frac{1}{2}+\frac{\alpha}{2}+\frac{\varepsilon}{2}}
  \langle(\mathcal{N}_{h}[1/3]+1)\mathcal{N}_{ex}\psi,\psi\rangle^{\frac{1}{2}}
  \langle\mathcal{K}_s\psi,\psi\rangle^{\frac{1}{2}}
  \lesssim N^{-\frac{1}{3}+\frac{\alpha}{2}+\frac{\varepsilon}{2}}
   \langle\mathcal{K}_s\psi,\psi\rangle
  \end{aligned}
\end{equation}
where we have used (\ref{jl2}) and the first inequality in (\ref{energy est bog}) in Lemma \ref{energy est bog lem}.
\vspace{1em}
\par For $[\Omega,\tilde{B}]$, $\tilde{\Phi}_3$ can be analysed similarly with a slight modification of (\ref{notation j}). For the rest of the estimate, we only need to replace the estimate $\Vert W\Vert_1\lesssim a$ by $N^{\frac{1}{3}}\Vert\nabla\eta\Vert_1\lesssim aN^{-\alpha}$, and use the second inequality in (\ref{energy est bog}) in Lemma \ref{energy est bog lem} rather that the first one. Collecting (\ref{Phi_31bound}) and (\ref{Phi_32bound}), we conclude via Lemmas \ref{control eBtiled on K_s lemma} and \ref{control eBtilde on N_ex lemma}:
\begin{equation}\label{Phi_3 bound}
\begin{aligned}
  \pm\int_{0}^{1}\int_{t}^{1}e^{-s\tilde{B}}
  (\Phi_3+\tilde{\Phi}_3)e^{s\tilde{B}}dsdt
  \lesssim& N^{-\frac{1}{6}+\alpha}(\mathcal{N}_{ex}+1)\\
  &+N^{-\frac{1}{6}-\frac{2}{33}+\alpha+\varepsilon}\big(\mathcal{K}_s+N^{\frac{1}{3}+\alpha}
  \big)
\end{aligned}
\end{equation}
for some $\varepsilon>0$ arbitrary but small.

\begin{flushleft}
    \textbf{Analysis of $\Phi_4$:}
\end{flushleft}
Using the notation defined in (\ref{notation j}), we write $\Phi_4$ by
\begin{equation}\label{Phi_4}
  \begin{aligned}
  \Phi_4&=-2\sum_{\sigma,\nu,\varpi,p_4}\chi_{p_4\notin B_F^\varpi}\int_{\Lambda^3}
  W(x_1-x_3)e^{ip_4x_2}a^*(j_{x_2,x_1,\sigma}^{(p_4,\varpi)})a^*_{p_4,\varpi}\\
  &\quad\quad\times
  a(g_{x_2,\varpi})a^*(g_{x_3,\nu})a(h_{x_3,\nu})a(h_{x_1,\sigma})dx_1dx_2dx_3+h.c.
  \end{aligned}
\end{equation}
We let $\omega=\frac{1}{3}+u$ for some $0<u\leq\frac{1}{6}$ to be determined. Using the fact that
\begin{equation*}
  \begin{aligned}
  a(h_{x_3,\nu})a(h_{x_1,\sigma})=&
  a(H_{x_3,\nu}[\omega])a(H_{x_1,\sigma}[\omega])
  +a(H_{x_3,\nu}[\omega])a(L_{x_1,\sigma}[\omega])\\
  &+a(L_{x_3,\nu}[\omega])a(H_{x_1,\sigma}[\omega])
  +a(L_{x_3,\nu}[\omega])a(L_{x_1,\sigma}[\omega])
  \end{aligned}
\end{equation*}
we decompose $\Phi_4=\sum_{j=1}^{4}\Phi_{4j}$. For $\Phi_{41}$, we first introduce the notation
\begin{equation*}
  b_{x_1,\nu,p_2}(z)\coloneqq\sum_{q_2\in B_F^\nu}W_{q_2-p_2}e^{iq_2x_1}f_{q_2,\nu}.
\end{equation*}
It can be easily verified that
\begin{equation}\label{bl2}
  \Vert a(b_{x_1,\nu,p_2})\Vert^2\leq \sum_{q_2\in B_F^\nu}\vert W_{q_2-p_2}\vert^2.
\end{equation}
We then rewrite $\Phi_{41}$ by
\begin{equation}\label{Phi_41}
  \begin{aligned}
  \Phi_{41}&=-2\sum_{\substack{\sigma,\nu,\varpi\\p_2,p_4}}\chi_{p_2\in P_{F,\omega}^\nu}
  \chi_{p_4\notin B_F^\varpi}\int_{\Lambda^2}e^{-ip_2x_1}e^{ip_4x_2}
  a^*(j_{x_2,x_1,\sigma}^{(p_4,\varpi)})a^*_{p_4,\varpi}\\
  &\quad\quad\times
  a(g_{x_2,\varpi})a^*(b_{x_1,\nu,p_2})a_{p_2,\nu}a(H_{x_1,\sigma}[\omega])dx_1dx_2+h.c.
  \end{aligned}
\end{equation}
With an observation similar to (\ref{OBS}), we bound it for any $\psi\in\mathcal{H}^{\wedge N}(\{N_{\varsigma_i}\})$:
\begin{equation}\label{Phi_41 bound}
  \begin{aligned}
  &\vert\langle\Phi_{41}\psi,\psi\rangle\vert\\
  &\lesssim\sum_{\substack{\sigma,\nu,\varpi\\p_2,p_4}}\chi_{p_2\in P_{F,\omega}^\nu}
  \chi_{p_4\notin B_F^\varpi}\int_{\Lambda^2}
  \Vert a(b_{x_1,\nu,p_2})a^*(g_{x_2,\varpi})(\mathcal{N}_h[\omega]+2)^{\frac{1}{2}}
  a_{p_4,\varpi}\psi\Vert\\
  &\quad\quad\times
  \Vert(\mathcal{N}_h[\omega]+2)^{-\frac{1}{2}}a^*(j_{x_2,x_1,\sigma}^{p_4,\varpi})
  a(H_{x_1,\sigma}[\omega])a_{p_2,\nu}\psi\Vert dx_1dx_2\\
  &\lesssim\Big(\sum_{\substack{\sigma,\nu,\varpi\\p_2,p_4}}\chi_{p_2\in P_{F,\omega}^\nu}
  \chi_{p_4\notin B_F^\varpi}\int_{\Lambda^2}
  \sum_{q_2\in B_F^\nu}\frac{N\vert W_{q_2-p_2}\vert^2}{\vert p_2\vert^2-
  \vert k_F^\nu\vert^2}
  \Vert (\mathcal{N}_h[\omega]+2)^{\frac{1}{2}}
  a_{p_4,\varpi}\psi\Vert^2\Big)^{\frac{1}{2}}\\
  &\quad\quad\times
  \Big(\sum_{\substack{\sigma,\nu,\varpi\\p_2,p_4}}\chi_{p_2\in P_{F,\omega}^\nu}
  \chi_{p_4\notin B_F^\varpi}\int_{\Lambda^2}
  \big(\vert p_2\vert^2-\vert k_F^\nu\vert^2\big)
  \Vert a^*(j_{x_2,x_1,\sigma}^{p_4,\varpi})\Vert^2\\
  &\quad\quad\quad\quad\quad\quad\times
   \Vert(\mathcal{N}_h[\omega]+2)^{-\frac{1}{2}}
  a(H_{x_1,\sigma}[\omega])a_{p_2,\nu}\psi\Vert^2\Big)^{\frac{1}{2}}\\
  &\lesssim N^{\frac{1}{6}+\alpha+\frac{\varepsilon}{2}-\frac{\omega}{2}}
  \langle(\mathcal{N}_h[\omega]+2)\mathcal{N}_{ex}\psi,\psi\rangle^{\frac{1}{2}}
  \langle\mathcal{K}_s\psi,\psi\rangle^{\frac{1}{2}}
  \lesssim N^{\frac{1}{6}+\alpha+\frac{\varepsilon}{2}-\frac{3}{2}u}
  \langle\mathcal{K}_s\psi,\psi\rangle
  \end{aligned}
\end{equation}
where we have used (\ref{jl2}), (\ref{bl2}) and (\ref{energy est bog u}) in Lemma \ref{energy est bog lem}, and we have set $\omega=\frac{1}{3}+u$ with $0<u<\frac{1}{6}$. We then choose
\begin{equation}\label{choose ubog}
  u=\frac{1}{9}+\frac{4}{3}\alpha+\frac{2}{3}\gamma+\frac{1}{3}\varepsilon
\end{equation}
and we further demand $0<\gamma<\alpha<\frac{1}{42}$ and $\varepsilon>0$ small enough but fixed, such that $0<u<\frac{1}{6}$. Thence we conclude that
\begin{equation}\label{Phi_41boundfinale}
  \pm\Phi_{41}\lesssim N^{-\alpha-\gamma}\mathcal{K}_s.
\end{equation}
\par The evaluations of $\Phi_{4j}$ for $j=2,3,4$ are similar, and we take $\Phi_{44}$ for example. We take $0<\varrho<\frac{1}{3}$ to be determined. Using the fact that $\chi_{p_4\notin B_F^\varpi}=\chi_{p_4\in A_{F,\varrho}^\varpi}+\chi_{p_4\in P_{F,\varrho}^\varpi}$, we can decompose $\Phi_{44}$ by
\begin{equation}\label{Phi_44j}
  \begin{aligned}
  \Phi_{441}&=2\sum_{\sigma,\nu,\varpi,p_4}\chi_{p_4\in A_{F,\varrho}^\varpi}\int_{\Lambda^3}
  W(x_1-x_3)e^{ip_4x_2}a^*_{p_4,\varpi}a(g_{x_2,\varpi})a^*(g_{x_3,\nu})\\
  &\quad\quad\times
  a^*(j_{x_2,x_1,\sigma}^{(p_4,\varpi)})
  a(L_{x_3,\nu}[\omega])a(L_{x_1,\sigma}[\omega])dx_1dx_2dx_3+h.c.\\
  \Phi_{442}&=2\sum_{\sigma,\nu,\varpi,p_4}\chi_{p_4\in P_{F,\varrho}^\varpi}\int_{\Lambda^3}
  W(x_1-x_3)e^{ip_4x_2}a^*_{p_4,\varpi}a(g_{x_2,\varpi})a^*(g_{x_3,\nu})\\
  &\quad\quad\times
  a^*(j_{x_2,x_1,\sigma}^{(p_4,\varpi)})
  a(L_{x_3,\nu}[\omega])a(L_{x_1,\sigma}[\omega])dx_1dx_2dx_3+h.c.
  \end{aligned}
\end{equation}
For the bound of $\Phi_{441}$, we let $-\frac{11}{48}<r,r^\prime<\frac{1}{3}$ to be chosen later, and introduce the notation
\begin{equation}\label{define Mbog}
  M_{x,\sigma}[r,\omega](z)=\sum_{k\in A_{F,\omega}^\sigma\backslash A_{F,r}^\sigma}e^{ikx}
  f_{k,\sigma}(z).
\end{equation}
It is easy to notice that $L_{x,\sigma}[\omega]=L_{x,\sigma}[r]+ M_{x,\sigma}[r,\omega]$, and by Lemma \ref{lemma ineqn K_s},
\begin{equation}\label{ineq M}
  \int_{\Lambda}\sum_{\sigma}a^*(M_{x,\sigma}[r,\omega])a(M_{x,\sigma}[r,\omega])dx
  \leq \mathcal{N}_{h}[r]\lesssim N^{-\frac{1}{3}-r}\mathcal{K}_s.
\end{equation}
On the other hand, by (\ref{jl2}) and (\ref{useful bound}-\ref{xi l2preposition}) in Lemma \ref{xi l2 lemma}, we have
\begin{equation}\label{jl2var}
  \begin{aligned}
  &\sum_{p_4\in A_{F,\varrho}^\varpi}\int_{\Lambda}\Vert a^*(j_{x_2,x_1,\sigma}^{(p_4,\varpi)})\Vert^2dx_2\leq
  \sum_{p_4\in A_{F,\varrho}^\varpi}\sum_{p_3\notin B_F^\sigma}\int_{\Lambda}
  \Big\vert\sum_{q_1\in B_F^\sigma}\zeta_{q_1-p_3,p_4,p_3}^{\varpi,\sigma}
  e^{iq_2(x_1-x_2)}\Big\vert^2dx_2\\
  &\leq\sum_{p_4\in A_{F,\varrho}^\varpi}\sum_{p_3\notin B_F^\sigma}
  \sum_{k}\vert\zeta_{k,p_4,p_3}^{\varpi,\sigma}\vert^2
  \leq\sum_{p_4\in A_{F,\varrho}^\varpi}\sum_{q_3\in B_F^\sigma}
  \sum_{k}\frac{a^2\chi_{q_3-k\notin B_F^\sigma}}
  {\big(\vert q_3-k\vert^2-\vert q_3\vert^2\big)^2}\lesssim N^{-\frac{2}{3}
  +\varrho+\varepsilon}.
  \end{aligned}
\end{equation}
With (\ref{define Mbog}-\ref{jl2var}) prepared above, we decompose $\Phi_{441}=\sum_{j=1}^{4}
\Phi_{441j}$ using the fact that
\begin{equation*}
  \begin{aligned}
  a(g_{x_2,\varpi})a(L_{x_1,\sigma}[\omega])=&
  a(I_{x_2,\varpi}[r^\prime])a(M_{x_1,\sigma}[r,\omega])
  +a(I_{x_2,\varpi}[r^\prime])a(L_{x_1,\sigma}[r])\\
  &+a(S_{x_2,\varpi}[r^\prime])a(M_{x_1,\sigma}[r,\omega])
  +a(S_{x_2,\varpi}[r^\prime])a(L_{x_1,\sigma}[r]).
  \end{aligned}
\end{equation*}
For $\Phi_{4411}$, using (\ref{ineq M}) and (\ref{jl2var}), and recall that $\omega=\frac{1}{3}+u$ with $u$ defined in (\ref{choose ubog}), we bound it by
\begin{equation}\label{Phi_4411bound}
  \begin{aligned}
  \vert\langle\Phi_{4411}\psi,\psi\rangle\vert&\lesssim
  \sum_{\sigma,\nu,\varpi,p_4}\chi_{p_4\in A_{F,\varrho}^\varpi}\int_{\Lambda^3}
  W(x_1-x_3)\Vert a(g_{x_3,\nu})a_{p_4,\varpi}a^*(I_{x_2,\varpi}[r^\prime])\psi\Vert\\
  &\quad\quad\times
  \Vert a^*(j_{x_2,x_1,\sigma}^{(p_4,\varpi)})
  a(L_{x_3,\nu}[\omega])a(M_{x_1,\sigma}[r,\omega])\psi\Vert dx_1dx_2dx_3\\
  &\lesssim\Big(\sum_{\sigma,\nu,\varpi,p_4}\chi_{p_4\in A_{F,\varrho}^\varpi}\int_{\Lambda^3}
  W(x_1-x_3)N\Vert a^*(I_{x_2,\varpi}[r^\prime])\psi\Vert^2\Big)^{\frac{1}{2}}\\
  &\quad\quad\times\Big(\sum_{\sigma,\nu,\varpi}\int_{\Lambda^2}
  W(x_1-x_3)N^{-\frac{2}{3}+\varrho+3\omega+\varepsilon}
  \Vert a(M_{x_1,\sigma}[r,\omega])\psi\Vert^2\Big)^{\frac{1}{2}}\\
  &\lesssim N^{\frac{2}{3}+\varrho+\frac{3}{2}u-\frac{r}{2}-\frac{r^\prime}{2}+\frac{\varepsilon}{2}}
  \Vert W\Vert_1\langle\mathcal{K}_s\psi,\psi\rangle
  \lesssim
  N^{-\frac{1}{3}+\varrho+\frac{3}{2}u-\frac{r}{2}-\frac{r^\prime}{2}+\frac{\varepsilon}{2}}
  \langle\mathcal{K}_s\psi,\psi\rangle.
  \end{aligned}
\end{equation}
For $\Phi_{4412}$, we need to use the fact that $\mathcal{N}_l[\omega]\leq\mathcal{N}_{ex}$:
\begin{equation}\label{Phi_4412bound}
  \begin{aligned}
  \vert\langle\Phi_{4412}\psi,\psi\rangle\vert&\lesssim
  \sum_{\sigma,\nu,\varpi,p_4}\chi_{p_4\in A_{F,\varrho}^\varpi}\int_{\Lambda^3}
  W(x_1-x_3)\Vert a(g_{x_3,\nu})a_{p_4,\varpi}a^*(I_{x_2,\varpi}[r^\prime])\psi\Vert\\
  &\quad\quad\times
  \Vert a^*(j_{x_2,x_1,\sigma}^{(p_4,\varpi)})a(L_{x_1,\sigma}[r])
  a(L_{x_3,\nu}[\omega])\psi\Vert dx_1dx_2dx_3\\
  &\lesssim\Big(\sum_{\sigma,\nu,\varpi,p_4}\chi_{p_4\in A_{F,\varrho}^\varpi}\int_{\Lambda^3}
  W(x_1-x_3)N\Vert a^*(I_{x_2,\varpi}[r^\prime])\psi\Vert^2\Big)^{\frac{1}{2}}\\
  &\quad\quad\times\Big(\sum_{\sigma,\nu,\varpi}\int_{\Lambda^2}
  W(x_1-x_3)N^{\varrho+r+\varepsilon}
  \Vert a(L_{x_3,\sigma}[\omega])\psi\Vert^2\Big)^{\frac{1}{2}}\\
  &\lesssim N^{\frac{2}{3}+\varrho+\frac{r}{2}-\frac{r^\prime}{2}+\frac{\varepsilon}{2}}
  \Vert W\Vert_1\langle\mathcal{K}_s\psi,\psi\rangle^\frac{1}{2}
  \langle\mathcal{N}_{ex}\psi,\psi\rangle^{\frac{1}{2}}\\
  &\lesssim
  N^{-\frac{1}{3}+\varrho+\frac{3}{2}u-\frac{r}{2}-\frac{r^\prime}{2}+\frac{\varepsilon}{2}}
  \langle\mathcal{K}_s\psi,\psi\rangle+
   N^{-\frac{1}{3}+\varrho-\frac{3}{2}u+\frac{3}{2}r-\frac{r^\prime}{2}+\frac{\varepsilon}{2}}
  \langle\mathcal{N}_{ex}\psi,\psi\rangle.
  \end{aligned}
\end{equation}
Similarly, we have
\begin{equation}\label{Phi_4413bound}
  \begin{aligned}
  \vert\langle\Phi_{4413}\psi,\psi\rangle\vert&\lesssim
  \sum_{\sigma,\nu,\varpi,p_4}\chi_{p_4\in A_{F,\varrho}^\varpi}\int_{\Lambda^3}
  W(x_1-x_3)\Vert a(g_{x_3,\nu})a^*(S_{x_2,\varpi}[r^\prime])a_{p_4,\varpi}\psi\Vert\\
  &\quad\quad\times
  \Vert a^*(j_{x_2,x_1,\sigma}^{(p_4,\varpi)})
  a(L_{x_3,\nu}[\omega])a(M_{x_1,\sigma}[r,\omega])\psi\Vert dx_1dx_2dx_3\\
  &\lesssim\Big(\sum_{\sigma,\nu,\varpi,p_4}\chi_{p_4\in B_{F}^\varpi}\int_{\Lambda^3}
  W(x_1-x_3)N^{\frac{5}{3}+r^\prime}\Vert a_{p_4,\varpi}\psi\Vert^2\Big)^{\frac{1}{2}}\\
  &\quad\quad\times\Big(\sum_{\sigma,\nu,\varpi}\int_{\Lambda^2}
  W(x_1-x_3)N^{-\frac{2}{3}+\varrho+3\omega+\varepsilon}
  \Vert a(M_{x_1,\sigma}[r,\omega])\psi\Vert^2\Big)^{\frac{1}{2}}\\
  &\lesssim N^{\frac{5}{6}+\frac{\varrho}{2}+\frac{3}{2}u+
  \frac{r^\prime}{2}-\frac{r}{2}+\frac{\varepsilon}{2}}
  \Vert W\Vert_1\langle\mathcal{K}_s\psi,\psi\rangle^\frac{1}{2}
  \langle\mathcal{N}_{ex}\psi,\psi\rangle^{\frac{1}{2}}\\
  &\lesssim
  N^{-\frac{1}{6}+\frac{\varrho}{2}+3u-\frac{3}{2}r
  +\frac{r^\prime}{2}+\frac{\varepsilon}{2}}
  \langle\mathcal{K}_s\psi,\psi\rangle+
   N^{-\frac{1}{6}+\frac{\varrho}{2}+\frac{r}{2}+\frac{r^\prime}{2}+\frac{\varepsilon}{2}}
  \langle\mathcal{N}_{ex}\psi,\psi\rangle
  \end{aligned}
\end{equation}
and
\begin{equation}\label{Phi_4414bound}
  \begin{aligned}
  \vert\langle\Phi_{4414}\psi,\psi\rangle\vert&\lesssim
  \sum_{\sigma,\nu,\varpi,p_4}\chi_{p_4\in A_{F,\varrho}^\varpi}\int_{\Lambda^3}
  W(x_1-x_3)\Vert a(g_{x_3,\nu})a^*(S_{x_2,\varpi}[r^\prime])a_{p_4,\varpi}\psi\Vert\\
  &\quad\quad\times
  \Vert a^*(j_{x_2,x_1,\sigma}^{(p_4,\varpi)})a(L_{x_1,\sigma}[r])
  a(L_{x_3,\nu}[\omega])\psi\Vert dx_1dx_2dx_3\\
  &\lesssim\Big(\sum_{\sigma,\nu,\varpi,p_4}\chi_{p_4\in B_{F}^\varpi}\int_{\Lambda^3}
  W(x_1-x_3)N^{\frac{5}{3}+r^\prime}\Vert a_{p_4,\varpi}\psi\Vert^2\Big)^{\frac{1}{2}}\\
  &\quad\quad\times\Big(\sum_{\sigma,\nu,\varpi}\int_{\Lambda^2}
  W(x_1-x_3)N^{\varrho+r+\varepsilon}
  \Vert a(L_{x_3,\nu}[\omega])\psi\Vert^2\Big)^{\frac{1}{2}}\\
  &\lesssim N^{\frac{5}{6}+\frac{\varrho}{2}+
  \frac{r^\prime}{2}+\frac{r}{2}+\frac{\varepsilon}{2}}
  \Vert W\Vert_1
  \langle\mathcal{N}_{ex}\psi,\psi\rangle\lesssim
   N^{-\frac{1}{6}+\frac{\varrho}{2}+\frac{r}{2}+\frac{r^\prime}{2}+\frac{\varepsilon}{2}}
  \langle\mathcal{N}_{ex}\psi,\psi\rangle
  \end{aligned}
\end{equation}
We then choose
\begin{equation}\label{choose rrprime}
  \begin{aligned}
  r&=-\frac{1}{4}+\frac{3}{4}\varrho+\frac{9}{4}u+\alpha+\gamma+\frac{\varepsilon}{2}
  =\frac{3}{4}\varrho+4\alpha+\frac{5}{2}\gamma+\frac{5}{4}\varepsilon\\
  r^\prime&=-\frac{5}{12}+\frac{5}{4}\varrho+\frac{3}{4}u+\alpha+\gamma+\frac{\varepsilon}{2}
  =-\frac{1}{3}+\frac{5}{4}\varrho+2\alpha+\frac{3}{2}\gamma+\frac{3}{4}\varepsilon
  \end{aligned}
\end{equation}
We will verify later that we indeed have $-\frac{11}{48}<r,r^\prime<\frac{1}{3}$. Combining (\ref{Phi_4411bound}-\ref{choose rrprime}), we conclude that
\begin{equation}\label{Phi_441bound}
  \pm\Phi_{441}\lesssim N^{-\alpha-\gamma}\mathcal{K}_s
  +N^{-\frac{1}{3}+\frac{3}{2}\varrho+3\alpha+2\gamma+\frac{3}{2}\varepsilon}\mathcal{N}_{ex}.
\end{equation}
\par For $\Phi_{442}$, we follow the notation given in (\ref{define Mbog}) and decompose $\Phi_{442}=\Phi_{4421}+\Phi_{4422}$ using the fact that
\begin{equation*}
  a(L_{x_1,\sigma})=a(M_{x_1,\sigma}[\vartheta,\omega])+a(L_{x_1,\sigma}[\vartheta]),
\end{equation*}
where $-\frac{11}{48}<\vartheta<\frac{1}{3}$ to be determined later. Moreover, by the second inequality of (\ref{xi l2 cut off}) in Lemma \ref{xi l2 lemma} (since we demand $0<\varrho<\frac{1}{3}$), together with (\ref{jl2}), we have
\begin{equation}\label{jl2cut}
  \begin{aligned}
  &\sum_{p_4\in P_{F,\varrho}^\varpi}
  \frac{1}{\vert p_4\vert^2-\vert k_F^\varpi\vert^2}
  \int_{\Lambda}\Vert a^*(j_{x_2,x_1,\sigma}^{(p_4,\varpi)})\Vert^2dx_2
  \leq\sum_{p_4\in P_{F,\varrho}^\varpi}\sum_{p_3\notin B_F^\sigma}
  \sum_{k}\frac{\vert\zeta_{k,p_4,p_3}^{\varpi,\sigma}\vert^2}
  {\vert p_4\vert^2-\vert k_F^\varpi\vert^2}\\
  &\leq\sum_{q_3\in B_F^\sigma}\sum_{q_4\in B_F^\varpi}
  \sum_{k}\frac{\vert\xi_{k,q_4,q_3}^{\varpi,\sigma}\vert^2
  \chi_{q_4+k\in P_{F,\varrho}^\varpi}}
  {\vert q_4+k\vert^2-\vert k_F^\varpi\vert^2}
  \lesssim N^{-\frac{2}{3}-\varrho}\ln N.
  \end{aligned}
\end{equation}
Therefore, by (\ref{ineq M}) and (\ref{jl2cut}), and recall that $\omega=\frac{1}{3}+u$ with $u>0$ chosen as in (\ref{choose ubog}), we have
\begin{equation}\label{Phi_4421bound}
  \begin{aligned}
  \vert\langle\Phi_{4421}\psi&,\psi\rangle\vert\lesssim
  \sum_{\sigma,\nu,\varpi,p_4}\chi_{p_4\in P_{F,\varrho}^\varpi}
  \int_{\Lambda^3}W(x_1-x_3)\Vert a(g_{x_3,\nu})a^*(g_{x_2,\varpi})a_{p_4,\varpi}\psi\Vert\\
  &\quad\quad\quad\quad\times
  \Vert a^*(j_{x_2,x_1,\sigma}^{p_4,\varpi})a(L_{x_3,\nu}[\omega])a(M_{x_1,\sigma}
  [\vartheta,\omega])\psi\Vert dx_1dx_2dx_3\\
  &\lesssim\Big(\sum_{\sigma,\nu,\varpi,p_4}\chi_{p_4\in P_{F,\varrho}^\varpi}
  \int_{\Lambda^3}W(x_1-x_3)N^2
  \big(\vert p_4\vert^2-\vert k_F^\varpi\vert^2\big)
  \Vert a_{p_4,\varpi}\psi\Vert^2\Big)^{\frac{1}{2}}\\
  &\quad\quad\times
  \Big(\sum_{\sigma,\nu,\varpi}
  \int_{\Lambda^3}W(x_1-x_3)N^{-\frac{2}{3}-\varrho+3\omega}\ln N
  \Vert a(M_{x_1,\sigma}
  [\vartheta,\omega])\psi\Vert^2\Big)^{\frac{1}{2}}\\
  &\lesssim N^{1+\frac{3}{2}u-\frac{\vartheta}{2}-\frac{\varrho}{2}}
  (\ln N)^{\frac{1}{2}}\Vert W\Vert_1
  \langle\mathcal{K}_s\psi,\psi\rangle
  \lesssim N^{\frac{3}{2}u-\frac{\vartheta}{2}-\frac{\varrho}{2}}(\ln N)^{\frac{1}{2}}
  \langle\mathcal{K}_s\psi,\psi\rangle.
  \end{aligned}
\end{equation}
Notice that $\mathcal{N}_l[\omega]\leq\mathcal{N}_{ex}$, we have
\begin{equation}\label{Phi_4422bound}
  \begin{aligned}
  \vert\langle\Phi_{4422}\psi&,\psi\rangle\vert\lesssim
  \sum_{\sigma,\nu,\varpi,p_4}\chi_{p_4\in P_{F,\varrho}^\varpi}
  \int_{\Lambda^3}W(x_1-x_3)\Vert a(g_{x_3,\nu})a^*(g_{x_2,\varpi})a_{p_4,\varpi}\psi\Vert\\
  &\quad\quad\quad\quad\times
  \Vert a^*(j_{x_2,x_1,\sigma}^{p_4,\varpi})a(L_{x_1,\sigma}
  [\vartheta])a(L_{x_3,\nu}[\omega])\psi\Vert dx_1dx_2dx_3\\
  &\lesssim\Big(\sum_{\sigma,\nu,\varpi,p_4}\chi_{p_4\in P_{F,\varrho}^\varpi}
  \int_{\Lambda^3}W(x_1-x_3)N^2
  \big(\vert p_4\vert^2-\vert k_F^\varpi\vert^2\big)
  \Vert a_{p_4,\varpi}\psi\Vert^2\Big)^{\frac{1}{2}}\\
  &\quad\quad\times
  \Big(\sum_{\sigma,\nu,\varpi}
  \int_{\Lambda^3}W(x_1-x_3)N^{\vartheta-\varrho}\ln N
  \Vert a(L_{x_3,\nu}[\omega])\psi\Vert^2\Big)^{\frac{1}{2}}\\
  &\lesssim N^{1+\frac{\vartheta}{2}-\frac{\varrho}{2}}
  (\ln N)^{\frac{1}{2}}\Vert W\Vert_1
  \langle\mathcal{K}_s\psi,\psi\rangle^\frac{1}{2}
  \langle\mathcal{N}_{ex}\psi,\psi\rangle^\frac{1}{2}\\
  &\lesssim N^{\frac{3}{2}u-\frac{\vartheta}{2}-\frac{\varrho}{2}}(\ln N)^{\frac{1}{2}}
  \langle\mathcal{K}_s\psi,\psi\rangle
  +N^{-\frac{3}{2}u+\frac{3}{2}\vartheta-\frac{\varrho}{2}}(\ln N)^{\frac{1}{2}}
  \langle\mathcal{N}_{ex}\psi,\psi\rangle.
  \end{aligned}
\end{equation}
We then choose
\begin{equation}\label{choose theta}
  \vartheta=3u-\varrho+2\alpha+2\gamma+2\varepsilon
  =\frac{1}{3}-\varrho+6\alpha+4\gamma+4\varepsilon.
\end{equation}
We will verify later that we indeed have $-\frac{11}{48}<\vartheta<\frac{1}{3}$. Combining (\ref{Phi_4421bound}-\ref{choose theta}), we conclude that
\begin{equation}\label{Phi_442bound}
  \pm\Phi_{442}\lesssim N^{-\alpha-\gamma}\mathcal{K}_s
  +N^{\frac{1}{3}-2\varrho+7\alpha+5\gamma+5\varepsilon}\mathcal{N}_{ex}.
\end{equation}
To optimize (\ref{Phi_441bound}) and (\ref{Phi_442bound}) (we demand $\varepsilon>0$ small enough), we choose
\begin{equation}\label{choose rho}
  \varrho=\frac{2}{7}\Big(\frac{2}{3}+4\alpha+3\gamma\Big).
\end{equation}
Since we have required $0<\gamma<\alpha<\frac{1}{42}$, indeed $0<\varrho<\frac{1}{3}$, plugging (\ref{choose rho}) into (\ref{choose rrprime}) and (\ref{choose theta}), we have
\begin{equation}\label{rrprimetheta}
  \begin{aligned}
  r&=\frac{1}{7}+\frac{34}{7}\alpha+\frac{22}{7}\gamma+\frac{5}{4}\varepsilon\\
  r^\prime&=-\frac{2}{21}+\frac{24}{7}\alpha+\frac{18}{7}\gamma+\frac{3}{4}\varepsilon\\
  \vartheta&=\frac{1}{7}+\frac{34}{7}\alpha+\frac{22}{7}\gamma+4\varepsilon
  \end{aligned}
\end{equation}
and we can verify that indeed $-\frac{11}{48}<r,r^\prime,\vartheta<\frac{1}{3}$ since we demand $\varepsilon>0$ small enough and fixed. We then conclude that
\begin{equation}\label{Phi_4bound}
  \pm\Phi_4\lesssim N^{-\alpha-\gamma}\mathcal{K}_s+
  N^{-\frac{1}{21}+\frac{33}{7}\alpha+\frac{23}{7}\gamma+\varepsilon}\mathcal{N}_{ex}
\end{equation}
with the further requirement that $0<\gamma<\alpha<\frac{1}{168}$, and for some $\varepsilon>0$ small enough but fixed.
\vspace{1em}
\par  For $[\Omega,\tilde{B}]$, $\tilde{\Phi}_4$ can be analysed similarly with a slight modification of (\ref{notation j}). For the rest of the estimate, we only need to replace the estimate $\Vert W\Vert_1\lesssim a$ by $N^{\frac{1}{3}}\Vert\nabla\eta\Vert_1\lesssim aN^{-\alpha}$, and use the second inequality in (\ref{energy est bog u}) in Lemma \ref{energy est bog lem} rather that the first one. By (\ref{Phi_4bound}), we conclude via Lemmas \ref{control eBtiled on K_s lemma} and \ref{control eBtilde on N_ex lemma}:
\begin{equation}\label{Phi_4 bound}
\begin{aligned}
  \pm\int_{0}^{1}\int_{t}^{1}e^{-s\tilde{B}}
  (\Phi_4+\tilde{\Phi}_4)e^{s\tilde{B}}&dsdt
  \lesssim N^{-\frac{1}{21}+\frac{33}{7}\alpha+\frac{23}{7}\gamma+\varepsilon}
  (\mathcal{N}_{ex}+1)\\
  &+\big(N^{-\alpha-\gamma}+
  N^{-\frac{1}{21}-\frac{2}{33}+\frac{33}{7}\alpha+\frac{23}{7}\gamma+\varepsilon}
  \big)\big(\mathcal{K}_s+N^{\frac{1}{3}+\alpha}
  \big)
\end{aligned}
\end{equation}
with the further requirement that $0<\gamma<\alpha<\frac{1}{168}$, and for some $\varepsilon>0$ arbitrary but small.

\begin{flushleft}
  \textbf{Conclusion of Lemma \ref{cal com bog lemma}:}
\end{flushleft}
Collecting (\ref{Phi_1 conclusion}), (\ref{E_phi_1}), (\ref{Phi_2boundfinale}), (\ref{Phi_3 bound}) and (\ref{Phi_4 bound}), we close the proof of Lemma \ref{cal com bog lemma}.
\end{proof}

\begin{proof}[Proof of Proposition \ref{p3}]
  \par Proposition \ref{p3} is deduced by combining Lemmas \ref{control eBtiled on K_s lemma}-\ref{cal com bog lemma} and equation (\ref{define Z_N}).
\end{proof}

\section*{Acknowledgments}
X.Chen is supported by NSF grant DMS-2406620, and Z. Zhang is partially supported by the National Key R\&D Program of China under Grant 2023YFA1008801 and NSF of China under Grant 12288101. We deeply appreciate the delightful and fruitful discussion with C.Brennecke and B.Schlein, who gave us many kind and constructive advices and comments, and different view points to the mathematics and physics.

\appendix
\setcounter{equation}{0}

\renewcommand\theequation{A.\arabic{equation}} 
\section{Analysis of (\ref{2nd})}\label{analysis}
\par We can rewrite
\begin{equation}\label{app2}
  \int_{\vert q\vert<r_0}\int_{\vert p\vert<r_0}\int_{k\in\mathbb{R}^3}
      \Big(\frac{1}{\vert k\vert^2}-
      \frac{2\chi_{p-k\notin B^{\sigma}_F}\chi_{q+k\notin B^{\nu}_F}}
    {\big(\vert q+k\vert^2+\vert p-k\vert^2-\vert q\vert^2-\vert p\vert^2\big)}\Big)
    =I_1+I_2
\end{equation}
with
\begin{equation}\label{app3}
  \begin{aligned}
  I_1\coloneqq&\int_{\vert q\vert<r_0}\int_{\vert p\vert<r_0}\int_{k\in\mathbb{R}^3}
      \Big(\frac{1}{\vert k\vert^2}-
      \frac{2}
    {\big(\vert q+k\vert^2+\vert p-k\vert^2-\vert q\vert^2-\vert p\vert^2\big)}\Big)\\
  I_2\coloneqq&\int_{\vert q\vert<r_0}\int_{\vert p\vert<r_0}\int_{k\in\mathbb{R}^3}
      \Big(\frac{2(1-\chi_{p-k\notin B^{\sigma}_F}\chi_{q+k\notin B^{\nu}_F})}
    {\big(\vert q+k\vert^2+\vert p-k\vert^2-\vert q\vert^2-\vert p\vert^2\big)}\Big)
  \end{aligned}
\end{equation}
$I_2$ is exactly the integral in the middle of (\ref{2nd}). In fact, we use
\begin{equation}\label{app4}
  1-\chi_{p-k\notin B^{\sigma}_F}\chi_{q+k\notin B^{\nu}_F}=
  \chi_{p-k\in B^{\sigma}_F}+\chi_{q+k\in B^{\nu}_F}-\chi_{p-k\in B^{\sigma}_F}\chi_{q+k\in B^{\nu}_F}
\end{equation}
to write $I_2=I_{21}+I_{22}+I_{23}$. By symmetry, it is easy to check that $I_{21}=I_{22}$ and $I_{23}=0$. By a simple change of variables we conclude that
\begin{equation}\label{app5}
  I_2= -4\int_{\vert x_1\vert<r_0}
        \int_{\vert x_2\vert<r_0}
        \int_{\vert x_3\vert<r_0}
        \int_{x_4}
        \frac{\delta(x_1+x_2=x_3+x_4)}{\vert x_1\vert^2+\vert x_2\vert^2
        -\vert x_3\vert^2-\vert x_4\vert^2}.
\end{equation}
On the other hand, notice that
\begin{equation}\label{app1}
  \frac{1}{2}\big(\vert q+k\vert^2+\vert p-k\vert^2-\vert q\vert^2-\vert p\vert^2\big)
  =\vert k\vert^2+k\cdot(q-p)
  =\Big\vert k+\frac{q-p}{2}\Big\vert^2-\Big\vert\frac{q-p}{2}\Big\vert^2.
\end{equation}
Since the Fourier transform of a radial function $f$ in $\mathbb{R}^n$ is given by
\begin{equation}\label{fourier transform radial}
 \widehat{f}(\xi)=\frac{2\pi}{\vert\xi\vert^{\frac{n-1}{2}}}
 \int_{0}^{\infty}f(r)J_{\frac{n}{2}-1}\big(2\pi\vert\xi\vert r\big)r^{\frac{n}{2}}dr.
\end{equation}
If we denote the distribution for $p,q$ fixed and $x\in\mathbb{R}^3$
\begin{equation}\label{define F distribution}
  F(x)\coloneqq \frac{1}{\vert x\vert^2}-
      \frac{2}
    {\big(\vert q+x\vert^2+\vert p-x\vert^2-\vert q\vert^2-\vert p\vert^2\big)}
\end{equation}
Using (\ref{app1}) and (\ref{fourier transform radial}) (in the sense of improper integral), and the fact that for $a>0$
\begin{equation}\label{app0000}
  \int_{0}^{\infty}\frac{r\sin(ar)}{b^2-r^2}dr=-\frac{\pi}{2}\cos(ab),
\end{equation}
we can calculate $\widehat{F}(\xi)$
\begin{equation}\label{fourier F}
  \widehat{F}(\xi)=\frac{\pi}{\vert\xi\vert}\Big(1-e^{\pi i(q-p)\xi}
  \cos\big(\pi\vert q-p\vert\cdot\vert\xi\vert\big)\Big).
\end{equation}
Since the limit
\begin{equation}\label{appp}
  \lim_{\vert\xi\vert\to0}\re \widehat{F}(\xi)=0
\end{equation}
exists, and obviously $I_1$ is real-valued, then we can check that
\begin{equation}\label{apdpfp}
  I_1=\re I_1=\lim_{\vert\xi\vert\to0}\re \widehat{F}(\xi)=0.
\end{equation}
This can be proved by testing $\widehat{F}$ with a family of approximation identity and taking the real part. We omit further details.

\renewcommand\theequation{B.\arabic{equation}} 
\section{Proof of (\ref{claim appB})}\label{claim}
\par Recall that
\begin{equation}\label{f_s_N}
  f_{s_N}=\bigwedge_{\sigma,k\in B_F^{\sigma}}f_{k,\sigma},
\end{equation}
where in Appendix \ref{claim}, $s_N$ always stands for the unique element of $\mathcal{P}^{(N)}$ defined in Section \ref{define P^(N)}. By direct calculation, we have
\begin{equation}\label{B1}
  \begin{aligned}
  \langle\mathcal{N}_{ex}^ne^{t\tilde{B}}f_{s_N},e^{t\tilde{B}}f_{s_N}\rangle
  &=\sum_{m_1,m_2=0}^\infty\frac{t^{(m_1+m_2)}}{m_1!m_2!}
  \langle\mathcal{N}_{ex}^nB^{m_1}f_{s_N},B^{m_2}f_{s_N}\rangle.
  \end{aligned}
\end{equation}
Recall the definition of $\tilde{B}$ in (\ref{define B tilde}), (\ref{define A tilde}) and (\ref{define xi_k,q,p,nu,sigma}), if we adopt the notations
\begin{equation}\label{define mathbbA}
  \mathbb{A}_{1}=\tilde{A},\quad\mathbb{A}_{-1}=-\tilde{A}^*,
\end{equation}
and we let for $m\in\mathbb{N}$,
\begin{equation}\label{define d_m}
 d_m=(j_1,j_2,\dots,j_m)\in\{+1,-1\}^m.
\end{equation}
Moreover, we let $d_0=(0)$. We define operations on $d_m$:
\begin{equation}\label{define op d_m}
  \delta_+(d_m)=\sum_{j_i>0}j_i,\quad\delta_-(d_m)=-\sum_{j_i<0}j_i,
\end{equation}
and
\begin{equation}\label{define op d_m sum}
  \delta(d_m)=\sum_{i=1}^{m}j_i=\delta_+(d_m)-\delta_-(d_m).
\end{equation}
Particularly, $\delta(d_0)=0$. We also let
\begin{equation}\label{define mathbbA(d_m)}
  \mathbb{A}(d_m)=\mathbb{A}_{j_1}\mathbb{A}_{j_2}\cdots\mathbb{A}_{j_m}.
\end{equation}
In particular, $\mathbb{A}(d_0)\coloneqq \mathbb{I}$ is the identical operator. If we adopt the notation
\begin{equation}\label{d_m conj}
  d_m^*=(-j_m,-j_{m-1},\dots,-j_1),
\end{equation}
we have
\begin{equation}\label{mathbbAconj}
  \mathbb{A}^*(d_m)=(-1)^m\mathbb{A}(d_m^*).
\end{equation}
 With notations above, we have
\begin{equation}\label{B^m}
  B^m=\frac{1}{2^m}\sum_{d_m}\mathbb{A}(d_m),
\end{equation}
and by (\ref{B1}), we obtain
\begin{equation}\label{BB1}
  \begin{aligned}
   &\langle\mathcal{N}_{ex}^ne^{t\tilde{B}}f_{s_N},e^{t\tilde{B}}f_{s_N}\rangle
  =\sum_{m_1,m_2=0}^\infty\frac{(-1)^{m_2}t^{(m_1+m_2)}}{m_1!m_2!}
  \langle {B^{m_2}}^*\mathcal{N}_{ex}^nB^{m_1}f_{s_N},f_{s_N}\rangle\\
  &=\sum_{m_1,m_2=0}^\infty\frac{1}{m_1!m_2!}\Big(\frac{t}{2}\Big)^{(m_1+m_2)}
  \sum_{d_{m_1},d_{m_2}}
  \langle \mathbb{A}(d^*_{m_2})\mathcal{N}_{ex}^n\mathbb{A}(d_{m_1})f_{s_N},f_{s_N}\rangle.
  \end{aligned}
\end{equation}
By (\ref{formulae of crea and annihi ops}) and (\ref{define A tilde}), we can check that
\begin{equation}\label{BBB1}
  \begin{aligned}
  &\langle\mathcal{N}_{ex}^ne^{t\tilde{B}}f_{s_N},e^{t\tilde{B}}f_{s_N}\rangle\\
  &=\sum_{m_1,m_2=0}^\infty
  \sum_{\substack{d_{m_1},d_{m_2}\\\delta(d_{m_2}^*,d_{m_1})=0}}
  \frac{\big(2\delta(d_{m_1})\big)^n}{m_1!m_2!}
  \Big(\frac{t}{2}\Big)^{(m_1+m_2)}
  \langle \mathbb{A}(d_{m_2}^*,d_{m_1})f_{s_N},f_{s_N}\rangle\\
  &=\sum_{m=1}^\infty\sum_{\substack{m_1,m_2\\m_1+m_2=2m}}
  \sum_{\substack{d_{m_1},d_{m_2}\\\delta(d_{m_2}^*,d_{m_1})=0}}
  \frac{\big(2\delta(d_{m_1})\big)^n}{m_1!m_2!}\Big(\frac{t}{2}\Big)^{2m}
  \langle \mathbb{A}(d_{m_2}^*,d_{m_1})f_{s_N},f_{s_N}\rangle.
  \end{aligned}
\end{equation}
Using (\ref{useful bound}-\ref{xi l2preposition}) in Lemma \ref{xi l2 lemma}, together with (\ref{formulae of crea and annihi ops}) and especially, we notice that $a_{k,\sigma}^2=0$, we obtain the following rough bound for some specific constant $c_0$:
\begin{equation}\label{rough bound}
  \begin{aligned}
  &\sum_{\substack{d_{m_1},d_{m_2}\\\delta(d_{m_2}^*,d_{m_1})=0}}
  \big(2\delta(d_{m_1})\big)^n
  \vert\langle \mathbb{A}(d_{m_2}^*,d_{m_1})f_{s_N},f_{s_N}\rangle\vert\\
  &\leq(4m)^nC_{2m}^m(2m)!\Big(\sum_{\substack{\sigma,\nu}}
  \sum_{p\in B_F^\sigma,q\in B_F^\nu}\sum_k\vert\xi_{k,q,p}^{\nu,\sigma}\vert^2\Big)^m\\
  &\leq(4m)^nC_{2m}^m(2m)!c_0^{2m}N^{-\frac{m}{3}+m\varepsilon}.
  \end{aligned}
\end{equation}
Since
\begin{equation*}
  (2m)!=(2m)!!(2m-1)!!\leq (2m)!!^2\leq 2^{2m}m!^2.
\end{equation*}
We attain
\begin{equation}\label{rough bound more}
  \begin{aligned}
  \langle\mathcal{N}_{ex}^ne^{t\tilde{B}}f_{s_N},e^{t\tilde{B}}f_{s_N}\rangle
  &\leq \sum_{m=1}^\infty\sum_{\substack{m_1,m_2\\m_1+m_2=2m}}
  \big(c_0^2\vert t\vert^2N^{-\frac{1}{3}+\varepsilon}\big)^m (4m)^n\frac{(2m)!}{m_1!m_2!}\\
  &\leq \sum_{m=1}^\infty (4c_0^2N^{-\frac{1}{3}+\varepsilon}\big)^m (4m)^n
  \lesssim N^{-\frac{1}{3}+\varepsilon}\lesssim 1.
  \end{aligned}
\end{equation}
Hence we reach (\ref{claim appB}).

\bibliographystyle{abbrv}
\bibliography{ref}

\begin{thebibliography}{10}

\bibitem{workshopfermi}
Mini-workshop: {M}athematics of many-body fermionic systems.
\newblock {\em Oberwolfach Rep.}, 20(4):2809--2849, 2023.
\newblock Abstracts from the mini-workshop held October 29--November 4, 2023,
  Organized by Nikolai Leopold, Phan Th\`anh Nam and Chiara Saffirio.

\bibitem{soviet}
A.~A. Abrikosov and I.~M. Khalatnikov.
\newblock Concerning a model for a non-ideal fermi gas.
\newblock {\em Soviet Phys. JETP}.

\bibitem{meanfieldfermi2019}
N.~Benedikter, P.~T. Nam, M.~Porta, B.~Schlein, and R.~Seiringer.
\newblock Optimal upper bound for the correlation energy of a {F}ermi gas in
  the mean-field regime.
\newblock {\em Comm. Math. Phys.}, 374(3):2097--2150, 2020.

\bibitem{meanfieldfermi2021}
N.~Benedikter, P.~T. Nam, M.~Porta, B.~Schlein, and R.~Seiringer.
\newblock Correlation energy of a weakly interacting {F}ermi gas.
\newblock {\em Invent. Math.}, 225(3):885--979, 2021.

\bibitem{Seiringerfermi}
N.~Benedikter, M.~Porta, B.~Schlein, and R.~Seiringer.
\newblock Correlation energy of a weakly interacting {F}ermi gas with large
  interaction potential.
\newblock {\em Arch. Ration. Mech. Anal.}, 247(4):Paper No. 65, 57, 2023.

\bibitem{BEC2018}
C.~Boccato, C.~Brennecke, S.~Cenatiempo, and B.~Schlein.
\newblock Complete {B}ose-{E}instein condensation in the {G}ross-{P}itaevskii
  regime.
\newblock {\em Comm. Math. Phys.}, 359(3):975--1026, 2018.

\bibitem{2018Bogoliubov}
C.~Boccato, C.~Brennecke, S.~Cenatiempo, and B.~Schlein.
\newblock Bogoliubov theory in the {G}ross-{P}itaevskii limit.
\newblock {\em Acta Math.}, 222(2):219--335, 2019.

\bibitem{boccatoBrenCena2020optimal}
C.~Boccato, C.~Brennecke, S.~Cenatiempo, and B.~Schlein.
\newblock Optimal rate for {B}ose-{E}instein condensation in the
  {G}ross-{P}itaevskii regime.
\newblock {\em Comm. Math. Phys.}, 376(2):1311--1395, 2020.

\bibitem{beyondGP}
C.~Brennecke, M.~Caporaletti, and B.~Schlein.
\newblock Excitation spectrum of {B}ose gases beyond the {G}ross-{P}itaevskii
  regime.
\newblock {\em Rev. Math. Phys.}, 34(9):Paper No. 2250027, 61, 2022.

\bibitem{dynamicChenHol}
X.~Chen and J.~Holmer.
\newblock The derivation of the {$\mathbb{T}^3$} energy-critical {NLS} from
  quantum many-body dynamics.
\newblock {\em Invent. Math.}, 217(2):433--547, 2019.

\bibitem{me}
X.~Chen, J.~Wu, and Z.~Zhang.
\newblock The second order 2d behaviors of a 3d bose gases in the
  gross-pitaevskii regime, 2024.
\newblock arXiv: 2401.15540, 142pp.

\bibitem{Phanfermimeanfield}
M.~R. Christiansen, C.~Hainzl, and P.~T. Nam.
\newblock The random phase approximation for interacting {F}ermi gases in the
  mean-field regime.
\newblock {\em Forum Math. Pi}, 11:Paper No. e32, 131, 2023.

\bibitem{christiansen2024correlationenergyelectrongas}
M.~R. Christiansen, C.~Hainzl, and P.~T. Nam.
\newblock The correlation energy of the electron gas in the mean-field regime,
  2024.
\newblock arXiv: 2405.01386, 60pp.

\bibitem{Linniksergodicmethod}
J.~S. Ellenberg, P.~Michel, and A.~Venkatesh.
\newblock Linnik's ergodic method and the distribution of integer points on
  spheres.
\newblock In {\em Automorphic representations and {$L$}-functions}, volume~22
  of {\em Tata Inst. Fundam. Res. Stud. Math.}, pages 119--185. Tata Inst.
  Fund. Res., Mumbai, 2013.

\bibitem{dy2006}
L.~Erd\H{o}s, B.~Schlein, and H.-T. Yau.
\newblock Derivation of the {G}ross-{P}itaevskii hierarchy for the dynamics of
  {B}ose-{E}instein condensate.
\newblock {\em Comm. Pure Appl. Math.}, 59(12):1659--1741, 2006.

\bibitem{dy2007}
L.~Erd\H{o}s, B.~Schlein, and H.-T. Yau.
\newblock Derivation of the cubic non-linear {S}chr\"odinger equation from
  quantum dynamics of many-body systems.
\newblock {\em Invent. Math.}, 167(3):515--614, 2007.

\bibitem{dy2009}
L.~Erd\H{o}s, B.~Schlein, and H.-T. Yau.
\newblock Rigorous derivation of the {G}ross-{P}itaevskii equation with a large
  interaction potential.
\newblock {\em J. Amer. Math. Soc.}, 22(4):1099--1156, 2009.

\bibitem{dy2010}
L.~Erd\H{o}s, B.~Schlein, and H.-T. Yau.
\newblock Derivation of the {G}ross-{P}itaevskii equation for the dynamics of
  {B}ose-{E}instein condensate.
\newblock {\em Ann. of Math. (2)}, 172(1):291--370, 2010.

\bibitem{dilutefermiBog}
M.~Falconi, E.~L. Giacomelli, C.~Hainzl, and M.~Porta.
\newblock The dilute {F}ermi gas via {B}ogoliubov theory.
\newblock {\em Ann. Henri Poincar\'e}, 22(7):2283--2353, 2021.

\bibitem{bose2ndthermo1}
S.~r. Fournais and J.~P. Solovej.
\newblock The energy of dilute {B}ose gases.
\newblock {\em Ann. of Math. (2)}, 192(3):893--976, 2020.

\bibitem{bose2ndthermo2}
S.~r. Fournais and J.~P. Solovej.
\newblock The energy of dilute {B}ose gases {II}: the general case.
\newblock {\em Invent. Math.}, 232(2):863--994, 2023.

\bibitem{fermiupper}
E.~L. Giacomelli.
\newblock An optimal upper bound for the dilute {F}ermi gas in three
  dimensions.
\newblock {\em J. Funct. Anal.}, 285(8):Paper No. 110073, 73, 2023.

\bibitem{giacomelli2024}
E.~L. Giacomelli.
\newblock An optimal lower bound for the low density fermi gas in three
  dimensions, 2024.
\newblock arXiv: 2410.08904, 38pp.

\bibitem{2024huangyangformulalowdensityfermi}
E.~L. Giacomelli, C.~Hainzl, P.~T. Nam, and R.~Seiringer.
\newblock The huang-yang formula for the low-density fermi gas: upper bound,
  2024.
\newblock arXiv: 2409.17914, 48pp.

\bibitem{haberberger2023}
F.~Haberberger, C.~Hainzl, P.~T. Nam, R.~Seiringer, and A.~Triay.
\newblock The free energy of dilute bose gases at low temperatures, 2024.
\newblock arXiv: 2304.02405, 68pp.

\bibitem{haberberger2024}
F.~Haberberger, C.~Hainzl, B.~Schlein, and A.~Triay.
\newblock Upper bound for the free energy of dilute bose gases at low
  temperature, 2024.
\newblock arXiv: 2405.03378, 54pp.

\bibitem{bogsimplified}
C.~Hainzl, B.~Schlein, and A.~Triay.
\newblock Bogoliubov theory in the {G}ross-{P}itaevskii limit: a simplified
  approach.
\newblock {\em Forum Math. Sigma}, 10:Paper No. e90, 39, 2022.

\bibitem{HeathBrown+1999+883+892}
D.~Heath-Brown.
\newblock {\em Lattice points in the sphere}, pages 883--892.
\newblock De Gruyter, Berlin, Boston, 1999.

\bibitem{huangyang}
K.~Huang and C.~N. Yang.
\newblock Quantum-mechanical many-body problem with hard-sphere interaction.
\newblock {\em Phys. Rev.}, 105:767--775, Feb 1957.

\bibitem{polarizedupper}
A.~r. B.~k. Lauritsen and R.~Seiringer.
\newblock Ground state energy of the dilute spin-polarized {F}ermi gas: upper
  bound via cluster expansion.
\newblock {\em J. Funct. Anal.}, 286(7):Paper No. 110320, 98, 2024.

\bibitem{polarizedlower}
A.~r. B.~k. Lauritsen and R.~Seiringer.
\newblock Pressure of a dilute spin-polarized {F}ermi gas: lower bound.
\newblock {\em Forum Math. Sigma}, 12:Paper No. e78, 37, 2024.

\bibitem{LHY106}
T.~D. Lee, K.~Huang, and C.~N. Yang.
\newblock Eigenvalues and eigenfunctions of a bose system of hard spheres and
  its low-temperature properties.
\newblock {\em Phys. Rev.}, 106:1135--1145, Jun 1957.

\bibitem{TDLEECNYANG112}
T.~D. Lee and C.~N. Yang.
\newblock Low-temperature behavior of a dilute bose system of hard spheres. i.
  equilibrium properties.
\newblock {\em Phys. Rev.}, 112:1419--1429, Dec 1958.

\bibitem{spectrum}
M.~Lewin, P.~T. Nam, S.~Serfaty, and J.~P. Solovej.
\newblock Bogoliubov spectrum of interacting {B}ose gases.
\newblock {\em Comm. Pure Appl. Math.}, 68(3):413--471, 2015.

\bibitem{BEC2002}
E.~H. Lieb and R.~Seiringer.
\newblock Proof of {B}ose-{E}instein condensation for dilute trapped gases.
\newblock {\em Phys. Rev. Lett.}, 88:170409, Apr 2002.

\bibitem{2005fermiphy}
E.~H. Lieb, R.~Seiringer, and J.~P. Solovej.
\newblock Ground-state energy of the low-density fermi gas.
\newblock {\em Phys. Rev. A}, 71:053605, May 2005.

\bibitem{lieb2005mathematics}
E.~H. Lieb, R.~Seiringer, J.~P. Solovej, and J.~Yngvason.
\newblock {\em {The mathematics of the Bose gas and its condensation}},
  volume~34.
\newblock Springer Science \& Business Media, 2005.

\bibitem{meanfield2011seiR}
R.~Seiringer.
\newblock The excitation spectrum for weakly interacting bosons.
\newblock {\em Comm. Math. Phys.}, 306(2):565--578, 2011.

\end{thebibliography}

\end{document}